%% file: thesis.tex
\pdfoutput=1 

\documentclass{scrbook}

\newcommand*{\Title}{Distributed Automata and Logic}
\newcommand*{\Author}{Fabian Reiter}
\newcommand*{\Subject}{PhD thesis in Theoretical Computer Science}
\newcommand*{\Keywords}{Automata,
                        Distributed algorithms,
                        Modal logic,
                        Monadic second-order logic,
                        Graphs}
\newcommand*{\Email}{fabian.reiter@gmail.com}

\usepackage{sty/main}        
\usepackage{sty/utilities}   
\usepackage{sty/theorems}    
\usepackage{sty/notations}   
\usepackage{sty/modallogic}  
\usepackage{sty/terminology} 
\usepackage{sty/colors}      
\usepackage{sty/unicode}     
\usepackage{sty/extraLDA}    

\begin{document}

\pagenumbering{Alph}
\input{tex/title.tex}
\maketitle

\frontmatter

\pdfbookmark[section]{\contentsname}{toc} 
\tableofcontents

\include{tex/abstract}
\include{tex/acknowledgments}

\mainmatter

\include{tex/introduction}
\include{tex/preliminaries}
\include{tex/local}
\include{tex/nonlocal}
\include{tex/emptiness}
\include{tex/alternation}
\include{tex/perspectives}

\backmatter

\bibliography{thesis.bib}
\printindex

\end{document}

%% file: tex/title.tex
\KOMAoptions{titlepage=firstiscover}
\newlength{\boffset}
\IfEq{\version}{print}{
  \setlength{\boffset}{12mm}
}{
  \setlength{\boffset}{0mm}
}
\renewcommand*{\coverpagetopmargin}{15mm}
\renewcommand*{\coverpagebottommargin}{15mm}
\renewcommand*{\coverpageleftmargin}{\dimexpr16mm+\boffset\relax}
\renewcommand*{\coverpagerightmargin}{16mm}

\setkomafont{subject}{\large\itshape}
\setkomafont{publishers}{\large}

\extratitle{%
  \centering%
  \Large%
  {%
    \fontsize{75}{0}\selectfont \color{black!11}%
    $\displaystyle \TransFunc \,\colon\, \StateSet \times 2^\StateSet \,\to\, \StateSet$%
  }%
  \\[5.5\baselineskip]%
  {%
    \usekomafont{subject}%
    \Subject%
  }%
  \\[2\baselineskip]%
  {%
    \usekomafont{title}\huge\color{chaptergrey}%
    \Title%
  }
  \\[2\baselineskip]%
  {%
    \usekomafont{author}%
    \Author%
  }%
  \\[5\baselineskip]%
    { \large%
      \osf{Dissertation defense on \printdate{2017-12-12}.}%
    }%
  \\[1\baselineskip]%
  { \large%
    \renewcommand{\arraystretch}{1.8}%
    \begin{tabular}{r>{\itshape}l}
      Olivier Carton    & supervisor \\ 
      Bruno Courcelle   & examiner   \\ 
      Pierre Fraigniaud & examiner   \\ 
      Nicolas Ollinger  & examiner   \\ 
      Jukka Suomela     & reviewer   \\ 
      Christine Tasson  & examiner   \\ 
      Wolfgang Thomas   & reviewer      
    \end{tabular}%
  } \\%
  \vfill%
  \includegraphics[width=8ex]{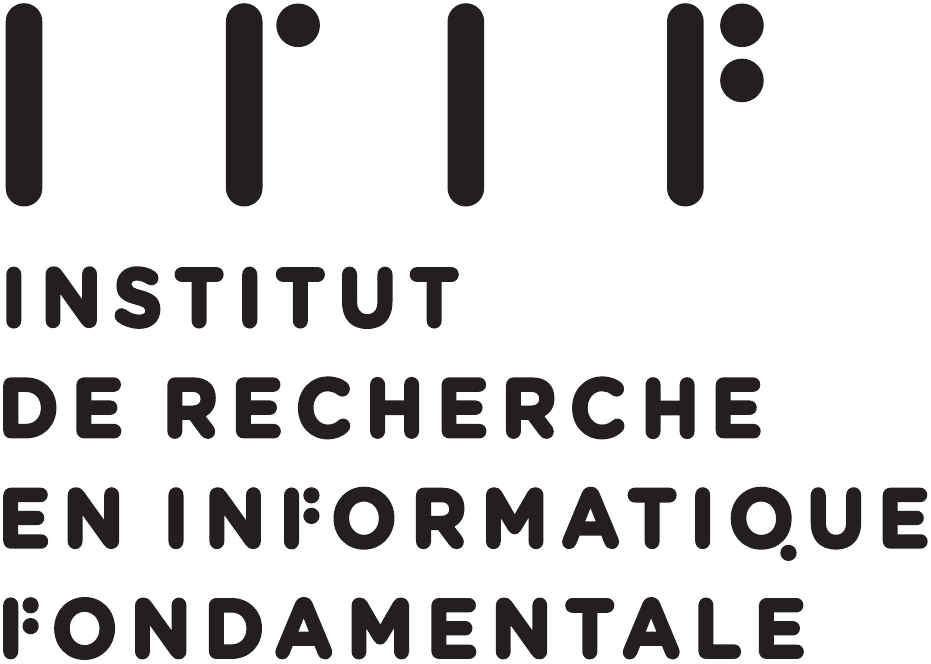}%
  \IfEq{\version}{print}{%
    \hspace{0.17\textwidth}%
  }{%
    \hspace{0.19\textwidth}%
  }%
  \parbox[b]{0.3\textwidth}{%
    \footnotesize\sffamily%
    {\scshape\lsstyle École Doctorale 386} \\%
    {\scriptsize Sciences Mathématiques de Paris Centre} \\[-1.5ex]%
  }%
  \hfill
  \includegraphics[width=20ex]{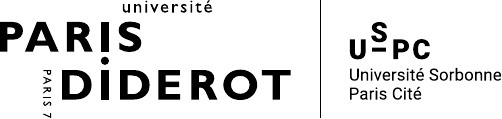}%
}

\subject{\Subject}

\title{\Title}

\author{\Author}

\date{}

\publishers{{%
    \osf{Paris,\, 2017}%
}}

\uppertitleback{%
  Last revision: \osf{\printdate{2017-11-24}} \\[\baselineskip]
  Contact address:\, %
  \href{mailto: \Author <\Email>}{\nolinkurl{\Email}}%
}

\lowertitleback{%
  \href{http://creativecommons.org/licenses/by/4.0/}%
       {\includegraphics[height=\baselineskip+\heightof{A}]{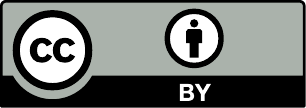}}%
  \quad%
  \parbox[b]{0.7\textwidth}{
    \textcopyright\ \Author.\, %
    This work is licensed under the \\ %
    \href{http://creativecommons.org/licenses/by/4.0/}%
         {Creative Commons Attribution 4.0 International License}.%
  }%
}

%% file: tex/abstract.tex
\chapter{Abstract}

\kl{Distributed automata} are finite-state machines
that operate on finite \kl{directed graphs}.
Acting as synchronous distributed algorithms,
they use their input \kl[digraph]{graph} as a \mbox{network}
in which identical processors communicate
for a possibly infinite number of synchronous rounds.
For the \emph{\kl{local}} variant
of those \kl[distributed automata]{automata},
where the number of rounds is bounded by a constant,
Hella~et~al.\ (\osf{2012,~2015})
have established a \kl[formula]{logical} characterization
in terms of basic \kl{modal logic}.
In this thesis,
we provide similar \kl[formula]{logical} characterizations
for two more expressive classes of \kl{distributed automata}.

The first class extends \kl{local automata}
with a \kl{global acceptance condition}
and the ability to \kl[\ALDAg]{alternate} between
nondeterministic and parallel \kl[alternating run]{computations}.
We show that it is \kl[device equivalent]{equivalent} to
\kl{monadic second-order logic} on \kl[digraphs]{graphs}.
By restricting \kl[global transition function]{transitions} to be
\kl[\NLDAg]{nondeterministic} or \kl[\DLDAg]{deterministic},
we also obtain two strictly weaker variants
for which the emptiness problem is decidable.

Our second class transfers the standard notion
of \kl[asynchronous automaton]{asynchronous algorithm}
to the setting of \kl{nonlocal distributed automata}.
The resulting machines
are shown to be \kl[device equivalent]{equivalent}
to a small fragment of least fixpoint logic,
and more specifically, to a
\kl[backward $\mu$-fragment]{restricted variant of the modal $\mu$-calculus}
that allows \kl{least fixpoints} but forbids greatest fixpoints.
Exploiting the connection with \kl[backward $\mu$-fragment]{logic},
we additionally prove that the expressive power
of those \kl{asynchronous automata}
is independent of whether or not messages can be~lost.

We then investigate the decidability of the
\kl{emptiness problem} for several classes
of \kl{nonlocal automata}.
We show that the problem is undecidable in general,
by simulating a Turing machine
with a \kl{distributed automaton}
that exchanges the roles of space and time.
On the other hand,
the problem is found to be decidable in $\LOGSPACE$
for a class of \kl{forgetful automata},
where the \kl{nodes} see the messages
received from their \kl[incoming neighbors]{neighbors}
but cannot remember their own \kl{state}.

As a minor contribution,
we also give new proofs of the strictness
of several \kl{set quantifier} alternation hierarchies
that are based on \kl{modal logic}.

\paragraph{Keywords.}

Automata,
Distributed algorithms,
Modal logic,
Monadic second-order logic,
Graphs.

\chapter{Résumé}

Les automates distribués sont des machines à états finis
qui opèrent sur des graphes orientés finis.
Fonctionnant en tant qu'algorithmes distribués synchrones,
ils utilisent leur graphe d'entrée comme un réseau
dans lequel des processeurs identiques communiquent entre eux
pendant un certain nombre (éventuellement infini) de rondes synchrones.
Pour la variante \emph{locale}
de ces automates,
où le nombre de rondes est borné par une constante,
Hella~et~al.\ (\osf{2012,~2015})
ont établi une caractérisation logique
par des formules de la logique modale de base.
Dans le cadre de cette thèse,
nous présentons des caractérisations logiques similaires
pour deux classes d'automates distribués plus expressives.

La première classe étend les automates locaux
avec une condition d'acceptation globale
et la capacité d'alterner entre
des modes de calcul non-déterministe et parallèle.
Nous montrons qu'elle est équivalente à
la logique monadique du second ordre sur les graphes.
En nous restreignant à des transitions
non-déterministes ou déterministes,
nous obtenons également deux variantes d'automates strictement plus faibles
pour lesquelles le problème du vide est décidable.

Notre seconde classe adapte la notion standard
d'algorithme asynchrone
au cadre des automates distribués non-locaux.
Les machines résultantes
sont prouvées équivalentes
à un petit fragment de la logique de point fixe,
et plus précisément, à une
variante restreinte du $\mu$-calcul modal
qui autorise les plus petits points fixes
mais interdit les plus grands points fixes.
Profitant du lien avec la logique,
nous montrons aussi que la puissance expressive
de ces automates asynchrones
est indépendante du fait que des messages puissent être perdus ou non.

Nous étudions ensuite la décidabilité du
problème du vide pour plusieurs classes
d'automates non-locaux.
Nous montrons que le problème est indécidable en général,
en simulant une machine de Turing
par un automate distribué
qui échange les rôles de l'espace et du temps.
En revanche,
le problème s'avère décidable en $\LOGSPACE$
pour une classe d'automates oublieux,
où les nœuds voient les messages
reçus de leurs voisins,
mais ne se souviennent pas de leur propre état.

Finalement, à titre de contribution mineure,
nous donnons également de nouvelles preuves de séparation
pour plusieurs hiérarchies d'alternance de quantificateurs
basées sur la logique modale.

\paragraph{Mots-clés.}

Automates,
Algorithmes distribués,
Logique modale,
Logique monadi\-que du second ordre,
Graphes.

%% file: tex/acknowledgments.tex
\chapter{Acknowledgments}

First and foremost,
I would like to thank my advisor, Olivier Carton,
for his continuous support during the past three years.
This included
finding a scholarship for me,
spending countless hours with me in front of a whiteboard,
as well as helping me in the writing
of several papers and this thesis.
I~am especially grateful
for the immense freedom he granted me
throughout the entire period,
letting me pursue ideas of my own,
but at the same time always being available for discussion.
He provided guidance whenever I needed it,
but never exerted pressure,
gave very good advice,
but always let me decide for myself.
In my opinion,
this is exactly how a doctoral thesis should be supervised,
but it can by no means be taken for granted.
I~therefore consider myself very fortunate
to have had Olivier as my advisor.

My sincere thanks extend to
Bruno Courcelle,
Pierre Fraigniaud,
Nicolas Ollinger,
Jukka Suomela,
Christine Tasson, and
Wolfgang Thomas,
who kindly accepted to be part of my thesis committee.

Jukka Suomela and Wolfgang Thomas
did me the great honor of reviewing
a preliminary version of the manuscript.
They gave detailed and extremely flattering feedback,
and both made an important remark
that led to a major improvement of this document:
I had failed to include a discussion
of perspectives for future research.
This shortcoming has now been addressed
by the addition of \cref{ch:perspectives}.
Furthermore,
Wolfgang compiled a very helpful collection of suggestions,
which I~have tried to incorporate into the present version.

Bruno Courcelle graciously gave his time
to read my master's thesis in \osf{2014},
although he had never heard of me before.
He then informed Géraud Sénizergues,
who most kindly invited me
to the final conference of the \textsc{frec} project in Marseille.
This opened the door for me
into the French community of automata theory,
especially since I met my future doctoral advisor
at that conference.
Thus,
it is indirectly through Bruno that I came to Paris.

There, at \textsc{irif} (formerly \textsc{liafa}),
Pierre Fraigniaud showed a very kind interest in my work
and opened another door for me,
this time into the community of distributed computing.
He did so by referring to my first paper
in several of his own collaborations
and by providing various opportunities for me to meet his colleagues,
in particular at two international workshops in Bertinoro and Oaxaca.
The latter was made possible
through a joint effort of Pierre and Sergio Rajsbaum.

In addition to the committee,
I am grateful
to Fabian Kuhn and Andreas Podelski,
who supervised my master's thesis
(the starting point for the present thesis),
to Antti Kuusisto,
who collaborated with me on the work in \cref{ch:emptiness},
to Laurent Feuilloley,
who proofread and corrected
the overview of distributed decision
in \cref{ssec:related-distributed-computing},
to Charles Paperman,
who was the first to tell me
that I was unknowingly working with some kind of \kl{modal logic},
to Nicolas Bacquey,
who informed me that
exchanging space and time is a common technique
in cellular automata theory,
and to Thomas Colcombet,
who developed the
\href{https://www.ctan.org/pkg/knowledge}{\textsf{knowledge}}
package
and encouraged me to use it~here.

\begin{center}
  \adfflourishleftdouble
\end{center}

We have just crossed the dividing line,
where I stop mentioning people by name.
This may seem hasty, or even harsh,
but there are two good reasons.
First,
I value my privacy very much
and do not want to share personal details in a document
that will be publicly available on the Internet.
Second,
the larger the circle of people I~include,
the greater the risk of forgetting someone.
A simple rule that helps to avoid both of these issues
is to mention only people
who stand in some direct professional relation to the thesis.
Nevertheless,
many more have helped me over the years
and had a tremendous influence
on my life and work.
I therefore sincerely hope
not to offend anyone
by expressing my gratitude
in the following simplistic manner:

\begin{center}
  \adfflatleafleft\quad
  \emph{Many thanks to my colleagues, friends, and family\:\!!}\!
  \quad\adfflatleafright
\end{center}

%% file: tex/introduction.tex
\chapter{Introduction}
\label{ch:introduction}

The present thesis aims to contribute
to the recently initiated development of
a \emph{descriptive complexity} theory
for \emph{distributed computing}.

\paragraph{What does this mean?}

Descriptive complexity~\cite{DBLP:books/daglib/Immerman99}
basically compares the expressive powers of
certain classes of \emph{algorithms}, or abstract machines,
\marginnote{We identify algorithms with abstract machines.}
with those of
certain classes of logical \emph{\kl{formulas}}.
The Holy Grail, so to speak,
is to establish \kl[device equivalences]{equivalences} of the form:
\begin{quote}
  “Algorithm class~$\AlgorithmSet$
  has exactly the same power as
  \kl{formula} class~$\FormulaSet$.”
\end{quote}
Probably the most famous result in this area
is Fagin's theorem from \osf{1974}~\cite{Fagin74},
which roughly states that
a \kl[digraph language]{graph property} can be recognized
by a nondeterministic Turing machine in polynomial time
if and only if
it can be \kl{defined}
by a \kl{formula} of existential second-order logic.
The theorem thereby provides
a logical characterization of the complexity class~$\NPTIME$.

Distributed computing~\cite{DBLP:books/mk/Lynch96,Peleg00},
on the other hand,
studies networks composed of several interconnected processors
that share a common goal.
The processors communicate with each other by passing messages
along the links of the network
in order to collectively solve some computational problem.
In many cases,
this is a \kl[digraph]{graph} problem,
where the considered problem instance
is precisely the \kl[digraph]{graph} defined by the network itself.
All processors run the same algorithm concurrently,
and often make no prior assumptions
about the size and topology of the \kl[digraph]{graph}.
Typical problems
that can be solved by such \emph{distributed algorithms}
include
\kl[coloring]{graph coloring}, leader election,
and the construction of spanning trees and maximal independent~sets.

Now,
the ultimate objective that motivates this thesis
is to develop an extension of descriptive complexity
for the classes of algorithms considered in distributed computing.
This means
that we seek to establish \kl[device equivalences]{equivalences} of the form:
\begin{quote}
  “Distributed algorithm class~$\AlgorithmSet$
  has the same power as
  \kl{formula} class~$\FormulaSet$.”
\end{quote}
\marginnote{Distributed algorithms \\ are abstract machines \\ that communicate.}
However,
such a statement can only be substantial
if we have a precise definition of class~$\AlgorithmSet$.
Therefore,
we will formally represent distributed algorithms
as abstract machines,
instead of the more common, but informal, representations in pseudocode.

\paragraph{Why is this interesting?}

First and foremost,
a descriptive complexity theory for distributed computing
would offer the same benefits as
its classical counterpart does for sequential computing:
\begin{enumerate}[wide,labelwidth=!,labelindent=0ex,nosep] 
\item \label{itm:naturalness}
  If distributed algorithm class~$\AlgorithmSet$
  turns out to be \kl[device equivalent]{equivalent}
  to \kl{formula} class~$\FormulaSet$,
  then this provides strong evidence
  for the naturalness of both classes.
  Indeed,
  the definition of any mathematical device may,
  by itself, seem arbitrary.
  \emph{Why should distributed machines communicate precisely that way?}
  \emph{Why should logical \kl{formulas} contain precisely those components?}
  But if two devices,
  that appear rather different on the surface,
  turn out to be descriptions of the exact same thing,
  then this is unlikely to be pure coincidence.
\item \label{itm:connecting-fields}
  Connecting two seemingly unrelated fields
  --~here, distributed computing and logic~--
  can provide new insights into both fields.
  Some proofs might be easier to perform
  if one adopts the point of view of one setting
  rather than the other.
  \marginnote{
    Formal logic dates back
    to the \osf{mid-19th} century, \\
    while distributed computing started in the \osf{1970s/1980s}.
  }
  Furthermore,
  some open questions in one field
  might already have well-known answers in the other.
  Especially the field of distributed computing
  could benefit from this,
  as it is more than a century younger than formal logic,
  and therefore has had less time to evolve.
\end{enumerate}

Second,
distributed computing also brings an interesting new perspective
to the field of descriptive complexity itself:
\begin{enumerate}[resume*]
\item Distributed algorithms can be evaluated
  on the same input as logical \kl{formulas},
  without any need for encoding that input.
  More precisely,
  the network in which a distributed algorithm is executed
  may be considered identical to
  the \kl{structure}
  on which the truth of a corresponding \kl{formula} is evaluated.
  This stands in sharp contrast
  to classical descriptive complexity theory.
  For instance,
  in the case of Fagin's theorem,
  the input of a Turing machine is
  a binary string that encodes a finite \kl[digraph]{graph}
  in form of an adjacency matrix.
  Hence,
  the equivalence of
  nondeterministic polynomial-time Turing machines
  and
  existential second-order logic
  is actually stated
  with respect to such an encoding.
\end{enumerate}

\section{Background}

Let us now take a step back and put the subject into context.
We start with a brief summary of some classical results in automata theory,
and then turn to more recent developments in distributed computing.

\subsection{Related work in automata theory}
\label{ssec:related-automata-theory}

Although the field of descriptive complexity theory
really started with Fagin's theorem in the \osf{1970s},
the idea of characterizing abstract machines through logical \kl{formulas}
had already appeared earlier in automata theory.
In the early \osf{1960s},
Büchi~\cite{Buchi60},
Elgot~\cite{Elgot61}
and Trakhtenbrot~\cite{Trakhtenbrot61}
discovered independently of each other
that the regular languages,
which are recognized by \emph{finite automata}
on \kl[pointed dipaths]{words},
are precisely
the languages \kl{definable} in \emph{\kl{monadic second-order logic}},
or~$\MSOL$
(see, e.g., \cite[Thm~3.1]{Thomas97b}).
The latter is an extension of \kl{first-order logic},
which in addition to allowing quantification
over elements of a given \kl{domain}
(e.g., positions in a \kl[pointed dipath]{word}),
also allows to quantify over sets of such elements.
Along with several other equivalent characterizations,
in particular through
regular expressions~\cite{Kleene56},
regular grammars~\cite{DBLP:journals/tit/Chomsky56},
and finite monoids~\cite{Nerode58},
the \kl[device equivalence]{equivalence} between automata and logic
helped to legitimize regularity
as a highly natural concept in formal language theory
(cf.\ \cref{itm:naturalness}, above).
Furthermore,
it proved that
the satisfiability and validity problems for $\MSOL$
on \kl[pointed dipaths]{words} are decidable,
because so are the corresponding problems for finite automata.
In this way,
the field of logic directly benefited
from the connection with automata theory
(cf.\ \cref{itm:connecting-fields}).
Nowadays,
such connections also play a central role in model checking,
where one needs to decide whether a system,
represented by an automaton,
satisfies a given specification,
expressed as a logical \kl{formula}.

About a decade later,
the result was generalized from \kl[pointed dipaths]{words}
to \kl{labeled} \kl[pointed ordered ditrees]{trees}
by Thatcher and Wright~\cite{DBLP:journals/mst/ThatcherW68}
and Doner~\cite{DBLP:journals/jcss/Doner70}
(see, e.g., \cite[Thm~3.6]{Thomas97b}).
The corresponding \kl[pointed ordered ditree]{tree} automata
can be seen as a canonical extension of finite automata
to \kl[pointed ordered ditrees]{trees};
as far as $\MSOL$ is concerned,
the generalization to \kl[pointed ordered ditrees]{trees}
is even more straightforward,
since both \kl[pointed dipaths]{words}
and \kl[pointed ordered ditrees]{trees} are merely special cases
of the relational \kl{structures}
on which logical \kl{formulas} are usually evaluated.
The other characterizations of regular languages
can also be generalized from \kl[pointed dipaths]{words}
to \kl[pointed ordered ditrees]{trees} in a natural manner,
and quite remarkably,
they all remain equivalent on \kl[pointed ordered ditrees]{trees}
(see, e.g., \cite{TATA2008}).
Hence,
the notion of regularity extends directly
to \kl[pointed ordered ditree]{tree} \kl{languages}.
Moreover,
similar \kl[device equivalences]{equivalences} have been established
for several other generalizations of \kl[pointed dipaths]{words},
such as nested words (see~\cite{DBLP:journals/jacm/AlurM09})
and Mazurkiewicz traces (see, e.g., \cite{DiekertM97}).

In contrast,
the situation becomes far more complicated
if we expand our field of interest to arbitrary finite \kl[digraphs]{graphs}
(possibly with \kl{node} \kl{labels} and multiple \kl{edge relations}).
Although some of the characterizations mentioned above
can be generalized to \kl[digraphs]{graphs} in a meaningful way,
they are, in general, no longer equivalent.
The logical approach
is certainly the easiest to generalize,
since \kl[digraphs]{graphs} are
yet another special case of relational \kl{structures}.
While on \kl[pointed dipaths]{words} and \kl[pointed ordered ditrees]{trees}
the existential fragment of $\MSOL$ ($\EMSOL$)
is already sufficient to characterize regularity,
it is strictly less expressive than full $\MSOL$ on \kl[digraphs]{graphs},
\marginnote{
  Fagin's result was later extended
  by Matz, \\
  Schweikardt and Thomas
  to yield a complete \\
  separation
  of the $\MSOL$ \kl[set quantifier]{quantifier} alternation hierarchy
  \textnormal{(\textit{see}~\cite{DBLP:journals/iandc/MatzST02})}.
}
as has been shown by Fagin in~\cite{DBLP:journals/mlq/Fagin75}.
Similarly,
the algebraic approach (based on monoids)
has been adapted to \kl[digraphs]{graphs} by \mbox{Courcelle}
in~\cite{DBLP:journals/iandc/Courcelle90},
and it turns out that $\MSOL$ is strictly less powerful
than his notion of recognizability.
(The latter is defined in terms of homomorphisms
into many-sorted algebras that are finite in each sort.)
A common pattern that emerges from such results
is that the different characterizations of regularity drift apart
as the complexity of the considered \kl{structures} increases.
In this sense,
regularity cannot be considered
a natural --~or even well-defined~-- property
of \kl[digraph languages]{graph languages}.

To complicate matters even further,
the automata-theoretic characterization
which is instrumental in the theory of \kl[pointed dipath]{word}
and \kl[pointed ordered ditree]{tree} \kl{languages},
does not seem to have a natural counterpart on \kl[digraphs]{graphs}.
A \kl[pointed dipath]{word} or \kl[pointed ordered ditree]{tree} automaton
can scan its entire input
in a single canonical traversal,
which is completely determined by the structure of the input
(i.e., left-to-right, for \kl[pointed dipaths]{words},
and bottom-up, for \kl[pointed ordered ditrees]{trees}).
On arbitrary \kl[digraphs]{graphs}, however,
there is no sense of a global direction
that the automaton could follow,
especially since we do not even require \kl{connectivity} or acyclicity.
This is one of the reasons
why much research in the area of \kl[digraph languages]{graph languages}
has focused on $\MSOL$.
In the words of Courcelle and Engelfriet~\cite[p.~3]{CourcelleE12}:
\begin{quote}
  \textit{
    \dots
    \kl{monadic second-order logic} can be viewed as playing
    the role of \\
    “finite automata on \kl[digraphs]{graphs}”
    \dots
  }
\end{quote}

Another approach,
investigated by Thomas
in~\cite{DBLP:conf/icalp/Thomas91} and~\cite{DBLP:conf/tapsoft/Thomas97},
is to nondeterministically assign a state of the automaton
to each \kl{node} of the \kl[digraph]{graph},
and then check that this assignment satisfies
certain local “transition” conditions for each \kl{node}
(specified with respect to \kl[neighbor]{neighboring} \kl{nodes} within a fixed radius)
as well as certain global occurrence conditions
at the level of the entire \kl[digraph]{graph}.
The \emph{\kl[digraph]{graph} acceptors} devised by Thomas
turn out to be equivalent to $\EMSOL$ on \kl[digraphs]{graphs} of bounded degree.
Following up on this idea in~\cite{DBLP:journals/dmtcs/SchwentickB99},
Schwentick and Barthelmann have also suggested a more general model,
which remains very close to a normal form of $\EMSOL$,
but overcomes the constraint of boundedness on the degree.
Both of these \kl[digraph]{graph} automaton models are legitimate generalizations
of classical finite automata,
in the sense that they are equivalent to them
and can easily simulate them if we restrict the input
to (\kl[digraphs]{graphs} representing) \kl[pointed dipaths]{words}
or \kl[pointed ordered ditrees]{trees}.
However,
on arbitrary \kl[digraphs]{graphs}, they are less well-behaved,
which is a direct consequence of their \kl[device equivalence]{equivalence} with $\EMSOL$.
In particular,
they do not satisfy closure under complementation,
and their emptiness problem is undecidable.
It is worth noting that
both models are somewhat similar to
the local distributed algorithms considered in the next
\lcnamecref{ssec:related-distributed-computing},
insofar as they take into account the local view
that each \kl{node} has of its fixed-radius neighborhood.
This connection has already been recognized and exploited
by Göös and Suomela
in~\cite{DBLP:conf/podc/GoosS11,DBLP:journals/toc/GoosS16};
we will mention it again below.

\subsection{Related work in distributed computing}
\label{ssec:related-distributed-computing}

Rather surprisingly,
the idea of extending descriptive complexity theory
to the setting of distributed computing
seems to be relatively new.
The first research in that direction
(of which the author is aware)
started in the early \osf{2010s}
as a collaboration between
the Finnish communities of logic and distributed algorithms.

In~\cite{DBLP:conf/podc/HellaJKLLLSV12,DBLP:journals/dc/HellaJKLLLSV15},
Hella~et~al.\ have presented a systematic study
of several models of distributed computing
that impose restrictions of varying degrees
on the communication between the \kl{nodes} of a network.
Their most permissive model
corresponds to the well-established port-numbering model,
where every \kl{node} has
a separate communication channel with each of its \kl{neighbors}
and is guaranteed that
the messages sent and received through that channel
relate consistently to the same \kl{neighbor};
the network is anonymous
in the sense that \kl{nodes} are not equipped with unique identifiers.
In the nomenclature
of~\cite{DBLP:conf/podc/HellaJKLLLSV12,DBLP:journals/dc/HellaJKLLLSV15},
the class of \kl[digraph]{graph} problems
solvable in this model
by deterministic \mbox{\emph{synchronous}} algorithms
is denoted by~$\VVc$.
\marginnote{
  The classes of Hella~et~al.:
  \normalfont
  \begin{tabular}{@{\,}l@{\;}r@{\;\;}l@{\,}}
    \toprule
           & \emph{Incoming} & \emph{Outgoing} \\
    \midrule
    $\VVc$ & vector          & vector \\
    $\VV$  & \threeemdash    & \threeemdash \\
    $\MV$  & multiset        & \threeemdash \\
    $\SV$  & set             & \threeemdash \\
    $\VB$  & vector          & singleton \\
    $\MB$  & multiset        & \threeemdash \\
    $\SB$  & set             & \threeemdash \\
    \bottomrule
  \end{tabular}
}
Here,
“synchronous” means that
all \kl{nodes} of the network share a global clock,
thereby allowing the computation
to proceed in an infinite sequence of rounds;
in each round,
all the \kl{nodes} simultaneously exchange messages with their \kl{neighbors},
and then update their local state
based on the newly obtained information.
Next,
by dropping the channel-consistency guarantee,
one obtains the class~$\VV$,
where in each round,
every \kl{node} sees a vector consisting of
all the incoming messages received from its \kl{neighbors},
and generates a vector of outgoing messages
that are sent to the \kl{neighbors};
the difference with~$\VVc$ is that
the two vectors are not necessarily sorted in the same order,
so the \kl{node} cannot assume that
the \kl{neighbor} who sends the $i$-th incoming message
is the same who receives the $i$-th outgoing message.
(However, the sorting orders do not change throughout the rounds.)
Communication is further restricted in the classes~$\MV$ and~$\SV$,
where the vector of incoming messages is replaced
by a multiset and a set, respectively.
In the former case,
a \kl{node} cannot identify the senders of its incoming messages,
whereas in the latter,
it cannot even distinguish between several identical messages.
Similarly,
the classes $\VB$, $\MB$, and \Intro*{$\SB$} are characterized
by the fact that the outgoing vector is replaced by a singleton,
meaning that
a \kl{node} must broadcast the same message to all of its \kl{neighbors}.

The main result
of~\cite{DBLP:conf/podc/HellaJKLLLSV12,DBLP:journals/dc/HellaJKLLLSV15}
is that the preceding classes satisfy the linear order
\begin{equation*}
  \SB \subsetneqq \MB = \VB \subsetneqq \SV = \MV = \VV \subsetneqq \VVc.
\end{equation*}
The same order holds for the so-called \emph{local}
(or constant-time) versions of these classes,
which contain only those \kl[digraph]{graph} problems
that can be solved in a constant number of communication rounds,
regardless of the size of the network.
(For a relatively recent survey of local algorithms,
see~\cite{DBLP:journals/csur/Suomela13}.)

Most relevant for the present thesis,
the same paper also establishes a very natural correspondence
between these local classes
and several variants of \kl{modal logic}.
In particular,
a \kl[digraph language]{graph property} lies in~$\SB[1]$, the local version of~$\SB$,
if and only if
it can be \kl{defined} by a \kl{formula} of \emph{\kl{backward modal logic}}.
\marginnote{
  Using the backward version of standard \kl{modal logic}
  is merely a presentational choice,
  motivated by the intuition
  that the messages of a distributed algorithm
  should flow in the same direction
  as the network links on which they travel.
  The presentation
  in~\textnormal{\cite{DBLP:conf/podc/HellaJKLLLSV12,DBLP:journals/dc/HellaJKLLLSV15}}
  is a bit different.
}
Just like a distributed algorithm,
such a \kl{formula} is evaluated
from the local point of view of a particular \kl{node} in the input \kl[digraph]{graph}.
In order to make a statement about the \kl{incoming neighborhood} of that \kl{node},
\kl{backward modal logic} allows to move the current point of evaluation
to one of the \kl{incoming neighbors}
by means of a special operator, called \emph{\kl{backward modality}}.
The key insight of Hella~et~al.\ is
that the nesting depth of these \kl{modalities} corresponds precisely to
the running time of the local algorithms that solve problems in~$\SB[1]$.
With this idea in mind,
it is possible to derive similar characterizations
for the other local classes
$\MB[1], \dots, \VVc[1]$
in terms of extensions of \kl{backward modal logic}
that offer additional types of \kl{modalities}
(viz., multimodal and graded modal logic).

Motivated by these results,
the connection between distributed algorithms and \kl{modal logic}
was further investigated by Kuusisto
in~\cite{DBLP:conf/csl/Kuusisto13} and~\cite{DBLP:journals/corr/Kuusisto14a}.
The first paper lifts the constraint of locality required
in~\cite{DBLP:conf/podc/HellaJKLLLSV12,DBLP:journals/dc/HellaJKLLLSV15},
thereby allowing algorithms with arbitrary running times.
Now,
for local algorithms,
it does not matter whether we impose a restriction
on the amount of memory space used by each \kl{node},
because in a constant number of rounds,
a \kl{node} can only visit a constant number of different states.
Therefore the local algorithms characterized by Hella~et~al.\
are implicitly finite-state machines.
On the other hand,
in the \kl[nonlocal automata]{nonlocal} case considered by Kuusisto,
space restrictions have to be made explicit.
His papers focus on algorithms for the class~$\SB$,
since results for that class can easily be adapted to the others.
In~\cite{DBLP:conf/csl/Kuusisto13},
particular attention is devoted to
a category of such algorithms
that act as finite-state semi-deciders;
we shall refer to them as \emph{\kl{distributed automata}}.
The main result establishes
a logical characterization of \kl{distributed automata}
in terms of a new recursive logic dubbed
\emph{modal substitution calculus}.
In the same vein,
it is also shown that
the infinite-state generalizations of \kl{distributed automata}
recognize precisely those \kl[digraph languages]{graph properties}
whose complement is definable by
the conjunction of a possibly infinite number
of \kl[backward modal logic]{backward modal} \kl{formulas}
(called modal theory).
Furthermore,
it is proven that on finite \kl[digraphs]{graphs},
\kl{distributed automata} are strictly more expressive
than the least-fixpoint fragment of the backward $\mu$-calculus.
This logic,
which we shall refer to simply as the \emph{\kl{backward $\mu$-fragment}},
extends \kl{backward modal logic} with a least fixpoint operator
that may not be negated.
It thus allows to express statements using least fixpoints,
but unlike in the full backward $\mu$-calculus,
greatest fixpoints are forbidden.
Finally,
the second paper~\cite{DBLP:journals/corr/Kuusisto14a}
makes crucial use of the connection with logic
to show that
universally halting distributed automata are necessarily local
if infinite graphs are allowed into the picture.

Closely related to the work mentioned above,
the last decade has also seen active research in
\emph{distributed decision}~\cite{DBLP:journals/eatcs/FeuilloleyF16},
a field that aims to develop
a counterpart of computational complexity theory
for distributed computing.
In that context,
the \kl{nodes} of a given network have to collectively decide
whether or not their network satisfies
some \kl[digraph language]{global property}.
Every \kl{node} first computes a local answer,
based on the information received from its \kl{neighbors}
over several rounds of communication,
and then all answers are aggregated
to produce a global verdict.
Typically,
the network is considered to be in a valid state
if it has been unanimously accepted by all \kl{nodes};
in other words,
the global answer is the logical conjunction
of the local answers.

Just as in classical complexity theory,
a common approach in distributed decision is
to start with some base class of deterministic algorithms,
and then extend it with additional features,
such as nondeterminism and randomness.
However,
depending on the underlying model of distributed computing,
these additional features
can quickly lead to excessive expressive power.
For instance,
if we add unrestricted nondeterminism
to the widely adopted \textsc{local} model,
\marginnote{
  The \textsc{local} model
  allows unbounded synchronous communication
  between Turing-complete processors
  that are equipped with unique identifiers.
  Despite the name,
  algorithms in this model are not necessarily local.
}
then the \kl{nodes} can simply guess
a representation of the entire network
and verify in one round that their guess was correct.
Consequently,
nondeterministic algorithms in the \textsc{local} model
can already decide every Turing-decidable
\kl[digraph language]{graph property}
in a single round of communication
(see, e.g., \cite[\S~4.1.1]{DBLP:journals/eatcs/FeuilloleyF16}).
To make things more interesting,
one therefore often imposes a restriction
on the number of bits
that each \kl{node} can nondeterministically choose;
viewing nondeterminism as the ability to “guess and verify”,
we refer to the bit strings guessed by the \kl{nodes}
as \emph{certificates}.
A typically chosen bound on the size of those certificates
is logarithmic in the size of the network
because this allows each \kl{node} to guess
only a constant number of processor identifiers.
In stark contrast to the unbounded case,
where Turing-decidability is the only limit,
there are natural decision problems
that cannot be solved by any nondeterministic local algorithm
whose certificates are logarithmically bounded.
An example of such a problem is to verify
whether a given \kl[ditree]{tree} is a minimum spanning tree,
as has been shown by Korman and Kutten
in~\cite{DBLP:journals/dc/KormanK07}.
Nevertheless,
on \kl{connected} \kl{graphs},
nondeterminism with logarithmic certificates provides enough power
to decide every \kl[digraph language]{property}
\kl{definable} in $\EMSOL$
within a constant number of rounds,
essentially by using nondeterministic bits
to construct a spanning tree
and simulate existential \kl{set quantifiers}.
This observation has been made by Göös and Suomela
in~\cite{DBLP:conf/podc/GoosS11,DBLP:journals/toc/GoosS16},
based on the work of Schwentick and Barthelmann
mentioned in the previous \lcnamecref{ssec:related-automata-theory}.

Once existential quantification
has been introduced into the system,
a natural follow-up
is to complement it with universal quantification;
for instance,
in classical complexity theory,
alternating the two types of quantifiers
leads to the polynomial hierarchy,
which generalizes the classes $\NPTIME$ and $\coNPTIME$.
While not very interesting for the unrestricted \textsc{local} model
with unbounded certificates
(where nondeterminism already suffices to decide everything possible),
this form of alternation
provides a genuine increase of power
if we consider distributed algorithms
that are oblivious to the \kl{node} identifiers.
In~\cite{DBLP:conf/stacs/BalliuDFO17},
Balliu, D'Angelo, Fraigniaud and Olivetti
showed that we require one alternation
between universal and existential quantifiers
in order to be able to decide
every Turing-decidable \kl[digraph language]{property}
in the identifier-oblivious variant of the \textsc{local} model
(with unbounded certificates);
hence the corresponding alternation hierarchy
collapses to its second level.
On the other hand,
the hierarchy of the standard \textsc{local} model
with certificates of logarithmic size
is much less well understood;
in particular,
it is still open whether or not that hierarchy is infinite.
As a first step towards an answer,
Feuilloley, Fraigniaud and Hirvonen showed
in~\cite{DBLP:conf/icalp/FeuilloleyFH16}
that if there is equality between the existential and universal versions
of a given level in the logarithmic hierarchy,
then the entire hierarchy collapses to that level.
Furthermore,
they could identify a decision problem
that lies outside of the hierarchy,
which shows that even with the full power of alternation,
algorithms whose certificates are logarithmically bounded
remain weaker than their unrestricted counterparts.

\section{Contributions}

Obviously,
developing a descriptive complexity theory for distributed computing
is a highly ambitious project,
of which the present work can only strive to be a small building block.
As its title suggests,
this thesis does not deal with the powerful models of computation
that are usually considered in distributed computing.
Instead,
it takes an automata-theoretic approach
and focuses on a rather weak model
that has already been explored by Hella~et~al.\ and Kuusisto,
namely \kl{distributed automata}.
The main contributions are two new logical characterizations
related to that model.

The first covers a variant of \kl{local distributed automata},
extended with a global acceptance condition
and the ability to alternate between
nondeterministic decisions of the individual processors
and the creation of parallel computation branches.
This kind of alternation constitutes
a canonical generalization of nondeterminism,
and is nowadays standard in automata theory.
We show that the resulting
\emph{\kl{alternating local automata with global acceptance}}
are \kl[device equivalent]{equivalent} to $\MSOL$ on finite \kl{directed graphs}.
In spirit,
they are similar to
the alternation hierarchies
considered in the distributed-decision community,
even though their expressive power is much more restricted.
They also share some similarities
with Thomas' \kl[digraph]{graph} acceptors,
as they use a combination of
local conditions,
checked by the \kl{nodes} based on their \kl{neighborhood},
and global conditions,
checked at the level of the entire \kl[digraph]{graph}.
However,
both types of conditions are much simpler than in Thomas' model,
which allows us to consider \kl[digraphs]{graphs} of unbounded degree.
To a certain extent,
the \kl[device equivalence]{equivalence} with $\MSOL$ can be considered as
a generalization to \kl[digraphs]{graphs}
of the classical result of Büchi, Elgot and Trakhtenbrot,
although the machines involved are by no means deterministic;
whereas on \kl[pointed dipaths]{words}
and \kl[pointed ordered ditrees]{trees},
alternation simply provides
a more succinct representation of deterministic automata,
it turns out to be a crucial ingredient in our case.
If we allow only nondeterminism,
we get a model that is not closed under complementation,
and is even strictly weaker than $\EMSOL$,
but has a decidable emptiness problem.
Interestingly,
that model is still powerful enough
to characterize precisely the regular languages
when restricted to \kl[pointed dipaths]{words}
or \kl[pointed ordered ditrees]{trees}.
Hence,
this work also contributes to the general observation,
made in \cref{ssec:related-automata-theory},
that regularity becomes a moving target
when lifted to the setting of \kl[digraphs]{graphs}.

The second main contribution consists in a logical characterization
of a fully deterministic class of \kl{nonlocal automata}.
As mentioned in \cref{ssec:related-distributed-computing},
Kuusisto has noticed that \kl{distributed automata},
in their unrestricted form,
are strictly more powerful
than the \kl{backward $\mu$-fragment} on finite \kl[digraphs]{graphs}.
While it is straightforward
to evaluate any \kl{formula} of the \kl{backward $\mu$-fragment}
via a \kl{distributed automaton},
there also exist \kl[distributed automata]{automata} that exploit the fact
that a \kl{node} can determine if it receives the same information
from all of its \kl{neighbors} at the exact same time.
Such a behavior cannot be simulated in the backward \mbox{$\mu$-fragment},
and actually not even in the much more expressive $\MSOL$.
However,
since the argument relies solely on synchrony,
it seems natural to ask
whether removing this feature can lead to a distributed automaton model
that has the same expressive power as the \kl{backward $\mu$-fragment}.
To answer this question,
we introduce several classes of \emph{\kl{asynchronous automata}}
that transfer the standard notion of asynchronous algorithm
to the setting of finite-state machines.
Basically,
this means that we eliminate the global clock from the network,
thus making it possible for \kl{nodes} to operate at different speeds
and for messages to be delayed for arbitrary amounts of time,
or even be lost.
From the syntactic point of view,
an \kl{asynchronous automaton} is the same as a synchronous one,
but it has to satisfy an additional semantic condition:
its \kl{acceptance behavior} must be independent of any timing-related issues.
Taking a closer look at the \kl[distributed automata]{automata}
obtained by translating \kl{formulas} of the \kl{backward $\mu$-fragment},
we can easily see that they are in fact asynchronous.
Furthermore,
their \kl{state} diagrams are almost acyclic,
except that the \kl{states} are allowed to have self-loops;
we call this property \emph{quasi-acyclic}.
As it turns out,
the two properties put together are sufficient
to give us the desired characterization:
\kl{quasi-acyclic} \kl{asynchronous automata}
are \kl[device equivalent]{equivalent} to the \kl{backward $\mu$-fragment} on finite \kl[digraphs]{graphs}.
Incidentally,
this remains true
even if we consider a seemingly more powerful variant of \kl{asynchronous automata},
where all messages are guaranteed to be delivered.

Another aspect of \kl{distributed automata} investigated in this thesis
are decision problems,
and more specifically emptiness problems,
where the task is to decide
whether a given \kl[distributed automaton]{automaton}
\kl{accepts} on at least one input \kl[digraph]{graph}.
As all the \kl[device equivalences]{equivalences} mentioned above are effective,
we can immediately settle the decidability
of the \kl{emptiness problem} for \kl{local automata}:
it is decidable for the basic variant
of Hella~et~al.,
but undecidable for the alternating extension
that we shall consider.
This is because the (finite) satisfiability problem
is $\PSPACE$-complete for (backward) \kl{modal logic}
but undecidable for $\MSOL$.
The problem is also decidable for our classes of \kl{asynchronous automata},
since (finite) satisfiability for the (backward) $\mu$-calculus
is $\EXPTIME$-complete.
However,
the corresponding question for unrestricted,
\kl[nonlocal automata]{\mbox{nonlocal} automata}
was left open in~\cite{DBLP:conf/csl/Kuusisto13}.
Here,
we answer this question negatively for the general case
and also consider it for three special cases.
On the positive side,
we obtain a $\LOGSPACE$ decision procedure
for a class of \kl[forgetful automata]{\emph{forgetful} automata},
where the \kl{nodes} see the messages received from their \kl{neighbors}
but cannot remember their own \kl{state}.
When restricted to the appropriate families of \kl[digraphs]{graphs},
these \kl{forgetful automata} are
\kl[device equivalent]{equivalent} to classical finite \kl[pointed dipath]{word} automata,
but strictly more expressive than finite \kl[pointed ordered ditree]{tree} automata.
On the negative side,
we show that the emptiness problem is already undecidable
for two heavily restricted classes of \kl{distributed automata}:
the \kl{quasi-acyclic} ones,
and those that reject immediately
if they receive more than one distinct message per round.
For the latter class,
we present a proof with an unusual twist,
where a Turing machine is simulated by a \kl{distributed automaton}
in such a way that the roles of space and time
are reversed between the two devices.

Finally, as a minor contribution,
we investigate the problem
of separating \kl[set quantifier]{quantifier} alternation hierarchies
for several classes of \kl{formulas} that are based on \kl{modal logic}.
Essentially,
these classes are hybrids,
obtained by adding the \kl{set quantifiers} of $\MSOL$
to some variant of \kl{modal logic}.
They are motivated by
the above characterizations of \kl{local distributed automata}
in terms of (backward) \kl{modal logic} and $\MSOL$.
The contribution is a toolbox of simple encoding techniques
that allow to easily transfer to the \kl[modal logic]{modal} setting
the separation results for $\MSOL$
established by Matz, Schweikardt and Thomas
in~\cite{DBLP:conf/lics/MatzT97,DBLP:conf/csl/Schweikardt97,DBLP:journals/iandc/MatzST02}.
We thereby provide alternative proofs to similar findings
previously reported by Kuusisto
in~\cite{DBLP:conf/aiml/Kuusisto08,DBLP:journals/apal/Kuusisto15}.

\section{Outline}

The structure of this thesis is rather straightforward.
All the notions that occur in several places
are defined in \cref{ch:preliminaries}.
In particular,
there is a simple definition of \kl{distributed automata}
that subsumes most of the variants we shall consider.
The subsequent four \lcnamecrefs{ch:local}
(i.e., \labelcref{ch:local,ch:nonlocal,ch:emptiness,ch:alternation})
are independent of each other
and thus can be read in any order.
In \cref{ch:local},
we focus on \kl{local distributed automata}
and present the \kl[\ALDAg]{alternating variant with global acceptance},
which is shown to be \kl[device equivalent]{equivalent} to $\MSOL$.
\Cref{ch:nonlocal}
shifts the focus to \kl{nonlocal automata};
there we introduce the semantic notion
of \kl[asynchronous automaton]{asynchrony}
and show that
\kl{quasi-acyclic} \kl{asynchronous automata}
are captured by the \kl{backward $\mu$-fragment}.
\kl{Nonlocal automata} are also the subject of \cref{ch:emptiness},
where we present both positive and negative decidability results
on the \kl{emptiness problem} for several restricted classes.
Then, in \cref{ch:alternation},
we switch completely to logic
and consider issues related to
\kl[set quantifier]{quantifier} alternation hierarchies.
Finally,
some perspectives for future research
are briefly outlined in \cref{ch:perspectives}.

\paragraph{Note to the reader of the electronic version.}
The PDF version of this document
makes extensive use of hyperlinks.
In addition to the cross-reference links inserted automatically
by the standard \LaTeX\ package
\href{https://www.ctan.org/pkg/hyperref}{\textsf{hyperref}},
most of the notions defined within the document
are linked to their point of definition.
This new feature,
which concerns both text and mathematical notation,
is based on the
\href{https://www.ctan.org/pkg/knowledge}{\textsf{knowledge}}
package developed by Thomas Colcombet.
Beware that there can be several links
within a single symbolic expression;
for instance,
the expression $\semF{\BC\SigmaMSO{\Level}(\ML)}[\pDIGRAPH]$
contains links to five different concepts:
$\semF{\dots}$,\, $\BC$,\, $\SigmaMSO{\Level}$,\, $\ML$, and $\pDIGRAPH$.

%% file: tex/preliminaries.tex
\chapter{Preliminaries}
\label{ch:preliminaries}

This \lcnamecref{ch:preliminaries}
introduces essential notation and terminology
that will be recurring throughout this thesis.
It is meant to be consulted for specific information
rather than for consecutive reading.
Concepts that are specific to a single \lcnamecref{ch:preliminaries},
will be introduced later,
along with the topic.

\section{Basic notation}

We denote
the empty set by~\Intro*{$\EmptySet$},
the set of Boolean values by
\Intro*{$\Boolean = \set{0,1}$},
the set of non-negative integers by
\Intro*{$\Natural = \set{0,1,2,\dots}$},
the set of positive integers by
\Intro*{$\Positive = \Natural \setminus \set{0}$},
and the set of integers by
\Intro*{$\Integer = \set{\dots,-1,0,1,\dots}$}.

Integer intervals of the form $\setbuilder{i\in\Integer}{m≤i≤n}$,
where $m,n \in \Integer$ and $m \leq n$,
will sometimes be denoted by \Intro*{$\range[m]{n}$}.
We may also use the shorthand \reintro*{$\range{n} \defeq \range[1]{n}$},
and, by analogy with the Bourbaki notation for real intervals,
we indicate that we exclude an endpoint
by reversing the square bracket corresponding to that endpoint,
e.g., \Intro*{$\lrange[m]{n} \defeq \range[m]{n} \setminus \set{m}$}.

For any two sets $\Set[1]$ and $\Set[2]$,
the set of all functions from $\Set[1]$ to $\Set[2]$
is denoted \Intro*{$\FunctionSpace{\Set[1]}{\Set[2]}$}.
This notation gives rise to two important special cases.
First,
we write \Intro*{$\powerset{\Set[1]}$} for the power set of $\Set[1]$,
since we can identify it
with the set of all functions from $\Set[1]$ to $\set{0,1}$.
Second,
given $k \in \Natural$,
we write
$\Pseudointro{\Set[1]^k} \defeq \FunctionSpace{\range{k}}{\Set[1]}$
for the set of all $k$-tuples over $\Set[1]$,
since we can identify it
with the set of functions from $\range{k}$ to $\Set[1]$.
All of these notations have another special case in common:
the set of binary strings of length $k$,
denoted $\Boolean^k$,
can be interpreted as either
the function space from $\range{k}$ to $\Boolean$,
or the power set of $\range{k}$,
or the set of $k$-tuples over~$\Boolean$.
By the first interpretation,
the individual letters
of a string~$\String$ of length~$k$
will be denoted $\String(1), \dots, \String(k)$.
Furthermore,
we write
\Intro*{$\card{\Set[1]}$} for the cardinality of~$\Set[1]$
and \Intro*{$\length{\String}$} for the length~of~$\String$.

\section{Symbols}
\label{sec:symbols}

Since logic plays an important role in this thesis,
it also has an influence
on how we present other concepts;
in particular,
our definition of \kl{directed graphs} in \cref{sec:digraphs}
will refer to the notion of (abstract) \kl{symbol}.

We shall not always make a sharp distinction
between \kl{variables} and (non-logical) \kl{constants}.
Instead,
there is simply a fixed supply of \kl{symbols},
which can serve both as \kl{variables} and as \kl{constants}.
Hence,
the terms “\Intro{variable}” and “\Intro{constant}”
are just synonyms for “\kl{symbol}”;
we will use them
whenever we want to clarify the intended role of a \kl{symbol}
within a given context.

The set~\Intro*{$\NodeSymbolSet$} contains our \Intro{node symbols},
which within \kl{formulas} will represent \kl{nodes} of \kl{structures}
such as \kl[digraphs]{graphs};
among them,
there is a special \Intro{position symbol}~\Intro*{$\PosSymbol$}.
Moreover,
for every integer $\Arity \geq 1$,
we let~\Intro*{$\RelSymbolSet{\Arity}$} denote
the set of \Intro[relation symbol]{$\Arity$-ary relation symbols}.
All of these sets are infinite and pairwise disjoint.
If a \kl{symbol} lies in~$\RelSymbolSet{\Arity}$,
for $\Arity \geq 0$,
then we call $\Arity$ the \Intro{arity} of that \kl{symbol}.
We also denote the set of all \Intro{symbols} by~$\SymbolSet$,
i.e., \Intro*{$\SymbolSet \defeq \bigcup_{\Arity\geq 0}\RelSymbolSet{\Arity}$},
\marginnote{
  $\SymbolSet$ contains both \\
  \kl{variables} and \\
  \kl{constants}.
}
and shall often refer to
the unary \mbox{\kl{relation symbols}} in $\SetSymbolSet$
as \Intro{set symbols}.

\kl{Node symbols} will always be represented by lower-case letters,
and \kl{relation symbols} by upper-case ones,
often decorated with subscripts.
Typically,
we use
$\NodeVariable[1], \NodeVariable[2], \NodeVariable[3]$
for \kl{node variables} or arbitrary \kl{node symbols},
$\SetVariable[1], \SetVariable[2], \SetVariable[3]$
for \kl{set variables} or arbitrary \kl{set symbols},
$\SetConstant[1], \SetConstant[2]$
for \kl{set constants},
and
$\RelSymbol[1], \RelSymbol[2]$
for \kl{relation constants} of higher \kl{arity}
or arbitrary \kl{symbols}.
(See \cref{sec:example-formulas} for some simple examples.)

\section{Structures}

Before we formally introduce \kl{directed graphs}
in the next \lcnamecref{sec:digraphs},
we define the more general concept of a relational \kl{structure}.
Although the present thesis focuses mainly
on variants of \kl{directed graphs},
this top-down approach will allow us to specify
the semantics of several types of logical \kl{formulas}
in a unified framework,
using a consistent notation.
In particular,
it will be apparent
that \kl{modal logic} simply provides an alternative syntax
for a certain fragment of \kl{first-order logic}
(see \cref{sec:logics}).

Let~$\Signature$ be any subset of~$\SymbolSet$.
A (relational) \Intro{structure}~$\Structure$ of \Intro{signature}~$\Signature$
consists of
a nonempty set of \Intro{nodes}~\Intro*{$\NodeSet{\Structure}$}
(also called the \Intro{domain} of $\Structure$),\,
a \kl{node}~\Intro*{$\inp{\NodeSymbol}{\Structure}$} of~$\NodeSet{\Structure}$
for each \kl{node symbol}~$\NodeSymbol$ in~$\Signature$,
and a $\Arity$-ary relation~\reintro*{$\inp{\RelSymbol}{\Structure}$}
on~$\NodeSet{\Structure}$
for each $\Arity$-ary \kl{relation symbol}~$\RelSymbol$ in~$\Signature$.
Here,
$\inp{\NodeSymbol}{\Structure}$ and $\inp{\RelSymbol}{\Structure}$ are called
$\Structure$'s \Intro{interpretations} of the \kl{symbols} $\NodeSymbol$ and $\RelSymbol$.
We may also say that
$\Structure$ is a \kl{structure} \reintro[signature]{over}~$\Signature$,
or that $\Signature$ is the \reintro{underlying signature} of $\Structure$,
and we denote $\Signature$ by \Intro*{$\signature(\Structure)$}.
In case the \kl{position symbol}~$\PosSymbol$ lies in $\signature(\Structure)$,
we call $\Structure$ a \Intro{pointed structure}
and~$\inp{\PosSymbol}{\Structure}$
the \Intro{distinguished node} of~$\Structure$.

A set of \kl{structures} will be referred to as
a \Intro{structure language}.
As is customary,
we are only interested in \kl{structures} up to isomorphism.
That is,
two \kl{structures} over $\Signature$ are considered to be equal
if there is a bijection between their \kl{domains}
that preserves the \kl{interpretations} of all \kl{symbols}~in~$\Signature$.
Consequently,
our \kl{structure languages} characterize only properties
that are invariant under isomorphism.

Let $\Structure$ be a \kl{structure}
and $\Assignment$ be a map
of the form
$\set{\map{\Symbol_1}{\Interpretation_1}, \,\dots\,, \map{\Symbol_n}{\Interpretation_n}}$
that
assigns to each \kl{symbol}~$\Symbol_i \in \SymbolSet$,
for $i \in \range{n}$,
a suitable \kl{interpretation}~$\Interpretation_i$
over the \kl{domain} of~$\Structure$.
That is,
if $\Symbol_i∈\NodeSymbolSet$,
then $\Interpretation_i∈\NodeSet{\Structure}$,
and if $\Symbol_i∈\RelSymbolSet{\Arity}$,
for $\Arity \geq 1$,
then $\Interpretation_i⊆(\NodeSet{\Structure})^\Arity$.
We use the notation
\Intro*{$\aver{\Structure}{\Assignment}$}
to designate the \Intro[extended variant]{$\Assignment$-extended variant}
of~$\Structure$,
which is the \kl{structure}~$\Structure'$
obtained from~$\Structure$
by \kl{interpreting} each \kl{symbol}~$\Symbol_i$ as~$\Interpretation_i$,
while maintaining the other \kl{interpretations} provided by~$\Structure$.
More formally,
letting $\Signature = \set{\Symbol_1,\dots,\Symbol_n}$,
we have
$\NodeSet{\Structure'} = \NodeSet{\Structure}$,\,
$\signature(\Structure') = \signature(\Structure)∪\Signature$,\,
$\inp{\Symbol_i}{\Structure'} = \Interpretation_i$
for $i \in \range{n}$,\,
and $\inp{\Symbol[2]}{\Structure'} = \inp{\Symbol[2]}{\Structure}$
for $\Symbol[2]∈\signature(\Structure)\setminus\Signature$.
Often,
we do not want to give an explicit name to the assignment~$\Assignment$,
in which case we may denote $\Structure'$ by
\Intro*{$\ver{\Structure}{\Symbol_1,\dots,\Symbol_n}{\Interpretation_1,\dots,\Interpretation_n}$}.
If the \kl{interpretations} of the \kl{symbols} in~$\Signature$
are clear from context,
we may also refer to~$\Structure'$ as
the \reintro[extended variant]{$\Signature$-extended variant} of $\Structure$.
Furthermore,
as we will often consider
\kl{pointed} \kl[extended variant]{variants} of \kl{structures},
we introduce the shorthand
\Intro*{$\pver{\Structure}{\Node} \defeq \ver{\Structure}{\PosSymbol}{\Node}$}
for $\Node \in \NodeSet{\Structure}$,
and refer to
$\pver{\Structure}{\Node}$ as
the \Intro[pointed variant]{$\Node$-pointed variant} of~$\Structure$
(i.e., the \kl[extended variant]{variant} of~$\Structure$
with \kl{distinguished node}~$\Node$).

\section{Different kinds of digraphs}
\label{sec:digraphs}

The \kl{structures} we are actually interested in
are several variants of \kl{directed graphs};
these are \kl{structures} with finite \kl{domains}
and relations of \kl{arity} at most $2$.
To facilitate lookup and comparison,
we present them all in the same \lcnamecref{sec:digraphs}.
In the following definitions,
let $\BitCount$ and $\RelCount$ be non-negative integers.

An \Intro[digraph]{$\BitCount$-bit labeled, $\RelCount$-relational directed graph},
abbreviated \reintro{digraph},
is a finite \kl{structure}~$\Digraph$ of \kl{signature}
$\set{\SetConstant_1,…,\SetConstant_\BitCount,\RelSymbol_1,…,\RelSymbol_\RelCount}$,
where
$\SetConstant_1,…,\SetConstant_\BitCount$
are \kl{set symbols},
and
$\RelSymbol_1,…,\RelSymbol_\RelCount$
are binary \kl{relation symbols}.

The sets
$\inp{\SetConstant_1}{\Digraph},…,\inp{\SetConstant_\BitCount}{\Digraph}$,
which we shall call \Intro{labeling sets},
determine a (\kl{node}) \Intro{labeling}
\Intro*{$\Labeling{\Graph} \colon \NodeSet{\Graph} \to \Boolean^\BitCount$}
that assigns a binary string of length~$\BitCount$ to each \kl{node}.
More precisely,
we define $\Labeling{\Graph}$
such that
\begin{equation*}
  \Labeling{\Graph}(\Node)(i) =
  \begin{cases*}
    0 & if $\Node \notin \SetConstant_i$, \\
    1 & otherwise,
  \end{cases*}
\end{equation*}
for all $\Node \in \NodeSet{\Graph}$ and $i \in \range{\BitCount}$.
Given another mapping
$\Relabeling \colon \NodeSet{\Digraph} \to \Boolean^{\BitCount'}$
with $\BitCount' \in \Natural$,
we shall denote by \Intro*{$\lver{\Digraph}{\Relabeling}$}
\marginnote{
  Our bracket notation is overloaded,
  but if one knows the type of $\Relabeling$, the \\
  \kl[relabeled variant]{$\Relabeling$-relabeled variant}
  $\lver{\Digraph}{\Relabeling}$ \\
  of\, $\Digraph$
  should be easy to distinguish from an \\
  \kl[extended variant]{$\Assignment$-extended variant}
  $\aver{\Digraph}{\Assignment}$, \\
  as well as from a \\
  \kl[pointed variant]{$\Node$-pointed variant}
  $\pver{\Digraph}{\Node}$.
}
the \Intro[relabeled variant]{$\Relabeling$-relabeled variant}
of $\Digraph$,
i.e., the $\BitCount'$-bit labeled \kl{digraph}~$\Digraph'$
that is the same as~$\Digraph$,
except that its \kl{labeling}~$\Labeling{\Digraph'}$
is equal to~$\Relabeling$.

It is often convenient to regard the \kl{labels}
of an $\BitCount$-bit \kl{labeled} \kl{digraph}
as the binary encodings of
letters of some finite alphabet~$\Alphabet$.
With respect to a given injective map
$f \colon \Alphabet \to \Boolean^\BitCount$,
a \reintro[digraph]{$\Alphabet$-labeled digraph}
is an $\BitCount$-bit \kl{labeled} \kl{digraph}~$\Digraph$
such that for every \kl{node} $\Node \in \NodeSet{\Graph}$,
we have
$\Labeling{\Graph}(\Node) = f(\Letter)$
for some $\Letter \in \Alphabet$.
Since we do not care about the specific encoding function~$f$,
we will never mention it explicitly,
and just call $\Digraph$ a
$\Alphabet$-\kl{labeled}, $\RelCount$-relational \kl{digraph}.

The binary relations
$\inp{\RelSymbol_1}{\Digraph},…,\inp{\RelSymbol_\RelCount}{\Digraph}$
will be referred to as \Intro{edge relations}.
If $\edge{\Node[1]}{\Node[2]}$ is an \Intro{edge}~in $\EdgeSet[i]{\Digraph}$,
then $\Node[1]$ is called
an \Intro[incoming neighbor]{incoming $i$-neighbor} of~$\Node[2]$,
or simply an \kl{incoming neighbor},
and $\Node[2]$ is called
an \Intro[outgoing neighbor]{outgoing $i$-neighbor} of~$\Node[1]$,
or just \kl{outgoing neighbor}.
We also say that $\Node[1]$ and $\Node[2]$ are \Intro{adjacent},
and without further qualification,
the term \Intro{neighbor} refers to both
\kl[incoming neighbor]{incoming} and \kl{outgoing neighbors}.
The (undirected) \Intro{neighborhood} of a \kl{node}
is the set of all of its \kl{neighbors},
and the \Intro[incoming neighborhood]{incoming}
and \Intro[outgoing neighborhood]{outgoing} \kl{neighborhoods}
are defined analogously.
A \kl{node} without \kl{incoming neighbors} is called a \Intro{source},
whereas a \kl{node} without \kl{outgoing neighbors} is called a \Intro{sink}.

The class of all
$\BitCount$-bit \kl{labeled}, $\RelCount$-relational \kl{digraphs}
is denoted by
\Intro*{$\DIGRAPH[\BitCount][\RelCount]$}.
In case the parameters are
$\BitCount=0$ and $\RelCount=1$,
we may omit them and use the shorthand
\reintro*{$\DIGRAPH \defeq \DIGRAPH[0][1]$}.
We shall also drop the subscripts on the \kl{symbols},
and just write $\SetConstant$ or $\RelSymbol$,
if there is only one \kl{symbol} of a given \kl{arity}.
Furthermore,
we denote by
\reintro*{$\DIGRAPH[\Alphabet][\RelCount]$}
the class of all
$\Alphabet$-\kl{labeled}, $\RelCount$-relational \kl{digraphs}.

As can be easily guessed from the previous definitions,
a \Intro{pointed digraph} is a \kl{digraph}
in which some \kl{node} has been marked
by the \kl{position symbol}~$\PosSymbol$,
i.e., it is a \kl{structure} of the form
$\ver{\Digraph}{\PosSymbol}{\Node}$,
with $\Digraph∈\DIGRAPH[\BitCount][\RelCount]$ and $\Node∈\NodeSet{\Digraph}$.
We write \Intro*{$\pDIGRAPH[\BitCount][\RelCount]$}
for the set of all
$\BitCount$-bit \kl{labeled}, $\RelCount$-relational \kl{pointed digraphs},
and define \reintro*{$\pDIGRAPH \defeq \pDIGRAPH[0][1]$}.

A \kl{digraph}~$\Digraph$ is called
an ($\BitCount$-bit \kl{labeled}, $\RelCount$-relational) \Intro{undirected graph},
or simply \reintro{graph},
if all of its \kl{edge relations} are irreflexive and symmetric,
i.e., if for all
$\Node[1],\Node[2]∈\NodeSet{\Digraph}$ and $i \in \range{\RelCount}$,
it holds that
$\edge{\Node[1]}{\Node[1]} \notin \EdgeSet[i]{\Digraph}$,
and
$\edge{\Node[1]}{\Node[2]} \in \EdgeSet[i]{\Digraph}$
if and only if
$\edge{\Node[2]}{\Node[1]} \in \EdgeSet[i]{\Digraph}$.
The corresponding class is denoted by
\Intro*{$\GRAPH[\BitCount][\RelCount]$},
and we may use the shorthand
\reintro*{$\GRAPH \defeq \GRAPH[0][1]$}.

A \kl{digraph}~$\Digraph$ is (weakly) \Intro{connected}
if for every nonempty proper subset $\NodeSubset$ of $\NodeSet{\Digraph}$,
there exist two \kl{nodes}
$\Node[1]∈\NodeSubset$ and $\Node[2]∈\NodeSet{\Digraph} \setminus \NodeSubset$
that are \kl{adjacent}.

The \kl[labeling]{node labeling} $\Labeling{\Digraph}$ of a $\Alphabet$-\kl{labeled} \kl{digraph}
constitutes a valid \Intro{coloring} of $\Digraph$
if no two \kl{adjacent} \kl{nodes} share the same \kl{label},
i.e.,
if $\edge{\Node[1]}{\Node[2]} \in \EdgeSet[i]{\Digraph}$
implies
$\Labeling{\Digraph}(\Node[1]) \neq \Labeling{\Digraph}(\Node[2])$,
for all $\Node[1],\Node[2]∈\NodeSet{\Digraph}$ and
$i \in \range{\RelCount}$.
If $\card{\Alphabet}=k$,
such a \kl{coloring} is called a \kl{$k$-coloring} of $\Digraph$,
and any $\RelCount$-relational \kl{digraph}
for which a \kl{$k$-coloring} exists
is said to be \Intro{$k$-colorable}.
Note that,
by definition,
a \kl{digraph} that contains self-loops
is not \kl{$k$-colorable} for any~$k$.

A \Intro{directed rooted tree}, or \reintro{ditree},
is an ($\BitCount$-bit \kl{labeled}) $\RelCount$-relational \kl{digraph}~$\Graph$
that has a distinct \kl{node}~$\Root$, called the \Intro{root},
such that
$\EdgeSet[i]{\Graph} \cap \EdgeSet[j]{\Graph} = \EmptySet$
for $i \neq j$,
and from each \kl{node}~$\Node$ in~$\NodeSet{\Digraph}$,
there is exactly one way to reach~$\Root$
by following the directed \kl{edges} in
$\bigcup_{1\leq i\leq\RelCount} \EdgeSet[i]{\Graph}$.
A \Intro{pointed ditree} is a \kl{pointed digraph}
$\pver{\Digraph}{\Root}$,
where $\Digraph$ is a \kl{ditree} and $\Root$ is its \kl{root}.
Moreover,
a (\kl[pointed ditree]{pointed}) $\RelCount$-relational \kl{ditree}
is called \Intro{ordered}
if for $1\leq i\leq\RelCount$,
every \kl{node} has at most one \kl[incoming neighbor]{incoming $i$-neighbor}
and every \kl{node} that has an \kl[incoming neighbor]{incoming $(i+1)$-neighbor}
also has an \kl[incoming neighbor]{incoming $i$-neighbor}.
As a special case,
an ordered $1$-relational \kl{ditree}
is referred to as a \Intro{directed path}, or \reintro{dipath}.
Accordingly,
the \kl{distinguished node} of a \Intro{pointed dipath}
is the last \kl{node} (the one with no \kl{outgoing neighbor}).
The classes of \kl{pointed dipaths} and \kl{pointed ordered ditrees}
can be identified with the \kl{structures}
on which classical word and tree automata are run.
We denote them by
\Intro*{$\pDIPATH[\BitCount]$}
and
\Intro*{$\pODITREE[\BitCount][\RelCount]$},
respectively.

We shall also consider an important subclass
of $\DIGRAPH[\BitCount][2]$
whose members represent rectangular labeled grids
(also called pictures).
In such a \kl{structure}~$\Grid$,
each \kl{node} is identified with a grid cell,
and the \kl{edge relations}
$\inp{\RelSymbol_1}{\Grid}$ and $\inp{\RelSymbol_2}{\Grid}$
are interpreted as
the “vertical” and “horizontal” successor relations,
respectively.
The unique \kl{node} that has no predecessor at all
is regarded as the “upper-left corner”,
and all the usual terminology of matrices applies.
Formally,
$\Grid$ is a \Intro[grid]{$\BitCount$-bit labeled grid}
if, for some $\GridHeight,\GridWidth≥1$,
it is isomorphic to a \kl{structure}
with \kl{domain}
$\set{1,…,\GridHeight}×\set{1,…,\GridWidth}$
and \kl{edge relations}
\begin{align*}
  \inp{\RelSymbol_1}{\Grid}
    &= \bigsetbuilder{\bigtuple{\tuple{i,j},\tuple{i+1,j}}}{1≤i<\GridHeight,\,1≤j≤\GridWidth},
    \\[0.5ex]
  \inp{\RelSymbol_2}{\Grid}
    &= \bigsetbuilder{\bigtuple{\tuple{i,j},\tuple{i,j+1}}}{1≤i≤\GridHeight,\,1≤j<\GridWidth}.
\end{align*}
If $\BitCount=0$,
we refer to $\Grid$ simply as a \reintro{grid}.
In alignment with the previous nomenclature,
we let
\reintro*{$\GRID$} and \Intro*{$\GRID[\BitCount]$}
denote
the classes of \kl{grids} and
\kl[grid]{$\BitCount$-bit labeled grids}.

A \Intro{digraph language} is a \kl{structure language}
that consist of \kl{digraphs}
with a fixed number of \kl{labeling sets} and \kl{edge relations},
i.e., a subset of $\DIGRAPH[\BitCount][\RelCount]$,
for some $\BitCount, \RelCount \in \Positive$.
The notion is defined analogously
for all the other classes of \kl{structures} introduced above.
In particular,
a \Intro{pointed-digraph language} is
a subset of $\pDIGRAPH[\BitCount][\RelCount]$.

\section{The considered logics}
\label{sec:logics}

As we shall contemplate both
classical logic and several variants of \kl{modal logic},
we introduce them all in a common framework.
First we define the syntax and semantics of a generalized language,
and then we specify
which particular syntactic fragments we are interested~in.
Some examples will follow in \cref{sec:example-formulas}.

\begin{table*}[p]
  \Phantomintro{\PosIs}
  \Phantomintro{\Equals}
  \Phantomintro{\PosIn}
  \Phantomintro{\InSet}
  \Phantomintro{\InRel}
  \Phantomintro{\NOT}
  \Phantomintro{\OR}
  \Phantomintro{\dm}
  \Phantomintro{\bdm}
  \Phantomintro{\gdm}
  \Phantomintro{\Exists}
  \centering
  \widefloat{
    \small
    \begin{tabularx}{\tablewidth}{llY}
      \toprule
      \textit{Syntax} & \textit{Free symbols} & \textit{Semantics}
        \\
      Formula $\Formula[2]$
        & Symbol set $\free(\Formula[2])$
        & Necessary and sufficient condition for $\Structure\Models\Formula[2]$
        \\
      \midrule\addlinespace
      \reintro*{$\PosIs{\NodeSymbol[1]}$}
        & $\set{\PosSymbol,\NodeSymbol[1]}$
        & $\inp{\PosSymbol}{\Structure} = \inp{\NodeSymbol[1]}{\Structure}$
        \\\addlinespace
      \reintro*{$(\NodeSymbol[1] \Equals \NodeSymbol[2])$}
        & $\set{\NodeSymbol[1],\NodeSymbol[2]}$
        & $\swl{\inp{\NodeSymbol[1]}{\Structure}}
               {\inp{\PosSymbol}{\Structure}}
           = \inp{\NodeSymbol[2]}{\Structure}$
        \\\addlinespace
      \reintro*{$\PosIn{\SetSymbol}$}
        & $\set{\PosSymbol,\SetSymbol}$
        & $\inp{\PosSymbol}{\Structure} ∈ \inp{\SetSymbol}{\Structure}$
        \\\addlinespace
      \reintro*{$\InSet{\SetSymbol}{\NodeSymbol[1]}$}
        & $\set{\NodeSymbol[1],\SetSymbol}$
        & $\swl{\inp{\NodeSymbol[1]}{\Structure}}
               {\inp{\PosSymbol}{\Structure}}
           ∈ \inp{\SetSymbol}{\Structure}$
        \\\addlinespace
      \reintro*{$\InRel{\RelSymbol}{\NodeSymbol[1]_0,…,\NodeSymbol[1]_\Arity}$}
        & $\set{\NodeSymbol[1]_0,…,\NodeSymbol[1]_\Arity,\RelSymbol}$
        & $\tuple{\inp{\NodeSymbol[1]_0}{\Structure},…,\inp{\NodeSymbol[1]_\Arity}{\Structure}}
           ∈ \inp{\RelSymbol}{\Structure}$
        \\\addlinespace
      \reintro*{$\NOT\Formula$}
        & $\free(\Formula)$
        & not\, $\Structure\Models\Formula$
        \\\addlinespace
      \reintro*{$(\Formula_1\OR\Formula_2)$}
        & $\free(\Formula_1)∪\free(\Formula_2)$
        & $\Structure\Models\Formula_1$ \,or\, $\Structure\Models\Formula_2$
        \\\addlinespace
      \reintro*{$\dm[\RelSymbol](\Formula_1,…,\Formula_\Arity)$}
        & $\set{\PosSymbol,\RelSymbol}
           ∪\, \smashoperator{\bigcup_{1≤i≤\Arity}}\free(\Formula_i)$
        & For some $\Node_1,\twodots,\Node_\Arity\mathbin∈\NodeSet{\Structure}$ such that
          $\tuple{\inp{\PosSymbol}{\Structure},\Node_1,\twodots,\Node_\Arity}
           ∈\inp{\RelSymbol}{\Structure}$,
          we have $\ver{\Structure}{\PosSymbol}{\Node_i}\Models\Formula_i$
          for each $i∈\set{1,\twodots,\Arity}$.
        \\\addlinespace[0.8\defaultaddspace]
      \reintro*{$\bdm[\RelSymbol](\Formula_1,…,\Formula_\Arity)$}
        & same as above
        & As above, except for the condition
          $\tuple{\Node_\Arity,…,\Node_1,\inp{\PosSymbol}{\Structure}}
           ∈\inp{\RelSymbol}{\Structure}$.
        \\\addlinespace
      \reintro*{$\gdm \Formula$}
        & $\free(\Formula)\setminus\set{\PosSymbol}$
        & $\swl{\ver{\Structure}{\PosSymbol}{\Node}}
               {\ver{\Structure}{\SetSymbol}{\NodeSubset}}
           \Models\Formula$
          for some $\Node∈\NodeSet{\Structure}$
        \\\addlinespace
      \reintro*{$\Exists{\NodeSymbol[1]}\Formula$}
        & $\free(\Formula)\setminus\set{\NodeSymbol[1]}$
        & $\swl{\ver{\Structure}{\NodeSymbol[1]}{\Node}}
               {\ver{\Structure}{\SetSymbol}{\NodeSubset}}
           \Models\Formula$
          for some $\Node∈\NodeSet{\Structure}$
        \\\addlinespace
      \reintro*{$\Exists{\SetSymbol}\Formula$}
        & $\free(\Formula)\setminus\set{\SetSymbol}$
        & $\ver{\Structure}{\SetSymbol}{\NodeSubset}\Models\Formula$
          for some $\NodeSubset⊆\NodeSet{\Structure}$
        \\\addlinespace
      \midrule
      \multicolumn{3}{l}{
        Here,\,
        $\NodeSymbol[1],\NodeSymbol[1]_0,…,\NodeSymbol[1]_\Arity,\NodeSymbol[2]
         ∈\NodeSymbolSet$,\:
        $\SetSymbol∈\SetSymbolSet$,\,
        $\RelSymbol∈\RelSymbolSet{\Arity+1}$,
        and $\Formula,\Formula_1,…,\Formula_\Arity$
        are \kl{formulas}, for $\Arity≥1$.}
        \\
      \bottomrule
    \end{tabularx}
  }
  \caption{Syntax and semantics of the considered logics.}
  \label{tab:syntax-semantics}
\end{table*}

\begin{table*}[p]
  \Phantomintro{\FOL}
  \Phantomintro{\EMSOL}
  \Phantomintro{\MSOL}
  \Phantomintro{\ML}
  \Phantomintro{\bML}
  \Phantomintro{\dML}
  \Phantomintro{\MLg}
  \Phantomintro{\bMLg}
  \Phantomintro{\dMLg}
  \Phantomintro{\MSO}
  \centering
  \newcommand*{\seplinespace}{\addlinespace[1.5\defaultaddspace]}
  \widefloat{
    \small
    \begin{tabularx}{\tablewidth}{lYl}
      \toprule
      \multicolumn{2}{l}{\textit{Class of formulas}} & \textit{Generating grammar}
        \\
      \midrule\addlinespace
      \reintro*{$\FOL$}
        & First-order
        & $\Formula
           \Coloneqq (\NodeSymbol[1] \Equals \NodeSymbol[2])
                \mid \InSet{\SetSymbol}{\NodeSymbol[1]}
                \mid \InRel{\RelSymbol}{\NodeSymbol[1]_0,…,\NodeSymbol[1]_\Arity}
                \mid \NOT\Formula
                \mid (\Formula_1\OR\Formula_2)
                \mid \Exists{\NodeSymbol[1]}\Formula$
        \\\seplinespace
      \reintro*{$\EMSOL$}
        & Existential $\MSOL$
        & $\Formula[1]
           \Coloneqq \Formula[2]
                \mid \Exists{\SetSymbol}\Formula[1]$,\:
          where $\Formula[2] \in \FOL$.
          \\
        && Equivalently, $\EMSOL \defeq \SigmaMSO{1}(\FOL)$;\,
           see \cref{sec:alternation-preliminaries}.
        \\\seplinespace
      \reintro*{$\MSOL$}
        & Monadic
        & $\Formula
           \Coloneqq (\NodeSymbol[1] \Equals \NodeSymbol[2])
                \mid \InSet{\SetSymbol}{\NodeSymbol[1]}
                \mid \InRel{\RelSymbol}{\NodeSymbol[1]_0,…,\NodeSymbol[1]_\Arity}
                \mid \NOT\Formula
                \mid (\Formula_1\OR\Formula_2)
                \mid \Exists{\NodeSymbol[1]}\Formula
                \mid \Exists{\SetSymbol}\Formula$
        \\
        & second-order
        & Equivalently, $\MSOL \defeq \MSO(\FOL)$.
        \\\seplinespace
      \reintro*{$\ML$}
        & Modal
        & $\Formula
           \Coloneqq \PosIs{\NodeSymbol[1]}
                \mid \PosIn{\SetSymbol}
                \mid \NOT\Formula
                \mid (\Formula_1\OR\Formula_2)
                \mid \dm[\RelSymbol](\Formula_1,…,\Formula_\Arity)$
        \\\seplinespace
      \reintro*{$\bML$}
        & Backward modal
        & $\Formula
           \Coloneqq \PosIs{\NodeSymbol[1]}
                \mid \PosIn{\SetSymbol}
                \mid \NOT\Formula
                \mid (\Formula_1\OR\Formula_2)
                \mid \bdm[\RelSymbol](\Formula_1,…,\Formula_\Arity)$
        \\\seplinespace
      \reintro*{$\dML$}
        & \mbox{Bidirectional modal}
        & $\Formula
           \Coloneqq \PosIs{\NodeSymbol[1]}
                \mid \PosIn{\SetSymbol}
                \mid \NOT\Formula
                \mid (\Formula_1\OR\Formula_2)
                \mid \dm[\RelSymbol](\Formula_1,…,\Formula_\Arity)
                \mid \bdm[\RelSymbol](\Formula_1,…,\Formula_\Arity)$
        \\\seplinespace
      \reintro*{$\MLg$}
        & Modal with \mbox{\kl{global modalities}}
        & $\Formula
           \Coloneqq \PosIs{\NodeSymbol[1]}
                \mid \PosIn{\SetSymbol}
                \mid \NOT\Formula
                \mid (\Formula_1\OR\Formula_2)
                \mid \dm[\RelSymbol](\Formula_1,…,\Formula_\Arity)
                \mid \gdm \Formula$
        \\\seplinespace
      \multicolumn{2}{l}{\reintro*{$\bMLg$},\, \reintro*{$\dMLg$}}
        & Analogous to the preceding grammars.
        \\\seplinespace
      \reintro*{$\MSO(\FormulaSet)$}
        & $\FormulaSet$ extended with \mbox{\kl{set quantifiers}}
        & Same grammar as $\FormulaSet$ with the additional choice
          “$\+\mid \Exists{\SetSymbol}\Formula\+$”.
        \\\addlinespace
      \midrule
      \multicolumn{3}{l}{
        Here,\,
        $\NodeSymbol[1],\NodeSymbol[1]_0,…,\NodeSymbol[1]_\Arity,\NodeSymbol[2]∈\NodeSymbolSet$,\:
        $\SetSymbol∈\SetSymbolSet$,\:
        $\RelSymbol∈\RelSymbolSet{\Arity+1}$,
        for $\Arity≥1$,
        and $\FormulaSet \in \set{\ML, \bML, \dots, \dMLg, \FOL}$.}
        \\
      \bottomrule
    \end{tabularx}
  }
  \caption{The considered classes of \kl{formulas}.}
  \label{tab:languages}
\end{table*}

\Cref{tab:syntax-semantics} shows
how \Intro{formulas} are built up,
and what they mean.
Furthermore,
it indicates how to obtain
the set \Intro*{$\free(\Formula)$} of \kl{symbols}
that occur \Intro{freely} in a given \kl{formula} $\Formula$,
i.e., outside the scope of a binding operator.
If $\free(\Formula) ⊆ \Signature$,
we say that $\Formula$ is a \Intro{sentence} over $\Signature$.
Sometimes,
when the notions of “\kl{variable}” and “\kl{constant}”
are clear from context,
we also use the notation
\Intro*{$\param{\Formula}{\NodeVariable_1,\dots,\NodeVariable_m,\SetVariable_1,\dots,\SetVariable_n}$}
to indicate that
at most the \kl{variables} given in brackets
occur \kl{freely} in $\Formula$,
i.e., that no other \kl{variables} than
$\NodeVariable_1,\dots,\NodeVariable_m,\SetVariable_1,\dots,\SetVariable_n$
lie in $\free(\Formula)$.
The relation $\Models$ defined in \cref{tab:syntax-semantics}
specifies in which cases
a \kl{structure} $\Structure$ \Intro{satisfies} $\Formula$,
written \Intro*{$\Structure\Models\Formula$},
assuming that $\Formula$ is
a \kl{sentence} over $\signature(\Structure)$.
Otherwise,
we stipulate that \Intro*{$\Structure\ModelsNot\Formula$}.

Of particular interest for this thesis are those \kl{formulas}
in which the \kl{node symbol}~$\PosSymbol$ is considered to be \kl{free},
although it might not occur explicitly.
They are evaluated on a \kl{pointed structure} $\Structure$
from the perspective of
the \kl{node} $\inp{\PosSymbol}{\Structure}$\!.
Atomic \kl{formulas} of the form
$\PosIs{\NodeSymbol}$ or $\PosIn{\SetSymbol}$,
with $\NodeSymbol∈\NodeSymbolSet$ and $\SetSymbol∈\SetSymbolSet$,
are \kl{satisfied} if
$\inp{\PosSymbol}{\Structure}$
is labeled by the corresponding \kl{symbol}.
Using the operator~$\dm[\RelSymbol]$\!,
which is called the \Intro[diamond]{$\RelSymbol$-diamond},
we can remap the \kl[position symbol]{symbol}~$\PosSymbol$
through existential quantification
over the \kl{nodes} in $\Structure$
that are reachable from $\inp{\PosSymbol}{\Structure}$
through the relation $\inp{\RelSymbol}{\Structure}$.
If we want to do the same
with respect to the inverse relation of
$\inp{\RelSymbol}{\Structure}$,
we can use the
\Intro[backward diamond]{backward $\RelSymbol$-diamond}~$\bdm[\RelSymbol]$.
In addition,
there is also the \Intro{global diamond}~$\gdm$
(unfortunately often called “universal modality”),
which ranges over all \kl{nodes} of $\Structure$.
It can be considered as the
\kl[diamond]{diamond operator}
corresponding to the relation
$\NodeSet{\Structure} \times \NodeSet{\Structure}$,
i.e., the \kl{edge relation}
of the complete \kl{digraph} over $\NodeSet{\Structure}$.
To facilitate certain descriptions,
we shall sometimes treat
$\bdm[\RelSymbol]$ and $\gdm$
as special cases of $\dm[\RelSymbol]$,
assuming that they are implicitly associated
with the reserved \kl{relation symbols}
\Intro*{$\invr{\RelSymbol}$} and \Intro*{$\GlobalRelSymbol$},
respectively.
These \kl[relation symbol]{symbols}
do not belong to~$\SymbolSet$,
and therefore cannot be \kl{interpreted}
by any \kl{structure}.

Allowing a bit of syntactic sugar,
we will make liberal use of
the remaining operators of predicate logic,
i.e., \Intro*{$\AND$}, \Intro*{$\IMP$}, \Intro*{$\IFF$}, \Intro*{$\Forall{}$},
and we may leave out some parentheses,
assuming that~$\OR$ and $\AND$
take precedence over~$\IMP$ and~$\IFF$.
Furthermore,
we define the abbreviations
\Phantomintro{\True}
\Phantomintro{\False}
\Phantomintro{\bx}
\begin{align*}
  &\reintro*{\True} \defeq \PosIs{\PosSymbol},
  \quad
  \reintro*{\False} \defeq \NOT\PosIs{\PosSymbol},
  \quad \text{and} \\
  &\reintro*{\bx[\RelSymbol](\Formula_1,…,\Formula_\Arity)}
  \,\defeq\, \NOT\dm[\RelSymbol](\NOT\Formula_1,…,\NOT\Formula_\Arity).
\end{align*}
Note that it makes sense to define
$\True$ (“true”) as $\PosIs{\PosSymbol}$,
since by definition,
the atomic \kl{formula} $\PosIs{\PosSymbol}$
is always \kl{satisfied} at the point of evaluation.
Also,
the second line remains applicable
if one substitutes $\invr{\RelSymbol}$ or $\GlobalRelSymbol$ for $\RelSymbol$.
The resulting
operators~$\bx[\RelSymbol]$,~\Intro*{$\bbx[\RelSymbol]$} and~\Intro*{$\gbx$}
provide universal quantification
and are called \Intro{boxes}
(using the same attributes as for \kl{diamonds}).
\kl{Diamonds} and \kl{boxes} are collectively referred to as
\Intro{modalities} or \reintro{modal operators}.
In case we restrict ourselves to
\kl{structures} that only have a single relation,
we may omit the \kl{relation symbol}~$\RelSymbol$,
and just use empty \kl{modalities} such as~$\dm$.
Similarly,
if the \kl{relation symbols} involved are indexed,
like $\RelSymbol_1,…,\RelSymbol_\RelCount$,
we associate them with \kl{modalities}
of the form~$\dm[i]$, for $1≤i≤\RelCount$.

Let us now turn to the particular classes of \kl{formulas}
considered in this thesis,
which are specified in \cref{tab:languages}.
The languages of
\Intro{first-order logic}~($\FOL$),
\Intro{existential monadic second-order logic}~($\EMSOL$), and
\Intro{monadic second-order logic}~($\MSOL$)
are defined in the usual way.
When evaluated on some \kl{structure}~$\Structure$,
their atomic \kl{formulas} allow to compare \kl{nodes}
assigned to \kl{node symbols} in $\signature(\Structure)$
with respect to
the equality relation
and any other relation
assigned to a \kl{relation symbol} in $\signature(\Structure)$.
In~$\FOL$,
we can assign new \kl{interpretations} to \kl{node symbols}
by means of existential and universal \Intro[node quantifiers]{quantification over nodes}.
In $\EMSOL$,
we may additionally \kl[interpret]{reinterpret} \kl{set symbols}
using existential \Intro[set quantifiers]{quantifiers over sets} of \kl{nodes},
and in $\MSOL$,
we can also use the corresponding universal \kl[set quantifiers]{quantifiers}.

The remaining classes of \kl{formulas}
can all be qualified as modal languages,
insofar as they include \kl{modal operators}
instead of the classical \kl[node quantifiers]{first-order quantifiers}.
By performing this change of paradigm,
we lose our “bird's-eye view” of the \kl{structure}~$\Structure$,
and now see it from the local point of view
of the \kl{node}~$\inp{\PosSymbol}{\Structure}$.
(For this,
$\Structure$ obviously has to be \kl[pointed structure]{pointed}.)
In basic \Intro{modal logic}~($\ML$),
a \kl{node} “sees” only its \kl{outgoing neighbors},
and thus our domain of quantification
is restricted to those \kl[outgoing neighbors]{neighbors}.
Furthermore,
the \kl{position symbol}~$\PosSymbol$
is the only \kl{node symbol}
whose \kl{interpretation} can be changed by a \kl{modal operator}.
\Intro{Backward modal logic}~($\bML$)
is the variant of~$\ML$
where a \kl{node} “sees” its \kl{incoming neighbors}
instead of its \kl{outgoing neighbors},
whereas \Intro{bidirectional modal logic}~($\dML$)
is the combination
where a \kl{node} “sees” both
\kl[incoming neighbor]{incoming} and \kl{outgoing neighbors}.
We will also look at
\Intro{modal logic with global modalities}~($\MLg$),
where we regain the possibility to quantify
over the entire \kl{domain} of the \kl{structure},
but are still confined
to remapping only the \kl{position symbol}~$\PosSymbol$.
The backward and bidirectional variants $\bMLg$ and $\dMLg$
are defined analogously.
Finally,
we also consider crossover versions of \kl{modal logic}
that are enriched with the \kl{set quantifiers} of $\MSOL$.
Given a class of \kl{formulas}~$\FormulaSet$,
we denote by $\MSO(\FormulaSet)$ the corresponding enriched class.
For instance,
the \kl{formulas} of $\MSO(\ML)$ are generated by the grammar
\begin{equation*}
  \Formula \Coloneqq \PosIs{\NodeSymbol}
                \mid \PosIn{\SetSymbol}
                \mid \NOT\Formula
                \mid (\Formula_1\OR\Formula_2)
                \mid \dm[\RelSymbol](\Formula_1,…,\Formula_\Arity)
                \mid \Exists{\SetSymbol} \Formula,
\end{equation*}
where
$\NodeSymbol∈\NodeSymbolSet$,\:
$\SetSymbol∈\SetSymbolSet$, and
$\RelSymbol∈\RelSymbolSet{\Arity+1}$.
Note that by this notation,
$\MSOL = \MSO(\FOL)$.

For any class of \kl{formulas} $\FormulaSet$,
we shall refer to its members as
\Intro[class-formula]{\mbox{$\FormulaSet$-formulas}}.
Given a \kl[class-formula]{$\FormulaSet$-formula}~$\Formula$,
a class of \kl{structures} $\StructClass$ (e.g., $\DIGRAPH$),
and a \kl{structure}~$\Structure$,
we use the semantic bracket notations
$\semf{\Formula}[\StructClass]$ and
$\lsemf{\Formula}[\Structure]$
to denote the \kl{structure language}
\Intro{defined}\, by~$\Formula$ over~$\StructClass$, and
the set of \kl{nodes} of~$\Structure$ at which~$\Formula$ holds.
More formally,
\Phantomintro{\semf}
\Phantomintro{\lsemf}
\newcommand*{\StrutSemPhiC}{\vphantom{\semf{\Formula}[\StructClass]}}
\begin{align*}
  \reintro*{\swl{\semf{\Formula}[\StructClass]}{\lsemf{\Formula}[\Structure]}}
    &\defeq
     \lrsetbuilder{\Structure∈\StructClass}{\Structure\Models\Formula \StrutSemPhiC},
     \quad \text{and} \\
  \reintro*{\lsemf{\Formula}[\Structure]}
    &\defeq
     \lrsetbuilder{\Node \in \NodeSet{\Structure}}
                  {\ver{\Structure}{\PosSymbol}{\Node} \Models \Formula}.
\end{align*}
Furthermore,
$\semF{\FormulaSet}[\StructClass]$
denotes the family of \kl{structure languages}
that are \Intro{definable} in $\FormulaSet$
(or \reintro[definable]{$\FormulaSet$-definable}\:\!)
over $\StructClass$, i.e.,
\Phantomintro{\semF}
\begin{equation*}
  \reintro*{\semF{\FormulaSet}[\StructClass]}
  \defeq
  \lrsetbuilder{\semf{\Formula}[\StructClass]}{\Formula∈\FormulaSet}.
\end{equation*}
If $\StructClass$ is equal to the set of all \kl{structures},
then we do not have to specify it explicitly as a subscript;
that is, we may simply write
\reintro*{$\semf{\Formula}$} and \reintro*{$\semF{\FormulaSet}$}
instead of
$\semf{\Formula}[\StructClass]$ and $\semF{\FormulaSet}[\StructClass]$.
Similarly,
we use
\Phantomintro{\eclf}
\begin{equation*}
  \reintro*{\eclf{\Formula[1]}[\StructClass]}
  \defeq
  \lrsetbuilder{\Formula[2]}
               {\semf{\Formula[2]}[\StructClass]=\semf{\Formula[1]}[\StructClass]}
\end{equation*}
for the \kl[device equivalence]{equivalence} class of $\Formula[1]$ over $\StructClass$,\,
and
\Phantomintro{\eclF}
\begin{equation*}
  \reintro*{\eclF{\FormulaSet}[\StructClass]}
  \defeq \bigcup_{\Formula∈\FormulaSet}\eclf{\Formula}[\StructClass]
\end{equation*}
for the set of all \kl{formulas}
that are \kl[device equivalent]{equivalent} over $\StructClass$
to some \kl{formula} in $\FormulaSet$.
Again,
we may drop the subscript
if we do not want to restrict to a particular class of \kl{structures}.

\section{Example formulas}
\label{sec:example-formulas}

In order to illustrate the syntax
introduced in the previous \lcnamecref{sec:logics},
we now look at two simple examples.

The first is a great classic
that is often used to show
how a widely known \kl[digraph language]{graph property}
can be expressed in $\MSOL$ without too much effort.

\begin{example}[3-Colorability] \label{ex:MSO-3-colorable}
  The following $\EMSOL$-\kl{formula} \kl{defines} the
  \kl[digraph language]{language}
  of \kl[colorable]{3-colorable} \kl{digraphs} over~$\DIGRAPH$.
  \begin{align*}
    \dphi[color]{3} \defeq
    \Exists{\SetVariable_1,\SetVariable_2,\SetVariable_3}
    \biggl( \,
      & \Forall{\NodeVariable[1]}
        \Bigl( \,
          \bigl(
            \InSet{\SetVariable_1}{\NodeVariable[1]} \OR
            \InSet{\SetVariable_2}{\NodeVariable[1]} \OR
            \InSet{\SetVariable_3}{\NodeVariable[1]}
          \bigr)
          \AND \vphantom{\biggl(} \\[-2.3ex]
      & \hspace{10.6ex}
          \NOT
          \bigl(
            \InSet{\SetVariable_1}{\NodeVariable[1]} \AND
            \InSet{\SetVariable_2}{\NodeVariable[1]}
          \bigr) \AND \vphantom{\biggl(} \\[-2.3ex]
      & \hspace{10.6ex}
          \NOT
          \bigl(
            \InSet{\SetVariable_1}{\NodeVariable[1]} \AND
            \InSet{\SetVariable_3}{\NodeVariable[1]}
          \bigr) \AND \vphantom{\biggl(} \\[-1.3ex]
      & \hspace{10.6ex}
          \NOT
          \bigl(
            \InSet{\SetVariable_2}{\NodeVariable[1]} \AND
            \InSet{\SetVariable_3}{\NodeVariable[1]}
          \bigr) \hspace{4.2ex}
        \Big) \AND \vphantom{\Big)} \\[-.4ex]
      & \Forall{\NodeVariable[1],\NodeVariable[2]}
        \Big(
          \InRel{\RelSymbol}{\NodeVariable[1],\NodeVariable[2]}
          \:\IMP\:
          \NOT
          \bigl(
            \InSet{\SetVariable_1}{\NodeVariable[1]} \AND
            \InSet{\SetVariable_1}{\NodeVariable[2]}
          \bigr) \AND \vphantom{\biggl(} \\[-2.3ex]
      & \hspace{15.7ex}
          \NOT
          \bigl(
            \InSet{\SetVariable_2}{\NodeVariable[1]} \AND
            \InSet{\SetVariable_2}{\NodeVariable[2]}
          \bigr) \AND \vphantom{\biggl(} \\[-2.3ex]
      & \hspace{15.7ex}
          \NOT
          \bigl(
            \InSet{\SetVariable_3}{\NodeVariable[1]} \AND
            \InSet{\SetVariable_3}{\NodeVariable[2]}
          \bigr) \hspace{4.2ex}
        \Bigr) \;
    \biggr)
  \end{align*}
  The existentially quantified \kl{set variables}
  $\SetVariable_1,\SetVariable_2,\SetVariable_3 \in \SetSymbolSet$
  represent the three possible colors.
  In the first four lines,
  we specify that the sets assigned
  to these \kl[set variables]{variables}
  form a partition of the set of \kl{nodes}
  (possibly with empty components).
  The remaining three lines constitute
  the actual definition of a valid \kl{coloring}:
  no two \kl{adjacent} \kl{nodes} share the same color,
  which means that \kl{adjacent} \kl{nodes} are in different sets.
\end{example}

Our second example is equally simple,
but less glamorous
because it illustrates a technical issue
that will concern us in \cref{ch:alternation},
where we shall work with $\MSO(\MLg)$ and some variants thereof.
As we do not allow \kl[node quantifiers]{first-order quantification}
in \kl{modal logic} with \kl{set quantifiers},
some properties that seem very natural in $\FOL$ (and thus~$\MSOL$)
become rather cumbersome to express.
Nevertheless,
translation from $\FOL$ to $\MSO(\MLg)$ is always possible
because we can simulate \kl[node quantifiers]{first-order quantifiers}
by \kl{set quantifiers} relativized to singletons,
which, by extension, also entails
the \kl[device equivalence]{equivalence} of $\MSOL$ and $\MSO(\MLg)$.
\Cref{ex:uniqueness} presents the basic construction
that allows us to do this.
We will refer to it several times in \cref{ch:alternation}.

\begin{example}[Uniqueness] \label{ex:uniqueness}
Consider the following \kl{formula} schema,
where $\SetVariable∈\SetSymbolSet$,\, $\RelSymbol∈\RelSymbolSet{2}$,
and $\Formula$ can be any \kl[class-formula]{$\dMLg$-formula}:
\Phantomintro{\seeone}
\begin{equation*}
  \reintro*{\seeone[\RelSymbol](\Formula)} \defeq
  \dm[\RelSymbol] \Formula \,\AND\;
  \Forall{\SetVariable} \bigl( \dm[\RelSymbol](\Formula {\,\AND\,} \PosIn{\SetVariable}) \IMP \bx[\RelSymbol](\Formula {\,\IMP\,} \PosIn{\SetVariable}) \bigr).
\end{equation*}
When evaluated on a \kl{pointed structure} $\Structure$
whose \kl{signature} includes $\set{\PosSymbol,\RelSymbol}∪\free(\Formula)$,
the \kl{formula} $\seeone[\RelSymbol](\Formula)$ states that there is exactly one \kl{node} $\Node∈\NodeSet{\Structure}$
reachable from $\inp{\PosSymbol}{\Structure}$ through an $\inp{\RelSymbol}{\Structure}$-\kl{edge},
such that $\Formula$ is \kl{satisfied} at $\Node$ (i.e., by the \kl{structure} $\ver{\Structure}{\PosSymbol}{\Node}$).
In the context of $1$-relational \kl{digraphs},
we may use the shorthand \reintro*{$\seeone(\Formula)$} to invoke this schema.
Using the same construction with \kl{global modalities},
we also~define
\Phantomintro{\totone}
\begin{equation*}
  \reintro*{\totone(\Formula)} \defeq \seeone[\GlobalRelSymbol](\Formula),
\end{equation*}
which states that
there is precisely one \kl{node} in the entire \kl{structure} $\Structure$
at which $\Formula$ is \kl{satisfied}.
Here,
$\Structure$ does not necessarily have to be \kl[pointed structure]{pointed},
and, of course, $\signature(\Structure)$ does not contain $\GlobalRelSymbol$
(since it is the \kl{symbol} reserved for the total symmetric relation).
\end{example}

Anticipating the notation of \cref{sec:alternation-preliminaries},
the \kl{formulas} obtained by the construction in \cref{ex:uniqueness}
can be classified as~\kl[class-formula]{$\eclF{\PiMSO{1}(\FormulaSet)}$-formulas},
where $\FormulaSet∈\set{\ML,\dML,\MLg,\dMLg}$
depends on the specific \kl{modalities} that occur in $\Formula$.

\section{Distributed automata}
\label{sec:distributed-automata}

We conclude this preliminary chapter
by introducing our primary objects of interest.
Simply put,
a \kl{distributed automaton}
is a deterministic finite-state machine~$\Automaton$
that reads sets of \kl{states} instead of the usual alphabetic symbols.
To \kl{run} $\Automaton$ on a $1$-relational \kl{digraph}~$\Digraph$,
we place a separate copy of the \kl[distributed automaton]{machine}
on every \kl{node}~$\Node$ of $\Digraph$,
\kl[initialization function]{initialize} it to a \kl{state}
that may depend on $\Node$'s \kl{label} $\Labeling{\Digraph}(\Node)$,
and then let all the \kl{nodes} communicate
in an infinite sequence of synchronous rounds.
In every round,
each \kl{node} computes its next \kl{state}
as a \kl[transition function]{function}
of its own current \kl{state}
and the set of \kl{states} of its \kl{incoming neighbors}.
Intuitively,
\kl{node}~$\Node$ broadcasts its current \kl{state}~$\State$
to every \kl{outgoing neighbor},
while at the same time collecting
the \kl{states} received from its \kl{incoming neighbors}
into a set~$\NeighborSet$;
the successor \kl{state} of~$\State$ is then computed
as a \kl[transition function]{function}
of~$\State$ and~$\NeighborSet$.
Since~$\NeighborSet$ is a set
(as opposed to a multiset or a vector),
$\Node$ cannot distinguish between
two \kl{incoming neighbors} that share the same \kl{state}.
Now,
acting as a semi-decider,
the \kl[distributed automaton]{machine} at \kl{node}~$\Node$ \kl{accepts}
precisely if it visits an \kl{accepting state} at some point in time.
Either way,
all \kl[distributed automaton]{machines} of the network
keep running and communicating forever.
This is because even if a \kl{node} has already \kl{accepted},
it may still obtain new information
that affects the \kl{acceptance behavior} of its \kl{outgoing neighbors}.

Let us now define the notion of \kl{distributed automaton} more formally,
and generalize it to \kl{digraphs}
with an arbitrary number of \kl{edge relations}.

\begin{definition}[Distributed automaton]
  \label{def:distributed-automaton}
  A (deterministic, \Intro{nonlocal}\,) \Intro{distributed automaton} (\Intro*{$\DA$})
  over $\Alphabet$-\kl{labeled}, $\RelCount$-relational \kl{digraphs}
  is a tuple
  $\Automaton = \tuple{\StateSet,\InitFunc,\TransFunc,\AcceptSet}$,
  where
  $\StateSet$
  is a finite nonempty set of \Intro{states},
  $\InitFunc \colon \Alphabet \to \StateSet$
  is an \Intro{initialization function},
  $\TransFunc \colon \StateSet \times (\powerset{\StateSet})^\RelCount \to \StateSet$
  is a \Intro{transition function}, and
  $\AcceptSet \subseteq \StateSet$
  is a set of \Intro{accepting states}.
\end{definition}

Let $\Digraph$ be a $\Alphabet$-\kl{labeled}, $\RelCount$-relational \kl{digraph}.
The \Intro{run} of $\Automaton$ on $\Digraph$ is an infinite sequence
$\Run = \tuple{\Run_0, \Run_1, \Run_2, \dots}$
of maps
$\Run_{\Time} \colon \NodeSet{\Digraph} \to \StateSet$,
called \Intro{configurations},
which are defined inductively as follows,
for $\Time \in \Natural$ and $\Node[2] \in \NodeSet{\Digraph}$:
\begin{equation*}
  \Run_0(\Node[2]) = \InitFunc(\Labeling{\Digraph}(\Node[2]))
  \qquad\text{and}\qquad
  \Run_{\Time+1}(\Node[2]) =
  \TransFunc
      \Bigl(\Run_{\Time}(\Node[2]),\,
            \bigtuple{\setbuilder{\Run_{\Time}(\Node[1])}
                                 {\edge{\Node[1]}{\Node[2]} \in \EdgeSet[i]{\Digraph}}
                     }_{1 \leq i \leq \RelCount}
      \Bigr).
\end{equation*}
For $\Node \in \NodeSet{\Digraph}$,
the automaton~$\Automaton$ \Intro{accepts}
the \kl{pointed digraph} $\pver{\Digraph}{\Node}$
if $\Node$ visits an \kl{accepting state} at some point
in the \kl{run} $\Run$ of $\Automaton$ on $\Digraph$,
i.e., if there exists $\Time \in \Natural$
such that $\Run_{\Time}(\Node) \in \AcceptSet$.
The \kl{pointed-digraph language} of~$\Automaton$,
or \kl{pointed-digraph language} \Intro{recognized}\, by~$\Automaton$,
is the set of all \kl{pointed digraphs}
that are \kl{accepted} by~$\Automaton$.
We denote this \kl{language}
by~\Intro*{$\sema{\Automaton}[\pDIGRAPH[\Alphabet][\RelCount]]$},
in analogy to our notation for logical \kl{formulas}.
Similarly,
given a class of \kl[distributed automata]{automata}~$\AutomatonSet$,
we write
\Intro*{$\semA{\AutomatonSet}[\pDIGRAPH[\Alphabet][\RelCount]]$}
for the class of \kl{pointed-digraph languages}
over $\pDIGRAPH[\Alphabet][\RelCount]$
that are \kl{recognized} by some member of $\AutomatonSet$;
we call them \Intro[recognizable]{$\AutomatonSet$-recognizable}.

As usual, two devices
(i.e., \kl[distributed automata]{automata} or \kl{formulas})
are \Intro[device equivalent]{equivalent} if they specify
(i.e., \kl{recognize} or \kl{define})
the same \kl{language}.

In distributed computing,
one often considers algorithms
that run in a constant number of synchronous rounds.
They are known as \emph{local algorithms}
(see, e.g., \cite{DBLP:journals/csur/Suomela13}).
Here,
we use the same terminology for \kl{distributed automata}
and give a syntactic definition of \kl{locality}
in terms of \kl{state} diagrams.
Basically,
a \kl{distributed automaton} is \kl{local}
if its \kl{state} diagram does not contain any directed cycles,
except for self-loops on \kl{sink} \kl{states}.
This is equivalent to requiring
that all \kl{nodes} stop changing their \kl{state}
after a constant number of rounds.
\begin{definition}[Local distributed automaton]
  \label{def:local-distributed-automaton}
  A \Intro{local distributed automaton} (\Intro*{$\LDA$})
  over $\RelCount$-relational \kl{digraphs}
  is a \kl{distributed automaton}
  $\Automaton = \tuple{\StateSet,\InitFunc,\TransFunc,\AcceptSet}$
  whose \kl{state} diagram satisfies the following two conditions:
  \begin{enumerate}
  \item \label{itm:quasi-acyclic}
    The only directed cycles are self-loops.
    That is,
    for every sequence
    $\State_1, \State_2, \dots, \State_n$ of \kl{states} in~$\StateSet$
    such that
    $\State_1 = \State_n$ and
    $\TransFunc(\State_i,\vec{\NeighborSet}_i) = \State_{i+1}$
    for some $\vec{\NeighborSet}_i \in (\powerset{\StateSet})^\RelCount$,
    it must be that all \kl{states} of the sequence are the same.
  \item \label{itm:looping-sink-states}
    Self-loops occur only on \kl{sink} \kl{states}.
    That is,
    for every \kl{state} $\State \in \StateSet$,
    if $\TransFunc(\State,\vec{\NeighborSet}) = \State$
    for some $\vec{\NeighborSet} \in (\powerset{\StateSet})^\RelCount$,
    then the same must hold
    for all $\vec{\NeighborSet} \in (\powerset{\StateSet})^\RelCount$.
    \qedhere
  \end{enumerate}
\end{definition}

Deviating only in nonessential details from the original presentation
given by Hella~et~al.\
in~\cite{DBLP:conf/podc/HellaJKLLLSV12,DBLP:journals/dc/HellaJKLLLSV15},
we can now restate their logical characterization of the class~$\SB[1]$
using the terminology introduced above.

\begin{theorem}[{
    $\semA{\LDA}[\pDIGRAPH[\Alphabet][\RelCount]] =
     \semF{\bML}[\pDIGRAPH[\Alphabet][\RelCount]]$;
    \cite{DBLP:conf/podc/HellaJKLLLSV12,DBLP:journals/dc/HellaJKLLLSV15}
  }]
  \label{thm:LDA-bML}
  \strut
  A \kl{pointed-digraph language} is \kl{recognizable}
  by a \kl{local distributed automaton}
  if and only if
  it is \kl{definable} by a \kl{formula} of \kl{backward modal logic}.
  There are effective translations in both directions.
\end{theorem}

The notion of \kl{locality} plays a major role in \cref{ch:local},
where we extend $\LDA$'s
with the capacity of alternation and a global acceptance condition.
Our extension leaves the realm of basic $\DA$'s,
since we show that it is \kl[device equivalent]{equivalent} to $\MSOL$,
which by~\mbox{\cite[Prp.~6~\&~8]{DBLP:conf/csl/Kuusisto13}}
is incomparable with $\DA$'s.

On the other hand,
in \cref{ch:nonlocal,ch:emptiness},
we consider a simpler extension of $\LDA$'s,
which can be seen as a natural intermediate stage
between $\LDA$'s and $\DA$'s.
Given the above definition of \kl{local automata},
a rather obvious generalization is
to allow self-loops on all \kl{states},
even if they are not \kl{sink} \kl{states};
we call this property \kl{quasi-acyclic}.
More formally,
a \Intro{quasi-acyclic distributed automaton} (\Intro*{$\QDA$})
is a $\DA$ that satisfies
\cref{itm:quasi-acyclic} of \cref{def:local-distributed-automaton},
but not necessarily \cref{itm:looping-sink-states}.
An example of such an \kl[quasi-acyclic automaton]{automaton}
will be provided in
\cref{sec:preliminaries-nonlocal}
(\cref{fig:automaton} on page \pageref{fig:automaton}).

%% file: tex/local.tex
\mychapterpreamble{%
  \nameCref{ch:local} based on the
  conference paper~\cite{DBLP:conf/lics/Reiter15}.}

\chapter{Alternating Local Automata}
\label{ch:local}

In this \lcnamecref{ch:local},
we transfer the well-established notion of alternating automaton
to the setting of \kl{local distributed automata}
and combine it with a \kl{global acceptance condition}.
This gives rise to a new class of graph automata
that \kl[global recognize]{recognize} precisely the
\kl[digraph languages]{languages} of finite \kl{digraphs}
\kl{definable} in $\MSOL$.
By restricting transitions to be nondeterministic or deterministic,
we also obtain two strictly weaker variants
for which the emptiness problem is decidable.

\section{Informal description}
\label{sec:dga-preview}

We start with an informal description of the adjustments
that we make to the basic model of \kl{local automata}
(see \cref{sec:distributed-automata}).
Formal definitions will follow in \cref{sec:dga-definitions}.

The term
“\Intro{local distributed automaton with global acceptance condition}”
(\Intro*{$\LDAg$})
will be used to refer collectively
to the deterministic, nondeterministic and alternating
versions of our model.
Let us first mention the properties they have in common.

\begin{description}[wide,labelindent=0ex]
\item[Levels of states.]
  As for basic \kl{local automata},
  the number of communication rounds is limited by a constant.
  To make this explicit
  and to simplify the subsequent definition of alternation,
  we associate a number, called \kl[state level]{level}, with every \kl[alternating state]{state}.
  In most cases,
  this number indicates the round in which the \kl[alternating state]{state} may occur.
  We require that potentially \kl[alternating initialization function]{initial} \kl[alternating states]{states} are at \kl[state level]{level} $0$, and outgoing
  transitions from \kl[alternating states]{states} at \kl[state level]{level} $i$ go to \kl[alternating states]{states} at \kl[state level]{level}
  $i+1$. There is an exception, however: the \kl[alternating states]{states} at the highest
  \kl[state level]{level}, called the \kl{permanent states}, can also be \kl[alternating initialization function]{initial} \kl[alternating states]{states}
  and can have incoming transitions from any \kl[state level]{level}. Moreover, all their
  outgoing transitions are self-loops. The idea is that, once a \kl{node} has
  reached a \kl{permanent state}, it terminates its local computation, and
  waits for the other \kl{nodes} in the \kl{digraph} to terminate too.
\item[Global acceptance.]
  Unlike for basic \kl{local automata},
  the considered input is a \kl{digraph},
  not a \kl{pointed digraph},
  and consequently the \kl{language} \kl[global recognized]{recognized} by an $\LDAg$
  is a \kl{digraph language}.
  For this reason,
  once all the \kl{nodes} have reached a \kl{permanent state},
  the $\LDAg$ ceases to operate as a distributed algorithm,
  and collects all the reached \kl{permanent states} into a set $F$.
  This set is the sole \kl[global acceptance]{acceptance} criterion:
  if $F$ is part of the $\LDAg$'s \kl{accepting sets},
  then the input \kl{digraph} is \kl[global accepted]{accepted},
  otherwise it is rejected.
  In particular,
  the \kl[\ALDAg]{automaton} cannot detect
  whether several \kl{nodes} have reached the same \kl{permanent state}.
  This limitation is motivated by the desire
  to have a simple finite representation of $\LDAg$'s;
  in other words,
  the same reason why we do not allow \kl{nodes}
  to distinguish between several \kl[incoming neighbors]{neighbors}
  that are in the same \kl[alternating state]{state}.
\end{description}

As an introductory example,
we translate the \kl[class-formula]{$\MSOL$-formula}
$\dphi[color]{3}$ from \cref{ex:MSO-3-colorable} 
in \cref{sec:example-formulas}
to the setting of $\LDAg$'s.

\begin{example}[3-colorability] \label{ex:ADGA-3-colorable}
  \Cref{fig:ADGA-3-colorable} shows the \kl[alternating state]{state} diagram of a simple
  \kl[\NLDAg]{nondeterministic} $\LDAg$ $\dA[color]{3}$\!. The \kl[alternating states]{states} are arranged in
  columns corresponding to their \kl[state levels]{levels}, ascending from left to right.
  $\dA[color]{3}$ expects an \kl[labeled]{unlabeled} \kl{digraph} as
  input, and \kl[global accepts]{accepts} it if and only if it is \kl[colorable]{3-colorable}.
  The \kl[\ALDAg]{automaton} proceeds
  as follows: All \kl{nodes} of the input \kl{digraph} are \kl[alternating initialization function]{initialized} to the
  \kl[alternating state]{state} $\q{ini}$. In the first round, each \kl{node} nondeterministically
  chooses to go to one of the \kl[alternating states]{states} $\State_1$, $\State_2$ and $\State_3$, which
  represent the three possible colors. Then, in the second round, the
  \kl{nodes} verify locally that the chosen \kl{coloring} is valid. If the set
  received from their \kl{incoming neighborhood} (only one, since there is
  only a single \kl{edge relation}) contains their own \kl[alternating state]{state}, they go to
  $\q{no}$, otherwise to $\q{yes}$. The \kl[\ALDAg]{automaton} then \kl[global accepts]{accepts} the
  input \kl{digraph} if and only if all the \kl{nodes} are in $\q{yes}$, i.e., $\{\q{yes}\}$
  is its only \kl{accepting set}. This is indicated by the bar to the right
  of the \kl[alternating state]{state} diagram. We shall refer to such a representation of
  sets using bars as \emph{barcode}.
\end{example}

\begin{figure}
  \centering
  \input{fig/ADGA-3-colorable.tikz}
  \caption{$\dA[color]{3}$\!, a nondeterministic $\LDAg$ over
    \kl[labeled]{unlabeled}, $1$-relational \kl{digraphs} whose \kl{digraph language}
    consists of the \kl[colorable]{3-colorable} \kl{digraphs}.}
  \label{fig:ADGA-3-colorable}
\end{figure}
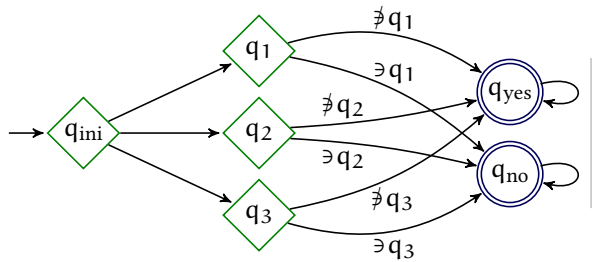

The last property,
which applies only to our most powerful version of $\LDAg$'s,
is \emph{alternation}, a generalization of nondeterminism
introduced by Chandra, Kozen and Stockmeyer
in~\cite{DBLP:journals/jacm/ChandraKS81}
(there, for Turing machines and other types of \kl[pointed dipath]{word} automata).

\begin{description}[resume*]
\item[Alternation.]
  In addition to being able to
  nondeterministically choose between different transitions, \kl{nodes} can
  also explore several choices in parallel. To this end, the
  \kl{nonpermanent states} of an \kl[\ALDAg]{alternating} $\LDAg$ ($\ALDAg$) are partitioned into
  two types, \kl[existential states]{existential} and \kl[universal states]{universal}, such that \kl[alternating states]{states}
  on the same \kl[state level]{level} are of the same type. If, in a given round, the
  \kl{nodes} are in \kl{existential states}, then they nondeterministically choose
  a single \kl[alternating state]{state} to go to in the next round, as described above. In
  contrast, if they are in \kl{universal states}, then the \kl[alternating run]{run} of the $\ALDAg$ is
  split into several parallel branches, called universal branches, one
  for each possible combination of choices of the \kl{nodes}. This procedure
  of splitting is repeated recursively for each round in which the \kl{nodes}
  are in \kl{universal states}. The $\ALDAg$ then \kl[global accepts]{accepts} the input \kl{digraph} if and only if
  its \kl[global acceptance condition]{acceptance condition} is satisfied in every universal branch of the
  \kl[alternating run]{run}.
\end{description}

\begin{example}[Non-3-colorability]
  To illustrate the notion of universal branching, consider the $\ALDAg$
  $\dcA[color]{3}$ shown in \cref{fig:ADGA-not-3-colorable}. It is a
  complement \kl[\ALDAg]{automaton} of $\dA[color]{3}$ from
  Example~\ref{ex:ADGA-3-colorable}, i.e., it \kl[global accepts]{accepts} precisely those
  (\kl[labeled]{unlabeled}) \kl{digraphs} that are \emph{not}
  \kl[colorable]{3-colorable}. \kl[alternating states]{States} represented as boxes are \kl[universal states]{universal}
  (whereas the diamonds in \cref{fig:ADGA-3-colorable} stand for
  \kl{existential states}). Given an input \kl{digraph} with $n$ \kl{nodes}, $\dcA[color]{3}$
  proceeds as follows: All \kl{nodes} are \kl[alternating initialization function]{initialized} to $\q{ini}$. In the
  first round, the \kl[alternating run]{run} is split into $3^n$ universal branches, each of
  which corresponds to one possible outcome of the first round of
  $\dA[color]{3}$ running on the same input \kl{digraph}. Then, in the second
  round, in each of the $3^n$ universal branches, the \kl{nodes} check
  whether the \kl{coloring} chosen in that branch is valid. As indicated by
  the barcode, the \kl[global acceptance condition]{acceptance condition} of $\dcA[color]{3}$ is
  satisfied if and only if at least one \kl{node} is in \kl[alternating state]{state} $\q{no}$, i.e., the
  \kl{accepting sets} are $\{\q{no}\}$ and $\{\q{yes},\q{no}\}$. Hence, the
  \kl[\ALDAg]{automaton} \kl[global accepts]{accepts} the input \kl{digraph} if and only if no valid \kl{coloring} was found
  in any universal branch. Note that we could also have chosen to make
  the \kl[alternating states]{states} $\State_1$, $\State_2$ and $\State_3$ \kl[existential state]{existential},
  since their outgoing
  transitions are deterministic. Regardless of their type, there is no
  branching in the second round.
\end{example}

\begin{figure}
  \centering
  \input{fig/ADGA-not-3-colorable.tikz}
  \caption{$\dcA[color]{3}$\!, an alternating $\LDAg$ over
    \kl[labeled]{unlabeled}, $1$-relational \kl{digraphs} whose \kl{digraph language}
    consists of the \kl{digraphs} that are \emph{not} \kl[colorable]{3-colorable}.}
  \label{fig:ADGA-not-3-colorable}
\end{figure}

\section{Formal definitions}
\label{sec:dga-definitions}

We now repeat and clarify the notions from \cref{sec:dga-preview} in a
more formal setting, beginning with our most general definition of
$\LDAg$'s.

\begin{definition}[Alternating local distributed automaton]
  \label{def:adga}
  An \Intro{alternating local distributed automaton with global acceptance condition} (\Intro*{$\ALDAg$}) over
  $\Alphabet$-\kl{labeled}, $\RelCount$-relational \kl{digraphs} is a tuple $\Automaton=\tuple{\StatePartition,\InitFunc,\TransFunc,\AcceptCondition}$, where
  \begin{itemize}
  \item $\StatePartition=\tuple{\ExistentialStateSet,\UniversalStateSet,\PermanentStateSet}$
    is a collection of \kl[alternating states]{states},
    with $\ExistentialStateSet$, $\UniversalStateSet$ and $\PermanentStateSet≠\EmptySet$ being pairwise disjoint finite sets of
    \Intro[existential states]{existential}, \Intro[universal states]{universal} and \Intro{permanent states},
    respectively, also referred to by the notational
    shorthands
    \begin{itemize}[topsep=0ex,itemsep=0ex]
    \item $\swl{\StateSet}{\NonpermanentStateSet} \defeq \ExistentialStateSet∪\UniversalStateSet∪\PermanentStateSet,$\, for the entire
      set of \Intro[alternating states]{states},
    \item $\NonpermanentStateSet \defeq \ExistentialStateSet∪\UniversalStateSet,$\, for the set of
      \Intro{nonpermanent states},
    \end{itemize}
  \item $\InitFunc\colon \Alphabet→\StateSet$ is an \Intro[alternating initialization function]{initialization function},
  \item $\TransFunc\colon \StateSet×(\powerset{\StateSet})^\RelCount→\powerset{\StateSet}$ is a (local)
    \Intro[alternating transition function]{transition function},\, and
  \item $\AcceptCondition⊆\powerset{\PermanentStateSet}$ is a set of \Intro{accepting sets}
    of \kl{permanent states}.
  \end{itemize}
  The functions $\InitFunc$ and $\TransFunc$ must be such that one can unambiguously
  associate with every \kl[alternating state]{state} $\State∈\StateSet$ a \Intro[state level]{level}\, \Intro*{$\level[\Automaton](\State)∈\Natural$}
  satisfying the following conditions:
  \begin{itemize}[itemsep=0ex]
    \setlength\abovedisplayskip{1ex}
    \setlength\belowdisplayskip{1ex}
  \item \kl[alternating states]{States} on the same \kl[state level]{level} are of the same type, i.e., for every
    $i∈\Natural$,
    \begin{equation*}
      \{\State∈\StateSet \mid \level[\Automaton](\State)=i\}∈(\powerset{\ExistentialStateSet}∪\powerset{\UniversalStateSet}∪\powerset{\PermanentStateSet}).
    \end{equation*}
  \item \kl[alternating initialization function]{Initial} \kl[alternating states]{states} are either on the lowest \kl[state level]{level} or \kl[permanent states]{permanent},
    i.e., for every $\State∈\StateSet$,
    \begin{equation*}
      \exists\Letter∈\Alphabet\colon \InitFunc(\Letter)=\State \quad \text{implies} \quad \level[\Automaton](\State)=0 \,\lor\, \State∈\PermanentStateSet.
    \end{equation*}
  \item \kl{Nonpermanent states} without incoming transitions are on the
    lowest \kl[state level]{level}, and transitions between \kl{nonpermanent states} go only
    from one \kl[state level]{level} to the next, i.e., for every $\State[2]∈\NonpermanentStateSet$,
    \begin{equation*}
      \level[\Automaton](\State[2])=
      \begin{cases}
        0 & \text{\parbox[t]{0.5\textwidth}{if for all $\State[1]∈\StateSet$ and $\vec{\NeighborSet}∈(\powerset{\StateSet})^\RelCount$\!, \\
            it holds that $\State[2]∉\TransFunc(\State[1],\vec{\NeighborSet})$,}} \\[3ex]
        i+1 & \text{\parbox[t]{0.5\textwidth}{if there are $\State[1]∈\NonpermanentStateSet$ and
              $\vec{\NeighborSet}∈(\powerset{\StateSet})^\RelCount$ \\ such that $\level[\Automaton](\State[1])=i$ and $\State[2]∈\TransFunc(\State[1],\vec{\NeighborSet})$.}}
      \end{cases}
    \end{equation*}
  \item The \kl{permanent states} are one \kl[state level]{level} higher than the highest
    \kl[nonpermanent states]{nonpermanent ones}, and have only self-loops as outgoing
    transitions, i.e., for every $\State∈\PermanentStateSet$,
    \begin{gather*}
      \level[\Automaton](\State)=
      \begin{cases}
        0                          & \text{if\; $\NonpermanentStateSet=\EmptySet$}, \\
        \max\{\level[\Automaton](\State)\mid \State∈\NonpermanentStateSet\}+1 & \text{otherwise},
      \end{cases} \\[.2ex]
      \TransFunc(\State,\vec{\NeighborSet})=\{\State\} \quad \text{for every $\vec{\NeighborSet}∈(\powerset{\StateSet})^\RelCount$}\!. \qedhere
    \end{gather*}
  \end{itemize}
\end{definition}

For any $\ALDAg$ $\Automaton=\tuple{\StatePartition,\InitFunc,\TransFunc,\AcceptCondition}$, we define its
\Intro[automaton length]{length} $\autolength(\Automaton)$ to be its highest \kl[state level]{level}, that is,
\Intro*{$\autolength(\Automaton) \defeq \max\{\level[\Automaton](\State)\mid \State∈\StateSet\}$}.

Next, we want to give a formal definition of a \kl[alternating run]{run}. For this, we need
the notion of a \kl[alternating configuration]{configuration}, which can be seen as the global state
of an $\ALDAg$.

\begin{definition}[Configuration]
  Consider a \kl{digraph}~$\Digraph$ and an $\ALDAg$ $\Automaton=\tuple{\StatePartition,\InitFunc,\TransFunc,\AcceptCondition}$.
  For any map $\Relabeling \colon \NodeSet{\Digraph} \to \StateSet$,
  we call the
  $\StateSet$-\kl{labeled} \kl[relabeled variant]{variant} $\lver{\Digraph}{\Relabeling}$ of $\Digraph$ a \Intro[alternating configuration]{configuration}
  of $\Automaton$ on $\Digraph$. If every \kl{node} in $\lver{\Digraph}{\Relabeling}$ is \kl{labeled} by a
  \kl{permanent state}, we refer to that \kl[alternating configuration]{configuration} as a \Intro{permanent configuration}.
  Otherwise, if $\lver{\Digraph}{\Relabeling}$ is a \kl[permanent configuration]{nonpermanent configuration}
  whose \kl{nodes} are \kl{labeled} exclusively by \kl[existential state]{existential} and (possibly)
  \kl{permanent states}, we say that $\lver{\Digraph}{\Relabeling}$ is an \Intro{existential configuration}.
  Analogously, the \kl[alternating configuration]{configuration} is \Intro[universal configuration]{universal} if it is
  \kl[permanent configuration]{nonpermanent} and only \kl{labeled} by \kl[universal states]{universal} and (possibly)
  \kl{permanent states}.

  Additionally, we say that a \kl{permanent configuration} $\lver{\Digraph}{\Relabeling}$ is
  \Intro[accepting configuration]{accepting} if the set of \kl[alternating states]{states} occurring in it is \kl[accepting set]{accepting},
  i.e., if $\{\Relabeling(\Node) \mid \Node∈\NodeSet{\Digraph}\}∈\AcceptCondition$. Any other \kl{permanent configuration}
  is called \Intro[rejecting configuration]{rejecting}.
  \kl[permanent configuration]{Nonpermanent configurations} are neither
  \kl[accepting configurations]{accepting} nor \kl[rejecting configurations]{rejecting}.
\end{definition}

The (local) \kl[alternating transition function]{transition function} of an $\ALDAg$ specifies for each \kl[alternating state]{state} a
set of potential successors, for a given family of sets of
\kl[alternating states]{states}. This can be naturally extended to \kl[alternating configurations]{configurations}, which leads
us to the definition of a \kl{global transition function}.

\pagebreak
\begin{definition}[Global transition function]
  The \Intro{global transition function} $\GlobTransFunc$ of an $\ALDAg$
  $\Automaton=\tuple{\StatePartition,\InitFunc,\TransFunc,\AcceptCondition}$ over $\Alphabet$-\kl{labeled}, $\RelCount$-relational \kl{digraphs} assigns to each \kl[alternating configuration]{configuration} $\lver{\Digraph}{\Relabeling[1]}$
  of $\Automaton$ the set of all of its \Intro{successor configurations} $\lver{\Digraph}{\Relabeling[2]}$,
  by combining all possible outcomes of local transitions on $\lver{\Digraph}{\Relabeling[1]}$,
  i.e.,
  \begin{align*}
    \GlobTransFunc \colon \DIGRAPH[\StateSet][\RelCount] &→ \powerset{(\DIGRAPH[\StateSet][\RelCount])} \\
    \lver{\Digraph}{\Relabeling[1]} &↦ \biggl\{\lver{\Digraph}{\Relabeling[2]} \biggm| \bigwedge_{\Node[2]∈\NodeSet{\Digraph}} \Relabeling[2](\Node[2])∈\TransFunc\Bigl(\Relabeling[1](\Node[2]),\,\bigtuple{\{\Relabeling[1](\Node[1])\mid \edge{\Node[1]}{\Node[2]} \in \EdgeSet[i]{\Digraph}\}}_{i \in \range{\RelCount}}\Bigr)\biggr\}.
    \qedhere
  \end{align*}
\end{definition}

We now have everything at hand to formalize the notion of a \kl[alternating run]{run}.

\begin{definition}[Run]
  A \Intro[alternating run]{run} of an $\ALDAg$ $\Automaton=\tuple{\StatePartition,\InitFunc,\TransFunc,\AcceptCondition}$ over $\Alphabet$-\kl{labeled}, $\RelCount$-relational \kl{digraphs} on a given \kl{digraph}
  $\Digraph∈\DIGRAPH[\Alphabet][\RelCount]$ is an acyclic \kl{digraph}~$\Run$ whose
  \kl{nodes} are \kl[alternating configurations]{configurations} of $\Automaton$ on $\Digraph$, such that
  \begin{itemize}
  \item the \Intro{initial configuration} $\lver{\Digraph}{\InitFunc∘\Labeling{\Digraph}} \in \NodeSet{\Run}$ is the only
    \kl{source},\marginnote{Here, the operator $∘$ denotes standard function
      composition, such that \mbox{$(\InitFunc∘\Labeling{\Digraph})(\Node)=\InitFunc(\Labeling{\Digraph}(\Node))$}.}
  \item every \kl[permanent configuration]{nonpermanent configuration}
    $\lver{\Digraph}{\Relabeling[1]} \in \NodeSet{\Run}$ with
    set of \kl{successor configurations}
      $\GlobTransFunc(\lver{\Digraph}{\Relabeling[1]})=\{\lver{\Digraph}{\Relabeling[2]_1},…,\lver{\Digraph}{\Relabeling[2]_m}\}$
    has
    \begin{itemize}[topsep=0ex,itemsep=0ex]
    \item exactly one \kl{outgoing neighbor} $\lver{\Digraph}{\Relabeling[2]_i}∈\GlobTransFunc(\lver{\Digraph}{\Relabeling[1]})$\, if
      $\lver{\Digraph}{\Relabeling[1]}$ is \kl[existential configuration]{existential},
    \item exactly $m$ \kl{outgoing neighbors} $\lver{\Digraph}{\Relabeling[2]_1},…,\lver{\Digraph}{\Relabeling[2]_m}$\, if
      $\lver{\Digraph}{\Relabeling[1]}$ is \kl[universal configuration]{universal},\, and
    \end{itemize}
  \item every \kl{permanent configuration} $\lver{\Digraph}{\Relabeling[1]} \in \NodeSet{\Run}$ is a \kl{sink}.
  \end{itemize}
  The \kl[alternating run]{run}~$\Run$ is \Intro[accepting run]{accepting} if every \kl{permanent configuration}
  $\lver{\Digraph}{\Relabeling[1]} \in \NodeSet{\Run}$ is \kl[accepting configuration]{accepting}.
\end{definition}

An $\ALDAg$ $\Automaton=\tuple{\StatePartition,\InitFunc,\TransFunc,\AcceptCondition}$ over $\Alphabet$-\kl{labeled}, $\RelCount$-relational \kl{digraphs} \Intro[global accepts]{accepts} a given \kl{digraph}
$\Digraph∈\DIGRAPH[\Alphabet][\RelCount]$ if and only if there exists an \kl{accepting run}~$\Run$ of $\Automaton$
on $\Digraph$. The \kl{digraph language} \Intro[global recognized]{recognized} by $\Automaton$ is the set
\Phantomintro{\semag}
\begin{equation*}
  \reintro*{\semag{\Automaton}[\DIGRAPH[\Alphabet][\RelCount]]} \defeq \bigl\{ \Digraph∈\DIGRAPH[\Alphabet][\RelCount] \bigm| \text{$\Automaton$ \kl[global accepts]{accepts} $\Digraph$} \bigr\}.
\end{equation*}
A \kl{digraph language} that is \kl[global recognized]{recognized} by some $\ALDAg$ is called
\Intro[global recognizable]{$\ALDAg$-recognizable}. We denote by \Intro*{$\semAg{\ALDAg}[\DIGRAPH[\Alphabet][\RelCount]]$} the class of
all such \kl{digraph languages}.

The $\ALDAg$ $\Automaton$ is \kl[device equivalent]{equivalent} to some \kl[class-formula]{$\MSOL$-formula} $\Formula$ if
it \kl[global recognizes]{recognizes} precisely the \kl{digraph language} \kl{defined} by $\Formula$ over $\DIGRAPH[\Alphabet][\RelCount]$, i.e., if
$\semag{\Automaton}[\DIGRAPH[\Alphabet][\RelCount]]=\semf{\Formula}[\DIGRAPH[\Alphabet][\RelCount]]$.

We inductively define that a \kl[alternating configuration]{configuration} $\lver{\Digraph}{\Relabeling[1]}∈\DIGRAPH[\StateSet][\RelCount]$ is
\Intro[reachable configuration]{reachable} by $\Automaton$ on $\Digraph$ if either $\lver{\Digraph}{\Relabeling[1]}=\lver{\Digraph}{\InitFunc∘\Labeling{\Digraph}}$, or
$\lver{\Digraph}{\Relabeling[1]}∈\GlobTransFunc(\lver{\Digraph}{\Relabeling[2]})$ for some \kl[alternating configuration]{configuration} $\lver{\Digraph}{\Relabeling[2]}∈\DIGRAPH[\StateSet][\RelCount]$
\kl[reachable configuration]{reachable} by $\Automaton$ on $\Digraph$. In case $\Digraph$ is irrelevant, we simply say
that $\lver{\Digraph}{\Relabeling[1]}$ is \kl[reachable configuration]{reachable} by~$\Automaton$.

The automaton $\Automaton$ is called a
\Intro[nondeterministic local distributed automaton with global acceptance condition]{nondeterministic $\LDAg$} (\Intro*{$\NLDAg$})
if
it has no \kl{universal states}, i.e., if $\UniversalStateSet=\EmptySet$. If additionally every
\kl[alternating configuration]{configuration} $\lver{\Digraph}{\Relabeling[1]}∈\DIGRAPH[\StateSet][\RelCount]$ that is \kl[reachable configuration]{reachable} by $\Automaton$ has
precisely one \kl{successor configuration}, i.e., $\card{\GlobTransFunc(\lver{\Digraph}{\Relabeling[1]})}=1$,
then we refer to $\Automaton$ as a
\Intro[deterministic local distributed automaton with global acceptance condition]{deterministic $\LDAg$} (\Intro*{$\DLDAg$}).
We denote
the classes of $\NLDAg$- and \kl[global recognizable]{$\DLDAg$-recognizable} \kl{digraph languages} by
\reintro*{$\semAg{\NLDAg}[\DIGRAPH[\Alphabet][\RelCount]]$} and \reintro*{$\semAg{\DLDAg}[\DIGRAPH[\Alphabet][\RelCount]]$}.

Let us now illustrate the notion of $\ALDAg$ using a slightly more
involved example.

\begin{example}[Concentric circles]
  Consider the $\ALDAg$ $\sA{centric}=\tuple{\StatePartition,\InitFunc,\TransFunc,\AcceptCondition}$ over $\set{\Letter[1],\Letter[2],\Letter[3]}$-\kl{labeled} \kl{digraphs} represented by the
  \kl[alternating state]{state} diagram in \cref{fig:ADGA-concentric-circles}. Again,
  \kl{existential states} are represented by diamonds, \kl{universal states} by
  boxes, and \kl{permanent states} by double circles. The short arrows
  mapping \kl[labels]{node labels} to \kl[alternating states]{states} indicate the \kl[alternating initialization function]{initialization function}
  $\InitFunc$. For instance, $\InitFunc(\Letter[1])=\qa$. The other arrows specify the
  \kl[alternating transition function]{transition function} $\TransFunc$. A label on such a transition arrow
  indicates a requirement on the set of \kl[alternating states]{states} that a \kl{node} receives
  from its \kl{incoming neighborhood} (only one set, since there is only a
  single \kl{edge relation}). For instance,
  $\TransFunc\bigl(\qb,\tuple{\{\qa,\qc\}}\bigl)=\{\State_{\Letter[2]:1},\State_{\Letter[2]:2}\}$. If there is
  no label, any set is permitted. Finally, as indicated by the barcode
  on the far right, the set of \kl{accepting sets} is
  $\AcceptCondition=\bigl\{\{\State_{\Letter[1]:3},\q{yes}\},\{\State_{\Letter[1]:4},\q{yes}\}\bigr\}$.

  Intuitively, $\sA{centric}$ proceeds as follows: In the first round,
  the $\Letter[1]$-\kl{labeled} \kl{nodes} do nothing but update their \kl[alternating state]{state}, while the
  $\Letter[2]$- and $\Letter[3]$-\kl{labeled} \kl{nodes} verify that the \kl{labels} in their
  \kl{incoming neighborhood} satisfy the condition of a valid graph
  \kl{coloring}. The $\Letter[3]$-\kl{labeled} \kl{nodes} additionally check that they do not
  see any $\Letter[1]$'s, and then directly terminate. Meanwhile, the
  $\Letter[2]$-\kl{labeled} \kl{nodes} nondeterministically choose one of the markers
  $1$ and $2$. In the second round, only the $\Letter[1]$-\kl{labeled} \kl{nodes} are
  busy. They verify that their \kl{incoming neighborhood} consists
  exclusively of $\Letter[2]$-\kl{labeled} \kl{nodes}, and that both of the markers $1$
  and $2$ are present, thus ensuring that they have at least two
  \kl{incoming neighbors}. Then, they simultaneously pick the markers $3$
  and $4$, thereby creating different universal branches, and the \kl[alternating run]{run}
  of the \kl[\ALDAg]{automaton} terminates. Finally, the $\ALDAg$ checks that all the
  \kl{nodes} approve of the \kl{digraph} (meaning that none of them has reached
  the \kl[alternating state]{state} $\q{no}$), and that in each universal branch, precisely
  one of the markers $3$ and $4$ occurs, which implies that there is a
  unique $\Letter[1]$-\kl{labeled} \kl{node}.

  To sum up, the \kl{digraph language} $\semag{\sA{centric}}[\DIGRAPH[\Alphabet][\RelCount]]$ consists of all the
  $\{\Letter[1],\Letter[2],\Letter[3]\}$-\kl{labeled}, \kl{digraphs} such that
  \begin{itemize}
    \item the \kl{labeling} constitutes a valid \kl[coloring]{3-coloring},
    \item there is precisely one $\Letter[1]$-\kl{labeled} \kl{node} $\Node_{\Letter[1]}$, and
    \item $\Node_{\Letter[1]}$ has only $\Letter[2]$-\kl{labeled} \kl{nodes} in its
      \kl{undirected neighborhood}, and at least two \kl{incoming neighbors}.
  \end{itemize}
  The name “$\sA{centric}$” refers to the fact that, in the (weakly)
  \kl{connected} component of $\Node_{\Letter[1]}$, the $\Letter[2]$- and $\Letter[3]$-\kl{labeled} \kl{nodes} form
  con\emph{centric} circles around $\Node_{\Letter[1]}$, i.e., \kl{nodes} at distance~1
  of $\Node_{\Letter[1]}$ are \kl{labeled} with $\Letter[2]$, \kl{nodes} at distance~2 (if existent)
  with $\Letter[3]$, \kl{nodes} at distance~3 (if existent) with $\Letter[2]$, and so
  forth.

  \Cref{fig:graph-labeled-pentagon} shows an example of a \kl{labeled}
  \kl{digraph} that lies in $\semag{\sA{centric}}[\DIGRAPH[\Alphabet][\RelCount]]$. A corresponding \kl{accepting run}
  can be seen in \cref{fig:run-accepting}. The leftmost \kl[alternating configuration]{configuration}
  is \kl[existential configuration]{existential}, the next one is \kl[universal configuration]{universal},
  and the two double-circled ones are \kl[permanent configuration]{permanent}. In the first round, the three
  \kl{nodes} that are in \kl[alternating state]{state} $\qb$ have a nondeterministic choice between
  $\State_{\Letter[2]:1}$ and $\State_{\Letter[2]:2}$. Hence, the second \kl[alternating configuration]{configuration} is one of
  eight possible choices. The branching in the second round is due to
  the \kl{node} in \kl[alternating state]{state} $\qaprime$ which goes simultaneously to $\State_{\Letter[1]:3}$
  and $\State_{\Letter[1]:4}$. In both branches, an \kl{accepting configuration} is
  reached, since $\{\State_{\Letter[1]:3},\q{yes}\}$ and $\{\State_{\Letter[1]:4},\q{yes}\}$ are
  both \kl{accepting sets}. Therefore, the entire \kl[alternating run]{run} is \kl[accepting run]{accepting}.
\end{example}

\begin{figure}[p]
  \centering
  \scalebox{0.99}{\input{fig/ADGA-concentric-circles.tikz}}
  \caption{$\sA{centric}$, an $\ALDAg$ over
    $\set{\Letter[1],\Letter[2],\Letter[3]}$-\kl{labeled} \kl{digraphs} whose \kl{digraph language}
    consists of the \kl{labeled} \kl{digraphs} that satisfy the following
    conditions: the \kl{labeling} constitutes a valid \kl[coloring]{3-coloring}, there is
    precisely one $\Letter[1]$-\kl{labeled} \kl{node} $\Node_{\Letter[1]}$, the
    \kl{undirected neighborhood} of $\Node_{\Letter[1]}$ contains only $\Letter[2]$-\kl{labeled} \kl{nodes}, and
    $\Node_{\Letter[1]}$ has at least two \kl{incoming neighbors}.}
  \label{fig:ADGA-concentric-circles}
\end{figure}

\begin{figure}[p]
  \centering
  \input{fig/graph-labeled-pentagon.tikz}
  \caption{An $\{\Letter[1],\Letter[2],\Letter[3]\}$-\kl{labeled}, \kl{digraph}.}
  \label{fig:graph-labeled-pentagon}
\end{figure}

\begin{figure}[p]
  \centering
  \scalebox{0.98}{\input{fig/run-accepting.tikz}}
  \caption{An \kl{accepting run} of the $\ALDAg$ of
    \cref{fig:ADGA-concentric-circles} on the \kl{labeled} \kl{digraph} shown in
    \cref{fig:graph-labeled-pentagon}.}
  \label{fig:run-accepting}
\end{figure}

In the following subsections (\ref{sec:hierarchy-closure},
\ref{sec:adga=mso} and \ref{sec:ndga-emptiness}), we derive our
results on several properties of $\LDAg$'s.

\section{Hierarchy and closure properties}
\label{sec:hierarchy-closure}

By a (\kl{node}) \Intro{projection} we mean a mapping $\Projection \colon \Alphabet → \Alphabet'$
between two alphabets~$\Alphabet$ and~$\Alphabet'$. With slight abuse of notation,
such a mapping is extended to \kl{labeled} \kl{digraphs} by applying it to each
node label, and to \kl{digraph languages} by applying it to each \kl{labeled}
\kl{digraph}. That is, for every $\Digraph∈\DIGRAPH[\Alphabet][\RelCount]$ and
$\Language⊆\DIGRAPH[\Alphabet][\RelCount]$,
\begin{equation*}
  \Projection(\Digraph) \defeq \lver{\Digraph}{\Projection∘\Labeling{\Digraph}},\!
  \quad \text{and} \quad
  \Projection(\Language) \defeq \setbuilder{\Projection(\Digraph)}{\Digraph∈\Language},
\end{equation*}
where the operator $∘$ denotes function composition, such that
$(\Projection∘\Labeling{\Digraph})(\Node)=\Projection(\Labeling{\Digraph}(\Node))$.

\begin{proposition}[{Closure properties of $\semAg{\ALDAg}[\DIGRAPH[\Alphabet][\RelCount]]$}] \label{prop:adga-closure}
  The class $\semAg{\ALDAg}[\DIGRAPH[\Alphabet][\RelCount]]$ of \kl[global recognizable]{$\ALDAg$-recognizable} \kl{digraph languages} is
  effectively closed under Boolean set operations and under
  \kl{projection}.
\end{proposition}

\begin{proofsketch}
  As usual for alternating automata, complementation can be achieved
  by simply swapping the \kl[existential states]{existential} and \kl{universal states}, and
  complementing the \kl[global acceptance condition]{acceptance condition}. That is, for an $\ALDAg$
  $\Automaton=\bigtuple{\tuple{\ExistentialStateSet,\UniversalStateSet,\PermanentStateSet},\InitFunc,\TransFunc,\AcceptCondition}$ over $\Alphabet$-\kl{labeled}, $\RelCount$-relational \kl{digraphs}, a
  complement \kl[\ALDAg]{automaton} is $\cA=\bigtuple{\tuple{\UniversalStateSet,\ExistentialStateSet,\PermanentStateSet},\InitFunc,\TransFunc,\powerset{\PermanentStateSet}\setminus\AcceptCondition}$. This
  can be easily seen by associating a two-player game with $\Automaton$ and
  any $\Alphabet$-\kl{labeled}, $\RelCount$-relational \kl{digraph} $\Digraph$. One player tries to come up with an
  \kl{accepting run} of $\Automaton$ on $\Digraph$, whereas the other player seeks to
  find a (path to a) \kl{rejecting configuration} in any \kl[alternating run]{run} proposed by
  the adversary. The first player has a winning strategy if and only if $\Automaton$
  \kl[global accepts]{accepts} $\Digraph$. (This game-theoretic characterization will be used
  and explained more extensively in the proof of \cref{thm:adga=mso}.)
  From this perspective, the construction of $\cA$ corresponds to
  interchanging the roles and winning conditions of the two players.

  For two $\ALDAg$'s $\Automaton_1$ and $\Automaton_2$, we can effectively construct an
  $\ALDAg$ $\Automaton_∪$ that \kl[global recognizes]{recognizes} $\semag{\Automaton_1}[\DIGRAPH[\Alphabet][\RelCount]]∪\semag{\Automaton_2}[\DIGRAPH[\Alphabet][\RelCount]]$ by taking advantage
  of nondeterminism. The approach is, in principle, very similar to
  the corresponding construction for nondeterministic finite automata
  on \kl[pointed dipaths]{words}. In the first round of $\Automaton_∪$, each \kl{node} in the input \kl{digraph}
  nondeterministically and independently decides whether to behave
  like in $\Automaton_1$ or in $\Automaton_2$. If there is a consensus, then the \kl[alternating run]{run}
  continues as it would in the unanimously chosen \kl[\ALDAg]{automaton} $\Automaton_j$,
  and it is \kl[accepting run]{accepting} if and only if it corresponds to an \kl{accepting run} of
  $\Automaton_j$. Otherwise, a conflict is detected, either locally by
  \kl{adjacent} \kl{nodes} that have chosen different \kl[\ALDAg]{automata}, or at the
  latest, when \kl[global acceptance]{acceptance} is checked globally (important for
  \kl[connected]{disconnected} \kl{digraphs}), and in either case the \kl[alternating run]{run} is \kl[accepting run]{rejecting}. (Note
  that we have omitted some technicalities that ensure that the
  construction outlined above satisfies all the properties of an
  $\ALDAg$.)

  Closure under node \kl{projection} is straightforward, again by
  exploiting nondeterminism. Given an $\ALDAg$ $\Automaton$ with node alphabet $\Alphabet$
  and a \kl{projection} $\Projection \colon \Alphabet → \Alphabet'$, we can effectively construct an
  $\ALDAg$ $\Automaton'$ that \kl[global recognizes]{recognizes} $\Projection(\semag{\Automaton}[\DIGRAPH[\Alphabet][\RelCount]])$ as follows: For every $\Letter[2]∈\Alphabet'$,
  each \kl{node} \kl{labeled} with $\Letter[2]$ nondeterministically chooses a new label
  $\Letter[1]∈\Alphabet$, such that $\Projection(\Letter[1])=\Letter[2]$. Then, the \kl[\ALDAg]{automaton} $\Automaton$ is simulated on
  that new input.
\end{proofsketch}

\begin{proposition}[{$\semAg{\NLDAg}[\DIGRAPH[\Alphabet][\RelCount]]⊂\semAg{\ALDAg}[\DIGRAPH[\Alphabet][\RelCount]]$}] \label{prop:ndga<adga}
  There are (infinitely many) \kl[global recognizable]{$\ALDAg$-recognizable} \kl{digraph languages} that
  are not \kl[global recognizable]{$\NLDAg$-recognizable}.
\end{proposition}

\begin{proof}
  For any constant $k≥1$, we consider the
  \kl[digraph language]{language} $\dL[card]{≤k}$ of all \kl{digraphs} that have at most $k$ \kl{nodes},
  i.e., $\dL[card]{≤k}=\bigl\{\Digraph∈\DIGRAPH \bigm| \card{\NodeSet{\Digraph}}≤k
  \bigr\}$. We can easily construct an $\ALDAg$ that \kl[global recognizes]{recognizes} this \kl{digraph language}:
  In a universal branching, each \kl{node} goes to $k+1$
  different \kl[alternating states]{states} in parallel. The \kl[\ALDAg]{automaton} \kl[global accepts]{accepts} if and only if there is no
  branch in which the $k+1$ \kl[alternating states]{states} occur all at once. Now, assume for
  sake of contradiction that $\dL[card]{≤k}$ is also \kl[global recognized]{recognized} by some
  $\NLDAg$ $\Automaton$, and let $\Digraph$ be a \kl{digraph} with $k$ \kl{nodes}. We construct a
  variant $\Digraph'$ of $\Digraph$ with $k+1$ \kl{nodes} by duplicating some \kl{node} $\Node$,
  together with all of its incoming and outgoing \kl{edges}. Observe that
  any \kl{accepting run} of $\Automaton$ on $\Digraph$ can be extended to an \kl{accepting run}
  on $\Digraph'$, where the copy of $\Node$ behaves exactly like $\Node$ in every
  round.
\end{proof}

\begin{proposition}[{Closure properties of $\semAg{\NLDAg}[\DIGRAPH[\Alphabet][\RelCount]]$}] \label{prop:ndga-closure}
  The class $\semAg{\NLDAg}[\DIGRAPH[\Alphabet][\RelCount]]$ of \kl[global recognizable]{$\NLDAg$-recognizable} \kl{digraph languages} is
  effectively closed under union, intersection and \kl{projection}, but
  \emph{not} closed under complementation.
\end{proposition}

\begin{proof}
  For union and \kl{projection}, we simply use the same constructions as
  for $\ALDAg$'s (see \cref{prop:adga-closure}).

  Intersection can be handled by a product construction, similar to
  the one for finite automata on \kl[pointed dipaths]{words}. Given two $\NLDAg$'s $\Automaton_1$ and
  $\Automaton_2$, we construct an $\NLDAg$ $\Automaton_⊗$ that operates on the Cartesian
  product of the \kl[alternating state]{state} sets of $\Automaton_1$ and $\Automaton_2$. It simulates the two
  \kl[\NLDAg]{automata} simultaneously and \kl[global accepts]{accepts} if and only if both of them reach an
  \kl{accepting configuration}.

  To see that $\semAg{\NLDAg}[\DIGRAPH[\Alphabet][\RelCount]]$ is not closed under complementation, we
  recall from the proof of \cref{prop:ndga<adga} that for any $k≥1$,
  the \kl[digraph language]{language} $\dL[card]{≤k}$ of all \kl{digraphs} that have at most $k$
  \kl{nodes} is not \kl[global recognizable]{$\NLDAg$-recognizable}. However, complementing the $\ALDAg$
  given for $\dL[card]{≤k}$ yields an $\NLDAg$ that \kl[global recognizes]{recognizes} the
  complement \kl[digraph language]{language} $\dL[card]{≥k+1}$.
\end{proof}

\begin{proposition}[{$\semAg{\DLDAg}[\DIGRAPH[\Alphabet][\RelCount]]⊂\semAg{\NLDAg}[\DIGRAPH[\Alphabet][\RelCount]]$}] \label{prop:ddga<ndga}
  There are (infinitely many) \kl[global recognizable]{$\NLDAg$-recognizable} \kl{digraph languages} that
  are not \kl[global recognizable]{$\DLDAg$-recognizable}.
\end{proposition}

\begin{proof}
  Let $k≥2$. As mentioned in the proof of \cref{prop:ndga-closure}, the
  \kl[digraph language]{language} $\dL[card]{≥k}$ of all \kl{digraphs} that have at least $k$ \kl{nodes} is
  \kl[global recognizable]{$\NLDAg$-recognizable}. To see that it is not \kl[global recognizable]{$\DLDAg$-recognizable}, consider
  (similarly to the proof of \cref{prop:ndga<adga}) a \kl{digraph} $\Digraph$ with
  $k-1$ \kl{nodes} and a variant $\Digraph'$ with $k$ \kl{nodes} obtained from $\Digraph$ by
  duplicating some \kl{node} $\Node$, together with all of its incoming and
  outgoing \kl{edges}. Given any $\DLDAg$ $\Automaton$, the determinism of $\Automaton$
  guarantees that $\Node$ and its copy $\Node'$ behave the same way in the
  (unique) \kl[alternating run]{run} of $\Automaton$ on $\Digraph'$. Hence, if that \kl[alternating run]{run} is \kl[accepting run]{accepting}, so is
  the \kl[alternating run]{run} on $\Digraph$.
\end{proof}

\begin{proposition}[{Closure properties of $\semAg{\DLDAg}[\DIGRAPH[\Alphabet][\RelCount]]$}] \label{prop:ddga-closure}
  The class $\semAg{\DLDAg}[\DIGRAPH[\Alphabet][\RelCount]]$ of \kl[global recognizable]{$\DLDAg$-recognizable} \kl{digraph languages} is
  effectively closed under Boolean set operations, but \emph{not}
  closed under \kl{projection}.
\end{proposition}

\begin{proof}
  To complement a $\DLDAg$, we can simply complement its set of \kl{accepting sets}.
  The product construction for intersection of $\NLDAg$'s mentioned
  in \cref{prop:ndga-closure} remains applicable when restricted to
  $\DLDAg$'s.

  Closure under node \kl{projection} does not hold because we can, for
  instance, construct a $\DLDAg$ that \kl[global recognizes]{recognizes} the \kl[digraph language]{language}
  $\dL[occur]{\Letter[1],\Letter[2],\Letter[3]}$ of all $\{\Letter[1],\Letter[2],\Letter[3]\}$-\kl{labeled} \kl{digraphs} in which
  each of the three node labels occurs at least once. However,
  \kl{projection} under the mapping $\Projection \colon \{\Letter[1],\Letter[2],\Letter[3]\}→\{\EmptyWord\}$, with
  \Phantomintro{\EmptyWord}
  $\Projection(\Letter[1])=\Projection(\Letter[2])=\Projection(\Letter[3])=\reintro*{\EmptyWord}$
  (the empty \kl[pointed dipath]{word}), yields the \kl{digraph language}
  $\Projection(\dL[occur]{\Letter[1],\Letter[2],\Letter[3]})=\dL[card]{≥3}$, which is not
  \kl[global recognizable]{$\DLDAg$-recognizable} (see the proof of \cref{prop:ddga<ndga}).
\end{proof}

\section{Equivalence with monadic second-order logic}
\label{sec:adga=mso}

\begin{theorem}[{$\semAg{\ALDAg}[\DIGRAPH[\Alphabet][\RelCount]]=\semF{\MSOL}[\DIGRAPH[\Alphabet][\RelCount]]$}]
  \label{thm:adga=mso}
  A \kl{digraph language} is \kl[global recognizable]{$\ALDAg$-recognizable} if and only if it is
  $\MSOL$-\kl{definable}. There are effective translations in both directions.
\end{theorem}

\begin{proofsketch}
  \newcommand{\Fsucc}[1]{\text{$\Formula_{#1 \vphantom{i}}^{\hspace{.05ex}\textnormal{succ\vphantom{g}}}$}\hspace{-.1ex}}
  \newcommand{\Fwin}[1]{\text{$\Formula_{#1 \vphantom{i}}^{\hspace{.05ex}\textnormal{win\vphantom{g}}}$}\hspace{-.1ex}}

  $(⇒)$ We start with the direction
  $\semAg{\ALDAg}[\DIGRAPH[\Alphabet][\RelCount]]⊆\semF{\MSOL}[\DIGRAPH[\Alphabet][\RelCount]]$.
  Let $\Automaton=\tuple{\StatePartition,\InitFunc,\TransFunc,\AcceptCondition}$
  be an $\ALDAg$ of \kl[automaton length]{length} $n$ over $\Alphabet$-\kl{labeled}, $\RelCount$-relational \kl{digraphs}.
  Without loss of generality,
  we may assume that every \kl[alternating configuration]{configuration} \kl[reachable configuration]{reachable} by $\Automaton$
  has at least one \kl{successor configuration} and that no \kl{permanent configuration}
  is \kl[reachable configuration]{reachable} in less than $n$ rounds. In order to
  encode the \kl[global acceptance behavior]{acceptance behavior} of $\Automaton$
  into an \kl[class-formula]{$\MSOL$-formula}~$\Formula_\Automaton$,
  we use again the game-theoretic characterization
  \marginnote{This characterization is heavily inspired by the work
    of Löding and Thomas in~\textnormal{\cite{DBLP:conf/ifipTCS/LodingT00}}.}
  briefly mentioned in the proof sketch of \cref{prop:adga-closure}. 
  Given $\Automaton$ and some
  $\Digraph∈\DIGRAPH[\Alphabet][\RelCount]$,
  we consider a game with two players: the
  \emph{automaton} (player~$\dgaEE$)
  and the \emph{pathfinder} (player~$\dgaAA$).
  This game is represented by an acyclic \kl{digraph}
  whose \kl{nodes} are precisely the \kl[alternating configurations]{configurations}
  \kl[reachable configuration]{reachable} by~$\Automaton$ on~$\Digraph$.
  For any two \kl[permanent configurations]{\emph{nonpermanent} configurations}
  $\lver{\Digraph}{\Relabeling[1]}$ and
  $\lver{\Digraph}{\Relabeling[2]}$,
  there is a directed \kl{edge} from $\lver{\Digraph}{\Relabeling[1]}$
  to $\lver{\Digraph}{\Relabeling[2]}$ if and only if
  $\lver{\Digraph}{\Relabeling[2]}∈\GlobTransFunc(\lver{\Digraph}{\Relabeling[1]})$.
  Starting at the \kl{initial configuration}
  $\lver{\Digraph}{\InitFunc∘\Labeling{\Digraph}}$,
  the two players move through the game together by
  following directed \kl{edges}. If the current \kl[alternating configuration]{configuration} is
  \kl[existential configuration]{existential}, then the automaton has to choose the next move, if it
  is \kl[universal configuration]{universal}, then the decision belongs to the pathfinder. This
  continues until some \kl{permanent configuration} is reached. The
  automaton wins if that \kl{permanent configuration} is \kl[accepting configuration]{accepting}, whereas
  the pathfinder wins if it is \kl[rejecting configuration]{rejecting}. A player is said to have a
  \emph{winning strategy} if it can always win, independently of its
  opponent's moves. It is straightforward to prove that the automaton
  has a winning strategy if and only if $\Automaton$ \kl[global accepts]{accepts} $\Digraph$.
  Our \kl[class-formula]{$\MSOL$-formula}
  $\Formula_\Automaton$ will express the existence of such a winning strategy, and
  thus be \kl[device equivalent]{equivalent} to~$\Automaton$.

  Within $\MSOL$,
  we represent a path $\lver{\Digraph}{\Relabeling[1]_0}\cdots
  \lver{\Digraph}{\Relabeling[1]_n}\strut$
  through the game by a sequence of families of \kl{set variables}
  $\vec{\SetVariable}_0,…,\vec{\SetVariable}_n$,
  where $\vec{\SetVariable}_0 = \tuple{\,}$ and
  $\vec{\SetVariable}_i = \tuple{\SetVariable_{i,\State}}_{\State∈\StateSet}$,
  for $1≤i≤n$.
  The intention is that each \kl{set variable} $\SetVariable_{i,\State}$
  is \kl{interpreted} as the set of \kl{nodes} $\Node∈\NodeSet{\Digraph}$
  for which $\Relabeling[1]_i(\Node)=\State$.
  (We do not need \kl{set variables} to represent $\lver{\Digraph}{\Relabeling[1]_0}$,
  since the players always start at $\lver{\Digraph}{\InitFunc∘\Labeling{\Digraph}}$.)

  Now, for every round $i$, we construct a \kl{formula}
  $\param{\Fwin{i}}{\vec{\SetVariable}_i}$
  (i.e., with \kl{free} \kl[set variables]{variables} in $\vec{\SetVariable}_i$),
  which expresses that the automaton has a winning strategy in the
  subgame starting at the \kl[alternating configuration]{configuration} $\lver{\Digraph}{\Relabeling[1]_i}$ represented by
  $\vec{\SetVariable}_i$.
  In case $\lver{\Digraph}{\Relabeling[1]_i}$ is \kl[existential configuration]{existential},
  this is true if the
  automaton has a winning strategy in some \kl{successor configuration} of
  $\lver{\Digraph}{\Relabeling[1]_i}$,
  whereas if $\lver{\Digraph}{\Relabeling[1]_i}$ is \kl[universal configuration]{universal},
  the automaton must
  have a winning strategy in all \kl{successor configurations} of
  $\lver{\Digraph}{\Relabeling[1]_i}$.
  This yields the following recursive definition for
  $0≤i≤n-1$:
  \begin{equation*}
    \param{\Fwin{i}}{\vec{\SetVariable}_i} \defeq
    \Exists{\vec{\SetVariable}_{i+1}}
      \Bigl(\param{\Fsucc{i+1}}{\vec{\SetVariable}_i,\vec{\SetVariable}_{i+1}} \,\AND\,
      \param{\Fwin{i+1}}{\vec{\SetVariable}_{i+1}} \Bigr),
  \end{equation*}
  if \kl[state level]{level} $i$ of $\Automaton$ is \kl[existential state]{existential}, and
  \begin{equation*}
    \param{\Fwin{i}}{\vec{\SetVariable}_i} \defeq
    \Forall{\vec{\SetVariable}_{i+1}}
      \Bigl(\param{\Fsucc{i+1}}{\vec{\SetVariable}_i,\vec{\SetVariable}_{i+1}} \IMP
      \param{\Fwin{i+1}}{\vec{\SetVariable}_{i+1}} \Bigr),
  \end{equation*}
  if \kl[state level]{level} $i$ of $\Automaton$ is \kl[universal state]{universal}. Here,
  $\param{\Fsucc{i+1}}{\vec{\SetVariable}_i,\vec{\SetVariable}_{i+1}}$
  is an \kl[class-formula]{$\FOL$-formula} expressing
  that $\vec{\SetVariable}_i$ and $\vec{\SetVariable}_{i+1}$ represent two \kl[alternating configurations]{configurations}
  $\lver{\Digraph}{\Relabeling[1]_i}$ and $\lver{\Digraph}{\Relabeling[1]_{i+1}}$ such that
  $\lver{\Digraph}{\Relabeling[1]_{i+1}}∈\GlobTransFunc(\lver{\Digraph}{\Relabeling[1]_i})$.
  As our recursion base, we can
  easily construct a \kl{formula} $\param{\Fwin{n}}{\vec{\SetVariable}_n}$ that is \kl{satisfied}
  if and only if $\vec{\SetVariable}_n$ represents an \kl{accepting configuration} of $\Automaton$.

  The desired \kl[class-formula]{$\MSOL$-formula} is
  $\Formula_\Automaton \defeq\, \param{\Fwin{0}}{\vec{\SetVariable}_0} =
  \param{\Fwin{0}}{\:}$.

  $(⇐)$ For the direction
  $\semAg{\ALDAg}[\DIGRAPH[\Alphabet][\RelCount]]⊇\semF{\MSOL}[\DIGRAPH[\Alphabet][\RelCount]]$,
  we can proceed by
  induction on the structure of an \kl[class-formula]{$\MSOL$-formula} $\Formula$.
  In order to deal with \kl{free} occurrences of \kl{node symbols},
  we encode their \kl{interpretations} into the node labels
  (which normally encode only the \kl{interpretations} of \kl{set symbols}).
  Let \Intro*{$\freeNodeSymbols(\Formula)$} be an abbreviation for $\free(\Formula) \cap \NodeSymbolSet$.
  For $\Digraph∈\DIGRAPH[\Alphabet][\RelCount]$ and $\Assignment\colon
  \freeNodeSymbols(\Formula) → \NodeSet{\Digraph}$,
  we represent $\aver{\Digraph}{\Assignment}$ as the \kl{labeled} \kl{digraph}
  $\lver{\Digraph}{\Labeling{\Digraph}×\Assignment^{-1}}$
  whose labeling $\Labeling{\Digraph}×\Assignment^{-1}$ assigns to each \kl{node}
  $\Node∈\NodeSet{\Digraph}$ the tuple
  $\bigtuple{\Labeling{\Digraph}(\Node),\:\Assignment^{-1}(\Node)}$,
  where
  $\Assignment^{-1}(\Node)$ is the set of all \kl{node symbols} in $\freeNodeSymbols(\Formula)$
  to which $\Assignment$ assigns~$\Node$.
  We now inductively construct an $\ALDAg$
  $\Automaton_{\Formula}=\tuple{\StatePartition,\InitFunc,\TransFunc,\AcceptCondition}$
  over
  $\RelCount$-relational \kl{digraphs} \kl{labeled} with the alphabet
  $\Alphabet' = \Alphabet×\powerset{\freeNodeSymbols(\Formula)}$,
  such that
  \begin{equation*}
    \lver{\Digraph}{\Labeling{\Digraph}×\Assignment^{-1}}∈
    \semag{\Automaton_{\Formula}}[\DIGRAPH[\Alphabet'][\RelCount]]
    \quad \text{if and only if} \quad
    \aver{\Digraph}{\Assignment} \Models \Formula.
  \end{equation*}

  \textit{Base case.}\,
  Let $\NodeVariable[1],\NodeVariable[2]∈\NodeSymbolSet$,\:
  $\SetVariable∈\SetSymbolSet$ and $i∈\range{\RelCount}$.
  If $\Formula$ is one of the atomic \kl{formulas}
  $\NodeVariable[1]\Equals\NodeVariable[2]$ or $\InSet{\SetVariable}{\NodeVariable[1]}$, then, in
  $\Automaton_{\Formula}$, each \kl{node} simply checks that its own label $\tuple{\Letter[1],M} ∈
  \Alphabet×\powerset{\freeNodeSymbols(\Formula)}$ satisfies the condition specified in $\Formula$ (which, in
  particular, is the case if $\NodeVariable[1],\NodeVariable[2]∉M$). Since this can be directly
  encoded into the \kl[alternating initialization function]{initialization function} $\InitFunc$, the $\ALDAg$ has \kl[automaton length]{length}
  $0$. It \kl[global accepts]{accepts} the input \kl{digraph} if and only if every \kl{node} reports that its
  label satisfies the condition.

  The case $\Formula=\InRel{\RelSymbol_i}{\NodeVariable[1],\NodeVariable[2]}$ is very
  similar, but $\Automaton_{\Formula}$ needs one communication round, after which the
  \kl{node} assigned to $\NodeVariable[2]$ can check whether it has received a message
  through an $i$-\kl{edge} from the \kl{node} assigned to $\NodeVariable[1]$. Accordingly,
  $\Automaton_{\Formula}$ has \kl[automaton length]{length} $1$.

  \medskip\textit{Inductive step.}\,
  In case $\Formula[1]$ is a composed \kl{formula}, we
  can obtain $\Automaton_{\Formula[1]}$ by means of the constructions outlined in the
  proof sketch of \cref{prop:adga-closure} (closure properties of
  $\semAg{\ALDAg}[\DIGRAPH[\Alphabet][\RelCount]]$). Let $\Formula[2]$ and $\Formula[2]'$
  be \kl[class-formula]{$\MSOL$-formulas} with
  \kl[device equivalent]{equivalent} $\ALDAg$'s $\Automaton_{\Formula[2]}$ and $\Automaton_{\Formula[2]'}$, respectively.

  If $\Formula[1]=\NOT\Formula[2]$, it suffices to define $\Automaton_{\Formula[1]}=\cA_{\Formula[2]}$.
  Similarly, if $\Formula[1]=\Formula[2]\OR\Formula[2]'$, we get $\Automaton_{\Formula[1]}$ by
  applying the union construction on $\Automaton_{\Formula[2]}$ and $\Automaton_{\Formula[2]'}$. (In
  general, we first have to extend $\Automaton_{\Formula[2]}$ and $\Automaton_{\Formula[2]'}$ such that they
  both operate on the same node alphabet
  $\Alphabet\!×\!\powerset{\freeNodeSymbols(\Formula[2])\,∪\,\freeNodeSymbols(\Formula[2]')}$.)

  Existential \kl[set quantifiers]{quantification} can be handled by node \kl{projection}. If
  $\Formula[1]=\Exists{\SetVariable}(\Formula[2])$, with $\SetVariable∈\SetSymbolSet$, we construct $\Automaton_{\Formula[1]}$
  by applying the \kl{projection} construction on $\Automaton_{\Formula[2]}$, using a mapping
  $\Projection \colon \Alphabet×\powerset{\freeNodeSymbols(\Formula[2])} → \tilde{\Alphabet}×\powerset{\freeNodeSymbols(\Formula[2])}$ that
  deletes the \kl{set variable} $\SetVariable$ from every label.
  In other words,
  the new alphabet~$\tilde{\Alphabet}$ encodes the subsets of
  $\free(\Formula[2]) \cap (\SetSymbolSet \setminus \set{\SetVariable})$.
  An analogous approach
  can be used if $\Formula[1]=\Exists{\NodeVariable[1]}(\Formula[2])$, with $\NodeVariable[1]∈\NodeSymbolSet$. The
  only difference is that, instead of applying the \kl{projection}
  construction directly on $\Automaton_{\Formula[2]}$, we apply it on a variant $\Automaton'_{\Formula[2]}$
  that operates just like $\Automaton_{\Formula[2]}$, but additionally checks that
  precisely one \kl{node} in the input \kl{digraph} is assigned to the \kl{node variable}
  $\NodeVariable[1]$.  \qedhere
\end{proofsketch}

From \cref{thm:adga=mso} we can immediately infer that it is
undecidable whether the \kl{digraph language} \kl[global recognized]{recognized} by some arbitrary
$\ALDAg$ is empty. Otherwise, we could decide the satisfiability problem
of $\MSOL$ on \kl{digraphs}, which is known to be undecidable (a direct
consequence of Trakhtenbrot's Theorem, see, e.g.,
\cite[Thm~9.2]{DBLP:books/sp/Libkin04}).

\begin{corollary}[Emptiness problem of $\ALDAg$'s] \label{cor:adga-emptiness}
  The emptiness problem for $\ALDAg$'s is undecidable.
\end{corollary}

\section{Emptiness problem for nondeterministic automata}
\label{sec:ndga-emptiness}

At the cost of reduced expressive power, we can also obtain a positive
decidability result.

\begin{proposition}[Emptiness problem of $\NLDAg$'s] \label{prop:ndga-emptiness}
  The emptiness problem of $\NLDAg$'s is decidable in doubly-exponential
  time. More precisely, for any $\NLDAg$ $\Automaton=\tuple{\StatePartition, \InitFunc,\TransFunc,\AcceptCondition}$ over $\Alphabet$-\kl{labeled}, $\RelCount$-relational \kl{digraphs},
  whether its \kl[global recognized]{recognized} \kl{digraph language} $\semag{\Automaton}[\DIGRAPH[\Alphabet][\RelCount]]$ is empty or not can
  be decided in time $2^k$\!, where $k∈\BigO(\:\!
  \RelCount·\NoHeight{\card{\StateSet}^{4\autolength(\Automaton)}}·\autolength(\Automaton) )$.

  Furthermore, whether or not the \kl{digraph language} $\semag{\Automaton}[\DIGRAPH[\Alphabet][\RelCount]]$ contains any \kl[connected]{\emph{connected}},
    \kl[graph]{\emph{undirected} graph} can be decided in time $\NoHeight{2^{2^{k'}}}$\!\!, where
  $k'∈\BigO(\:\!  \RelCount·\card{\StateSet}·\autolength(\Automaton) )$\strut.
\end{proposition}

\begin{proofsketch}
  Let $\Digraph∈\DIGRAPH[\Alphabet][\RelCount]$. Since $\NLDAg$'s cannot perform universal
  branching, we can consider any \kl[alternating run]{run} of $\Automaton$ on $\Digraph$ as a sequence of
  \kl[alternating configurations]{configurations} $\Run=\lver{\Digraph}{\Relabeling[1]_0} \cdots \lver{\Digraph}{\Relabeling[1]_n}$,
  with $n≤\autolength(\Automaton)$.
  In~$\Run$, each \kl{node} of $\Digraph$ traverses one of at most
  $\card{\StateSet}^{\autolength(\Automaton)+1}$ possible sequences of \kl[alternating states]{states}. Now, assume
  that $\Digraph$ has more than $\card{\StateSet}^{\autolength(\Automaton)+1}$ \kl{nodes}. Then, by the
  Pigeonhole Principle, there must be two distinct \kl{nodes} $\Node,\Node'∈\NodeSet{\Digraph}$
  that traverse the same sequence of \kl[alternating states]{states} in~$\Run$. We construct a
  smaller \kl{digraph} $\Digraph'$ by removing $\Node'$ from $\Digraph$, together with its
  adjacent \kl{edges}, and adding directed \kl{edges} from $\Node$ to all of the
  former \kl[outgoing neighbors]{\emph{outgoing} neighbors} of $\Node'$. If all the \kl{nodes} in $\Digraph'$
  maintain their nondeterministic choices from~$\Run$, none of them will
  notice that $\Node'$ is missing, and consequently they all behave just
  as in~$\Run$. The resulting \kl[alternating run]{run}~$\Run'$ on $\Digraph'$ is \kl[accepting run]{accepting} if and only if $\Run$ is
  \kl[accepting run]{accepting}.

  Applying this argument recursively, we conclude that if $\semag{\Automaton}[\DIGRAPH[\Alphabet][\RelCount]]$ is
  not empty, then it must contain some \kl{labeled} \kl{digraph} that has at most
  $\card{\StateSet}^{\autolength(\Automaton)+1}$ \kl{nodes}. Hence, the emptiness problem is
  decidable because the search space is finite. The time complexity
  indicated above corresponds to the naive approach of checking every
  \kl{digraph} with at most $\card{\StateSet}^{\autolength(\Automaton)+1}$ \kl{nodes}.

  If we are only interested in \kl{connected}, \kl{undirected graphs}, the
  reasoning is very similar, but we have to require a larger minimum
  number of \kl{nodes} in order to be able to remove some \kl{node} without
  influencing the behavior of the others. In a \kl{graph} $\Digraph$ with more
  than
  \mbox{$\bigl(\card{\StateSet}·2^{\RelCount·\card{\StateSet}}\bigr)^{\autolength(\Automaton)+1}$}
  \kl{nodes}, there must be two distinct \kl{nodes} $\Node,\Node'∈\NodeSet{\Digraph}$ that, in addition
  to traversing the same sequence of \kl[alternating states]{states}, also receive the same
  family of sets of \kl[alternating states]{states} from their \kl{neighborhood} in every
  round. Observe that the \kl[\NLDAg]{automaton} will not notice if we merge $\Node$
  and $\Node'$. The rest of the argument is analogous to the previous
  scenario.
\end{proofsketch}

\section{Summary and discussion}
\label{sec:summary}

We have introduced $\ALDAg$'s, which are probably the first graph automata
in the literature to be \kl[device equivalent]{equivalent} to $\MSOL$ on \kl{digraphs}. However,
their expressive power results mainly from the use of alternation: we
have seen that the deterministic, nondeterministic and alternating
variants form a strict hierarchy, i.e.,
\begin{equation*}
  \semAg{\DLDAg}[\DIGRAPH[\Alphabet][\RelCount]] \hyperref[prop:ddga<ndga]{⊂} \semAg{\NLDAg}[\DIGRAPH[\Alphabet][\RelCount]]
  \hyperref[prop:ndga<adga]{⊂} \semAg{\ALDAg}[\DIGRAPH[\Alphabet][\RelCount]].
\end{equation*}
The corresponding closure and decidability properties are summarized
in \cref{tab:closure-decidability}.

\begin{table}[t]
  \centering
  \setlength{\tabcolsep}{1ex}
  \begin{tabular}{lccccc}
    \toprule
         & \multicolumn{4}{c}{Closure Properties} & Decidability \\
    \cmidrule(rl){2-5} \cmidrule(rl){6-6}
         & Complement & Union & Intersection & Projection & Emptiness \\
    \addlinespace
    $\ALDAg$ & \hyperref[prop:adga-closure]{\cmark} & \hyperref[prop:adga-closure]{\cmark}
         & \hyperref[prop:adga-closure]{\cmark} & \hyperref[prop:adga-closure]{\cmark}
         & \hyperref[cor:adga-emptiness]{\xmark} \\
    \addlinespace
    $\NLDAg$ & \hyperref[prop:ndga-closure]{\xmark} & \hyperref[prop:ndga-closure]{\cmark}
         & \hyperref[prop:ndga-closure]{\cmark} & \hyperref[prop:ndga-closure]{\cmark}
         & \hyperref[prop:ndga-emptiness]{\cmark} \\
    \addlinespace
    $\DLDAg$ & \hyperref[prop:ddga-closure]{\cmark} & \hyperref[prop:ddga-closure]{\cmark}
         & \hyperref[prop:ddga-closure]{\cmark} & \hyperref[prop:ddga-closure]{\xmark}
         & \hyperref[prop:ndga-emptiness]{\cmark} \\
    \bottomrule
  \end{tabular}
  \caption{Closure and decidability properties of alternating,
    nondeterministic, and deterministic $\LDAg$'s.}
  \label{tab:closure-decidability}
\end{table}

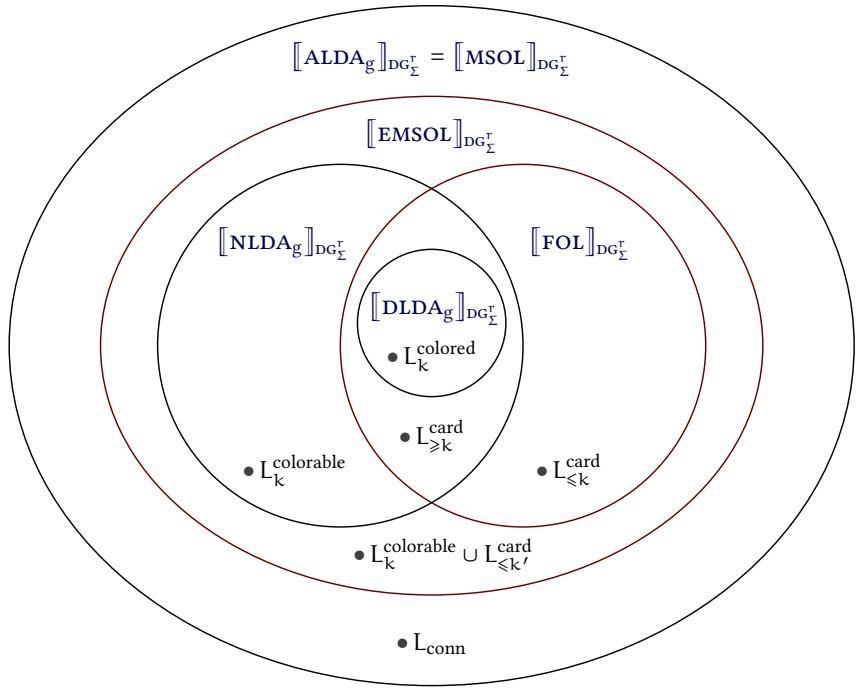
\begin{figure}[t]
  \centering
  \scalebox{0.91}{\input{fig/venn-diagram.tikz}}
  \caption{Venn diagram relating the classes of \kl{digraph languages}
    \kl[global recognizable]{recognizable} by our three flavors of $\LDAg$'s to those \kl{definable} in
    $\MSOL$, $\EMSOL$ and~$\FOL$.}
  \label{fig:venn-diagram}
\end{figure}

On an intuitive level, this hierarchy and these closure properties do
not seem very surprising. One might even ask: \emph{are $\ALDAg$'s just
  another syntax for $\MSOL$?} Indeed, universal branchings
correspond to universal \kl[set quantifiers]{quantification}, and nondeterministic choices
to existential \kl[set quantifiers]{quantification}. By disallowing universal
\kl[set quantifiers]{set quantification} in $\MSOL$ we obtain $\EMSOL$, and further
disallowing existential \kl[set quantifiers]{set quantification} yields
$\FOL$. Analogously to $\LDAg$'s, the classes of \kl{digraph languages}
\kl{definable} in these logics form a strict hierarchy, i.e.,
\begin{equation*}
  \semF{\FOL}[\DIGRAPH[\Alphabet][\RelCount]] ⊂ \semF{\EMSOL}[\DIGRAPH[\Alphabet][\RelCount]] ⊂ \semF{\MSOL}[\DIGRAPH[\Alphabet][\RelCount]].
\end{equation*}
Furthermore, the closure properties of $\semF{\EMSOL}[\DIGRAPH[\Alphabet][\RelCount]]$ and $\semF{\FOL}[\DIGRAPH[\Alphabet][\RelCount]]$
coincide with those of $\semAg{\NLDAg}[\DIGRAPH[\Alphabet][\RelCount]]$ and $\semAg{\DLDAg}[\DIGRAPH[\Alphabet][\RelCount]]$,
respectively. Given that $\semAg{\ALDAg}[\DIGRAPH[\Alphabet][\RelCount]]$ and $\semF{\MSOL}[\DIGRAPH[\Alphabet][\RelCount]]$ are equal, one
might therefore expect that the analogous equalities hold for the
weaker classes. However, as already hinted by the positive
decidability properties in \cref{tab:closure-decidability}, this is
not the case.  The actual relationships between the different classes
of \kl{digraph languages} are depicted in \cref{fig:venn-diagram}. A glance
at this diagram suggests that $\ALDAg$'s are not simply a one-to-one
reproduction of $\MSOL$.

\begin{proof}[Justification of \cref{fig:venn-diagram}]
  \newcommand{\Fstate}[2]{\text{$\Formula_{#1:#2}^{\hspace{.05ex}\textnormal{state\vphantom{g}}}$}\hspace{-.1ex}}
  Fagin has shown in \cite{DBLP:journals/mlq/Fagin75} that the \kl[digraph language]{language} $\sL{conn}$ of all
  (weakly) \kl{connected} \kl{digraphs} separates $\semF{\EMSOL}[\DIGRAPH[\Alphabet][\RelCount]]$ from
  $\semF{\MSOL}[\DIGRAPH[\Alphabet][\RelCount]]$. (Since non-\kl{connectivity} is $\EMSOL$-\kl{definable}, this also
  implies that $\semF{\EMSOL}[\DIGRAPH[\Alphabet][\RelCount]]$ is not closed under complementation.) The
  inclusion $\semAg{\NLDAg}[\DIGRAPH[\Alphabet][\RelCount]]⊆\semF{\EMSOL}[\DIGRAPH[\Alphabet][\RelCount]]$ holds because we can encode every
  $\NLDAg$ into an \kl[class-formula]{$\EMSOL$-formula}, using the same construction as in the
  proof sketch of \cref{thm:adga=mso}. It is also easy to see that we
  can encode a $\DLDAg$ by inductively constructing a family of
  \kl[class-formula]{$\FOL$-formulas} $\param{\Fstate{i}{\State}}{\NodeVariable[1]}$, stating that in round $i$ (of the
  \emph{unique} \kl[alternating run]{run} of the \kl[\DLDAg]{automaton}), the \kl{node} assigned to $\NodeVariable[1]$ is in
  \kl[alternating state]{state} $\State$. Hence, $\semAg{\DLDAg}[\DIGRAPH[\Alphabet][\RelCount]]⊆\semF{\FOL}[\DIGRAPH[\Alphabet][\RelCount]]$. In the following, let
  $k,k'≥2$. The incomparability of $\semAg{\NLDAg}[\DIGRAPH[\Alphabet][\RelCount]]$ and $\semF{\FOL}[\DIGRAPH[\Alphabet][\RelCount]]$ is
  witnessed by the \kl[digraph language]{language} $\dL[colorable]{k}$ of \kl{$k$-colorable}
  \kl{digraphs}, which lies within $\semAg{\NLDAg}[\DIGRAPH[\Alphabet][\RelCount]]$ (see
  Example~\ref{ex:ADGA-3-colorable}) but outside of $\semF{\FOL}[\DIGRAPH[\Alphabet][\RelCount]]$ (see,
  e.g., \cite{DBLP:books/sp/Libkin04}), and the \kl[digraph language]{language} $\dL[card]{≤k}$ of \kl{digraphs} with
  at most $k$ \kl{nodes}, which lies outside of $\semAg{\NLDAg}[\DIGRAPH[\Alphabet][\RelCount]]$ (see the proof
  of \cref{prop:ndga<adga}) but obviously within $\semF{\FOL}[\DIGRAPH[\Alphabet][\RelCount]]$. Considering
  the union \kl[digraph language]{language} $\dL[colorable]{k}∪\dL[card]{≤k'}$ also tells us
  that the inclusion of $\semAg{\NLDAg}[\DIGRAPH[\Alphabet][\RelCount]]∪\semF{\FOL}[\DIGRAPH[\Alphabet][\RelCount]]$ in $\semF{\EMSOL}[\DIGRAPH[\Alphabet][\RelCount]]$ is
  strict. Finally, the \kl[digraph language]{language} $\dL[card]{≥k}$ of \kl{digraphs} with at
  least $k$ \kl{nodes} separates $\semAg{\DLDAg}[\DIGRAPH[\Alphabet][\RelCount]]$ from $\semAg{\NLDAg}[\DIGRAPH[\Alphabet][\RelCount]]∩\semF{\FOL}[\DIGRAPH[\Alphabet][\RelCount]]$ (see
  the proof of \cref{prop:ddga<ndga}). A simple example of a \kl[digraph language]{language}
  that lies within $\semAg{\DLDAg}[\DIGRAPH[\Alphabet][\RelCount]]$ is the set $\dL[colored]{k}$ of
  $\Alphabet$-\kl{labeled} \kl{digraphs} whose \kl{labelings} are valid \kl{$k$-colorings},
  with~$\card{\Alphabet}=k$.
\end{proof}

Nevertheless,
based on the \kl[device equivalence]{equivalence} of $\LDA$'s and $\bML$
established by Hella~et~al.\ (\cref{thm:LDA-bML}),
we can actually obtain precise logical characterizations
of $\DLDAg$'s and $\NLDAg$'s
by extending $\bML$ with
\kl{global modalities} and existential \kl{set quantifiers}.
Adapting the proofs
of~\cite{DBLP:conf/podc/HellaJKLLLSV12,DBLP:journals/dc/HellaJKLLLSV15}
to our setting,
it is relatively easy to show~that
\begin{equation*}
  \semAg{\DLDAg}[\DIGRAPH[\Alphabet][\RelCount]] =
  \semF{\bMLg}[\DIGRAPH[\Alphabet][\RelCount]]
  \quad \text{and} \quad
  \semAg{\NLDAg}[\DIGRAPH[\Alphabet][\RelCount]] =
  \semF{\SigmaMSO{1}(\bMLg)}[\DIGRAPH[\Alphabet][\RelCount]].
\end{equation*}
More generally,
one can show a levelwise \kl[device equivalence]{equivalence}
with the \kl{set quantifier} alternation hierarchy of $\MSO(\bMLg)$,
a rather unconventional logic that is \kl[device equivalent]{equivalent} to~$\MSOL$.
In other words,
two corresponding levels of alternation
in the frameworks of $\ALDAg$'s and $\MSO(\bMLg)$
characterize exactly the same \kl{digraph languages}.
Against this backdrop,
$\ALDAg$'s may be best described as
a machine-oriented syntax for $\MSO(\bMLg)$.
We shall pick up on this point
in the introduction of \cref{ch:alternation}.

As of the time of writing this thesis, no new results on
$\semF{\MSOL}[\DIGRAPH[\Alphabet][\RelCount]]$
have been inferred from the alternative characterization through
$\ALDAg$'s. On the other hand, the notion of $\NLDAg$ contributes to the
general observation, mentioned in \cref{ssec:related-automata-theory}, that many
characterizations of regularity, which are equivalent on \kl[pointed dipaths]{words} and
\kl[pointed ordered ditrees]{trees}, drift apart on \kl{digraphs}. To see this, consider $\NLDAg$'s whose input
is restricted to those $\Alphabet$-\kl{labeled}, $\RelCount$-relational \kl{digraphs} that represent \kl[pointed dipaths]{words} or
\kl[pointed ordered ditrees]{trees} over the alphabet $\Alphabet$. For \kl[pointed dipaths]{words}, $\RelCount = 1$ and \kl{edges}
simply go from one position to the next, whereas for \kl[pointed ordered ditrees]{ordered trees} of
arity $k$, we set $\RelCount = k$ and require \kl{edge relations} such that
$\NoHeight{\edge{\Node[1]}{\Node[2]} \in \EdgeSet[i]{\Digraph}}$ if and only if $\Node[1]$ is the $i$-th child of $\Node[2]$. Observe
that we can easily simulate any \kl[pointed dipath]{word} or \kl[pointed ordered ditree]{tree} automaton by an $\NLDAg$ of
\kl[automaton length]{length} $2$: guess a run of the automaton in the first round (each \kl{node}
nondeterministically chooses some state), then check whether it is a
valid accepting run in the second round (transitions are verified
locally, and \kl[global acceptance]{acceptance} is determined by the unique \kl{sink}). This
implies that the classes of \kl[global recognizable]{$\NLDAg$-recognizable} and $\MSOL$-\kl{definable}
\kl{languages} collapse on \kl[pointed dipaths]{words} and \kl[pointed ordered ditrees]{trees}, and hence that $\NLDAg$'s
\kl[global recognize]{recognize}
precisely the regular languages on those restricted \kl{structures}. (Note
that this does not hold for $\DLDAg$'s;
for instance, it is easy to see that
they cannot decide whether a given unary \kl[pointed dipath]{word} is of even length.)
In this light, the decidability of
the emptiness problem for $\NLDAg$'s can be seen as an extension to
arbitrary \kl{digraphs} of the corresponding decidability results for finite
automata on \kl[pointed dipaths]{words} and \kl[pointed ordered ditrees]{trees}.

%% file: fig/ADGA-3-colorable.tikz
\begin{tikzpicture}[automaton, half row sep]
  \matrix[states] {
        & \node[existential] (q1) {$q_1$}; &[6ex] \\
        &     & \node[permanent] (q-yes) {$\q{yes}$}; \\
    \node[initial,existential] (q-ini) {$\q{ini}$}; & \node[existential] (q2) {$q_2$}; \\
        &     & \node[permanent] (q-no) {$\q{no}$}; \\
        & \node[existential] (q3) {$q_3$}; \\
  };
  \path[use as bounding box]
        (q-ini)  edge (q1)
                 edge (q2)
                 edge (q3)
        (q1)     edge[bend left=25] node[above=-.1ex,xshift=.4ex] {$\xnni q_1$} (q-yes)
                 edge[bend left=15] node[above=.2ex] {$\xni q_1$} (q-no)
        (q2.10)  edge[bend right=5] node[above left=-.5ex,xshift=-1.5ex] {$\xnni q_2$} (q-yes)
        (q2.350) edge[bend left=5] node[below left=-.5ex,xshift=-1.5ex] {$\xni q_2$} (q-no)
        (q3)     edge[bend right=15] node[below=.4ex,xshift=-.2ex] {$\xnni q_3$} (q-yes)
                 edge[bend right=25] node[below=.1ex,xshift=.4ex] {$\xni q_3$} (q-no);
  \matrix[accepting sets] {
       \\
    \x \\
       \\
       \\
       \\
  };
  \DrawColumnBackground{2}{4}{1}
\end{tikzpicture}

%% file: fig/ADGA-not-3-colorable.tikz
\begin{tikzpicture}[automaton, half row sep]
  \matrix[states] {
        & \node[universal] (q1) {$q_1$}; &[6ex] \\
        &     & \node[permanent] (q-yes) {$\q{yes}$}; \\
    \node[initial,universal] (q-ini) {$\q{ini}$}; & \node[universal] (q2) {$q_2$}; \\
        &     & \node[permanent] (q-no) {$\q{no}$}; \\
        & \node[universal] (q3) {$q_3$}; \\
        & \\
  };
  \path[use as bounding box]
        (q-ini.20) edge (q1)
        (q-ini)    edge (q2)
                   edge (q3)
        (q1)       edge[bend left=25] node[above=-.1ex,xshift=.4ex] {$\xnni q_1$} (q-yes)
                   edge[bend left=15] node[above=.2ex] {$\xni q_1$} (q-no)
        (q2.10)    edge[bend right=5] node[above left=-.5ex,xshift=-1.5ex] {$\xnni q_2$} (q-yes)
        (q2.350)   edge[bend left=5] node[below left=-.5ex,xshift=-1.5ex] {$\xni q_2$} (q-no)
        (q3)       edge[bend right=15] node[below=.4ex,xshift=-.2ex] {$\xnni q_3$} (q-yes)
                   edge[bend right=25] node[below=.1ex,xshift=.4ex] {$\xni q_3$} (q-no);
  \matrix[accepting sets] {
            \\
       & \x \\
            \\
    \x & \x \\
            \\
            \\
  };
  \DrawColumnBackground{2}{4}{2}
\end{tikzpicture}
\vspace{-2ex}

%% file: fig/ADGA-concentric-circles.tikz
\begin{tikzpicture}[automaton, half row sep]
  \matrix[states] {
        \node[existential,draw=none] {}; &[3.5ex]&[7ex] \node[permanent] (q-a3) {$q_{\Letter[1]:3}$}; \\
    \node[existential] (q-a) {$\qa$}; & \node[universal] (q-apr) {$\qaprime$};  \\
        &     & \node[permanent] (q-a4) {$q_{\Letter[1]:4}$}; \\
        & \node[universal] (q-b1) {$q_{\Letter[2]:1}$}; \\
    \node[existential] (q-b) {$\qb$}; &     & \node[permanent] (q-yes) {$\q{yes}$}; \\
        & \node[universal] (q-b2) {$q_{\Letter[2]:2}$}; \\
        \\
    \node[existential] (q-c) {$\qc$}; &     & \node[permanent] (q-no) {$\q{no}$}; \\
  };
  \matrix[symbols] {
        \\
    \node (a) {$\Letter[1]$}; \\
        \\
        \\
    \node (b) {$\Letter[2]$}; \\
        \\
        \\
    \node (c) {$\Letter[3]$}; \\
  };
  \matrix[accepting sets] {
    \x &    \\
       &    \\
       & \x \\
       &    \\
    \x & \x \\
       &    \\
       &    \\
       &    \\
  };
  \DrawColumnBackground{1}{8}{2}
  \path (a)         edge (q-a)
        (b)         edge (q-b)
        (c)         edge (q-c)
        (q-a)       edge (q-apr)
        (q-b.25)    edge node[above] {$\xnni \qb$} (q-b1)
        (q-b.0)     edge node[above,xshift=-0.3ex] {$\xnni \qb$} (q-b2)
        (q-b.330)   edge[bend right=20] node[above=4.2ex,xshift=-9.2ex] {$\xni \qb$} (q-no)
        (q-c.0)     edge[bend right=25] node[above=-.5ex,xshift=-12ex] {$\xnni \qc ∧ \xnni \qa$} (q-yes)
        (q-c.330)   edge[bend right=22] node[above=.2ex] {$\xni \qc ∨ \xni \qa$} (q-no)
        (q-apr.22)  edge[bend left=8] node[above=.7ex,xshift=1ex] {$\xeq \{q_{\Letter[2]:1},q_{\Letter[2]:2}\}$} (q-a3)
        (q-apr)     edge node[above=.4ex,xshift=.9ex] {$\xeq \{q_{\Letter[2]:1},q_{\Letter[2]:2}\}$} (q-a4)
        (q-apr.320) edge[bend right=9] node[above=4.5ex,xshift=1.5ex] {$\xneq \{q_{\Letter[2]:1},q_{\Letter[2]:2}\}$} (q-no)
        (q-b1)      edge (q-yes)
        (q-b2)      edge (q-yes);
\end{tikzpicture}

%% file: fig/graph-labeled-pentagon.tikz
\pentagraphPic{input graph}{\lnodedistIG}{$\Letter[1]$}{$\Letter[2]$}{$\Letter[3]$}{$\Letter[2]$}{$\Letter[2]$}{$\Letter[3]$}

%% file: fig/run-accepting.tikz
\begin{tikzpicture}[run or game]
  \matrix {
      & &[-2ex] \node[config,pacc] (c3a)
           {\pentagraphPic{configuration}{\lnodedistC}{$q_{\Letter[1]:3}$}{$\q{yes}$}{$\q{yes}$}{$\q{yes}$}{$\q{yes}$}{$\q{yes}$}}; \\
    \node[config,exis] (c1)
     {\pentagraphPic{configuration}{\lnodedistC}{$\qa$}{$\qb$}{$\qc$}{$\qb$}{$\qb$}{$\qc$}};
      & \node[config,univ] (c2)
         {\pentagraphPic{configuration}{\lnodedistC}{$\qaprime$}{$q_{\Letter[2]:1}$}{$\q{yes}$}{$q_{\Letter[2]:1}$}{$q_{\Letter[2]:2}$}{$\q{yes}$}}; \\
      & & \node[config,pacc] (c3b)
           {\pentagraphPic{configuration}{\lnodedistC}{$q_{\Letter[1]:4}$}{$\q{yes}$}{$\q{yes}$}{$\q{yes}$}{$\q{yes}$}{$\q{yes}$}}; \\
  };
  \path (c1) edge (c2)
        (c2) edge (c3a)
             edge (c3b);
\end{tikzpicture}

%% file: fig/venn-diagram.tikz
\begin{tikzpicture}[semithick]
  \tikzstyle{classL} = [font=\large]
  \tikzstyle{exampleL} = []
  \newcommand{\langelement}[1]{$\color{darkgray}\bullet\,\color{black}#1$}

  \def\Ladga{(0,0) ellipse (37ex and 30ex)}
  \def\Lemso{(0,0) ellipse (29ex and 22ex)}
  \def\Lndga{(180:8ex) circle (16ex)}
  \def\Lfo{(0:8ex) circle (16ex)}
  \def\Lddga{(90:2ex) ellipse (6.5ex and 6.5ex)}

  \draw[black] \Ladga node[above=23ex,classL] {$\semAg{\ALDAg}[\DIGRAPH[\Alphabet][\RelCount]]={\color{darkred}\semF{\MSOL}[\DIGRAPH[\Alphabet][\RelCount]]}$}
      node[below=24.5ex,exampleL] {\langelement{\sL{conn}}};
  \draw[darkred] \Lemso node[above=16.2ex,classL] {$\semF{\EMSOL}[\DIGRAPH[\Alphabet][\RelCount]]$}
      node[xshift=1ex,yshift=-18.5ex,exampleL] {\langelement{\dL[colorable]{k}∪\dL[card]{≤k'}}};
  \draw[darkred] \Lfo node[xshift=5ex,yshift=9ex,classL] {$\semF{\FOL}[\DIGRAPH[\Alphabet][\RelCount]]$}
      node[xshift=4ex,yshift=-11ex,exampleL] {\langelement{\dL[card]{≤k}}};
  \draw[black] \Lndga node[xshift=-5ex,yshift=9ex,classL] {$\semAg{\NLDAg}[\DIGRAPH[\Alphabet][\RelCount]]$}
      node[xshift=-4ex,yshift=-11ex,exampleL] {\langelement{\dL[colorable]{k}}};
  \draw[black] \Lddga node[above=-1ex,xshift=0.3ex,classL] {$\semAg{\DLDAg}[\DIGRAPH[\Alphabet][\RelCount]]$}
      node[below=.7ex,exampleL] {\langelement{\dL[colored]{k}}};
  \node[exampleL] at (-90:8ex) {\langelement{\dL[card]{≥k}}};
\end{tikzpicture}

%% file: tex/nonlocal.tex
\mychapterpreamble{%
  \nameCref{ch:nonlocal} based on the
  conference paper~\cite{DBLP:conf/icalp/Reiter17}.}

\chapter{Asynchronous Nonlocal Automata}
\label{ch:nonlocal}

In this \lcnamecref{ch:nonlocal},
we introduce a particular class of \kl{nonlocal distributed automata}
and show that on finite \kl{digraphs},
they are \kl[device equivalent]{equivalent} to the
\kl[backward $\mu$-fragment]{least-fixpoint fragment of the backward $\mu$-calculus},
or simply \kl[backward $\mu$-fragment]{\emph{backward $\mu$-fragment}}.

For the general case,
a logical characterization has been provided by Kuusisto
in~\cite{DBLP:conf/csl/Kuusisto13};
there he introduced a modal-logic-based variant of Datalog,
called \emph{modal substitution calculus},
that captures exactly the class of \kl{nonlocal automata}.
Furthermore,
\mbox{\cite[Prp.~7]{DBLP:conf/csl/Kuusisto13}} shows that
these \kl[nonlocal automata]{automata} can easily \kl{recognize}
all the \kl[pointed-digraph languages]{properties}
\kl{definable} in the \kl{backward $\mu$-fragment} on finite \kl{digraphs}.
On the other hand,
the reverse conversion from \kl{nonlocal automata} to the \kl{backward $\mu$-fragment}
is not possible in general.
As explained in \cite[Prp.~6]{DBLP:conf/csl/Kuusisto13},
it is easy to come up with an \kl[nonlocal automaton]{automaton}
that makes crucial use of the fact
that a \kl{node} can determine whether it receives the same information
from all of its \kl{incoming neighbors} at exactly the same time.
Such synchronous behavior
cannot be simulated in the \kl{backward $\mu$-fragment}
(and not even in $\MSOL$).
This leaves open the problem
of identifying a subclass of \kl{distributed automata}
for which the conversion works in both directions.

Here,
we present a very simple solution:
it basically suffices
to transfer the standard notion of asynchronous algorithm
to the setting of \kl{distributed automata}.

The organisation of this \lcnamecref{ch:nonlocal} is as follows.
After giving the necessary formal definitions in Section~\ref{sec:preliminaries-nonlocal},
we state and briefly discuss the main result in Section~\ref{sec:result}.
The proof is then developed in the last two sections.
Section~\ref{sec:logic-to-automata} presents
the rather straightforward translation from logic to automata.
The reverse translation is given in Section~\ref{sec:automata-to-logic},
which is a bit more involved
and therefore occupies the largest part of the \lcnamecref{ch:nonlocal}.

\section{Preliminaries}
\label{sec:preliminaries-nonlocal}

The class of \kl{asynchronous distributed automata}
introduced in this \lcnamecref{ch:nonlocal},
is a special case of the \kl{distributed automata}
defined in \cref{sec:distributed-automata}.
We maintain the same syntax as in \cref{def:distributed-automaton},
but reintroduce the semantics of (unrestricted) \kl{distributed automata}
from a slightly different perspective.
In order to keep notation simple,
we do this only for $1$-relational \kl{digraphs},
but everything presented here can easily be extended
to the multi-relational case.

To \kl[asynchronous run]{run} a \kl{distributed automaton}~$\Automaton$
on a \kl{digraph}~$\Digraph$,
we now regard the \kl{edges} of~$\Digraph$ as $\FIFO$ buffers.
Each buffer~$\edge{\Node[2]}{\Node[3]}$ will always contain a sequence of \kl{states}
previously traversed by \kl{node}~$\Node[2]$.
An adversary chooses
when $\Node[2]$ evaluates~$\TransFunc$
to push a new \kl{state} to the back of the buffer,
and when the current first \kl{state} gets popped from the front.
The details are clarified in the following.

A \Intro{trace} of an \kl[distributed automaton]{automaton}
$\Automaton = \tuple{\StateSet,\InitFunc,\TransFunc,\AcceptSet}$
is a finite nonempty sequence
$\Trace = \State_1 \dots \State_n$ of \kl{states} in $\StateSet$
such that
$\State_i \neq \State_{i+1}$
and
$\TransFunc(\State_i,\NeighborSet_i) = \State_{i+1}$
for some $\NeighborSet_i \subseteq \StateSet$.
Notice that
$\Automaton$ is \kl{quasi-acyclic}
if and only if its set of \kl{traces} $\TraceSet$ is finite.

For any \kl{states} $\State[1],\State[2] \in \StateSet$ and
any (possibly empty) sequence $\Trace$ of \kl{states} in $\StateSet$,
we define the unary postfix operators
\Intro*{$\first$}, \Intro*{$\last$}, \Intro*{$\pushlast$} and \Intro*{$\popfirst$}
as follows:
\begin{align*}
  \State[1]\Trace.\first &= \Trace\State[1].\last = \State[1],
  \\
  \Trace\State[1].\pushlast(\State[2]) &=
  \begin{cases*}
    \Trace\State[1]\State[2] & if $\State[1] \neq \State[2]$, \\
    \Trace\State[1]          & if $\State[1] = \State[2]$,
  \end{cases*}
  \\
  \State[1]\Trace.\popfirst &=
  \begin{cases*}
    \Trace          & if $\Trace$ is nonempty, \\
    \State[1]\Trace & if $\Trace$ is empty.
  \end{cases*}
\end{align*}

An (asynchronous) \Intro{timing} of a \kl{digraph}
$\Digraph = \tuple{\NodeSet{\Digraph}, \EdgeSet{\Digraph}, \Labeling{\Digraph}}$
is an infinite sequence
$\Timing = \tuple{\Timing_1, \Timing_2, \Timing_3, \dots}$
of maps
$\Timing_{\Time} \colon \NodeSet{\Digraph} \cup \EdgeSet{\Digraph} \to \Boolean$,
indicating which \kl{nodes} and \kl{edges} are active at time~$\Time$,
where $1$ is assigned infinitely often
to every \kl{node} and every \kl{edge}.
More formally,
for all $\Time[1] \in \Positive$, $\Node \in \NodeSet{\Digraph}$ and $\Edge \in \EdgeSet{\Digraph}$,
there exist $\Time[2],\Time[3] > \Time[1]$
such that $\Timing_{\Time[2]}(\Node) = 1$ and $\Timing_{\Time[3]}(\Edge) = 1$.
We refer to this as the \Intro{fairness property} of $\Timing$.
As a restriction,
we say that $\Timing$ is \Intro[lossless-asynchronous timing]{lossless-asynchronous}
if $\Timing_{\Time}(\edge{\Node[1]}{\Node[2]}) = 1$ implies $\Timing_{\Time}(\Node[2]) = 1$
for all $\Time \in \Positive$ and $\edge{\Node[1]}{\Node[2]} \in \EdgeSet{\Digraph}$.
Furthermore,
$\Timing$ is called the (unique) \Intro{synchronous timing} of~$\Digraph$
if $\Timing_{\Time}(\Node) = \Timing_{\Time}(\Edge) = 1$
for all $\Time \in \Positive$, $\Node \in \NodeSet{\Digraph}$ and $\Edge \in \EdgeSet{\Digraph}$.

\begin{definition}[Asynchronous Run]
  Let
  $\Automaton = \tuple{\StateSet,\InitFunc,\TransFunc,\AcceptSet}$
  be a \kl{distributed automaton} over $\BitCount$-bit \kl{labeled} \kl{digraphs}
  and~$\TraceSet$ be its set of \kl{traces}.
  Furthermore, let
  $\Digraph = \tuple{\NodeSet{\Digraph}, \EdgeSet{\Digraph}, \Labeling{\Digraph}}$
  be an $\BitCount$-bit \kl{labeled} \kl{digraph}
  and
  $\Timing = \tuple{\Timing_1, \Timing_2, \Timing_3, \dots}$
  be a \kl{timing} of $\Digraph$.
  The (asynchronous) \Intro[asynchronous run]{run}
  of $\Automaton$ on~$\Digraph$ timed by $\Timing$
  is the infinite sequence
  $\Run = \tuple{\Run_0, \Run_1, \Run_2, \dots}$
  of \Intro[asynchronous configurations]{configurations}
  $\Run_{\Time} \colon \NodeSet{\Digraph} \cup \EdgeSet{\Digraph} \to \TraceSet$,
  with $\Run_{\Time}(\NodeSet{\Digraph}) \subseteq \StateSet$,
  which are defined inductively as follows,
  for $\Time \in \Natural$, $\Node[2] \in \NodeSet{\Digraph}$ and $\edge{\Node[2]}{\Node[3]} \in \EdgeSet{\Digraph}$:
  \begin{align*}
    \Run_0(\Node[2]) &= \Run_0(\edge{\Node[2]}{\Node[3]}) = \InitFunc(\Labeling{\Digraph}(\Node[2])), \\[1ex]
    \Run_{\Time+1}(\Node[2]) &=
    \begin{cases*}
      \swl{\Run_{\Time}(\Node[2])}{\Run_{\Time}(\edge{\Node[2]}{\Node[3]}).\pushlast(\Run_{\Time+1}(\Node[2])).\popfirst}
        & if $\Timing_{\Time+1}(\Node[2]) = 0$, \\
      \TransFunc \bigl( \Run_{\Time}(\Node[2]),
                        \setbuilder{\Run_{\Time}(\edge{\Node[1]}{\Node[2]}).\first}{\edge{\Node[1]}{\Node[2]} \in \EdgeSet{\Digraph}}
                 \bigr)
        & if $\Timing_{\Time+1}(\Node[2]) = 1$,
    \end{cases*} \\[1ex]
    \Run_{\Time+1}(\edge{\Node[2]}{\Node[3]}) &=
    \begin{cases*}
      \Run_{\Time}(\edge{\Node[2]}{\Node[3]}).\pushlast(\Run_{\Time+1}(\Node[2]))
        & if $\Timing_{\Time+1}(\edge{\Node[2]}{\Node[3]}) = 0$, \\
      \Run_{\Time}(\edge{\Node[2]}{\Node[3]}).\pushlast(\Run_{\Time+1}(\Node[2])).\popfirst
        & if $\Timing_{\Time+1}(\edge{\Node[2]}{\Node[3]}) = 1$.
    \end{cases*}
  \end{align*}
  If $\Timing$ is the \kl{synchronous timing} of $\Digraph$,
  we refer to $\Run$ as the \Intro{synchronous run} of~$\Automaton$ on~$\Digraph$.
\end{definition}

Throughout this \lcnamecref{ch:nonlocal},
we assume that our \kl{digraphs}, \kl[distributed automata]{automata} and logical \kl{formulas}
agree on the number $\BitCount$ of \kl{labeling} bits.
An \kl[distributed automaton]{automaton} $\Automaton$ \Intro[asynchronous accepts]{accepts}
a \kl{pointed digraph} $\pver{\Digraph}{\Node}$
under \kl{timing} $\Timing$
if $\Node$ visits an \kl{accepting state} at some point
in the \kl[asynchronous run]{run} $\Run$ of $\Automaton$ on $\Digraph$ timed by $\Timing$,
i.e., if there exists $\Time \in \Natural$
such that $\Run_{\Time}(\Node) \in \AcceptSet$.
If we simply say that $\Automaton$ \kl[asynchronous accepts]{accepts} $\pver{\Digraph}{\Node}$,
without explicitly specifying a \kl{timing} $\Timing$,
then we stipulate that $\Run$ is the \kl{synchronous run} of $\Automaton$ on $\Digraph$.
Notice that this is coherent with
the definition of \kl{acceptance} presented in \cref{sec:distributed-automata}.

Given a \kl{digraph}
$\Digraph = \tuple{\NodeSet{\Digraph}, \EdgeSet{\Digraph}, \Labeling{\Digraph}}$
and a class $\TimingSet$ of \kl{timings} of $\Digraph$,
the \kl[distributed automaton]{automaton} $\Automaton$ is called \Intro{consistent}
for $\Digraph$ and $\TimingSet$
if for all $\Node \in \NodeSet{\Digraph}$,
either $\Automaton$ \kl[asynchronous accepts]{accepts} $\pver{\Digraph}{\Node}$
under every \kl{timing} in $\TimingSet$,
or $\Automaton$ does not \kl[asynchronous accept]{accept} $\pver{\Digraph}{\Node}$
under any \kl{timing} in $\TimingSet$.
We say that $\Automaton$ is \Intro[asynchronous automaton]{asynchronous}
if it is \kl{consistent} for every possible choice of $\Digraph$ and $\TimingSet$,
and \Intro[lossless-asynchronous automaton]{lossless-asynchronous}
if it is \kl{consistent} for every choice
where $\TimingSet$ contains only \kl{lossless-asynchronous timings}.
By contrast,
we call an \kl[distributed automaton]{automaton}
\Intro[synchronous automaton]{synchronous}
if we wish to emphasize
that no such consistency requirements are imposed.
Intuitively,
all \kl[distributed automata]{automata} can operate in the synchronous setting,
but only some of them also work reliably
in environments that provide fewer guarantees.

We denote by \Intro*{$\aDA$},\, \Intro*{$\laDA$} and \reintro*{$\DA$}
the classes of
\kl[asynchronous automata]{asynchronous},
\kl[lossless-asynchronous automata]{lossless-asynchro\-nous}
and \kl{synchronous automata},
respectively.
Similarly,
\Intro*{$\aQDA$},\, \Intro*{$\laQDA$} and \reintro*{$\QDA$}
are the corresponding classes of \kl{quasi-acyclic automata}.

Next,
we want to introduce the \kl{backward $\mu$-fragment},
for which it is convenient
to distinguish explicitly between \kl{constants} and \kl{variables}.
As our starting point,
we consider
$\bML$ restricted to~$\BitCount$ \kl{set constants}
and (arbitrarily many) unnegated \kl{set variables}.
Its \kl{formulas} are generated by the grammar
\begin{equation*}
  \Formula \Coloneqq \False
                \mid \True
                \mid \PosIn{\SetConstant_i}
                \mid \NOT \PosIn{\SetConstant_i}
                \mid \PosIn{\SetVariable}
                \mid (\Formula \OR \Formula)
                \mid (\Formula \AND \Formula)
                \mid \bdm \Formula
                \mid \bbx \Formula \,,
\end{equation*}
where
$\SetConstant_i \in \SetSymbolSet$
is considered to be a \kl{set constant},
for $1 \leq i \leq \BitCount$,
and
$\SetVariable \in
 \SetSymbolSet \setminus \set{\SetConstant_1,\dots,\SetConstant_{\BitCount}}$
is considered to be a \kl{set variable}.
Note that this syntax ensures that \kl{set variables} cannot be negated.

Traditionally,
the modal $\mu$-calculus is defined to comprise individual \kl{fixpoints}
which may be nested.
However,
it is well-known that we can add simultaneous \kl{fixpoints}
to the $\mu$-calculus
without changing its expressive power,
and that nested \kl{fixpoints} of the same type (i.e., least or greatest)
can be rewritten as non-nested simultaneous ones
(see, e.g., \cite[\S~3.7]{BradfieldS07} or \cite[\S~4.3]{Lenzi05}).
The following definition directly takes advantage of this fact.
We shall restrict ourselves to the
\Intro[backward $\mu$-fragment]{$\mu$-fragment of the backward $\mu$-calculus},
abbreviated \reintro{backward $\mu$-fragment},
where only \kl{least fixpoints} are allowed,
and where the usual \kl{modal operators} are replaced by their backward-looking variants.
Without loss of generality,
we stipulate that each \kl{formula} of
the \kl{backward $\mu$-fragment} with $\BitCount$ \kl{set constants}
is of the form
\Phantomintro{\MU}
\begin{equation*}
  \Formula \:=\: \reintro*{\MU}
  \begin{pmatrix}
    \SetVariable_1 \\
    \vdots \\
    \SetVariable_\VarMax
  \end{pmatrix}
  .
  \begin{pmatrix}
    \Formula_1(\SetConstant_1, \dots, \SetConstant_\BitCount, \SetVariable_1, \dots, \SetVariable_\VarMax) \\
    \vdots \\
    \Formula_\VarMax(\SetConstant_1, \dots, \SetConstant_\BitCount, \SetVariable_1, \dots, \SetVariable_\VarMax)
  \end{pmatrix},
\end{equation*}
where
$\SetVariable_1,\dots,\SetVariable_\VarMax \in
 \SetSymbolSet \setminus \set{\SetConstant_1,\dots,\SetConstant_{\BitCount}}$
are considered to be \kl{set variables},
and $\Formula_1,\dots,\Formula_\VarMax$ are \kl{formulas} of
$\bML$ with~$\BitCount$ \kl{set constants} and unnegated \kl{set variables}
that may contain no other \kl{set variables} than
$\SetVariable_1, \dots, \SetVariable_\VarMax$.
We shall denote the set of \kl{formulas}
of the \kl{backward $\mu$-fragment}
by~\Intro*{$\SigmaMu{1}(\bML)$}.

For every \kl{digraph}
$\Digraph = \tuple{\NodeSet{\Digraph}, \EdgeSet{\Digraph}, \Labeling{\Digraph}}$,
the tuple $\tuple{\Formula_1,\dots,\Formula_\VarMax}$
gives rise to an operator
$\Operator \colon (\powerset{\NodeSet{\Digraph}})^{\VarMax} \to (\powerset{\NodeSet{\Digraph}})^{\VarMax}$
that takes some valuation of
$\vec{\SetVariable} = \tuple{\SetVariable_1,\dots,\SetVariable_\VarMax}$
and reassigns to each $\SetVariable_i$ the resulting valuation of $\Formula_i$.
More formally,
$\Operator$~maps $\vec{\NodeSubset} = \tuple{\NodeSubset_1,\dots,\NodeSubset_\VarMax}$
to $\tuple{\NodeSubset'_1,\dots,\NodeSubset'_\VarMax}$
such that
$\NodeSubset'_i =
 \lsemf{\Formula_i}[\ver{\Digraph}{\vec{\SetVariable}}{\vec{\NodeSubset}}]$.
Here,
$\ver{\Digraph}{\vec{\SetVariable}}{\vec{\NodeSubset}}$
is the \kl{extended variant} of~$\Digraph$
that \kl{interprets} each $\SetVariable_i$~as~$\NodeSubset_i$.
A~(simultaneous) \Intro{fixpoint} of the operator~$\Operator$
is a tuple
$\vec{\NodeSubset} \in (\powerset{\NodeSet{\Digraph}})^{\VarMax}$
such that
$\Operator(\vec{\NodeSubset}) = \vec{\NodeSubset}$.
Since, by definition,
\kl{set variables} occur only positively in \kl{formulas},
the operator $\Operator$ is \Intro{monotonic}.
This means that
$\vec{\NodeSubset} \subseteq \vec{\NodeSubset}'$
implies
$\Operator(\vec{\NodeSubset}) \subseteq \Operator(\vec{\NodeSubset}')$
for all $\vec{\NodeSubset},\vec{\NodeSubset}' \in (\powerset{\NodeSet{\Digraph}})^{\VarMax}$,
where set inclusions are to be understood componentwise
(i.e., $\NodeSubset_i \subseteq \NodeSubset'_i$ for each $i$).
Therefore,
by virtue of a theorem due to Knaster and Tarski,
$\Operator$~has a \Intro{least fixpoint},
which is defined as the unique \kl{fixpoint}
$\vec{\LFixpoint} = \tuple{\LFixpoint_1,\dots,\LFixpoint_\VarMax}$
of~$\Operator$
such that
$\vec{\LFixpoint} \subseteq \vec{\NodeSubset}$
for every other \kl{fixpoint} $\vec{\NodeSubset}$ of $\Operator$.
As a matter of fact,
the Knaster-Tarski theorem even tells us that
$\vec{\LFixpoint}$ is equal to
$\bigcap \setbuilder{\vec{\NodeSubset} \in (\powerset{\NodeSet{\Digraph}})^{\VarMax}}
                    {\Operator(\vec{\NodeSubset}) \subseteq \vec{\NodeSubset}}$,
where set operations must also be understood componentwise.
Another, perhaps more intuitive, way of characterizing $\vec{\LFixpoint}$
is to consider the inductively constructed sequence of approximants
$\tuple{\vec{\LFixpoint}^0, \vec{\LFixpoint}^1, \vec{\LFixpoint}^2, \dots}$,
where
$\vec{\LFixpoint}^0 = \tuple{\EmptySet, \dots, \EmptySet}$
and
$\vec{\LFixpoint}^{j+1} = \Operator(\vec{\LFixpoint}^j)$.
Since this sequence is monotonically increasing and $\NodeSet{\Digraph}$ is finite,
there exists $n \in \Natural$ such that
$\vec{\LFixpoint}^n = \vec{\LFixpoint}^{n+1}$.
It is easy to check that
$\vec{\LFixpoint}^n$ coincides with the \kl{least fixpoint} $\vec{\LFixpoint}$.
For more details and proofs, see, e.g.,
\cite[\S~3.3.1]{DBLP:series/txtcs/GradelKLMSVVW07}.

\begin{figure}[tp]
  \centering
  \input{fig/automaton.tex}
  \captionsetup{singlelinecheck=off}
  \caption[A quasi-acyclic asynchronous distributed automaton.]{
    A \kl{quasi-acyclic} \kl{asynchronous distributed automaton}
    that is \kl[device equivalent]{equivalent} to the \kl{formula}
    $$\MU
    \begin{pmatrix}
      \SetVariable[1] \\[0.5ex]
      \SetVariable[2]
    \end{pmatrix}
    .
    \begin{pmatrix}
      (\PosIn{\SetConstant_1} \AND \PosIn{\SetVariable[2]}) \,\OR\, \bdm \PosIn{\SetVariable[1]} \\[0.5ex]
      \bbx \, \PosIn{\SetVariable[2]}
    \end{pmatrix}$$
    of the \kl{backward $\mu$-fragment}.
    A given $1$-bit \kl{labeled} \kl{pointed digraph} $\pver{\Digraph}{\Node[2]}$
    is \kl[asynchronous accepted]{accepted} by this \kl[distributed automaton]{automaton}
    if and only if,
    starting at $\Node[2]$
    and following $\Digraph$'s \kl{edges} in the backward direction,
    it is possible to reach some \kl{node}~$\Node[1]$ \kl{labeled} with~$1$
    from which it is impossible to reach any directed cycle.
  }
  \label{fig:automaton}
\end{figure}

Having introduced the necessary background,
we can finally establish the semantics of~$\Formula$
with respect to $\Digraph$:
the set
$\lsemf{\Formula}[\Digraph] =
 \lrsetbuilder{\Node \in \NodeSet{\Digraph}}
              {\pver{\Digraph}{\Node} \Models \Formula}$
of \kl{nodes} at which $\Formula$ holds
is precisely~$\LFixpoint_1$,
the first component of~$\vec{\LFixpoint}$.
Accordingly,
the \kl{pointed digraph}~$\pver{\Digraph}{\Node}$ lies in
the \kl[pointed-digraph language]{language}
$\semf{\Formula}[\pDIGRAPH[\BitCount][1]]$
\kl{defined} by~$\Formula$
if and only if
$\Node \in \LFixpoint_1$,
and we denote by
$\semF{\SigmaMu{1}(\bML)}[\pDIGRAPH[\BitCount][1]]$
the class of all \kl{pointed-digraph languages}
\kl{defined} by some \kl{formula} of the \kl{backward $\mu$-fragment}.

Figure~\ref{fig:automaton} provides an example
of a \kl{quasi-acyclic} \kl{asynchronous distributed automaton}
and an \kl[device equivalent]{equivalent} \kl{formula} of the \kl{backward $\mu$-fragment}.

\section{Equivalence with the backward mu-fragment}
\label{sec:result}

\marginnote{
  This may seem counterintuitive at first sight,
  but it is actually consistent
  with the standard terminology of distributed computing:
  an asynchronous algorithm can always serve as a synchronous algorithm
  (i.e., it can be executed in a synchronous environment),
  but the converse is not~true.
}
Based on the definitions given in Section~\ref{sec:preliminaries-nonlocal},
\kl{asynchronous automata} are a special case of \kl{lossless-asynchronous automata},
which in turn are a special case of \kl{synchronous automata}.
Furthermore,
quasi-acyclicity constitutes an additional
(possibly orthogonal) restriction on these models.
We thus immediately obtain the hierarchy of classes
depicted in Figure~\ref{fig:before-result}.

Our main result provides a simplification of this hierarchy:
the classes $\semA{\aQDA}[\pDIGRAPH[\BitCount][1]]$ and $\semA{\laQDA}[\pDIGRAPH[\BitCount][1]]$ are actually equal to
the class of \kl{pointed-digraph languages} \kl{definable} in the \kl{backward $\mu$-fragment}.
This yields the revised diagram shown in Figure~\ref{fig:after-result}.

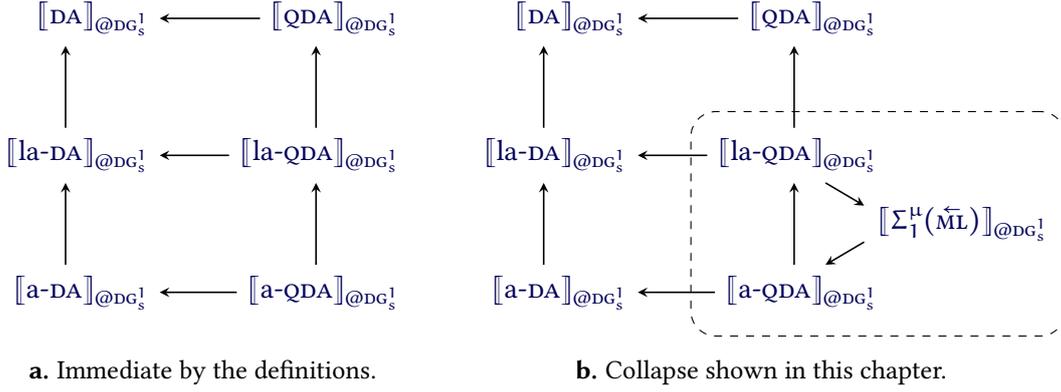
\begin{figure}[tp]
  \centering
  \widefloat{
    \begin{subfigure}[t]{0.45\textwidth}
      \captionsetup{skip=3ex}
      \centering
      \input{fig/before-result.tex}
      \caption{Immediate by the definitions.}
      \label{fig:before-result}
    \end{subfigure}
    \hspace{2ex}
    \begin{subfigure}[t]{0.6\textwidth}
      \captionsetup{skip=3ex}
      \centering
      \input{fig/after-result.tex}
      \caption{Collapse shown in this \lcnamecref{ch:nonlocal}.}
      \label{fig:after-result}
    \end{subfigure}
  }
  \caption{
    Hierarchy of the classes of \protect\kl{pointed-digraph languages}
    \kl{recognizable} by \kl{distributed automata}~($\DA$),
    depending on whether the \kl[distributed automata]{automata} are
    \kl[synchronous automata]{synchronous} (neither~“la” nor~“a”),
    \kl[lossless-asynchronous automata]{lossless-asynchronous}~(“la”),
    \kl[asynchronous automata]{asynchronous}~(“a”), or
    \kl{quasi-acyclic}~(“$\Acronym{q}$”).
    The arrows denote inclusion
    (e.g., $\semA{\laDA}[\pDIGRAPH[\BitCount][1]] \subseteq \semA{\DA}[\pDIGRAPH[\BitCount][1]]$).
  }
  \label{fig:result}
\end{figure}

\begin{theorem}[{$\semF{\SigmaMu{1}(\bML)}[\pDIGRAPH[\BitCount][1]] = \semA{\aQDA}[\pDIGRAPH[\BitCount][1]] = \semA{\laQDA}[\pDIGRAPH[\BitCount][1]]$}]
  \label{thm:main-result}
  When restricted to finite \kl{digraphs},
  the \kl{backward $\mu$-fragment}
  is effectively \kl[device equivalent]{equivalent} to the classes of
  \kl{quasi-acyclic} \kl{asynchronous automata} and
  \kl{quasi-acyclic} \kl{lossless-asynchronous automata}.
\end{theorem}
\begin{proof}
  The forward direction is given by
  Proposition~\ref{prp:logic-to-automata} (in Section~\ref{sec:logic-to-automata}),
  which asserts that $\semF{\SigmaMu{1}(\bML)}[\pDIGRAPH[\BitCount][1]] \subseteq \semA{\aQDA}[\pDIGRAPH[\BitCount][1]]$,
  and the trivial observation that $\semA{\aQDA}[\pDIGRAPH[\BitCount][1]] \subseteq \semA{\laQDA}[\pDIGRAPH[\BitCount][1]]$.
  For the backward direction,
  we use Proposition~\ref{prp:automata-to-logic} (in Section~\ref{sec:automata-to-logic}),
  which asserts that $\semA{\laQDA}[\pDIGRAPH[\BitCount][1]] \subseteq \semF{\SigmaMu{1}(\bML)}[\pDIGRAPH[\BitCount][1]]$.
\end{proof}

As stated before,
\kl{synchronous automata} are more powerful than the \kl{backward $\mu$-fragment}
(and incomparable with $\MSOL$).
This holds even if we consider only \kl{quasi-acyclic automata},
i.e., the inclusion $\semF{\SigmaMu{1}(\bML)}[\pDIGRAPH[\BitCount][1]] \subset \semA{\QDA}[\pDIGRAPH[\BitCount][1]]$ is known to be strict
(see \cite[Prp.~6]{DBLP:conf/csl/Kuusisto13}).
Moreover,
an upcoming paper will show that
the inclusion $\semA{\QDA}[\pDIGRAPH[\BitCount][1]] \subset \semA{\DA}[\pDIGRAPH[\BitCount][1]]$ is also strict.

In contrast,
it remains open whether \kl{quasi-acyclicity}
is in fact necessary for characterizing $\semF{\SigmaMu{1}(\bML)}[\pDIGRAPH[\BitCount][1]]$.
On the one hand,
this notion is crucial for our proof
(see Proposition~\ref{prp:automata-to-logic}),
but on the other hand,
no \kl{pointed-digraph language} separating $\semA{\aDA}[\pDIGRAPH[\BitCount][1]]$ or $\semA{\laDA}[\pDIGRAPH[\BitCount][1]]$ from $\semF{\SigmaMu{1}(\bML)}[\pDIGRAPH[\BitCount][1]]$
has been found so far.

\section{Computing least fixpoints using asynchronous automata}
\label{sec:logic-to-automata}

In this section,
we prove the easy direction of the main result.
Given a \kl{formula}~$\Formula$ of the \kl{backward $\mu$-fragment},
it is straightforward to construct
a (synchronous) \kl{distributed automaton}~$\Automaton$
that computes on any \kl{digraph} the \kl{least fixpoint}~$\vec{\LFixpoint}$
of the operator associated with~$\Formula$.
As long as it operates in the synchronous setting,
$\Automaton$ simply follows the sequence of approximants
$\tuple{\vec{\LFixpoint}^0, \vec{\LFixpoint}^1, \dots}$
described in Section~\ref{sec:preliminaries-nonlocal}.
It is important to stress that
the very same observation has previously been made
in~\cite[Prp.~7]{DBLP:conf/csl/Kuusisto13}
(formulated from a different point of view).
In the following proposition,
we refine this observation
by giving a more precise characterization
of the obtained \kl[distributed automaton]{automaton}:
it is always quasi-acyclic and
capable of operating in a (possibly lossy) asynchronous environment.

\begin{proposition}[{$\semF{\SigmaMu{1}(\bML)}[\pDIGRAPH[\BitCount][1]] \subseteq \semA{\aQDA}[\pDIGRAPH[\BitCount][1]]$}]
  \label{prp:logic-to-automata}
  For every \kl{formula} of the \kl{backward $\mu$-fragment},
  we can effectively construct
  an \kl[device equivalent]{equivalent} \kl{quasi-acyclic} \kl{asynchronous automaton}.
\end{proposition}
\begin{proof}
  Let 
  $\Formula = \MU \tuple{\SetVariable_1,\dots,\SetVariable_\VarMax} . \tuple{\Formula_1,\dots,\Formula_\VarMax}$
  be a \kl{formula} of the \kl{backward $\mu$-fragment}
  with $\BitCount$ \kl{set constants}.
  Without loss of generality,
  we may assume that the \kl[formulas]{subformulas} $\Formula_1,\dots,\Formula_\VarMax$
  do not contain any nested \kl{modalities}.
  To see this,
  suppose that ${\Formula_i = \bdm \Formula[2]}$.
  Then $\Formula$ is \kl[device equivalent]{equivalent} to
  $\Formula' = \MU \tuple{\SetVariable_1,\dots,\SetVariable_i,\dots,\SetVariable_\VarMax,\SetVariable[2]}
                 . \tuple{\Formula_1,\dots,\Formula'_i,\dots,\Formula_\VarMax,\Formula[2]}$,
  where $\SetVariable[2]$ is a fresh \kl{set variable}
  and $\Formula'_i = \bdm \SetVariable[2]$.
  The operator $\bbx$ and Boolean combinations of $\bdm$ and $\bbx$ are handled analogously.

  We now convert $\Formula$ into an \kl[device equivalent]{equivalent} \kl[distributed automaton]{automaton}
  $\Automaton = \tuple{\StateSet,\InitFunc,\TransFunc,\AcceptSet}$
  with \kl{state} set
  $\StateSet =
   \powerset{\set{\SetConstant_1,\dots,\SetConstant_\BitCount,\SetVariable_1,\dots,\SetVariable_\VarMax}}$.
  The idea is that
  each \kl{node} $\Node$ of the input \kl{digraph} has to remember
  which of the atomic propositions
  $\PosIn{\SetConstant_1},\dots,\PosIn{\SetConstant_\BitCount},
   \PosIn{\SetVariable_1},\dots,\PosIn{\SetVariable_\VarMax}$
  have, so far, been verified to hold at~$\Node$.
  Therefore,
  we define the \kl{initialization function} such that
  $\InitFunc(\String) = \setbuilder{\SetConstant_i}{\String(i)=1}$
  for all $\String \in \Boolean^\BitCount$.
  Let us write
  $\tuple{\State,\NeighborSet} \models \Formula_i$
  to indicate that a pair
  $\tuple{\State,\NeighborSet} \in \StateSet \times \powerset{\StateSet}$
  satisfies a \kl[formula]{subformula}~$\Formula_i$ of~$\Formula$.
  This is the case precisely when $\Formula_i$ holds at any \kl{node}~$\Node$
  that \kl{satisfies} exactly the atomic propositions in $\State$
  and whose \kl{incoming neighbors} \kl{satisfy} exactly
  the propositions specified by $\NeighborSet$.
  Note that this satisfaction relation is well-defined in our context
  because the nesting depth of \kl{modal operators} in $\Formula_i$
  is at most $1$.
  With that,
  the \kl{transition function} of~$\Automaton$ can be succinctly described by
  $\TransFunc(\State,\NeighborSet) =
   \State \cup \setbuilder{\SetVariable_i}{\tuple{\State,\NeighborSet} \models \Formula_i}$.
  Since $\State \subseteq \TransFunc(\State,\NeighborSet)$,
  we are guaranteed that the \kl[distributed automaton]{automaton} is \kl{quasi-acyclic}.
  Finally,
  the \kl[accepting state]{accepting} set is given by
  $\AcceptSet = \setbuilder{\State}{\SetVariable_1 \in \State}$.

  It remains to prove that $\Automaton$ is asynchronous and \kl[device equivalent]{equivalent} to $\Formula$.
  For this purpose,
  let
  $\Digraph = \tuple{\NodeSet{\Digraph}, \EdgeSet{\Digraph}, \Labeling{\Digraph}}$
  be an $\BitCount$-bit \kl{labeled} \kl{digraph}
  and
  $\vec{\LFixpoint} =
   \tuple{\LFixpoint_1,\dots,\LFixpoint_\VarMax} \in (\powerset{\NodeSet{\Digraph}})^{\VarMax}$
  be the \kl{least fixpoint} of the operator~$\Operator$
  associated with $\tuple{\Formula_1,\dots,\Formula_\VarMax}$.
  Due to the asynchrony condition,
  we must consider an arbitrary \kl{timing}
  $\Timing = \tuple{\Timing_1, \Timing_2, \dots}$
  of~$\Digraph$.
  The corresponding \kl[asynchronous run]{run} $\Run = \tuple{\Run_0, \Run_1, \dots}$
  of $\Automaton$ on $\Digraph$ timed by~$\Timing$
  engenders an infinite sequence
  $\tuple{\vec{\NodeSubset}^0, \vec{\NodeSubset}^1, \dots}$,
  where each tuple
  $\vec{\NodeSubset}^{\Time} =
   \tuple{\NodeSubset^{\Time}_1,\dots,\NodeSubset^{\Time}_\VarMax} \in (\powerset{\NodeSet{\Digraph}})^{\VarMax}$
  specifies the valuation of every \kl{set variable} $\SetVariable_i$ at time~$\Time$,
  i.e., $\NodeSubset^{\Time}_i = \setbuilder{\Node \in \NodeSet{\Digraph}}{\SetVariable_i \in \Run_{\Time}(\Node)}$.
  Since $\Automaton$ is quasi-acyclic and $\NodeSet{\Digraph}$ is finite,
  this sequence must eventually stabilize at some value $\vec{\NodeSubset}^\infty$,
  and each \kl{node} \kl[asynchronous accepts]{accepts} if and only if it belongs to~$\NodeSubset^\infty_1$.
  Reformulated this way,
  our task is to demonstrate that $\vec{\NodeSubset}^\infty$ equals $\vec{\LFixpoint}$,
  regardless of the \kl{timing}~$\Timing$.

  “$\vec{\NodeSubset}^\infty \subseteq \vec{\LFixpoint}$”:
  We show by induction that
  $\vec{\NodeSubset}^{\Time} \subseteq \vec{\LFixpoint}$
  for all $\Time \in \Natural$.
  This obviously holds for $\Time = 0$,
  since $\vec{\NodeSubset}^0 = \tuple{\EmptySet,\dots,\EmptySet}$.
  Now,
  consider any node $\Node \in \NodeSet{\Digraph}$ at an arbitrary time~$\Time$.
  Let $\State$ be the current state of $\Node$
  and $\NeighborSet$ be the set of current \kl{states} of its \kl{incoming neighbors}.
  Depending on~$\Timing$,
  it might be the case that $\Node$ actually receives
  some outdated information $\NeighborSet'$
  instead of~$\NeighborSet$.
  However,
  given that the \kl[incoming neighbors]{neighbors}' previous \kl{states}
  cannot contain more \kl{set variables} than their current ones
  (by construction),
  and that \kl{set variables} can only occur positively in each~$\Formula_i$,
  we know that
  $\tuple{\State,\NeighborSet'} \Models \Formula_i$
  implies
  $\tuple{\State,\NeighborSet} \Models \Formula_i$.
  Hence,
  if $\Node$ performs a local transition at time~$\Time$,
  then the only new \kl{set variables} that can be added to its \kl{state} must lie in
  $\setbuilder{\SetVariable_i}{\tuple{\State,\NeighborSet} \Models \Formula_i}$.
  On a global scale,
  this means that
  $\vec{\NodeSubset}^{\Time+1} \setminus \vec{\NodeSubset}^{\Time}
   \subseteq \Operator(\vec{\NodeSubset}^{\Time})$.
  Furthermore,
  by the induction hypothesis, the \kl{monotonicity} of $\Operator$,
  and the fact that $\vec{\LFixpoint}$ is a \kl{fixpoint},
  we have
  $\Operator(\vec{\NodeSubset}^{\Time}) \subseteq \Operator(\vec{\LFixpoint})
   = \vec{\LFixpoint}$.
  Putting both together,
  and again relying on the induction hypothesis,
  we obtain~$\vec{\NodeSubset}^{\Time+1} \subseteq \vec{\LFixpoint}$.

  “$\vec{\NodeSubset}^\infty \supseteq \vec{\LFixpoint}$”:
  For the converse direction,
  we make use of the Knaster-Tarski theorem,
  which gives us the equality
  $\vec{\LFixpoint} =
   \bigcap \setbuilder{\vec{\NodeSubset} \in (\powerset{\NodeSet{\Digraph}})^{\VarMax}}
                      {\Operator(\vec{\NodeSubset}) \subseteq \vec{\NodeSubset}}$.
  With this,
  it suffices to show that
  $\Operator(\vec{\NodeSubset}^\infty) \subseteq \vec{\NodeSubset}^\infty$.
  Consider some time $\Time \in \Natural$
  such that $\vec{\NodeSubset}^{\Time'} = \vec{\NodeSubset}^\infty$
  for all $\Time' \geq \Time$.
  Although we know
  that every \kl{node} has reached its final \kl{state} at time $\Time$,
  the $\FIFO$ buffers of some \kl{edges} might still contain obsolete \kl{states} from previous times.
  However,
  the \kl{fairness property} of~$\Timing$ guarantees that
  our customized $\popfirst$ operation is executed infinitely often at every \kl{edge},
  while the $\pushlast$ operation has no effect
  because all the \kl{states} remain unchanged.
  Therefore,
  there must be a time $\Time' \geq \Time$
  from which on each buffer contains only the current \kl{state}
  of its incoming \kl{node},
  i.e., $\Run_{\Time''}(\edge{\Node[1]}{\Node[2]}) = \Run_{\Time''}(\Node[1])$
  for all $\Time'' \geq \Time'$ and $\edge{\Node[1]}{\Node[2]} \in \EdgeSet{\Digraph}$.
  Moreover,
  the \kl{fairness property} of $\Timing$ also ensures that
  every \kl{node} $\Node$ reevaluates the local \kl{transition function}~$\TransFunc$ infinitely often,
  based on its own current \kl{state} $\State$
  and the set $\NeighborSet$ of \kl{states} in the buffers
  associated with its \kl{incoming neighbors}.
  As this has no influence on $\Node$'s \kl{state},
  we can deduce that
  $\setbuilder{\SetVariable_i}{\tuple{\State,\NeighborSet} \Models \Formula_i} \subseteq \State$.
  Consequently,
  we have
  $\Operator(\vec{\NodeSubset}^{\Time'}) \subseteq \vec{\NodeSubset}^{\Time'}$,
  which is equivalent to
  $\Operator(\vec{\NodeSubset}^\infty) \subseteq \vec{\NodeSubset}^\infty$.
\end{proof}

\section{Capturing asynchronous runs using least fixpoints}
\label{sec:automata-to-logic}

This section is dedicated to proving
the converse direction of the main result,
which will allow us to translate any
\kl{quasi-acyclic} \kl{lossless-asynchronous automaton}
into an \kl[device equivalent]{equivalent} \kl{formula} of the \kl{backward $\mu$-fragment}
(see Proposition~\ref{prp:automata-to-logic}).
Our proof builds on two concepts:
the invariance of \kl{distributed automata} under \kl{backward bisimulation}
(stated in Proposition~\ref{prp:bisimulation-invariance})
and an ad-hoc relation “$\enables$”
that captures the possible behaviors
of a fixed \kl{lossless-asynchronous automaton}~$\Automaton$
(in a specific sense described in Lemma~\ref{lem:enables-relation}).

We start with the notion of \kl{backward bisimulation},
which is defined like the standard notion of bisimulation
(see, e.g., \cite[Def.~2.16]{BlackburnRV02} or \cite[Def.~5]{BlackburnB07}),
except that \kl{edges} are followed in the backward direction.
Formally,
a \Intro{backward bisimulation}
between two $\BitCount$-bit \kl{labeled} \kl{digraphs}
$\Digraph = \tuple{\NodeSet{\Digraph}, \EdgeSet{\Digraph}, \Labeling{\Digraph}}$
and
$\Digraph' = \tuple{\NodeSet{\Digraph'}, \EdgeSet{\Digraph'}, \Labeling{\Digraph'}}$
is a binary relation
$\Bisimulation \subseteq \NodeSet{\Digraph} \times \NodeSet{\Digraph'}$
that fulfills the following conditions
for all $\edge{\Node[2]}{\Node[2]'} \in \Bisimulation$:
\begin{enumerate}
\item $\Labeling{\Digraph}(\Node[2]) = \Labeling{\Digraph'}(\Node[2]')$,
\item if $\edge{\Node[1]}{\Node[2]} \in \EdgeSet{\Digraph}$,
  then there exists $\Node[1]' \in \NodeSet{\Digraph'}$
  such that
  $\Node[1]'\Node[2]' \in \EdgeSet{\Digraph'}$
  and
  $\edge{\Node[1]}{\Node[1]'} \in \Bisimulation$, \\
  and, conversely,
\item if $\Node[1]'\Node[2]' \in \EdgeSet{\Digraph'}$,
  then there exists $\Node[1] \in \NodeSet{\Digraph}$
  such that
  $\edge{\Node[1]}{\Node[2]} \in \EdgeSet{\Digraph}$
  and
  $\edge{\Node[1]}{\Node[1]'} \in \Bisimulation$.
\end{enumerate}
We say that
the \kl{pointed digraphs}
$\pver{\Digraph}{\Node}$ and $\pver{\Digraph'}{\Node'}$
are \Intro{backward bisimilar}
if there exists such a \kl{backward bisimulation} $\Bisimulation$
relating $\Node$ and $\Node'$.
It is easy to see that
\kl{distributed automata} cannot distinguish
between \kl{backward bisimilar} \kl{structures}:

\begin{proposition}
  \label{prp:bisimulation-invariance}
  \kl{Distributed automata} are invariant under \kl{backward bisimulation}.
  That is,
  for every \kl[distributed automaton]{automaton} $\Automaton$,
  if two \kl{pointed digraphs}
  $\pver{\Digraph}{\Node[2]}$ and $\pver{\Digraph'}{\Node[2]'}$
  are \kl{backward bisimilar},
  then
  $\Automaton$ \kl[asynchronous accepts]{accepts} $\pver{\Digraph}{\Node[2]}$
  if and only if
  it \kl[asynchronous accepts]{accepts} $\pver{\Digraph'}{\Node[2]'}$.
\end{proposition}
\begin{proof}
  Let $\Bisimulation$ be a \kl{backward bisimulation}
  between $\Digraph$ and $\Digraph'$
  such that $\edge{\Node[2]}{\Node[2]'} \in \Bisimulation$.
  Since \kl[asynchronous acceptance]{acceptance} is defined with respect to
  the synchronous behavior of the \kl[distributed automaton]{automaton},
  we need only consider the \kl{synchronous runs}
  $\Run = \tuple{\Run_0, \Run_1, \dots}$ and
  $\Run' = \tuple{\Run'_0, \Run'_1, \dots}$
  of~$\Automaton$ on~$\Digraph$ and~$\Digraph'$, respectively.
  Now,
  given that the $\FIFO$ buffers on the \kl{edges} of the \kl{digraphs}
  merely contain the current \kl{state} of their incoming \kl{node},
  it is straightforward to prove by induction on $\Time$ that
  every pair of \kl{nodes} $\edge{\Node[1]}{\Node[1]'} \in \Bisimulation$ satisfies
  $\Run_{\Time}(\Node[1]) = \Run'_{\Time}(\Node[1]')$
  for all $\Time \in \Natural$.
\end{proof}

We now turn to the mentioned relation “$\enables$”,
which is defined with respect to
a fixed \kl[distributed automaton]{automaton}.
For the remainder of this section,
let~$\Automaton$ denote an \kl[distributed automaton]{automaton}
$\tuple{\StateSet,\InitFunc,\TransFunc,\AcceptSet}$,
and let~$\TraceSet$ denote its set of \kl{traces}.
The relation
\Intro*{${\enables} \subseteq (\powerset{\TraceSet} \times \TraceSet)$}
specifies whether,
in a lossless-asynchronous environment,
a given \kl{trace}~$\Trace$ can be traversed by a \kl{node}
whose \kl{incoming neighbors} traverse the \kl{traces} of a given set~$\NeighborHistory$.
Loosely speaking,
the intended meaning of
$\NeighborHistory \enables \Trace$ (“$\NeighborHistory$~enables~$\Trace$”)
is the following:
Take an appropriately chosen \kl{digraph}
under some \kl{lossless-asynchronous timing}~$\Timing$,
and observe the corresponding \kl[asynchronous run]{run} of $\Automaton$
up to a specific time~$\Time$;
if \kl{node}~$\Node$ was initially in \kl{state} $\Trace.\first$
and at time~$\Time$ it has \emph{seen} its \kl{incoming neighbors}
traversing precisely the \kl{traces} in~$\NeighborHistory$,
then it is possible for~$\Timing$ to be such that
at time~$\Time$,
\kl{node} $\Node$ has traversed exactly the \kl{trace}~$\Trace$.
This relation can be defined inductively:
As the base case,
we specify that for every
$\State \in \StateSet$ and $\NeighborSet \subseteq \StateSet$,
we have
${\NeighborSet \enables \State.\pushlast(\TransFunc(\State,\NeighborSet))}$.
For the inductive clause,
consider a \kl{trace} $\Trace[2] \in \TraceSet$
and two finite (possibly equal) sets of \kl{traces}
$\NeighborHistory, \NeighborHistory' \subseteq \TraceSet$
such that the \kl{traces} in $\NeighborHistory'$ can be obtained
by appending at most one \kl{state} to the \kl{traces} in $\NeighborHistory$.
More precisely,
if $\Trace[1] \in \NeighborHistory$,
then $\Trace[1].\pushlast(\State[1]) \in \NeighborHistory'$
for some $\State[1] \in \StateSet$,
and conversely,
if $\Trace[1]' \in \NeighborHistory'$,
then $\Trace[1]' = \Trace[1].\pushlast(\Trace[1]'.\last)$
for some $\Trace[1] \in \NeighborHistory$.
We shall denote this auxiliary relation by
\Intro*{$\NeighborHistory \becomes \NeighborHistory'$}.
If it holds,
then
$\NeighborHistory \enables \Trace[2]$
implies
$\NeighborHistory' \enables \Trace[2].\pushlast(\State[2])$,
where
$\State[2] =
 \TransFunc(\Trace[2].\last,
            \setbuilder{\Trace[1]'.\last}{\Trace[1]' \in \NeighborHistory'})$.

The next step is to show (in Lemma~\ref{lem:enables-relation})
that our definition of “$\enables$”
does indeed capture the intuition given above.
To formalize this,
we first introduce two further pieces of terminology.

First,
the notions of \kl[asynchronous configuration]{configuration} and \kl[asynchronous run]{run} can be enriched
to facilitate discussions about the past.
Let $\Run = \tuple{\Run_0, \Run_1, \dots}$
be a \kl[asynchronous run]{run} of $\Automaton$ on a \kl{digraph}
$\Digraph = \tuple{\NodeSet{\Digraph}, \EdgeSet{\Digraph}, \Labeling{\Digraph}}$
(timed by some \kl{timing}~$\Timing$).
The corresponding \Intro{enriched run} is the sequence
$\RichRun = \tuple{\RichRun_0, \RichRun_1, \dots}$
of \Intro{enriched configurations}
that we obtain from $\Run$
by requiring each \kl{node} to remember the entire \kl{trace} it has traversed so far.
Formally,
for $\Time \in \Natural$, $\Node \in \NodeSet{\Digraph}$ and $\Edge \in \EdgeSet{\Digraph}$,
\begin{equation*}
  \RichRun_0(\Node) = \Run_0(\Node),
  \quad\;
  \RichRun_{\Time+1}(\Node) = \RichRun_{\Time}(\Node).\pushlast(\Run_{\Time+1}(\Node))
  \quad \text{and} \quad
  \RichRun_{\Time}(\Edge) = \Run_{\Time}(\Edge).
\end{equation*}

Second,
we will need to consider finite segments of \kl{timings} and \kl{enriched runs}.
A \Intro{lossless-asynchronous timing segment} of a \kl{digraph}~$\Digraph$
is a sequence
${\Timing = \tuple{\Timing_1, \dots, \Timing_\EndTime}}$
that could be extended to a whole \kl{lossless-asynchronous timing}
$\tuple{\Timing_1, \dots, \Timing_\EndTime, \Timing_{\EndTime+1}, \dots}$.
Likewise,
for an initial \kl{enriched configuration}~$\RichRun_0$ of~$\Digraph$,
the corresponding \Intro{enriched run segment} timed by $\Timing$
is the sequence $\tuple{\RichRun_0, \dots, \RichRun_\EndTime}$,
where each $\RichRun_{\Time+1}$ is computed
from $\RichRun_{\Time}$ and $\Timing_{\Time+1}$
in the same way as for an entire \kl{enriched run}.

Equipped with the necessary terminology,
we can now state and prove a (slightly technical) lemma
that will allow us to derive benefit from the relation “$\enables$”.
This lemma essentially states that
if $\NeighborHistory \enables \Trace[2]$ holds
and we are given enough \kl{nodes}
that traverse the \kl{traces} in~$\NeighborHistory$,
then we can take those \kl{nodes} as the \kl{incoming neighbors}
of a new \kl{node}~$\Node[2]$
and delay the messages received by~$\Node[2]$
in such a way that $\Node[2]$ traverses $\Trace[2]$,
without losing any messages.

\begin{lemma}
  \label{lem:enables-relation}
  For every \kl{trace} $\Trace[2] \in \TraceSet$
  and every finite (possibly empty) set of \kl{traces}
  $\NeighborHistory = \set{\Trace[1]_1, \dots, \Trace[1]_\ell} \subseteq \TraceSet$
  that satisfy the relation $\NeighborHistory \enables \Trace[2]$,
  there are lower bounds $m_1, \dots, m_\ell \in \Positive$
  such that the following statement holds true:

  For any $n_1, \dots, n_\ell \in \Positive$ satisfying $n_i \geq m_i$,
  let $\Digraph$ be a \kl{digraph} consisting of
  the \kl{nodes} $\tuple{\Node[1]_i^j}_{i,j}$ and $\Node[2]$,
  and the \kl{edges} $\tuple{\Node[1]_i^j\Node[2]}_{i,j}$,
  with index ranges $1 \leq i \leq \ell$ and $1 \leq j \leq n_i$.
  If we start from the \kl{enriched configuration}~$\RichRun_0$ of~$\Digraph$,
  where
  \begin{equation*}
    \RichRun_0(\Node[1]_i^j) = \Trace[1]_i,
    \quad\;
    \RichRun_0(\Node[1]_i^j\Node[2]) = \Trace[1]_i
    \quad \text{and} \quad
    \RichRun_0(\Node[2]) = \Trace[2].\first,
  \end{equation*}
  then we can construct a (nonempty) \kl{lossless-asynchronous timing segment}
  $\Timing = \tuple{\Timing_1, \dots, \Timing_\EndTime}$ of~$\Digraph$,
  where
  $\Timing_{\Time}(\Node[1]_i^j) = 0$ and
  $\Timing_{\Time}(\Node[2]) = 1$
  for $1 \leq \Time \leq \EndTime$,
  such that the corresponding \kl{enriched run segment}
  $\RichRun = \tuple{\RichRun_0, \dots, \RichRun_\EndTime}$ timed by $\Timing$
  satisfies
  \begin{equation*}
    \RichRun_{\EndTime-1}(\Node[1]_i^j\Node[2]) = \Trace[1]_i.\last
    \quad \text{and} \quad
    \RichRun_\EndTime(\Node[2]) = \Trace[2].
    \qedhere
  \end{equation*}
\end{lemma}
\begin{proof}
  We proceed by induction on the definition of~“$\enables$”.
  In the base case,
  where
  $\NeighborHistory = \set{\State[1]_1, \dots, \State[1]_\ell} \subseteq \StateSet$
  and
  $\Trace[2] = \State[2].\pushlast(\TransFunc(\State[2],\NeighborHistory))$
  for some $\State[2] \in \StateSet$,
  the statement holds with $m_1 = \dots = m_\ell = 1$.
  This is witnessed by a \kl[lossless-asynchronous timing segment]{timing segment}
  $\Timing = \tuple{\Timing_1}$,
  where
  $\Timing_1(\Node[1]_i^j) = 0$,\,
  $\Timing_1(\Node[2]) = 1$,
  and $\Timing_1(\Node[1]_i^j\Node[2])$ can be chosen as desired.

  For the inductive step,
  assume that the statement holds for
  $\Trace[2]$ and
  $\NeighborHistory = \set{\Trace[1]_1, \dots, \Trace[1]_\ell}$
  with some values $m_1, \dots, m_\ell$.
  Now consider any other set of \kl{traces}
  $\NeighborHistory' = \set{\Trace[1]'_1, \dots, \Trace[1]'_{\ell'}}$
  such that $\NeighborHistory \becomes \NeighborHistory'$,
  and let $\Trace[2]' = \Trace[2].\pushlast(\State[2])$,
  where
  $\State[2] =
   \TransFunc(\Trace[2].\last,
              \setbuilder{\Trace[1]'_k.\last}{\Trace[1]'_k \in \NeighborHistory'})$.
  Since $\NeighborHistory \enables \Trace[2]$,
  we have $\NeighborHistory' \enables \Trace[2]'$.
  The remainder of the proof consists in showing that
  the statement also holds for $\Trace[2]'$ and $\NeighborHistory'$
  with some large enough integers $m'_1, \dots, m'_{\ell'}$.
  Let us fix
  $m'_k = \sum \setbuilder{m_i}{\Trace[1]_i.\pushlast(\Trace[1]'_k.\last) = \Trace[1]'_k}$.
  (As there is no need to find minimal values, we opt for easy expressibility.)

  Given any numbers $n'_1, \dots, n'_{\ell'}$ with $n'_k \geq m'_k$,
  we choose suitable values $n_1, \dots, n_\ell$ with $n_i \geq m_i$,
  and consider the corresponding \kl{digraph} $\Digraph$ described in the lemma.
  Because we have $\NeighborHistory \becomes \NeighborHistory'$,
  we can assign to each \kl{node} $\Node[1]_i^j$ a \kl{state} $\State[1]_i^j$
  such that
  $\Trace[1]_i.\pushlast(\State[1]_i^j) \in \NeighborHistory'$.
  Moreover,
  provided our choice of $n_1, \dots, n_\ell$ was adequate,
  we can also ensure that for each $\Trace[1]'_k \in \NeighborHistory'$,
  there are exactly $n'_k$ \kl{nodes} $\Node[1]_i^j$ such that
  $\Trace[1]_i.\pushlast(\State[1]_i^j) = \Trace[1]'_k$.
  (Note that \kl{nodes} with distinct \kl{traces}
  $\Trace[1]_i, \Trace[1]_{i'} \in \NeighborHistory$
  might be mapped to the same \kl{trace}
  $\Trace[1]'_k \in \NeighborHistory'$,
  in case $\Trace[1]_{i'} = \Trace[1]_i\State[1]_i^j$.)
  It is straightforward to verify
  that such a choice of numbers and such an assignment of \kl{states}
  are always possible,
  given the lower bounds $m'_1, \dots, m'_{\ell'}$ specified above.

  Let us now consider
  the \kl{lossless-asynchronous timing segment}
  $\Timing = \tuple{\Timing_1, \dots, \Timing_\EndTime}$
  and the corresponding \kl{enriched run segment}
  $\RichRun = \tuple{\RichRun_0, \dots, \RichRun_\EndTime}$
  provided by the induction hypothesis.
  Since the $\popfirst$ operation has no effect on a \kl{trace} of length $1$,
  we may assume without loss of generality
  that
  $\Timing_{\Time}(\Node[1]_i^j\Node[2]) = 0$
  if
  $\RichRun_{\Time-1}(\Node[1]_i^j\Node[2])$ has length $1$,
  for $\Time < \EndTime$.
  Consequently,
  if we start from the alternative \kl{enriched configuration}~$\RichRun'_0$,
  where
  \begin{equation*}
    \RichRun'_0(\Node[1]_i^j) = \Trace[1]_i.\pushlast(\State[1]_i^j),
    \quad\;
    \RichRun'_0(\Node[1]_i^j\Node[2]) = \Trace[1]_i.\pushlast(\State[1]_i^j)
    \quad \text{and} \quad
    \RichRun'_0(\Node[2]) = \Trace[2].\first,
  \end{equation*}
  then the corresponding \kl{enriched run segment}
  $\tuple{\RichRun'_0, \dots, \RichRun'_\EndTime}$ timed by $\Timing$
  can be derived from~$\RichRun$ by simply applying “$\pushlast(\State[1]_i^j)$”
  to
  $\RichRun_{\Time}(\Node[1]_i^j)$ and $\RichRun_{\Time}(\Node[1]_i^j\Node[2])$,
  for $\Time < \EndTime$.
  We thus~get
  \begin{equation*}
    \RichRun'_{\EndTime-1}(\Node[1]_i^j\Node[2]) = \Trace[1]_i.\last.\pushlast(\State[1]_i^j)
    \quad \text{and} \quad
    \RichRun'_\EndTime(\Node[2]) = \Trace[2].
  \end{equation*}
  We may also assume without loss of generality
  that $\Timing_\EndTime(\Node[1]_i^j\Node[2]) = 1$
  if $\RichRun'_{\EndTime-1}(\Node[1]_i^j\Node[2])$ has length~$2$,
  since this does not affect $\RichRun$
  and lossless-asynchrony is ensured by $\Timing_\EndTime(\Node[2]) = 1$.
  Hence,
  it suffices to extend $\Timing$ by an additional map $\Timing_{\EndTime+1}$,
  where
  $\Timing_{\EndTime+1}(\Node[1]_i^j) = 0$,\,
  $\Timing_{\EndTime+1}(\Node[2]) = 1$,
  and $\Timing_{\EndTime+1}(\Node[1]_i^j\Node[2])$ can be chosen as desired.
  The resulting \kl{enriched run segment}
  $\tuple{\RichRun'_0, \dots, \RichRun'_{\EndTime+1}}$
  satisfies
  \begin{align*}
    \RichRun'_\EndTime(\Node[1]_i^j\Node[2]) &= \State[1]_i^j = \Trace[1]'_k.\last
    \quad \text{(for some $\Trace[1]'_k \in \NeighborHistory'$)}
    \quad \text{and} \\
    \RichRun'_{\EndTime+1}(\Node[2]) &= \Trace[2].\pushlast(\State[2]) = \Trace[2]'.
    \qedhere
  \end{align*}
\end{proof}

Finally,
we can put all the pieces together
and prove the converse direction of Theorem~\ref{thm:main-result}:

\begin{proposition}[{$\semA{\laQDA}[\pDIGRAPH[\BitCount][1]] \subseteq \semF{\SigmaMu{1}(\bML)}[\pDIGRAPH[\BitCount][1]]$}]
  \label{prp:automata-to-logic}
  For every \kl{quasi-acyclic} \kl{lossless-asynchronous automaton},
  we can effectively construct
  an \kl[device equivalent]{equivalent} \kl{formula} of the \kl{backward $\mu$-fragment}.
\end{proposition}
\begin{proof}
  Assume that
  $\Automaton = \tuple{\StateSet,\InitFunc,\TransFunc,\AcceptSet}$
  is a \kl{quasi-acyclic} \kl{lossless-asynchronous automaton}
  over $\BitCount$-bit \kl{labeled} \kl{digraphs}.
  Since it is quasi-acyclic,
  its set of \kl{traces} $\TraceSet$ is finite,
  and thus we can afford to introduce
  a separate \kl{set variable} $\SetVariable_{\Trace}$
  for each \kl{trace} $\Trace \in \TraceSet$.
  Making use of the relation~“$\enables$”,
  we convert~$\Automaton$ into an \kl[device equivalent]{equivalent} \kl{formula}
  $\Formula =
   \MU \bigl[ \SetVariable_1, \tuple{\SetVariable_{\Trace}}_{\Trace \in \TraceSet} \bigr] .
       \bigl[ \Formula_1, \tuple{\Formula_{\Trace}}_{\Trace \in \TraceSet} \bigr]$
  of the \kl{backward $\mu$-fragment},
  where
  \begin{align}
    \Formula_1 &=
    \smashoperator[r]{\bigvee_{\substack{\Trace \in \TraceSet \\ \Trace\ldotp\last \in \AcceptSet}}}
    \SetVariable_{\Trace},
    \tag{a} \label{eq:accept} \\
    \Formula_{\State} &=
    \smashoperator[r]{\bigvee_{\substack{\String \in \Boolean^\BitCount \\ \InitFunc(\String)=\State}}} \:\:
    \Bigl(
      \smashoperator[r]{\bigwedge_{\String(i)=1}} \SetConstant_i
      \,\AND\!
      \smashoperator[r]{\bigwedge_{\String(i)=0}} \NOT \SetConstant_i
    \Bigr)
    \quad\; \text{for $\State \in \StateSet$,\quad and}
    \tag{b} \label{eq:init} \\
    \Formula_{\Trace[2]} &=
    \SetVariable_{\Trace[2]\ldotp\first}
    \,\AND
    \bigvee_{\substack{\NeighborHistory \subseteq \TraceSet \\ \NeighborHistory \:\!\enables\:\! \Trace[2]}}
    \Bigl(
      \bigl(\, \smashoperator{\bigwedge_{\Trace[1] \in \NeighborHistory}} \bdm \SetVariable_{\Trace[1]} \bigr)
      \AND
      \bigl( \bbx \smashoperator{\bigvee_{\Trace[1] \in \NeighborHistory}} \SetVariable_{\Trace[1]} \bigr)
    \Bigr)
    \quad\; \text{for $\Trace[2] \in \TraceSet$ with $\length{\Trace[2]} \geq 2$.}
    \tag{c} \label{eq:transition}
  \end{align}
  Note that this \kl{formula} can be constructed effectively
  because an inductive computation of~“$\enables$” must terminate
  after at most $\card{\TraceSet} \cdot 2^{\card{\TraceSet}}$ iterations.

  To prove that $\Formula$ is indeed \kl[device equivalent]{equivalent} to $\Automaton$,
  let us consider an arbitrary $\BitCount$-bit \kl{labeled} \kl{digraph}
  $\Digraph = \tuple{\NodeSet{\Digraph}, \EdgeSet{\Digraph}, \Labeling{\Digraph}}$
  and the corresponding \kl{least fixpoint}
  $\vec{\LFixpoint} =
   \tuple{\LFixpoint_1, \tuple{\LFixpoint_{\Trace}}_{\Trace \in \TraceSet}}
   \in (\powerset{\NodeSet{\Digraph}})^{\card{\TraceSet}+1}$
  of the operator $\Operator$ associated with
  $\tuple{\Formula_1, \tuple{\Formula_{\Trace}}_{\Trace \in \TraceSet}}$.

  The easy direction is to show that
  for all \kl{nodes} $\Node \in \NodeSet{\Digraph}$,
  if $\Automaton$ \kl[asynchronous accepts]{accepts} $\pver{\Digraph}{\Node}$,
  then $\pver{\Digraph}{\Node}$ \kl{satisfies} $\Formula$.
  For that,
  it suffices to consider the \kl[synchronous run]{synchronous} \kl{enriched run}
  $\RichRun = \tuple{\RichRun_0, \RichRun_1, \dots}$
  of~$\Automaton$ on~$\Digraph$.
  (Any other \kl[asynchronous run]{run} timed by a \kl{lossless-asynchronous timing}
  would exhibit the same \kl[asynchronous acceptance behavior]{acceptance behavior}.)
  As in the proof of Proposition~\ref{prp:bisimulation-invariance},
  we can simply ignore the $\FIFO$ buffers on the \kl{edges} of~$\Digraph$
  because $\RichRun_{\Time}(\edge{\Node[1]}{\Node[2]}) = \RichRun_{\Time}(\Node[1]).\last$.
  Using this,
  a straightforward induction on $\Time$ shows that
  every node $\Node[2] \in \NodeSet{\Digraph}$ satisfies
  ${\setbuilder{\RichRun_{\Time}(\Node[1])}{\edge{\Node[1]}{\Node[2]} \in \EdgeSet{\Digraph}}
    \enables \RichRun_{\Time+1}(\Node[2])}$
  for all $\Time \in \Natural$.
  (For~$\Time = 0$,
   the claim follows from the base case of the definition of~“$\enables$”;
   for the step from~$\Time$ to~$\Time+1$,
   we can immediately apply the inductive clause of the definition.)
  This in turn allows us to prove that
  each \kl{node}~$\Node[2]$ is contained in all the components of $\vec{\LFixpoint}$
  that correspond to a \kl{trace} traversed by~$\Node[2]$ in~$\RichRun$,
  i.e., $\Node[2] \in \LFixpoint_{\RichRun_{\Time}(\Node[2])}$ for all $\Time \in \Natural$.
  Naturally, we proceed again by induction:
  For $\Time = 0$,
  we have $\RichRun_0(\Node[2]) = \InitFunc(\Labeling{\Digraph}(\Node[2])) \in \StateSet$,
  hence the \kl[formula]{subformula} $\Formula_{\RichRun_0(\Node[2])}$
  defined in equation~\eqref{eq:init}
  holds at~$\Node[2]$,
  and thus $\Node[2] \in \LFixpoint_{\RichRun_0(\Node[2])}$.
  For the step from~$\Time$ to~$\Time+1$,
  we need to distinguish two cases.
  If $\RichRun_{\Time+1}(\Node[2])$ is of length $1$,
  then it is equal to $\RichRun_{\Time}(\Node[2])$,
  and there is nothing new to prove.
  Otherwise,
  we must consider the appropriate \kl[formula]{subformula}
  $\Formula_{\RichRun_{\Time+1}(\Node[2])}$
  given by equation~\eqref{eq:transition}.
  We already know from the base case that the conjunct
  $\SetVariable_{\RichRun_{\Time+1}(\Node[2])\ldotp\first} = \SetVariable_{\RichRun_0(\Node[2])}$
  holds at $\Node[2]$,
  with respect to any \kl[set variable]{variable} assignment
  that \kl{interprets} each $\SetVariable_{\Trace}$ as $\LFixpoint_{\Trace}$.
  Furthermore,
  by the induction hypothesis,
  $\SetVariable_{\RichRun_{\Time}(\Node[1])}$
  holds at every \kl{incoming neighbor} $\Node[1]$ of $\Node[2]$.
  Since
  ${\setbuilder{\RichRun_{\Time}(\Node[1])}{\edge{\Node[1]}{\Node[2]} \in \EdgeSet{\Digraph}}
    \enables \RichRun_{\Time+1}(\Node[2])}$,
  we conclude that the second conjunct of $\Formula_{\RichRun_{\Time+1}(\Node[2])}$
  must also hold at $\Node[2]$,
  and thus $\Node[2] \in \LFixpoint_{\RichRun_{\Time+1}(\Node[2])}$.
  Finally,
  assuming $\Automaton$ \kl[asynchronous accepts]{accepts} $\pver{\Digraph}{\Node}$,
  we know by definition that
  $\RichRun_{\Time}(\Node).\last \in \AcceptSet$ for some $\Time \in \Natural$.
  Since $\Node \in \LFixpoint_{\RichRun_{\Time}(\Node)}$,
  this implies that
  the \kl[formula]{subformula}~$\Formula_1$ defined in equation~\eqref{eq:accept}
  holds at $\Node$,
  and therefore that $\pver{\Digraph}{\Node}$ \kl{satisfies} $\Formula$.

  For the converse direction of the \kl[device equivalence]{equivalence},
  we have to overcome the difficulty
  that~$\Formula$ is more permissive than~$\Automaton$,
  in the sense that
  a \kl{node} $\Node$ might lie in $\LFixpoint_{\Trace}$,
  and yet not be able to follow the trace $\Trace$ under any \kl{timing} of $\Digraph$.
  Intuitively,
  the reason why we still obtain an \kl[device equivalence]{equivalence}
  is that $\Automaton$ cannot take advantage
  of all the information provided by any particular \kl[asynchronous run]{run},
  because it must ensure that for \emph{all} \kl{digraphs},
  its \kl[asynchronous acceptance behavior]{acceptance behavior} is independent of the \kl{timing}.
  It turns out that
  even if $\Node$ cannot traverse $\Trace$,
  some other \kl{node}~$\Node'$ in an indistinguishable \kl{digraph}
  will be able to do so.
  More precisely, we will show that
  \begin{equation}
    \parbox{0.75\textwidth}{
      if $\Node \in \LFixpoint_{\Trace}$,
      then there exists a \kl{pointed digraph} $\pver{\Digraph'}{\Node'}$,
      \kl{backward bisimilar} to $\pver{\Digraph}{\Node}$,
      and a \kl{lossless-asynchronous timing} $\Timing'$ of $\Digraph'$,
      such that $\RichRun'_{\Time}(\Node') = \Trace$
      for some $\Time \in \Natural$,
    }
    \tag{$\ast$} \label{stm:traversable}
  \end{equation}
  where $\RichRun'$ is the \kl{enriched run}
  of $\Automaton$ on $\Digraph'$ timed by $\Timing'$.
  Now suppose that $\pver{\Digraph}{\Node}$ \kl{satisfies} $\Formula$.
  By equation~\eqref{eq:accept},
  this means that $\Node \in \LFixpoint_{\Trace}$
  for some trace $\Trace$ such that $\Trace.\last \in \AcceptSet$.
  Consequently,
  $\Automaton$~\kl[asynchronous accepts]{accepts}
  the \kl{pointed digraph} $\pver{\Digraph'}{\Node'}$
  postulated in~\eqref{stm:traversable},
  based on the claim that~$\Node'$ traverses~$\Trace$ under \kl{timing}~$\Timing'$
  and the fact that $\Automaton$ is lossless-asynchronous.
  Since $\pver{\Digraph}{\Node}$ and $\pver{\Digraph'}{\Node'}$
  are \kl{backward bisimilar},
  it follows from Proposition~\ref{prp:bisimulation-invariance}
  that~$\Automaton$ also \kl[asynchronous accepts]{accepts}~$\pver{\Digraph}{\Node}$.

  It remains to verify \eqref{stm:traversable}.
  We achieve this by
  computing the \kl{least fixpoint}~$\vec{\LFixpoint}$ inductively
  and proving the statement by induction on the sequence of approximants
  $\tuple{\vec{\LFixpoint}^0, \vec{\LFixpoint}^1, \dots}$.
  Note that we do not need to consider the limit case,
  since $\vec{\LFixpoint} = \vec{\LFixpoint}^n$ for some $n \in \Natural$.

  The base case is trivially true
  because all the components of $\vec{\LFixpoint}^0$ are empty.
  Furthermore,
  if~$\Trace[2]$~consists of a single \kl{state}~$\State$,
  then we do not even need to argue by induction,
  as it is evident from equation~\eqref{eq:init} that
  for all $j \geq 1$,
  \kl{node} $\Node[2]$ lies in $\LFixpoint^j_{\State}$
  precisely when
  $\InitFunc(\Labeling{\Digraph}(\Node[2]))= \State$.
  It thus suffices to set $\pver{\Digraph'}{\Node[2]'} = \pver{\Digraph}{\Node[2]}$
  and choose the \kl{timing}~$\Timing'$ arbitrarily.
  Clearly,
  we have $\RichRun'_0(\Node[2]') = \InitFunc(\Labeling{\Digraph}(\Node[2]))= \State$
  if $\Node[2] \in \LFixpoint^j_{\State}$.

  On the other hand,
  if~$\Trace[2]$ is of length at least $2$,
  we must assume that
  statement~\eqref{stm:traversable} holds for the components of~$\vec{\LFixpoint}^j$
  in order to prove it for $\LFixpoint^{j+1}_{\Trace[2]}$.
  To this end,
  consider an arbitrary \kl{node} $\Node[2] \in \LFixpoint^{j+1}_{\Trace[2]}$.
  By the first conjunct in~\eqref{eq:transition}
  and the preceding remarks regarding the trivial cases,
  we know that $\InitFunc(\Labeling{\Digraph}(\Node[2])) = \Trace[2].\first$
  (and incidentally that $j \geq 1$).
  Moreover,
  the second conjunct ensures
  the existence of a (possibly empty) set of \kl{traces} $\NeighborHistory$
  that satisfies
  $\NeighborHistory \enables \Trace[2]$
  and that represents
  a “projection” of~$\Node[2]$'s \kl{incoming neighborhood} at stage~$j$.
  By the latter we mean that
  for all $\Trace[1] \in \NeighborHistory$,
  there exists $\Node[1] \in \NodeSet{\Digraph}$
  such that
  $\edge{\Node[1]}{\Node[2]} \in \EdgeSet{\Digraph}$ and $\Node[1] \in \LFixpoint^j_{\Trace[1]}$,
  and conversely,
  for all $\Node[1] \in \NodeSet{\Digraph}$ with $\edge{\Node[1]}{\Node[2]} \in \EdgeSet{\Digraph}$,
  there exists $\Trace[1] \in \NeighborHistory$
  such that $\Node[1] \in \LFixpoint^j_{\Trace[1]}$.

  Now,
  for each \kl{trace} $\Trace[1] \in \NeighborHistory$
  and each \kl{incoming neighbor}~$\Node[1]$ of~$\Node[2]$
  that is contained in $\LFixpoint^j_{\Trace[1]}$,
  the induction hypothesis provides us with a \kl{pointed digraph}
  $\pver{\Digraph'_{\Node[1]:\Trace[1]}}{\Node[1]'_{\Trace[1]}}$
  and a corresponding \kl{timing} $\Timing'_{\Node[1]:\Trace[1]}$,
  as described in~\eqref{stm:traversable}.
  We make $n_{\Node[1]:\Trace[1]} \in \Natural$ distinct copies
  of each such \kl{digraph}~$\Digraph'_{\Node[1]:\Trace[1]}$.
  From this,
  we construct $\Digraph' = \tuple{\NodeSet{\Digraph'}, \EdgeSet{\Digraph'}, \Labeling{\Digraph'}}$
  by taking the disjoint union of all the $\sum n_{\Node[1]:\Trace[1]}$ \kl{digraphs},
  and adding a single new node $\Node[2]'$
  with $\Labeling{\Digraph'}(\Node[2]') = \Labeling{\Digraph}(\Node[2])$,
  together with all the \kl{edges} of the form~$\Node[1]'_{\Trace[1]}\Node[2]'$
  (i.e., one such \kl{edge} for each copy of every~$\Node[1]'_{\Trace[1]}$).
  Given that every $\pver{\Digraph'_{\Node[1]:\Trace[1]}}{\Node[1]'_{\Trace[1]}}$
  is \kl{backward bisimilar} to $\pver{\Digraph}{\Node[1]}$,
  we can guarantee that the same holds
  for $\pver{\Digraph'}{\Node[2]'}$ and $\pver{\Digraph}{\Node[2]}$
  by choosing the numbers of \kl{digraph} copies in $\Digraph'$ such that
  each \kl{incoming neighbor}~$\Node[1]$ of~$\Node[2]$ is represented
  by at least one \kl{incoming neighbor} of~$\Node[2]'$.
  That is, for every~$\Node[1]$,
  we require that ${n_{\Node[1]:\Trace[1]} \geq 1}$ for some~$\Trace[1]$.

  Finally,
  we construct a suitable \kl{lossless-asynchronous timing} $\Timing'$ of~$\Digraph'$,
  which proceeds in two phases
  to make $\Node[2]'$ traverse $\Trace[2]$
  in the corresponding \kl{enriched run}~$\RichRun'$.
  In the first phase,
  where $0 < \Time \leq \Time_1$,
  node~$\Node[2]'$ remains inactive,
  which means that
  every~$\Timing_{\Time}$ assigns~$0$ to~$\Node[2]'$ and its incoming \kl{edges}.
  The \kl{state} of~$\Node[2]'$ at time~$\Time_1$ is thus still $\Trace[2].\first$.
  Meanwhile,
  in every copy of each \kl{digraph}~$\Digraph'_{\Node[1]:\Trace[1]}$,
  the \kl{nodes} and \kl{edges} behave according to \kl{timing}~$\Timing'_{\Node[1]:\Trace[1]}$
  until the respective copy of~$\Node[1]'_{\Trace[1]}$
  has completely traversed~$\Trace[1]$,
  whereupon the entire subgraph becomes inactive.
  By choosing $\Time_1$ large enough,
  we make sure that the $\FIFO$ buffer on each \kl{edge}
  of the form~$\Node[1]'_{\Trace[1]}\Node[2]'$
  contains precisely $\Trace[1]$ at time~$\Time_1$.
  In the second phase,
  which lasts from~$\Time_1 + 1$ to~$\Time_2$,
  the only active parts of $\Digraph'$ are $\Node[2]'$ and its incoming \kl{edges}.
  Since the number~$n_{\Node[1]:\Trace[1]}$ of copies
  of each \kl{digraph}~$\Digraph'_{\Node[1]:\Trace[1]}$
  can be chosen as large as required,
  we stipulate that for every \kl{trace}~$\Trace[1] \in \NeighborHistory$,
  the sum of $n_{\Node[1]:\Trace[1]}$ over all $\Node[1]$
  exceeds the lower bound~$m_{\Trace[1]}$ that is associated with $\Trace[1]$
  when invoking Lemma~\ref{lem:enables-relation}
  for~$\Trace[2]$ and~$\NeighborHistory$.
  Applying that lemma,
  we obtain a \kl{lossless-asynchronous timing segment}
  of the subgraph induced by $\Node[2]'$ and its \kl{incoming neighbors}.
  This segment determines our \kl{timing}~$\Timing'$ between~$\Time_1 + 1$ and~$\Time_2$
  (the other parts of~$\Digraph'$ being inactive),
  and gives us $\RichRun'_{\Time_2}(\Node[2]') = \Trace[2]$, as desired.
  Naturally,
  the remainder of~$\Timing'$, starting at~$\Time_2 + 1$, can be chosen arbitrarily,
  so long as it satisfies the properties of a \kl{lossless-asynchronous timing}.

  As a closing remark,
  note that the \kl{pointed digraph}~$\pver{\Digraph'}{\Node[2]'}$
  constructed above is very similar to the standard unraveling
  of $\pver{\Digraph}{\Node[2]}$ into a (possibly infinite) tree.
  (The set of \kl{nodes} of that tree-unraveling is precisely
  the set of all directed paths in $\Digraph$ that start at $\Node[2]$;
  see, e.g., \cite[Def.~4.51]{BlackburnRV02} or \cite[\S~3.2]{BlackburnB07}).
  However, there are a few differences:
  First,
  we do the unraveling backwards,
  because we want to generate a \kl{backward bisimilar} \kl{structure},
  where all the \kl{edges} point toward the \kl{root}.
  Second,
  we may duplicate the \kl{incoming neighbors} (i.e., children) of each \kl{node} in the tree,
  in order to satisfy the lower bounds imposed by Lemma~\ref{lem:enables-relation}.
  Third,
  we stop the unraveling process at a finite depth
  (not necessarily the same for each subtree),
  and place a copy of the original \kl{digraph}~$\Digraph$ at every~leaf.
\end{proof}

%% file: fig/automaton.tex
\begin{tikzpicture}[->, >=stealth, shorten >=1pt, semithick, auto, node distance=25ex]
  \tikzset{every state/.append style={inner sep=0ex, minimum size=6.9ex}}
  \tikzset{every edge/.append style={inner sep=0.5ex,font=\footnotesize}}
  \tikzset{every loop/.style={inner sep=0.3ex,looseness=5}}
  \node[initial,state,initial text={\normalsize $1$},initial distance=2.5ex,align=center]
    (1)
    {$1$ \\[-1ex] $\scriptstyle (\SetConstant_1)$};
  \node[initial,state,initial text={\normalsize $0$},initial distance=2.5ex]
    (2) [below of=1,yshift=3ex]
    {$2$};
  \node[accepting,state,align=center]
    (3) [right of=1]
    {$3$ \\[-1ex] $\scriptstyle (\SetVariable[1])$};
  \node[state,align=center]
    (4) [right of=2]
    {$4$ \\[-1ex] $\scriptstyle (\SetVariable[2])$};
  \node[accepting,state,align=center]
    (5) [right of=3,xshift=-2.5ex]
    {$5$ \\[-1ex] $\scriptstyle (\SetVariable[1],\SetVariable[2])$};

  \path (1) edge [loop below] node {otherwise} (1)
            edge              node [align=center]
                                   {if $\NeighborSet \nsubseteq \set{4,5}$ \\
                                    and $\NeighborSet \nsubseteq \set{1,2,4}$} (3)
        (2) edge [loop below] node {otherwise} (2)
            edge              node [sloped,anchor=south,pos=0.38,align=left]
                                   {if $\NeighborSet \nsubseteq \set{4,5}$ \\
                                    and $\NeighborSet \nsubseteq \set{1,2,4}$} (3)
            edge              node [swap,pos=0.42]
                                   {if $\NeighborSet \subseteq \set{4}$} (4)
            edge              node [sloped,swap,anchor=north,pos=0.28]
                                   {if $\set{5} \subseteq \NeighborSet \subseteq \set{4,5}$} (5)
        (3) edge [loop below] node {otherwise} (3)
            edge              node [align=left]
                                   {if $\NeighborSet \subseteq \set{4,5}$} (5)
        (4) edge [loop below] node {otherwise} (4)
            edge              node [sloped,anchor=north,pos=0.29]
                                   {if $\set{5} \subseteq \NeighborSet$} (5)
        (5) edge [loop below] node {always} (5);
  \draw [postaction={decorate,decoration={raise=1ex,text along path,text align={left,left indent=2.6ex},
                                          text={|\footnotesize|if {$\NeighborSet \subseteq \set{4,5}$}}}}]
        (1) to [bend left=39] (5);
  \node [right of=4,xshift=-6ex,yshift=0ex,align=right,font={\small\itshape}]
        {$\NeighborSet$: set of \\ received \\ states};
\end{tikzpicture}

%% file: fig/before-result.tex
\begin{tikzpicture}[->, >=stealth, semithick, node distance=11ex]
  \node (SA) {$\semA{\DA}[\pDIGRAPH[\BitCount][1]]$};
  \node (LA) [below=of SA.east,anchor=east] {$\semA{\laDA}[\pDIGRAPH[\BitCount][1]]$};
  \node (AA) [below=of LA.east,anchor=east] {$\semA{\aDA}[\pDIGRAPH[\BitCount][1]]$};
  \node (QSA) [right=8ex of SA] {$\semA{\QDA}[\pDIGRAPH[\BitCount][1]]$};
  \node (QLA) [below=of QSA.east,anchor=east] {$\semA{\laQDA}[\pDIGRAPH[\BitCount][1]]$};
  \node (QAA) [below=of QLA.east,anchor=east] {$\semA{\aQDA}[\pDIGRAPH[\BitCount][1]]$};

  \path
    (QSA) edge (SA)
    (QLA) edge (LA)
    (QAA) edge (AA)
    ([xshift=-7.5ex]AA.north east)  edge ([xshift=-7.5ex]LA.south east)
    ([xshift=-7.5ex]LA.north east)  edge ([xshift=-7.5ex]SA.south east)
    ([xshift=-7.5ex]QAA.north east) edge ([xshift=-7.5ex]QLA.south east)
    ([xshift=-7.5ex]QLA.north east) edge ([xshift=-7.5ex]QSA.south east);
\end{tikzpicture}

%% file: fig/after-result.tex
\begin{tikzpicture}[->, >=stealth, semithick, node distance=11ex]
  \node (SA) {$\semA{\DA}[\pDIGRAPH[\BitCount][1]]$};
  \node (LA) [below=of SA.east,anchor=east] {$\semA{\laDA}[\pDIGRAPH[\BitCount][1]]$};
  \node (AA) [below=of LA.east,anchor=east] {$\semA{\aDA}[\pDIGRAPH[\BitCount][1]]$};
  \node (QSA) [right=8ex of SA] {$\semA{\QDA}[\pDIGRAPH[\BitCount][1]]$};
  \node (QLA) [below=of QSA.east,anchor=east] {$\semA{\laQDA}[\pDIGRAPH[\BitCount][1]]$};
  \node (QAA) [below=of QLA.east,anchor=east] {$\semA{\aQDA}[\pDIGRAPH[\BitCount][1]]$};
  \node (mu)  [xshift=13ex] at ($(QLA)!0.5!(QAA)$)
              [inner sep=0ex] {$\semF{\SigmaMu{1}(\bML)}[\pDIGRAPH[\BitCount][1]]$};
  \path
    (QSA) edge (SA)
    (QLA) edge (LA)
    (QAA) edge (AA)
    ([xshift=-7.5ex]AA.north east)  edge ([xshift=-7.5ex]LA.south east)
    ([xshift=-7.5ex]LA.north east)  edge ([xshift=-7.5ex]SA.south east)
    ([xshift=-7.5ex]QAA.north east) edge ([xshift=-7.5ex]QLA.south east)
    ([xshift=-7.5ex]QLA.north east) edge ([xshift=-7.5ex]QSA.south east)
    ([xshift=-5ex]QLA.south east)   edge ([xshift=-1ex]mu.north west)
    ([xshift=-1ex]mu.south west)          edge ([xshift=-5ex]QAA.north east);
  \begin{scope}[overlay,on background layer]
    \node[fit=(QLA) (mu) (QAA),
          draw,dashed,rounded corners=2ex,inner sep=1.3ex] {};
  \end{scope}
\end{tikzpicture}

%% file: tex/emptiness.tex
\mychapterpreamble{%
  \nameCref{ch:emptiness} based on the
  conference paper~\cite{DBLP:journals/corr/KuusistoR17}.}

\chapter{Emptiness Problems}
\label{ch:emptiness}

This \lcnamecref{ch:emptiness} is concerned with
the decidability of the \kl{emptiness problem}
for several classes of \kl{nonlocal distributed automata}.
Given such an \kl[distributed automaton]{automaton},
the task is to decide algorithmically
whether it \kl{accepts} on at least one input \kl{digraph}.
For our main variants of \kl{local automata},
we can easily determine if this is possible,
simply on the basis of their logical characterizations:
emptiness is decidable for $\LDA$'s
because they are effectively \kl[device equivalent]{equivalent} to $\bML$,
for which the (finite) satisfiability problem
is known to be $\PSPACE$-complete;
on the other hand,
it is undecidable for $\ALDAg$'s
because they are effectively \kl[device equivalent]{equivalent} to $\MSOL$,
for which (finite) satisfiability is undecidable.
We have also shown in \cref{sec:ndga-emptiness},
that the corresponding problem for $\NLDAg$'s is decidable,
using a simple finite-model argument.
Furthermore,
by the results on \kl{nonlocal automata} presented in \cref{ch:nonlocal},
we know that emptiness is decidable for $\aQDA$'s and $\laQDA$'s,
since (finite) satisfiability for the (backward) $\mu$-calculus
is $\EXPTIME$-complete.
However,
for \kl{nonlocal automata} in general,
the decidability question
has been left open by Kuusisto in~\cite{DBLP:conf/csl/Kuusisto13}.
Indeed,
since the logical characterization given there
is in terms of the newly introduced modal substitution calculus
(for which no decidability results have been previously established),
it does not provide us with an immediate answer.
Here,
we obtain a negative answer for the general case
and also consider the question
for three subclasses of \kl{nonlocal distributed automata}.

Our first variant, dubbed \kl{forgetful automata},
is characterized by the fact
that \kl{nodes} can see their \kl{incoming neighbors}' \kl{states}
but cannot remember their own \kl{state}.
Although this restriction might seem very artificial,
it bears an intriguing connection to classical automata theory:
\kl{forgetful distributed automata} turn out to be
\kl[device equivalent]{equivalent} to finite \kl[pointed dipath]{word} automata
(and hence $\MSOL$)
when restricted to \kl{pointed dipaths},
but strictly more expressive
than finite \kl[pointed ordered ditree]{tree} automata
(and hence $\MSOL$)
when restricted to \kl{pointed ordered ditrees}.
As shown in~\mbox{\cite[Prp.~8]{DBLP:conf/csl/Kuusisto13}},
the situation is different on arbitrary \kl{digraphs},
where \kl{distributed automata}
(and hence \kl[forgetful automata]{forgetful} ones)
are unable to \kl{recognize} non-reachability properties
that can be easily expressed in $\MSOL$.
Hence,
none of the two formalisms can simulate the other in general.
However,
while satisfiability for $\MSOL$ is undecidable,
we obtain a $\LOGSPACE$ algorithm
that decides the \kl{emptiness problem} for \kl{forgetful distributed automata}.

The preceding decidability result begs the question of
what happens if we drop the forgetfulness condition.
Motivated by the \kl[device equivalence]{equivalence} of
finite \kl[pointed dipath]{word} automata and
\kl{forgetful distributed automata},
we first investigate this question
when restricted to \kl{dipaths}.
In sharp contrast to the forgetful case,
we find that
for arbitrary \kl{distributed automata},
it is undecidable
whether an \kl[distributed automaton]{automaton}
\kl{accepts} on some \kl{dipath}.
Although our proof follows the standard approach
of simulating a Turing machine,
it has an unusual twist:
we exchange the roles of space and time,
in the sense that
the space of the simulated Turing machine $\Machine$
is encoded into
the time of the simulating \kl{distributed automaton} $\Automaton$,
and conversely,
the time of $\Machine$ is encoded into the space of $\Automaton$.
To lift this result to arbitrary \kl{digraphs},
we introduce the class of \kl{monovisioned distributed automata},
where \kl{nodes} enter a rejecting \kl{sink} \kl{state} as soon as
they see more than one \kl{state} in their \kl{incoming neighborhood}.
For every \kl{distributed automaton}~$\Automaton$,
one can construct a \kl{monovisioned automaton}~$\Automaton'$
that satisfies the emptiness property
if and only if $\Automaton$ does so on \kl{dipaths}.
Hence,
the \kl{emptiness problem} is undecidable for \kl{monovisioned automata},
and thus also in general.

Our third and last class consists
of the \kl{quasi-acyclic distributed automata}.
The motivation for considering this particular class is threefold.
First,
quasi-acyclicity may be seen as a natural intermediate stage
between \kl[local automata]{local} and \kl{nonlocal distributed automata},
because \kl{local automata}
(for which the \kl{emptiness problem} is decidable)
can be characterized as those \kl[distributed automata]{automata}
whose \kl{state} diagram is acyclic as long as we ignore \kl{sink} \kl{states}
(see \cref{sec:distributed-automata}).
Second,
the Turing machine simulation mentioned above
makes crucial use of directed cycles
in the diagram of the simulating \kl[distributed automaton]{automaton},
which suggests that cycles might be the source of undecidability.
Third,
the notion of quasi-acyclic \kl{state} diagrams also plays a major role
in \cref{ch:nonlocal},
where it serves as an ingredient
for $\aQDA$'s and $\laQDA$'s
(for which the \kl{emptiness problem} is also decidable).
However,
contrary to what one might expect from these clues,
we show that quasi-acyclicity alone is not sufficient
to make the \kl{emptiness problem} decidable,
thereby giving an alternative proof
of undecidability for the general case.

The remainder of this \lcnamecref{ch:emptiness} is organized as follows:
We first introduce some formal definitions in \cref{sec:preliminaries}
and establish the connections between \kl{forgetful distributed automata}
and classical \kl[pointed dipath]{word} and
\kl[pointed ordered ditree]{tree} automata in \cref{sec:classical}.
Then,
in \cref{sec:forgetfulness},
we show the positive decidability result for \kl{forgetful automata}.
Finally,
we establish the negative results for \kl{monovisioned automata}
in \cref{sec:space-time}
and for \kl{quasi-acyclic automata} in \cref{sec:fireworks}.

\section{Preliminaries}
\label{sec:preliminaries}

Given a \kl{distributed automaton}~$\Automaton$,
the (general) \Intro{emptiness problem} consists in deciding effectively
whether the \kl[pointed-digraph language]{language} of~$\Automaton$ is nonempty,
i.e., whether there is a \kl{pointed digraph} $\pver{\Digraph}{\Node}$
that is \kl{accepted} by~$\Automaton$.
Similarly,
the \Intro{dipath-emptiness problem}
is to decide whether $\Automaton$ \kl{accepts} some \kl{pointed dipath}.

We now define \kl{forgetful distributed automata},
which are characterized by the fact that
in each communication round,
the \kl{nodes} of the input \kl{digraph}
can see their \kl[incoming neighbors]{neighbors}' \kl{states}
but cannot remember their own \kl{state}.
As this entails
that they are not able to access their own label
by storing it in their \kl{state},
we instead let them reread that label in each round.

\begin{definition}[Forgetful distributed automaton]
  A \Intro{forgetful distributed automaton} (\Intro*{$\FDA$})
  over $\Alphabet$-\kl{labeled}, $\RelCount$-relational \kl{digraphs}
  is a tuple
  $\Automaton = \tuple{\StateSet,\InitState,\tuple{\TransFunc_{\Letter}}_{\Letter\in\Alphabet},\AcceptSet}$,
  where
  $\StateSet$
  is a finite nonempty set of \kl{states},
  $\InitState \in \StateSet$
  is an \Intro{initial state},
  $\TransFunc_{\Letter} \colon (\powerset{\StateSet})^\RelCount \to \StateSet$
  is a \reintro[transition function]{transition function associated with label~$\Letter \in \Alphabet$},
  and
  $\AcceptSet \subseteq \StateSet$
  is a set of \kl{accepting states}.
\end{definition}

The semantics is completely analogous to the one defined
in \cref{sec:distributed-automata},
for the unrestricted \kl[distributed automata]{automata}
of \cref{def:distributed-automaton}.
For a given $\Alphabet$-\kl{labeled}, $\RelCount$-relational \kl{digraph}~$\Digraph$,
the \kl{run}~$\Run$ of $\Automaton$ on $\Digraph$ is
the infinite sequence of \kl{configurations}
$\tuple{\Run_0, \Run_1, \Run_2, \dots}$,
which are defined inductively as follows,
for $\Time \in \Natural$ and $\Node[2] \in \NodeSet{\Digraph}$:
\begin{equation*}
  \Run_0(\Node[2]) = \InitState
  \qquad\text{and}\qquad
  \Run_{\Time+1}(\Node[2]) =
  \TransFunc_{\Labeling{\Digraph}(\Node[2])}
      \Bigl(\bigtuple{\setbuilder{\Run_{\Time}(\Node[1])}{\edge{\Node[1]}{\Node[2]} \in \EdgeSet[i]{\Digraph}}
                     }_{1 \leq i \leq \RelCount}
      \Bigr).
\end{equation*}
The definition of \kl{acceptance} remains
exactly the same as in \cref{sec:distributed-automata},
i.e., for $\Node \in \NodeSet{\Digraph}$,
the \kl{pointed digraph} $\pver{\Digraph}{\Node}$
is \kl{accepted} by~$\Automaton$
if and only if
there exists $\Time \in \Natural$
such that $\Run_{\Time}(\Node) \in \AcceptSet$.

\section{Comparison with classical automata}
\label{sec:classical}

The purpose of this \lcnamecref{sec:classical} is to motivate
our interest in \kl{forgetful distributed automata}
by establishing their connection
with classical \kl[pointed dipath]{word}
and \kl[pointed ordered ditree]{tree} automata.

\begin{proposition}[{
    $\semA{\FDA}[\pDIPATH[\Alphabet]] =
     \semF{\MSOL}[\pDIPATH[\Alphabet]]$
  }]
  \label{prp:equivalence-word-automata}
  When restricted to the class of \kl{pointed dipaths},
  \kl{forgetful distributed automata} are
  \kl[device equivalent]{equivalent} to finite \kl[pointed dipath]{word} automata,
  and thus to $\MSOL$.
\end{proposition}
\begin{proof}
  Let us denote a (deterministic) finite \kl[pointed dipath]{word} automaton
  over some finite alphabet~$\Alphabet$
  by a tuple
  $\ClAutomaton = \tuple{\ClStateSet, \ClInitState, \ClTransFunc, \ClAcceptSet}$,
  where
  $\ClStateSet$ is the set of states,
  $\ClInitState$ is the initial state,
  $\ClTransFunc \colon \ClStateSet \times \Alphabet \to \ClStateSet$
  is the transition function, and
  $\ClAcceptSet$ is the set of accepting states.

  Given such a \kl[pointed dipath]{word} automaton~$\ClAutomaton$,
  we construct a \kl{forgetful distributed automaton}
  ${\Automaton = \tuple{\StateSet,\InitState,\tuple{\TransFunc_{\Letter}}_{\Letter\in\Alphabet},\AcceptSet}}$
  that simulates~$\ClAutomaton$ on $\Alphabet$-\kl{labeled} \kl{dipaths}.
  For this,
  it suffices to set $\StateSet = \ClStateSet \cup \set{\Waiting}$,\,
  $\InitState = \Waiting$,\, $\AcceptSet = \ClAcceptSet$, and
  \begin{equation*}
    \TransFunc_{\Letter}(\NeighborSet) =
    \begin{cases*}
      \ClTransFunc(\ClInitState, \Letter) & if $\NeighborSet = \EmptySet$, \\
      \ClTransFunc(\ClState, \Letter)     & if $\NeighborSet = \set{\ClState}$
                                            for some $\ClState \in \ClStateSet$, \\
      \Waiting                            & otherwise.
    \end{cases*}
  \end{equation*}
  When $\Automaton$ is \kl{run} on a \kl{dipath},
  each \kl{node}~$\Node$ starts in a waiting phase, represented by~$\Waiting$,
  and remains idle until its predecessor
  has computed the state~$\ClState$
  that $\ClAutomaton$ would have reached
  just before reading the local letter~$\Letter$ of~$\Node$.
  (If there is no predecessor, $\ClState$ is set to~$\ClInitState$.)
  Then, $\Node$ switches to the state~$\ClTransFunc(\ClState, \Letter)$
  and stays there forever.
  Consequently,
  the \kl[distinguished node]{distinguished} last \kl{node} of the \kl{dipath} will end up
  in the state reached by~$\ClAutomaton$
  at the end of the \kl[pointed dipath]{word},
  and it \kl{accepts} if and only if~$\ClAutomaton$ does.

  For the converse direction,
  we convert a given \kl{forgetful distributed automaton}
  $\Automaton = \tuple{\StateSet,\InitState,\tuple{\TransFunc_{\Letter}}_{\Letter\in\Alphabet},\AcceptSet}$
  into the \kl[pointed dipath]{word} automaton
  $\ClAutomaton = \tuple{\ClStateSet, \ClInitState, \ClTransFunc, \ClAcceptSet}$
  with components
  $\ClStateSet = \powerset{\StateSet}$,\,
  $\ClInitState = \EmptySet$,\,
  $\ClAcceptSet = \setbuilder{\StateSubset \subseteq \StateSet}
                             {\StateSubset \cap \AcceptSet \neq \EmptySet}$,
  and
  \begin{equation*}
    \ClTransFunc(\ClState, \Letter) =
    \set{\InitState} \cup
    \begin{cases*}
      \set{\TransFunc_{\Letter}(\EmptySet)}
      & if $\ClState = \ClInitState$, \\
      \setbuilder{\TransFunc_{\Letter}(\set{\State})}{\State \in \ClState}
      & otherwise.
    \end{cases*}
  \end{equation*}
  On any $\Alphabet$-\kl{labeled} \kl{dipath}~$\Digraph$,
  our construction guarantees that
  the set of \kl{states} visited by~$\Automaton$ at the $i$-th \kl{node}
  is equal to
  the state that~$\ClAutomaton$ reaches
  just after processing the $i$-th letter
  of the \kl[pointed dipath]{word} associated with~$\Digraph$.
  We can easily verify this by induction on~$i$:
  At the first \kl{node},
  which is \kl{labeled} with $\Letter_1$,
  \kl[distributed automaton]{automaton}~$\Automaton$ starts in \kl{state}~$\InitState$
  and then remains forever in \kl{state}~$\TransFunc_{\Letter_1}(\EmptySet)$.
  \kl{Node} number~$i + 1$ also starts in~$\InitState$,
  and transitions to
  $\TransFunc_{\Letter_{i+1}}(\set{\State_{\Time}^i})$ at time~$\Time + 1$,
  where $\Letter_{i+1}$ is the \kl{node}'s own label
  and~$\State_{\Time}^i$ is the \kl{state} of its predecessor at time~$\Time$.
  In agreement with this behavior,
  we know by the induction hypothesis and the definition of $\ClTransFunc$
  that the state of~$\ClAutomaton$ after reading $\Letter_{i+1}$
  is precisely
  $\set{\InitState} \cup
   \setbuilder{\TransFunc_{\Letter_{i+1}}(\set{\State_{\Time}^i})}{\Time \in \Natural}$.
  As a result,
  the final state reached by~$\ClAutomaton$ will be accepting
  if and only if
  $\Automaton$ visits some \kl{accepting state} at the last \kl{node}.
\end{proof}

A (deterministic, bottom-up) finite \kl[pointed ordered ditree]{tree} automaton
over $\Alphabet$-\kl{labeled}, $\RelCount$-relational \kl{ordered ditrees}
can be defined as a tuple
$\ClAutomaton = \tuple{\ClStateSet, \tuple{\ClTransFunc_k}_{0\leq k\leq\RelCount}, \ClAcceptSet}$,
where
$\ClStateSet$ is a finite nonempty set of states,
$\ClTransFunc_k \colon \ClStateSet^k \times \Alphabet \to \ClStateSet$
is a transition function of arity~$k$, and
$\ClAcceptSet \subseteq \ClStateSet$ is a set of accepting states.
Such an automaton assigns a state of~$\ClStateSet$
to each \kl{node} of a given \kl{pointed ordered ditree},
starting from the leaves and working its way up to the \kl{root}.
If \kl{node}~$\Node$ is \kl{labeled} with letter~$\Letter$
and its $k$ children have been assigned the states
$\ClState_1, \dots, \ClState_k$
(following the numbering order of the $k$ first \kl{edge relations}),
then $\Node$ is assigned the state
$\ClTransFunc_k(\ClState_1, \dots, \ClState_k, \Letter)$.
Note that leaves are covered by the special case $k = 0$.
Based on this,
the \kl{pointed ditree} is accepted
if and only if
the state at the \kl{root} belongs to~$\ClAcceptSet$.
For a more detailed presentation
see, e.g., \cite[\S~3.3]{DBLP:books/ws/automata2012/Loding12}.

\begin{proposition}[{
    $\semA{\FDA}[\pODITREE[\Alphabet][\RelCount]] \supsetneqq
     \semF{\MSOL}[\pODITREE[\Alphabet][\RelCount]]$
  }]
  \label{prp:inclusion-tree-automata}
  When restricted to the class of \kl{pointed ordered ditrees},
  \kl{forgetful distributed automata} are
  strictly more expressive than finite \kl[pointed ordered ditree]{tree} automata,
  and thus than $\MSOL$.
\end{proposition}
\begin{proof}
  To convert a \kl[pointed ordered ditree]{tree} automaton
  $\ClAutomaton = \tuple{\ClStateSet, \tuple{\ClTransFunc_k}_{0\leq k\leq\RelCount}, \ClAcceptSet}$
  into a \kl{forgetful distributed automaton}
  $\Automaton = \tuple{\StateSet,\InitState,\tuple{\TransFunc_{\Letter}}_{\Letter\in\Alphabet},\AcceptSet}$
  that is \kl[device equivalent]{equivalent} to $\ClAutomaton$
  over $\Alphabet$-\kl{labeled}, $\RelCount$-relational \kl{ordered ditrees},
  we use a simple generalization of the construction
  in the proof of Proposition~\ref{prp:equivalence-word-automata}:
  $\StateSet = \ClStateSet \cup \set{\Waiting}$,\,
  $\InitState = \Waiting$,\,
  $\AcceptSet = \ClAcceptSet$, and
  \begin{equation*}
    \TransFunc_{\Letter}(\vec{\NeighborSet}) =
    \begin{cases*}
      \ClTransFunc_k(\ClState_1, \dots, \ClState_k, \Letter)
      & if $\vec{\NeighborSet} =
        \bigtuple{\set{\ClState_1}, \dots, \set{\ClState_k}, \EmptySet, \dots, \EmptySet}$
        with $\ClState_1,\dots,\ClState_k \in \ClStateSet$, \\
      \Waiting
      & otherwise.
    \end{cases*}
  \end{equation*}

  In contrast,
  a conversion in the other direction is not always possible,
  as can be seen from the following example on binary \kl{ditrees}.
  Consider the \kl{forgetful distributed automaton}
  $\Automaton' = \tuple{\set{\Waiting,\Finished,\Accepting}, \Waiting, \TransFunc, \set{\Accepting}}$,
  with
  \begin{equation*}
    \TransFunc(\NeighborSet_1,\NeighborSet_2) =
    \begin{cases*}
      \Waiting
      & if $\NeighborSet_1 = \NeighborSet_2 = \set{\Waiting}$ \\
      \Finished
      & if $\NeighborSet_1, \NeighborSet_2 \in \set{\EmptySet, \set{\Finished}}$ \\
      \Accepting
      & otherwise.
    \end{cases*}
  \end{equation*}
  When \kl{run} on an \kl[labeled]{unlabeled}, $2$-relational \kl{ordered ditree},
  $\Automaton'$ \kl{accepts} at the \kl{root}
  precisely if the \kl{ditree} is \emph{not} perfectly balanced,
  i.e., if there exists a \kl{node}
  whose left and right subtrees have different heights.
  To achieve this,
  each \kl{node} starts in the waiting \kl{state}~$\Waiting$,
  where it remains as long as it has two children
  and those children are also in~$\Waiting$.
  If the \kl{ditree} is perfectly balanced,
  then all the leaves switch permanently from~$\Waiting$ to~$\Finished$
  in the first round,
  their parents do so in the second round,
  their parents' parents in the third round,
  and so forth,
  until the signal reaches the \kl{root}.
  Therefore,
  the \kl{root} will transition directly from~$\Waiting$ to~$\Finished$,
  never visiting \kl{state}~$\Accepting$,
  and hence the \kl{pointed ditree} is rejected.
  On the other hand,
  if the \kl{ditree} is not perfectly balanced,
  then there must be some lowermost internal \kl{node}~$\Node$
  that does not have two subtrees of the same height
  (in particular, it might have only one child).
  Since its subtrees are perfectly balanced,
  they behave as in the preceding case.
  At some point in time,
  only one of~$\Node$'s children will be in \kl{state}~$\Waiting$,
  at which point $\Node$ will switch to \kl{state}~$\Accepting$.
  This triggers an upward-propagating chain reaction,
  eventually causing the \kl{root} to also visit~$\Accepting$,
  and thus to \kl{accept}.
  Note that $\Accepting$ is just an intermediate \kl{state};
  regardless of whether or not the \kl{ditree} is perfectly balanced,
  every \kl{node} will ultimately end up in~$\Finished$.

  To prove that $\Automaton'$ is not \kl[device equivalent]{equivalent} to any \kl[pointed ordered ditree]{tree} automaton,
  one can simply invoke the pumping lemma for regular \kl[pointed ordered ditree]{tree} \kl{languages}
  to show that the complement \kl[pointed-digraph language]{language} of~$\Automaton'$
  is not recognizable by any \kl[pointed ordered ditree]{tree} automaton.
  The claim then follows from the fact
  that regular \kl[pointed ordered ditree]{tree} \kl{languages} are closed under complementation.
\end{proof}

\section{Exploiting forgetfulness}
\label{sec:forgetfulness}

We now give an algorithm deciding
the \kl{emptiness problem} for \kl{forgetful distributed automata}
(on arbitrary \kl{digraphs}).
Its space complexity is linear in the number of \kl{states}
of the given \kl[forgetful automaton]{automaton}.
However,
as an uncompressed binary encoding of a \kl{distributed automaton}
requires space exponential in the number of \kl{states},
this results in $\LOGSPACE$ complexity.
Obviously,
the statement might not hold anymore
if the \kl[forgetful automaton]{automaton} were instead represented by a more compact device,
such as a logical \kl{formula}.

\begin{theorem}
  \label{thm:forgetful-emptiness}
  We can decide the \kl{emptiness problem} for \kl{forgetful distributed automata}
  with $\LOGSPACE$ complexity.
\end{theorem}
\begin{proof}
  Let
  $\Automaton = \tuple{\StateSet,\InitState,\tuple{\TransFunc_{\Letter}}_{\Letter\in\Alphabet},\AcceptSet}$
  be some \kl{forgetful distributed automaton}
  over $\Alphabet$-\kl{labeled}, $\RelCount$-relational \kl{digraphs}.
  Consider the infinite sequence of sets of \kl{states}
  $\StateSubset_0, \StateSubset_1, \StateSubset_2 \cdots$
  such that
  $\StateSubset_{\Time}$ contains precisely those \kl{states}
  that can be visited by~$\Automaton$
  at some \kl{node} in some \kl{digraph} at time~$\Time$.
  That is,
  $\State \in \StateSubset_{\Time}$
  if and only if
  there exists a \kl{pointed digraph} $\pver{\Digraph}{\Node}$
  such that
  $\Run_{\Time}(\Node) = \State$,
  where $\Run$ is the \kl{run} of $\Automaton$ on $\Digraph$.
  From this point of view,
  the \kl{pointed-digraph language} of $\Automaton$ is nonempty
  precisely if
  there is some $\Time \in \Natural$
  for which
  $\StateSubset_{\Time} \cap \AcceptSet \neq \EmptySet$.

  By definition,
  we have $\StateSubset_0 = \set{\InitState}$.
  Furthermore,
  exploiting the fact that $\Automaton$ is \kl{forgetful},
  we can specify a simple function
  $\BigTransFunc \colon \powerset{\StateSet} \to \powerset{\StateSet}$
  such that
  $\StateSubset_{\Time + 1} = \BigTransFunc(\StateSubset_{\Time})$:
  \begin{equation*}
    \BigTransFunc(\StateSubset[1]) =
    \bigsetbuilder
      {\TransFunc_{\Letter}(\vec{\StateSubset[2]})}
      {\text{
          $\Letter \in \Alphabet$ and
          $\vec{\StateSubset[2]} \in (\powerset{\StateSubset[1]})^\RelCount$
        }
      }
  \end{equation*}
  Obviously,
  $\StateSubset_{\Time + 1} \subseteq \BigTransFunc(\StateSubset_{\Time})$.
  To see that
  $\StateSubset_{\Time + 1} \supseteq \BigTransFunc(\StateSubset_{\Time})$,
  assume we are given a \kl{pointed digraph}
  $\pver{\Digraph_{\State}}{\Node_{\State}}$
  for each state~$\State \in \StateSubset_{\Time}$
  such that
  $\Node_{\State}$ visits $\State$ at time~$\Time$
  in the \kl{run} of~$\Automaton$ on~$\Digraph_{\State}$.
  (Such a \kl{pointed digraph}
  must exist by the definition of $\StateSubset_{\Time}$.)
  Now,
  for any
  $\Letter \in \Alphabet$ and
  $\vec{\StateSubset[2]} =
   \tuple{\StateSubset[2]_1, \dots, \StateSubset[2]_\RelCount}
   \in (\powerset{\StateSubset_{\Time}})^\RelCount$,
  we construct a new \kl{digraph}~$\Digraph$ as follows:
  Starting with a single $\Letter$-\kl{labeled} \kl{node}~$\Node$,
  we add a (disjoint) copy of~$\Digraph_{\State}$
  for each state~$\State$ that occurs in some set~$\StateSubset[2]_k$.
  Then,
  we add a $k$-\kl{edge} from $\Node_{\State}$ to $\Node$
  if and only if $\State \in \StateSubset[2]_k$.
  Each \kl{node}~$\Node_{\State}$ behaves
  the same way in~$\Digraph$ as in~$\Digraph_{\State}$
  because $\Node$ has no influence on its \kl{incoming neighbors}.
  Since $\Automaton$ is \kl{forgetful},
  the state of~$\Node$ at time~$\Time + 1$ depends solely on
  its own label and
  its \kl{incoming neighbors}' \kl{states} at time~$\Time$.
  Consequently,
  $\Node$ visits
  the state~$\TransFunc_{\Letter}(\vec{\StateSubset[2]})$
  at time $\Time + 1$,
  and thus
  $\TransFunc_{\Letter}(\vec{\StateSubset[2]}) \in \StateSubset_{\Time + 1}$.

  Now,
  we know that the sequence
  $\StateSubset_0, \StateSubset_1, \StateSubset_2 \cdots$
  must be eventually periodic
  because its generator function $\BigTransFunc$
  maps the finite set $\powerset{\StateSet}$ to itself.
  Hence,
  it suffices to consider the prefix of length $\card{\powerset{\StateSet}}$
  in order to determine whether
  $\StateSubset_{\Time} \cap \AcceptSet \neq \EmptySet$
  for some $\Time \in \Natural$.
  This leads to the following simple algorithm,
  which decides the \kl{emptiness problem} for \kl{forgetful automata}.
  \begin{align*}
    \textsc{empty}(\Automaton): \quad
    & \StateSubset \gets \set{\InitState} \\
    & \text{\textbf{repeat} at most $\card{\powerset{\StateSet}}$ times}: \\
    & \qquad \StateSubset \gets \BigTransFunc(\StateSubset) \\
    & \qquad \text{\textbf{if} $\StateSubset \cap \AcceptSet \neq \EmptySet$}:\;
             \text{\textbf{return} true} \\
    & \text{\textbf{return} false}
  \end{align*}

  It remains to analyze the space complexity of this algorithm.
  For that,
  we assume that the binary encoding of~$\Automaton$
  given to the algorithm
  contains a lookup table
  for each \kl{transition function}~$\TransFunc_{\Letter}$
  and a bit array representing~$\AcceptSet$,
  which amounts to an asymptotic size of
  $\Theta \bigl(
     \card{\Alphabet} \cdot
     \card{\powerset{\StateSet}}^\RelCount \cdot
     \log\card{\StateSet}
   \bigr)$
  input bits.
  To implement the procedure \textsc{empty},
  we need
  $\card{\StateSet}$ bits of working memory
  to represent the set~$\StateSubset$
  and another $\card{\StateSet}$ bits for the loop counter.
  Furthermore,
  we can compute $\BigTransFunc(\StateSubset[1])$
  for any given set $\StateSubset[1] \subseteq \StateSet$
  by simply iterating over all $\Letter \in \Alphabet$ and
  $\vec{\StateSubset[2]} \in (\powerset{\StateSet})^\RelCount$,
  and adding $\TransFunc_{\Letter}(\vec{\StateSubset[2]})$
  to the returned set
  if all components of~$\vec{\StateSubset[2]}$
  are subsets of~$\StateSubset[1]$.
  This requires
  $\log\card{\Alphabet} + \card{\StateSet} \cdot \RelCount$
  additional bits to keep track of the iteration progress,
  $\Theta \bigl(
     \log\card{\Alphabet} +
     \card{\StateSet} \cdot \RelCount +
     \log\log\card{\StateSet}
   \bigr)$
  bits to store pointers into the lookup tables,
  and $\card{\StateSet}$ bits to store the intermediate result.
  In total,
  the algorithm uses
  $\Theta \bigl(
     \log\card{\Alphabet} + \card{\StateSet} \cdot \RelCount
   \bigr)$
  bits of working memory,
  which is logarithmic in the size of the input.
\end{proof}

\section{Exchanging space and time}
\label{sec:space-time}

In this \lcnamecref{sec:space-time},
we first show the undecidability of
the \kl{dipath-emptiness problem} for arbitrary \kl{distributed automata},
and then lift that result to the general \kl{emptiness problem}.

\begin{theorem}
  \label{thm:dipath-emptiness}
  The \kl{dipath-emptiness problem} for \kl{distributed automata} is undecidable.
\end{theorem}
\begin{proof}[Proof sketch]
  We proceed by reduction from the halting problem for Turing machines.
  For our purposes,
  a Turing machine operates deterministically with one head on a single tape,
  which is one-way infinite to the right and initially empty.
  The problem consists of determining
  whether the machine will eventually reach a designated halting state.
  We show a way of encoding the computation of a Turing machine~$\Machine$
  into the \kl{run} of a \kl{distributed automaton}~$\Automaton$
  over \kl[labeled]{unlabeled} \kl{digraphs},
  such that the \kl[pointed-digraph language]{language} of $\Automaton$ contains a \kl{pointed dipath}
  if and only if~$\Machine$ reaches its halting~state.

\begin{figure}[tbp]
  \centering
  \widefloat{\input{fig/exchange-space-time.tex}}
  \caption{
    Exchanging space and time to prove Theorem~\ref{thm:dipath-emptiness}.
    The left-hand side depicts the computation of a
    Turing machine with state set $\set{\mathbold{0},\mathbold{1},\mathbold{2},\mathbold{3}}$
    and tape alphabet
    $\set{
      \protect\tikz{\protect\node[minimum width=2ex]{};\protect\node[minimum width=1ex,draw]{};},
      \protect\tikz{\protect\node[minimum width=2ex]{};\protect\node[minimum width=1ex,draw,fill=lightgray]{};}
    }$.
    On the right-hand side,
    this machine is simulated by a \kl{distributed automaton} \kl{run} on a \kl{dipath}.
    Waiting \kl{nodes} are represented in black,
    whereas active \kl{nodes} display the content of the “currently visited” cell of the Turing machine
    (i.e., only the third component of the \kl{states} is shown).
  }
  \label{fig:exchange-space-time}
\end{figure}

  Note that since \kl{dipaths} are oriented,
  the communication between their \kl{nodes} is only one-way.
  Hence,
  we cannot simply represent (a section of) the Turing tape as a \kl{dipath}.
  \marginnote{
    It turns out that
    this corresponds to a well-known construction
    in cellular automata theory;
    see \cref{sec:cellular-automata}.
  }
  Instead,
  the key idea of our simulation is to exchange the roles of space and time,
  in the sense that
  the space of $\Machine$ is encoded into the time of $\Automaton$,
  and the time of $\Machine$ into the space of $\Automaton$.
  Assuming the \kl[pointed-digraph language]{language} of~$\Automaton$ contains a \kl{dipath},
  we will think of that \kl{dipath} as representing the timeline of $\Machine$,
  such that
  each \kl{node} corresponds to
  a single point in time in the computation of $\Machine$.
  Roughly speaking,
  when \kl[run]{running}~$\Automaton$,
  the \kl{node}~$\Node_{\Time}$ corresponding to time~$\Time$ will “traverse”
  the configuration~$\Config_{\Time}$ of $\Machine$ at time~$\Time$.
  Here, “traversing” means that
  the sequence of \kl{states} of $\Automaton$ visited by $\Node_{\Time}$
  is an encoding of $\Config_{\Time}$ read from left to right,
  supplemented with some additional bookkeeping information.

  The first element of the \kl{dipath}, \kl{node}~$\Node_0$,
  starts by visiting a \kl{state} of $\Automaton$ representing
  an empty cell that is currently read by $\Machine$ in its initial state.
  Then it transitions to another \kl{state} that simply represents an empty cell,
  and remains in such a \kl{state} forever after.
  Thus $\Node_0$ does indeed “traverse” $\Config_0$.
  We will show that
  it is also possible for any other \kl{node}~$\Node_{\Time}$
  to “traverse” its corresponding configuration~$\Config_{\Time}$,
  based on the information it receives from $\Node_{\Time-1}$.
  In order for this to work,
  we shall give $\Node_{\Time-1}$ a head start of two cells,
  so that $\Node_{\Time}$ can compute the content of cell $\Cell$ in $\Config_{\Time}$
  based on the contents of cells $\Cell-1$,\, $\Cell$ and $\Cell+1$ in $\Config_{\Time-1}$.

  \kl{Node}~$\Node_{\Time}$ enters an \kl{accepting state} of $\Automaton$ precisely if
  it “sees” the halting state of $\Machine$ during its “traversal” of $\Config_{\Time}$.
  Hence,
  $\Automaton$ \kl{accepts} the \kl{pointed dipath} of length~$\Time$
  if and only if $\Machine$ reaches its halting state at time~$\Time$.

  We now describe the inner workings of $\Automaton$ in a semi-formal way.
  In parallel,
  the reader might want to have a look at Figure~\ref{fig:exchange-space-time},
  which illustrates the construction by means of an example.
  Let $\Machine$ be represented by the tuple
  $\tuple{\MStateSet,\MSymbolSet,\MInitState,\MBlank,\MTransFunc,\MHaltState}$,
  where
  $\MStateSet$ is the set of states,
  $\MSymbolSet$ is the tape alphabet,
  $\MInitState$ is the initial state,
  $\MBlank$ is the blank symbol,
  $\MTransFunc\colon
   (\MStateSet\setminus\set{\MHaltState}) \times \MSymbolSet \to
   \MStateSet \times \MSymbolSet \times \set{\LEFT,\RIGHT}$
  is the transition function, and
  $\MHaltState$ is the halting state.
  From this,
  we construct $\Automaton$ as $\tuple{\StateSet,\InitState,\TransFunc,\AcceptSet}$,
  with
  the \kl{state} set
  $\StateSet = (\set{\Waiting} \,\cup\, (\MStateSet\times\MSymbolSet) \,\cup\, \MSymbolSet)^3$,
  the initial \kl{state}
  $\InitState = \tuple{\Waiting,\Waiting,\Waiting}$,
  the \kl{transition function} $\TransFunc$ specified informally below, and
  the \kl[accepting state]{accepting} set $\AcceptSet$ that contains precisely those \kl{states}
  that have $\MHaltState$ in their third component.
  In keeping with the intuition that each \kl{node} of the \kl{dipath}
  “traverses” a configuration of $\Machine$,
  the third component of its \kl{state} indicates
  the content of the “currently visited” cell $\Cell$.
  The two preceding components keep track of the recent history,
  i.e.,
  the second component always holds the content of the previous cell $\Cell-1$,
  and the first component that of $\Cell-2$.
  In the following explanation,
  we concentrate on updating the third component,
  tacitly assuming that the other two are kept up to date.
  The special symbol $\Waiting$ indicates that no cell has been “visited”,
  and we say that a \kl{node} is in the waiting phase
  while its third component is~$\Waiting$.

  In the first round,
  $\Node_0$ sees that it does not have any \kl{incoming neighbor},
  and thus exits the waiting phase
  by setting its third component to $\tuple{\MInitState,\MBlank}$,
  and after that, it sets it to $\MBlank$ for the remainder of the \kl{run}.
  Every other \kl{node} $\Node_{\Time}$ remains in the waiting phase
  as long as its \kl{incoming neighbor}'s second component is $\Waiting$.
  This ensures a delay of two cells with respect to $\Node_{\Time-1}$.
  Once $\Node_{\Time}$ becomes active,
  given the \emph{current} \kl{state} $\tuple{c_1,c_2,c_3}$ of $\Node_{\Time-1}$,
  it computes the third component $d_3$ of its own \emph{next} \kl{state} $\tuple{d_1,d_2,d_3}$
  as follows:
  If none of the components $c_1$, $c_2$, $c_3$ “contain the head of $\Machine$”,
  i.e., if none of them lie in $\MStateSet\times\MSymbolSet$,
  then it simply sets~$d_3$ to be equal to~$c_2$.
  Otherwise,
  a computation step of $\Machine$ is simulated in the natural way.
  For instance,
  if $c_3$ is of the form $\tuple{\MState,\MSymbol}$,
  and $\MTransFunc(\MState,\MSymbol) = \tuple{\MState',\MSymbol',\LEFT}$,
  then $d_3$ is set to $\tuple{\MState',c_2}$.
  This corresponds to the case where, at time $\Time-1$,
  the head of $\Machine$ is located to the right of $\Node_{\Time}$'s next “position”
  and moves to the left.
  As another example,
  if $c_2$ is of the form $\tuple{\MState,\MSymbol}$,
  and $\MTransFunc(\MState,\MSymbol) = \tuple{\MState',\MSymbol',\RIGHT}$,
  then $d_3$ is set to $\MSymbol'$.
  The remaining cases are handled analogously.

  Note that, thanks to the two-cell delay between \kl{adjacent} \kl{nodes},
  the head of $\Machine$ always “moves forward” in the time of $\Automaton$,
  although it may move in both directions with respect to the space of $\Machine$
  (see Figure~\ref{fig:exchange-space-time}).
\end{proof}

To infer from Theorem~\ref{thm:dipath-emptiness}
that the general \kl{emptiness problem} for \kl{distributed automata}
is also undecidable,
we now introduce the notion of
\Intro{monovisioned automata},
which have the property
that \kl{nodes} “expect” to see
no more than one \kl{state} in their \kl{incoming neighborhood} at any given time.
More precisely,
a \kl{distributed automaton}
$\Automaton = \tuple{\StateSet,\InitFunc,\TransFunc,\AcceptSet}$
is \kl{monovisioned} if it has a rejecting \kl{sink} \kl{state}
$\RejectState \in \StateSet \setminus \AcceptSet$,
such that
$\TransFunc(\State,\NeighborSet) = \RejectState$
whenever $\card{\NeighborSet} > 1$
or $\RejectState \in \NeighborSet$
or $\State = \RejectState$,
for all $\State \in \StateSet$ and $\NeighborSet \subseteq \StateSet$.
Obviously,
for every \kl{distributed automaton},
we can construct a \kl{monovisioned automaton}
that has the same \kl{acceptance behavior} on \kl{dipaths}.
Furthermore,
as shown by means of the next two lemmas,
the \kl{emptiness problem} for \kl{monovisioned automata}
is equivalent to its restriction to \kl{dipaths}.
All put together,
we get the desired reduction
from the \kl{dipath-emptiness problem} to the general \kl{emptiness problem}.

\begin{lemma}
  \label{lem:tree-model-property}
  The \kl[pointed-digraph language]{language} of a \kl{distributed automaton} is nonempty
  if and only if it contains a \kl{pointed ditree}.
\end{lemma}
\begin{proof}[Proof sketch]
  We slightly adapt the notion of \emph{tree-unraveling},
  which is a standard tool in modal logic
  (see, e.g., \cite[Def.~4.51]{BlackburnRV02} or \cite[\S~3.2]{BlackburnB07}).
  Consider any \kl{distributed automaton}~$\Automaton$.
  Assume that $\Automaton$ \kl{accepts} some \kl{pointed digraph}
  $\pver{\Digraph}{\Node[2]}$,
  and let~$\Time \in \Natural$ be the first point in time
  at which $\Node$ visits an \kl{accepting state}.
  Based on that,
  we can easily construct a \kl{pointed ditree}
  $\pver{\Digraph'}{\Node[2]'}$
  that is also \kl{accepted} by~$\Automaton$.
  First of all,
  the \kl{root}~$\Node[2]'$ of~$\Digraph'$
  is chosen to be a copy of~$\Node[2]$.
  On the next level of the \kl{ditree},
  the \kl{incoming neighbors} of~$\Node[2]'$
  are chosen to be fresh copies
  $\Node[1]'_1,\dots,\Node[1]'_n$
  of $\Node[2]$'s \kl{incoming neighbors}
  $\Node[1]_1,\dots,\Node[1]_n$.
  Similarly,
  the \kl{incoming neighbors} of
  $\Node[1]'_1,\dots,\Node[1]'_n$
  are fresh copies of the \kl{incoming neighbors} of
  $\Node[1]_1,\dots,\Node[1]_n$.
  If $\Node[1]_i$ and $\Node[1]_j$
  have \kl{incoming neighbors} in common,
  we create distinct copies
  of those \kl[incoming neighbors]{neighbors}
  for $\Node[1]_i'$ and $\Node[1]_j'$.
  This process is iterated until we obtain
  a \kl{ditree} of height~$\Time$.
  It is easy to check
  that~$\Node$ and~$\Node'$ visit
  the same sequence of \kl{states}
  $\State_0, \State_1, \dots, \State_{\Time}$
  during the first~$\Time$ communication rounds.
\end{proof}

\begin{samepage}
  \begin{lemma}
    \label{lem:dipath-model-property}
    The \kl[pointed-digraph language]{language} of a \kl{monovisioned distributed automaton} is nonempty
    if and only if it contains a \kl{pointed dipath}.
  \end{lemma}
\end{samepage}
\begin{proof}[Proof sketch]
  Consider any \kl{monovisioned distributed automaton}~$\Automaton$
  whose \kl[pointed-digraph language]{language} is nonempty.
  By Lemma~\ref{lem:tree-model-property},
  $\Automaton$ \kl{accepts} some \kl{pointed ditree} $\pver{\Digraph}{\Node[2]}$.
  Let~$\Time \in \Natural$ be the first point in time
  at which $\Node[2]$ visits an \kl{accepting state}.
  Now,
  it is easy to prove by induction that
  for all $i \in \set{0, \dots, \Time}$,
  sibling \kl{nodes} at depth~$i$
  traverse the same sequence of \kl{states}
  $\State_0, \State_1, \dots, \State_{\Time - i}$
  between times~$0$ and~$\Time - i$,
  and this sequence does not contain the rejecting \kl{state}~$\RejectState$.
  Thus,
  $\Automaton$ also \kl{accepts}
  any \kl{dipath} from some \kl{node} at depth~$\Time$ to the \kl{root}.
\end{proof}

\section{Timing a firework show}
\label{sec:fireworks}

We now show
that the \kl{emptiness problem} is undecidable
even for \kl{quasi-acyclic automata}.
This also provides an alternative, but more involved
undecidability proof for the general case.
Notice that our proof of Theorem~\ref{thm:dipath-emptiness}
does not go through
if we consider only \kl{quasi-acyclic automata}.

It is straightforward to see that
\kl{quasi-acyclicity} is preserved under a standard product construction,
similar to the one employed for finite automata on \kl[pointed dipaths]{words}.
Hence, we have the following closure property,
which will be used in the subsequent undecidability proof.

\begin{lemma}
  \label{lem:closure-quasi-acyclic}
  The class of \kl[pointed-digraph language]{languages}
  \kl{recognizable} by \kl{quasi-acyclic distributed automata}
  is closed under union and intersection.
\end{lemma}

\begin{theorem}
  \label{thm:emptiness-quasi-acyclic}
  The \kl{emptiness problem} for \kl{quasi-acyclic distributed automata} is undecidable.
\end{theorem}
\begin{proof}[Proof sketch]
  We show this by reduction from Post's correspondence problem ($\PCP$).
  An instance $\Instance$ of $\PCP$ consists of
  a collection of pairs of nonempty finite words
  $\tuple{\UpperWord_\Index,\LowerWord_\Index}_{\Index\in\IndexSet}$
  over the alphabet $\set{0,1}$,
  indexed by some finite set of integers $\IndexSet$.
  It is convenient to view each pair
  $\tuple{\UpperWord_\Index,\LowerWord_\Index}$
  as a domino tile
  labeled with $\UpperWord_\Index$ on the upper half
  and $\LowerWord_\Index$ on the lower half.
  The problem is to decide if there exists a nonempty sequence
  $\Solution = \tuple{\Index_1, \dots, \Index_\SolutionSize}$
  of indices in~$\IndexSet$,
  such that the concatenations
  $\UpperWord_\Solution = \UpperWord_{\Index_1} \cdots\, \UpperWord_{\Index_\SolutionSize}$ and
  $\LowerWord_\Solution = \LowerWord_{\Index_1} \cdots\, \LowerWord_{\Index_\SolutionSize}$
  are equal.
  We construct a \kl{quasi-acyclic automaton} $\Automaton$
  whose \kl[pointed-digraph language]{language} is nonempty if and only if
  $\Instance$ has such a solution~$\Solution$.

  Metaphorically speaking,
  our construction can be thought of as a perfectly timed “firework show”,
  whose only “spectator” will see a putative solution
  $\Solution = \tuple{\Index_1, \dots, \Index_\SolutionSize}$,
  and be able to check whether it is indeed a valid solution of $\Instance$.
  Our “spectator” is the \kl{distinguished node}~$\Root$ of the \kl{pointed digraph}
  on which $\Automaton$ is \kl{run}.
  We assume that $\Root$ has $\SolutionSize$ \kl{incoming neighbors},
  one for each element of $\Solution$.
  Let $\Node_\IIndex$ denote
  the \kl[incoming neighbor]{neighbor} corresponding to $\Index_\IIndex$,
  for $1 \leq \IIndex \leq \SolutionSize$.
  Similarly to our proof of Theorem~\ref{thm:dipath-emptiness},
  we use the time of $\Automaton$ to represent
  the spatial dimension of the words $\UpperWord_\Solution$ and $\LowerWord_\Solution$.
  On an intuitive level,
  $\Root$ will “witness” simultaneous left-to-right traversals of
  $\UpperWord_\Solution$ and $\LowerWord_\Solution$,
  advancing by one bit per time step,
  and it will check that the two words match.
  It is the task of each \kl{node} $\Node_\IIndex$
  to send to $\Root$ the required bits of
  the subwords $\UpperWord_{\Index_\IIndex}$ and $\LowerWord_{\Index_\IIndex}$
  at the appropriate times.
  In keeping with the metaphor of fireworks,
  the correct timing can be achieved by attaching to $\Node_\IIndex$
  a carefully chosen “fuse”,
  which is “lit” at time~$0$.
  Two separate “fire” signals will travel at different speeds
  along this (admittedly sophisticated) “fuse”,
  and once they reach $\Node_\IIndex$,
  they trigger the “firing” of
  $\UpperWord_{\Index_\IIndex}$ and $\LowerWord_{\Index_\IIndex}$, respectively.

  We now go into more details.
  Using the \kl{labeling} of the input \kl[pointed digraph]{graph},
  the \kl[distributed automaton]{automaton} $\Automaton$ distinguishes between
  $2\InstanceSize+1$ different types of \kl{nodes}:
  two types $\Index$ and $\Index'$ for each index $\Index\in\IndexSet$,
  and one additional type $\Spectator$ to identify the “spectator”.
  Motivated by Lemma~\ref{lem:tree-model-property},
  we suppose that the input \kl[pointed digraph]{graph} is a \kl{pointed ditree},
  with a very specific shape that encodes a putative solution
  $\Solution = \tuple{\Index_1, \dots, \Index_\SolutionSize}$.
  An example illustrating the following description
  of such a \kl{ditree}-encoding
  is given in Figure~\ref{fig:firework-show}.
  Although $\Automaton$ is not able to enforce
  all aspects of this particular shape,
  we will make sure that it \kl{accepts} such a \kl{structure}
  if its \kl[pointed-digraph language]{language} is nonempty.
  The \kl{root} (and \kl{distinguished node}) $\Root$
  is the only \kl{node} of type $\Spectator$.
  Its children $\Node_1, \dots, \Node_\SolutionSize$
  are of types $\Index_1, \dots, \Index_\SolutionSize$,
  respectively.
  The “fuse” attached to each child~$\Node_\IIndex$
  is a chain of $\IIndex-1$ \kl{nodes}
  that represents the multiset of indices occurring in
  the $(\IIndex-1)$-prefix of $\Solution$.
  More precisely,
  there is an induced \kl{dipath}
  $\Node_{\IIndex,1} \rightarrow \cdots\, \Node_{\IIndex,\IIndex-1} \rightarrow \Node_\IIndex$,
  such that the multiset of types of the \kl{nodes}
  $\Node_{\IIndex,1}, \dots, \Node_{\IIndex,\IIndex-1}$
  is equal to the multiset of indices occurring in
  $\tuple{\Index_1, \dots, \Index_{\IIndex-1}}$.
  We do not impose any particular order on those \kl{nodes}.
  Finally,
  each \kl{node} of type $\Index \in \IndexSet$
  also has an incoming chain of \kl{nodes} of type $\Index'$
  \mbox{(depicted in gray in Figure~\ref{fig:firework-show})},
  whose length corresponds exactly to
  the product of the types occurring on the part of the “fuse”
  below that \kl{node}.
  That is,
  if we define the alias $\Node_{\IIndex,\IIndex} \defeq \Node_\IIndex$,
  then for every \kl{node} $\Node_{\IIndex,\IIIndex}$ of type $\Index \in \IndexSet$,
  there is an induced \kl{dipath}
  $\Node_{\IIndex,\IIIndex,1} \rightarrow \cdots\,
   \Node_{\IIndex,\IIIndex,\PrimeProduct} \rightarrow
   \Node_{\IIndex,\IIIndex}$,
  where all the \kl{nodes}
  $\Node_{\IIndex,\IIIndex,1},\dots,\Node_{\IIndex,\IIIndex,\PrimeProduct}$
  are of type $\Index'$,
  and the number $\PrimeProduct$ is equal to the product of the types
  of the \kl{nodes} $\Node_{\IIndex,1},\dots,\Node_{\IIndex,\IIIndex-1}$
  (which is $1$ if $\IIIndex = 1$).
  We shall refer to such a chain
  $\Node_{\IIndex,\IIIndex,1},\dots,\Node_{\IIndex,\IIIndex,\PrimeProduct}$
  as a “side fuse”.

  The \kl[distributed automaton]{automaton} $\Automaton$ has to perform two tasks simultaneously:
  First,
  assuming it is \kl{run} on a \kl{ditree}-encoding of a sequence $\Solution$,
  exactly as specified above,
  it must verify that $\Solution$ is a valid solution,
  i.e., that the words $\UpperWord_\Solution$ and $\LowerWord_\Solution$ match.
  Second,
  it must ensure that the input \kl[pointed digraph]{graph} is indeed
  sufficiently similar to such a \kl{ditree}-encoding.
  In particular,
  it has to check that
  the “fuses” used for the first task are consistent with each other.
  Since, by Lemma~\ref{lem:closure-quasi-acyclic},
  \kl{quasi-acyclic distributed automata} are closed under intersection,
  we can consider the two tasks separately,
  and implement them using two independent
  \kl[distributed automata]{automata} $\Automaton_1$ and~$\Automaton_2$.
  In the following,
  we describe both devices in a rather informal manner.
  The important aspect to note is
  that they can be easily formalized using \kl{quasi-acyclic} \kl{state} diagrams.

  \begin{figure}[tbp]
    \centering
    \widefloat{\input{fig/firework-show.tex}}
    \caption{
      Timing a “firework show” to prove Theorem~\ref{thm:emptiness-quasi-acyclic}.
      The domino tiles on the bottom-left visualize the solution $\tuple{5,3,7,3}$
      for the instance
      $\set{3 {\,\mapsto\,} \tuple{00,100},\,
        5 {\,\mapsto\,} \tuple{010,0},\,
        7 {\,\mapsto\,} \tuple{11,01}}$
      of $\PCP$.
      This solution is encoded into the \kl{labeled} \kl{ditree} above,
      with node types $\Spectator$, $3$, $5$, $7$, $3'$, $5'$, $7'$.
      Each domino is represented by a bold-highlighted white \kl{node} of the appropriate type.
      The “fuse” of such a bold \kl{node} consists of
      the chain of white \kl{nodes} below it,
      which lists the indices of the preceding dominos in an arbitrary order.
      Each white \kl{node} also has a gray “side fuse”
      whose length is equal to the product of the white types occurring below that \kl{node}.
      The “firework show” observed at the \kl{root} will feature two simultaneous bitstreams,
      which both represent the sequence $010001100$.
    }
    \label{fig:firework-show}
  \end{figure}
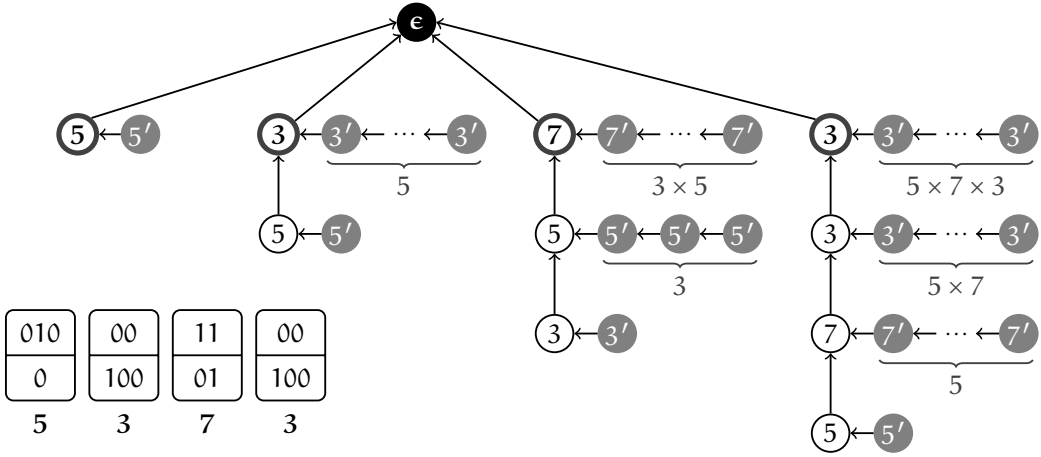
  
  We start with $\Automaton_1$,
  which verifies the solution $\Solution$.
  It takes into account
  only \kl{nodes} with types in $\IndexSet \cup \set{\Spectator}$
  (thus ignoring the gray \kl{nodes} in Figure~\ref{fig:firework-show}).
  At \kl{nodes} of type $\Index \in \IndexSet$,
  the \kl{states} of $\Automaton_1$ have two components,
  associated with the upper and lower halves of the
  domino~$\tuple{\UpperWord_\Index,\LowerWord_\Index}$.
  If a \kl{node} of type $\Index$ sees
  that it does not have any \kl{incoming neighbor},
  then the upper and lower components of its \kl{state}
  immediately start traversing sequences of substates
  representing the bits of $\UpperWord_\Index$ and $\LowerWord_\Index$,
  respectively.
  Since those substates must keep track of
  the respective positions within $\UpperWord_\Index$ and $\LowerWord_\Index$,
  none of them can be visited twice.
  After that,
  both components loop forever on a special substate $\Finished$,
  which indicates the end of transmission.
  The other \kl{nodes} of type $\Index$
  keep each of their two components in a waiting status,
  indicated by another substate~$\Waiting$,
  until the corresponding component of their \kl{incoming neighbor}
  reaches its last substate before~$\Finished$.
  This constitutes the aforementioned “fire” signal.
  Thereupon,
  they start traversing the same sequences of substates as in the previous case.
  Note that both components are updated independently of each other,
  hence there can be an arbitrary time lag between
  the “traversals” of $\UpperWord_\Index$ and $\LowerWord_\Index$.
  Now, assuming the “fuse” of each \kl{node} $\Node_\IIndex$
  really encodes the multiset of indices occurring in
  $\tuple{\Index_1, \dots, \Index_{\IIndex-1}}$,
  the delay accumulated along that “fuse” will be such that
  $\Node_\IIndex$ starts “traversing”
  $\UpperWord_{\Index_\IIndex}$ and $\LowerWord_{\Index_\IIndex}$
  at the points in time corresponding to their respective starting positions
  within $\UpperWord_\Solution$ and $\LowerWord_\Solution$.
  That is,
  for~$\UpperWord_{\Index_\IIndex}$ it starts at time
  $\length{\UpperWord_{\Index_1} \cdots\, \UpperWord_{\Index_{\IIndex-1}}} + 1$,
  and for~$\LowerWord_{\Index_\IIndex}$ at time
  $\length{\LowerWord_{\Index_1} \cdots\, \LowerWord_{\Index_{\IIndex-1}}} + 1$.
  Consequently,
  in each round
  $\Time \leq \min\set{\card{\UpperWord_\Solution}, \card{\LowerWord_\Solution}}$,
  the \kl{root} $\Root$ receives the $\Time$-th bits of
  $\UpperWord_\Solution$ and $\LowerWord_\Solution$.
  At most two distinct children send bits at the same time,
  while the others remain
  in some \kl{state} $\State \in \set{\Waiting,\Finished}^2$.
  With this,
  the behavior of $\Automaton_1$ at $\Root$ is straightforward:
  It enters its only \kl{accepting state} precisely if
  all of its children
  have reached the \kl{state}~$\tuple{\Finished,\Finished}$
  and it has never seen any mismatch between the upper and lower bits.

  We now turn to $\Automaton_2$,
  whose job is to verify that
  the “fuses” used by $\Automaton_1$ are reliable.
  Just like $\Automaton_1$,
  it works under the assumption
  that the input \kl{digraph} is a \kl{ditree} as specified previously,
  but with significantly reduced guarantees:
  The \kl{root} could now have an arbitrary number of children,
  the “fuses” and “side fuses” could be of arbitrary lengths,
  and each “fuse” could represent an arbitrary multiset of indices in $\IndexSet$.
  Again using an approach reminiscent of fireworks,
  we devise a protocol in which
  each child~$\Node$ will send
  two distinct signals to the \kl{root}~$\Root$.
  The first signal $\Signal_1$ indicates that
  the current time~$\Time$ is equal to the product of
  the types of all the \kl{nodes} on $\Node$'s “fuse”.
  Similarly,
  the second signal $\Signal_2$ indicates that
  the current time is equal to that same product
  multiplied by $\Node$'s own type.
  To achieve this,
  we make use of the “side fuses”,
  along which two additional signals
  $\SideSignal_1$ and $\SideSignal_2$ are propagated.
  For each \kl{node} of type $\Index \in \IndexSet$,
  the \kl{nodes} of type $\Index'$ on the corresponding “side fuse”
  operate in a way such that
  $\SideSignal_1$ advances by one \kl{node} per time step,
  whereas $\SideSignal_2$ is delayed by $\Index$ time units at every \kl{node}.
  Hence,
  $\SideSignal_1$ travels $\Index$ times faster than~$\SideSignal_2$.
  Building on that,
  each \kl{node} $\Node$ of type~$\Index$
  (not necessarily a child of the \kl{root})
  sends $\Signal_1$ to its parent,
  either at time~$1$, if it does not have any predecessor on the “fuse”,
  or one time unit before receiving $\Signal_2$ from its predecessor.
  The latter is possible,
  because the predecessor also sends a pre-signal~$\PreSignal_2$
  before sending $\Signal_2$.
  Then,
  $\Node$ checks that signal $\SideSignal_1$ from its “side fuse”
  arrives exactly at the same time as
  $\Signal_2$ from its predecessor,
  or at time $1$ if there is no predecessor.
  Otherwise,
  it immediately enters a rejecting \kl{state}.
  This will guarantee, by induction, that the length of the “side fuse”
  is equal to the product of the types on the “fuse” below.
  Finally,
  two rounds prior to receiving $\SideSignal_2$,
  while that signal is still being delayed
  by the last \kl{node} on the “side fuse”,
  $\Node$ first sends the pre-signal $\PreSignal_2$\!,
  and then the signal $\Signal_2$ in the following round.
  For this to work,
  we assume that
  each \kl{node} on the “side fuse” waits for at least two rounds
  between receiving $\SideSignal_2$ from its predecessor
  and forwarding the signal to its successor,
  i.e., all indices in $\IndexSet$ must be strictly greater than~$2$.
  Due to the delay accumulated by $\SideSignal_2$ along the “side fuse”,
  the time at which $\Signal_2$ is sent corresponds precisely to
  the length of the “side fuse” multiplied by $\Index$.

  Without loss of generality,
  we require that
  the set of indices $\IndexSet$ contains only prime numbers
  (as in Figure~\ref{fig:firework-show}).
  Hence,
  by the unique-prime-factorization theorem,
  each multiset of numbers in $\IndexSet$ is uniquely determined
  by the product of its elements.
  This leads to a simple verification procedure
  performed by $\Automaton_2$ at the \kl{root}:
  At time $1$,
  \kl{node} $\Root$ checks that it receives $\Signal_1$ and not $\Signal_2$.
  After that,
  it expects to never again see $\Signal_1$ without $\Signal_2$,
  and remains in a loop as long as it gets
  either no signal at all or both $\Signal_1$~and~$\Signal_2$.
  Upon receiving $\Signal_2$ alone,
  it exits the loop
  and verifies that all of its children have sent both signals,
  which is apparent from the \kl{state} of each child.
  The \kl{root} rejects immediately
  if any of the expectations above are violated,
  or if two \kl{nodes} with different types
  send the same signal at the same time.
  Otherwise,
  it enters an \kl{accepting state} after leaving the loop.
  Now,
  consider the sequence
  $\TimeSequence = \tuple{\Time_1,\dots,\Time_{\SolutionSize+1}}$
  of rounds in which $\Root$ receives
  at least one of the signals $\Signal_1$ and $\Signal_2$.
  It is easy to see by induction on $\TimeSequence$ that
  successful completion of the procedure above
  ensures that there is a sequence
  $\Solution = \tuple{\Index_1, \dots, \Index_\SolutionSize}$
  of indices in $\IndexSet$
  with the following properties:
  For each $\IIndex \in \set{1,\dots,\SolutionSize}$,
  the \kl{root} has at least one child $\Node_\IIndex$ of type $\Index_\IIndex$
  that sends $\Signal_1$ at time~$\Time_\IIndex$
  and $\Signal_2$ at time~$\Time_{\IIndex+1}$,
  and the “fuse” of $\Node_\IIndex$ encodes precisely
  the multiset of indices occurring in
  $\tuple{\Index_1, \dots, \Index_{\IIndex-1}}$.
  Conversely,
  each child of $\Root$ can be associated in the same manner
  with a unique element of $\Solution$.

  To conclude our proof,
  we have to argue that the \kl[distributed automaton]{automaton} $\Automaton$,
  which simulates $\Automaton_1$ and $\Automaton_2$ in parallel,
  \kl{accepts} some \kl{labeled} \kl{pointed digraph}
  if and only if
  $\Instance$ has a solution~$\Solution$.
  The “if” part is immediate,
  since, by construction,
  $\Automaton$ \kl[accept]{accepting} a \kl{ditree}-encoding of $\Solution$
  is equivalent to $\Solution$ being a valid solution of $\Instance$.
  To show the “only if” part,
  we start with a \kl{pointed digraph} \kl{accepted} by $\Automaton$,
  and incrementally transform it into a \kl{ditree}-encoding of a solution $\Solution$,
  while maintaining \kl{acceptance} by $\Automaton$:
  First of all,
  by Lemma~\ref{lem:tree-model-property},
  we may suppose that the \kl{digraph} is a \kl{ditree}.
  Its \kl{root} must be of type~$\Spectator$,
  since $\Automaton$ would not \kl{accept} otherwise.
  Next,
  we require that $\Automaton$ raises an alarm
  at \kl{nodes} that see an unexpected set of \kl{states} in their \kl{incoming neighborhood},
  and that this alarm is propagated up to the \kl{root},
  which then reacts by entering a rejecting \kl{sink} \kl{state}.
  This ensures that the repartition of types is consistent with our specification;
  for example,
  that the children of a \kl{node} of type~$\Index'$ must be of type~$\Index'$ themselves.
  We now prune the \kl{ditree}
  in such a way that
  \kl{nodes} of type~$\Index$ keep at most two children
  and \kl{nodes} of type~$\Index'$ keep at most one child.
  (The behavior of the deleted children must be indistinguishable from the behavior of the remaining children,
  since otherwise an alarm would be raised.)
  This leaves us with a \kl{ditree}
  corresponding exactly to the input “expected”
  by the \kl[distributed automaton]{automaton} $\Automaton_2$.
  Since it is \kl{accepted} by $\Automaton_2$,
  this \kl{ditree} must be very close to an encoding
  of a solution $\Solution = \tuple{\Index_1, \dots, \Index_\SolutionSize}$,
  with the only difference that
  each element $\Index_\IIndex$ of~$\Solution$ may be represented
  by several \kl{nodes} $\Node_\IIndex^1,\dots,\Node_\IIndex^\CloneCount$.
  However,
  we know by construction that $\Automaton$ behaves the same
  on all of these representatives.
  We can therefore remove the subtrees
  rooted at $\Node_\IIndex^2,\dots,\Node_\IIndex^\CloneCount$,
  and thus
  we obtain a \kl{ditree}-encoding of $\Solution$
  that is \kl{accepted} by~$\Automaton$.
\end{proof}

%% file: fig/exchange-space-time.tex
\begin{tikzpicture}[
    line width=0.7pt,
    every matrix/.style = {matrix of math nodes, nodes in empty cells, row sep={4.55ex,between origins}, inner sep=0ex,
                           nodes={draw,minimum size=3ex,inner sep=0ex,anchor=center,font=\boldmath}},
    b/.style  = {fill=black},               
    g/.style  = {fill=lightgray},           
    hg/.style = {draw=darkgray},            
    gg/.style = {g, hg, line width=0.33ex}, 
    G/.style  = {g, hg, line width=0.41ex}, 
    ww/.style = {hg, line width=0.33ex},    
    W/.style  = {hg, line width=0.41ex},    
    plain/.style = {draw opacity=0, font=\unboldmath}
  ]
  \matrix[column sep={3ex,between origins}] (machine) {
    0        &          &          &        &        \\
    |[g]|    & 1        &          &        &        \\
    |[gg]|   & |[gg]|   & |[ww]| 2 & |[ww]| & |[ww]| \\
    |[g]|    & |[g]| 1  & |[g]|    &        &        \\
    |[g]| 0  & |[g]|    & |[g]|    &        &        \\
             & |[g]| 3  & |[g]|    &        &        \\
  };
  \matrix[right=0.4ex of machine] (machine-dots) {
    |[plain]|\cdots \\ |[plain]|\cdots \\ |[plain]|\cdots \\ |[plain]|\cdots \\ |[plain]|\cdots \\ |[plain]|\cdots \\
  };
  \matrix[right=6ex of machine-dots, column sep={4.45ex,between origins}, nodes={circle}] (automaton) {
    |[b]| & 0       &       &       &       &       &       &         &         &         &       &       &        \\
    |[b]| & |[b]|   & |[b]| & |[g]| & 1     &       &       &         &         &         &       &       &        \\
    |[b]| & |[b]|   & |[b]| & |[b]| & |[b]| & |[G]| & |[G]| & |[W]| 2 & |[W]|   & |[W]|   &       &       &        \\
    |[b]| & |[b]|   & |[b]| & |[b]| & |[b]| & |[b]| & |[b]| & |[g]|   & |[g]| 1 & |[g]|   &       &       &        \\
    |[b]| & |[b]|   & |[b]| & |[b]| & |[b]| & |[b]| & |[b]| & |[b]|   & |[b]|   & |[g]| 0 & |[g]| & |[g]| &        \\
    |[b]| & |[b]|   & |[b]| & |[b]| & |[b]| & |[b]| & |[b]| & |[b]|   & |[b]|   & |[b]|   & |[b]| &       & |[g]|3 \\
  };
  \matrix[right=1ex of automaton] (automaton-dots) {
    |[plain]|\cdots \\ |[plain]|\cdots \\ |[plain]|\cdots \\ |[plain]|\cdots \\ |[plain]|\cdots \\ |[plain]|\cdots \\
  };
  \foreach \x [evaluate = \x as \xplusone using int(\x+1)] in {1,...,5}
    \foreach \y in {1,...,13}
      \draw[->] (automaton-\x-\y) -- (automaton-\xplusone-\y);
  \node[above=2.5ex of machine.north west,anchor=mid west] (machine-space) {\small{space}};
  \draw[->] ([xshift=0.5ex]machine-space.mid east) -- ++(5ex,0);
  \node[overlay,left=2.5ex of machine.north west,anchor=mid west,rotate=-90] (machine-time) {\small{time}};
  \draw[overlay,->] ([yshift=-0.5ex]machine-time.mid east) -- ++(0,-6ex);
  \node[above=2.5ex of automaton.north west,anchor=mid west] (automaton-time) {\small{time}};
  \draw[->] ([xshift=0.5ex]automaton-time.mid east) -- ++(5ex,0);
  \node[left=2.5ex of automaton.north west,anchor=mid west,rotate=-90] (automaton-space) {\small{space}};
  \draw[->] ([yshift=-0.5ex]automaton-space.mid east) -- ++(0,-5ex);
  \node[below=3ex of machine,anchor=mid] {\small{\emph{Turing machine}}};
  \node[below=3ex of automaton,anchor=mid] {\small{\emph{Distributed automaton}}};
\end{tikzpicture}

%% file: fig/firework-show.tex
\begin{tikzpicture}[
    line width=0.7pt,
    edge from parent/.style = {draw, <-},
    level 1/.style = {level distance=9ex},
    level 2/.style = {level distance=8ex},
    sibling distance = 22ex,
    every node/.style = {inner sep=0ex},
    nodraw/.style = {draw=none, fill=none, text=black},
    bold/.style = {draw=darkgray, line width=0.41ex, font=\boldmath},
    main node/.style = {draw, circle, inner sep=0ex, minimum size=3ex},
    root node/.style = {main node, fill=black, text=white, font=\boldmath},
    side node/.style = {main node, draw=gray, fill=gray, text=white},
    side line/.style = {grow=right, every node/.style={side node}, every child/.style={}, level distance=5ex},
    brace/.style = {decorate, decoration={brace,mirror,raise=1ex}, darkgray},
    brace label/.style = {minimum size=0ex, below=2.2ex},
    domino/.style = {draw, rectangle split, rectangle split parts=2, rounded corners=0.6ex, inner sep=1ex, minimum width=5.5ex}
  ]
  \path[every node/.style=main node] node[root node] (root) {$\epsilon$}
     child[xshift=6ex] {node[bold] {$5$} [child anchor=50,parent anchor=west,
                                          every child/.style={grow=down,child anchor=border,parent anchor=border}]
       child[side line,parent anchor=border] {node {$5'$}}
     }
     child {node[bold] {$3$} [every child/.style={grow=down}]
       child {node {$5$}
         child[side line] {node {$5'$}}
       }
       child[side line] {node {$3'$} child {node[nodraw] {\footnotesize\,\dots} child {node {$3'$}}}}
     }
     child {node[bold] {$7$} [every child/.style={grow=down}]
       child {node {$5$}
         child {node {$3$}
           child[side line] {node {$3'$}}
         }
         child[side line] {node {$5'$} child {node {$5'$} child {node {$5'$}}}}
       }
       child[side line] {node {$7'$} child {node[nodraw] {\footnotesize\,\dots} child {node {$7'$}}}}
     }
     child {node[bold] {$3$} [child anchor=north west,parent anchor=east,
                              every child/.style={grow=down,child anchor=border,parent anchor=border}]
       child {node {$3$}
         child {node {$7$}
           child {node {$5$}
             child[side line] {node {$5'$}}
           }
           child[side line] {node {$7'$} child {node[nodraw] {\footnotesize\,\dots} child {node {$7'$}}}}
         }
         child[side line] {node {$3'$} child {node[nodraw] {\footnotesize\,\dots} child {node {$3'$}}}}
       }
       child[side line] {node {$3'$} child {node[nodraw] {\footnotesize\,\dots} child {node {$3'$}}}}
     }
  ;
  \draw[brace] (root-2-2.south west) -- node[brace label] {$5$} (root-2-2-1-1.south east);
  \draw[brace] (root-3-2.south west) -- node[brace label] {$3 \times 5$} (root-3-2-1-1.south east);
  \draw[brace] (root-3-1-2.south west) -- node[brace label] {$3$} (root-3-1-2-1-1.south east);
  \draw[brace] (root-4-2.south west) -- node[brace label] {$5 \times 7 \times 3$} (root-4-2-1-1.south east);
  \draw[brace] (root-4-1-2.south west) -- node[brace label] {$5 \times 7$} (root-4-1-2-1-1.south east);
  \draw[brace] (root-4-1-1-2.south west) -- node[brace label] {$5$} (root-4-1-1-2-1-1.south east);
  \matrix[xshift=-20ex,yshift=-28ex,row sep=1ex,column sep=1ex,matrix of nodes] {
    \node[domino] {$010$ \nodepart{two} $0$};  &
    \node[domino] {$00$ \nodepart{two} $100$}; &
    \node[domino] {$11$ \nodepart{two} $01$};  &
    \node[domino] {$00$ \nodepart{two} $100$}; \\
    $\mathbold{5}$ & $\mathbold{3}$ & $\mathbold{7}$ & $\mathbold{3}$ \\
  };
\end{tikzpicture}

%% file: tex/alternation.tex
\mychapterpreamble{%
  \nameCref{ch:alternation} based on the
  preprint~\cite{DBLP:journals/corr/Reiter16}.}

\chapter{Alternation Hierarchies}
\label{ch:alternation}

In this \lcnamecref{ch:alternation},
we transfer the \kl{set quantifiers} of $\MSOL$
to the setting of \kl{modal logic}
and investigate the resulting alternation hierarchies.
More precisely,
we establish separation results for the hierarchies
that one obtains
by alternating existential and universal \kl{set quantifiers}
in several logics of the form $\MSO(\FormulaSet)$,
where $\FormulaSet$ is some variant of \kl{modal logic}.

Within the context of this thesis,
the motivation for such hybrids
between modal logic and classical logic
stems from their close connection to \kl{local distributed automata}.
By~\cite{DBLP:conf/podc/HellaJKLLLSV12,DBLP:journals/dc/HellaJKLLLSV15},
$\LDA$'s are \kl[device equivalent]{equivalent} to $\bML$ (\cref{thm:LDA-bML}),
and by \cref{ch:local},
$\ALDAg$'s are \kl[device equivalent]{equivalent} to $\MSOL$ (\cref{thm:adga=mso}).
As mentioned in \cref{sec:summary},
the combination of those two results
suggests an alternative logical characterization of $\ALDAg$'s
using $\MSO(\bMLg)$ instead of $\MSOL$.
The \kl[device equivalence]{equivalence} of $\MSO(\bMLg)$ and $\MSOL$
can be easily proven by a standard technique
that simulates \kl{node quantifiers} through \kl{set quantifiers}
(see, e.g., \cite[\S~3]{DBLP:conf/aiml/Kuusisto08,DBLP:journals/apal/Kuusisto15}).
Yet in some sense,
$\MSO(\bMLg)$ provides a more faithful representation of $\ALDAg$'s
because it preserves the expressive power
of each \kl[set quantifier]{quantifier} alternation level.
For instance,
the existential fragment $\SigmaMSO{1}(\bMLg)$
specifies exactly the same \kl{digraph languages} as $\NLDAg$'s,
whereas $\EMSOL$ is strictly more powerful
(see \cref{sec:summary}).
Therefore,
if we want to precisely examine the power of alternation
between nondeterministic decisions and universal branchings
in $\ALDAg$'s,
then we can do so
from a purely logical perspective using $\MSO(\bMLg)$.
This has the advantage that,
compared to \kl{state} diagrams,
\kl{formulas} take up less space and are usually easier to manipulate.

As it turns out,
the above considerations are closely related
to an old problem in modal logic.
Already in \osf{1983},
van~Benthem asked in \cite{Benthem83}
whether the syntactic hierarchy
obtained by alternating existential and universal \kl{set quantifiers}
in $\MSO(\ML)$
induces a corresponding hierarchy on the semantic side.
\marginnote{
  In~\textnormal{\cite{DBLP:journals/jphil/Cate06}} and
  \textnormal{\cite{DBLP:conf/aiml/Kuusisto08,DBLP:journals/apal/Kuusisto15}},
  $\MSO(\ML)$ is called $\SOPML$
  (second-order propositional modal~logic).
}
Remaining unanswered,
the question was raised again by ten~Cate
in~\cite{DBLP:journals/jphil/Cate06},
and finally a positive answer was provided
by Kuusisto in~\cite{DBLP:conf/aiml/Kuusisto08,DBLP:journals/apal/Kuusisto15}:
he showed that $\MSO(\ML)$
induces an \emph{infinite} hierarchy over \kl{pointed digraphs}.
This tells us that the hierarchy does not completely collapse at some level,
but a priori leaves open whether or not
each number of \kl[set quantifier]{quantifier} alternations
corresponds to a separate semantic level.

Kuusisto's proof builds upon the work
of Matz, Schweikardt and Thomas in \cite{DBLP:journals/iandc/MatzST02}
(elaborating on their previous results
in~\cite{DBLP:conf/lics/MatzT97} and~\cite{DBLP:conf/csl/Schweikardt97}),
where they have shown that
in the case of $\MSOL$ on \kl{digraphs},
the alternation hierarchy is \emph{strict}.
Thus,
each additional alternation between the two types of \kl{set quantifiers}
properly extends the family of \kl{definable} \kl{digraph languages}.
Significantly,
this separation also holds on \kl{grids},
a more restrictive class of \kl{structures},
where it can be established
using techniques from classical automata theory.
Furthermore,
taken in conjunction
with the \kl[device equivalence]{equivalence} of $\MSO(\bMLg)$ and $\MSOL$,
the result on \kl{digraphs} immediately implies that
the corresponding hierarchy of $\MSO(\bMLg)$ is \emph{infinite}.
But since the alternation levels of that logic
are not the same as those of $\MSOL$,
it does not seem obvious how strictness could be inferred.

The present \lcnamecref{ch:alternation}
provides an alternative way of transferring
the results of Matz, Schweikardt and Thomas
to the \kl[modal logic]{modal} setting.
In particular,
our method allows to show directly that
the \kl{set quantifier} alternation hierarchies
of $\MSO(\ML)$ and $\MSO(\MLg)$
are \emph{strict} over (\kl[pointed digraph]{pointed}) \kl{digraphs}.
\marginnote{
  To avoid the \kl{backward modalities} of $\MSO(\bMLg)$,
  we work instead with $\MSO(\MLg)$,
  which is called $\SOPMLE$
  in~\mbox{\textnormal{\cite{DBLP:conf/aiml/Kuusisto08,DBLP:journals/apal/Kuusisto15}}}.
  By duality,
  separating one alternation hierarchy also separates the other.
}
At first sight,
this seems to expand the existing body of knowledge,
especially since the strictness question for $\MSO(\ML)$
has been mentioned as an open problem
in~\cite{DBLP:conf/aiml/Kuusisto08} and~\cite{Kuusisto13}.
However,
it turns out that in both cases,
strictness is actually a consequence of infiniteness
[\osf{A. Kuusisto, personal communication, \printdate{2016-03-03}}].
Although this observation has so far not been formally published,
it appears to be folklore in the model-theory community.
Hence,
this \lcnamecref{ch:alternation} contributes new proofs
to essentially known results.
Just as Kuusisto has done in \cite{DBLP:conf/aiml/Kuusisto08,DBLP:journals/apal/Kuusisto15},
we use as a starting point
the strictness result of~\cite{DBLP:journals/iandc/MatzST02} for $\MSOL$ on \kl{grids}.
But from there on,
the two proof methods diverge considerably.

The original approach of Kuusisto
is mainly based on the fact
that one can simulate \kl{first-order quantifiers}
by means of \kl{set quantifiers},
combined with a \kl{formula} stating that a set is a singleton.
As already mentioned,
this can be used to show that $\MSO(\MLg)$ is \kl[device equivalent]{equivalent} to $\MSOL$.
The spirit of the proof in~\cite{DBLP:conf/aiml/Kuusisto08,DBLP:journals/apal/Kuusisto15}
is essentially the same for $\MSO(\ML)$,
although the details are much more technical,
since this logic is less expressive than $\MSOL$
on arbitrary \kl{pointed structures}.
It is precisely the use of additional
\kl[set quantifiers]{second-order quantifiers}
that leads to the temporary loss
of the specific separation results provided by~\cite{DBLP:journals/iandc/MatzST02}.

In contrast,
one simple insight
will allow us to directly transfer those results:
When restricted to the class of \kl{grids},
$\MSO(\MLg)$ and $\MSOL$ are more than just \kl[device equivalent]{equivalent}~--
they are \emph{levelwise} \kl[device equivalent]{equivalent},
and consequently
all the separation results shown for $\MSOL$
also hold for $\MSO(\MLg)$ on \kl{grids}.
This approach is based on the observation
that the existential fragment of $\MSO(\MLg)$
can simulate another model, called \kl[tiling systems]{\emph{tiling systems}},
which has been shown to be \kl[device equivalent]{equivalent} to the existential fragment of $\MSOL$
in \cite{DBLP:journals/iandc/GiammarresiRST96}.
On the basis of this new finding,
we can then transfer the given separation results from $\MSO(\MLg)$ on \kl{grids}
to other classes of \kl{digraphs} and other extensions of \kl{modal logic},
such as $\MSO(\ML)$.
While this works along the same general principle as the
\emph{strong first-order reductions} used in \cite{DBLP:journals/iandc/MatzST02},
the additional limitations imposed by \kl{modal logic}
force us to introduce custom encoding techniques
that cope with the lack of expressive power.

The remainder of this \lcnamecref{ch:alternation} is organized in a top-down manner.
After introducing the necessary notation in \cref{sec:alternation-preliminaries},
we present the main results in \cref{sec:results},
and almost immediately get to the central proof in \cref{sec:proofs}.
The latter relies on several other propositions,
but since those are treated as “black boxes”,
the main line of reasoning might be comprehensible
without reading any further.
We then provide all the missing details
in the last two \lcnamecrefs{sec:grids},
which are independent of each other.
\Cref{sec:grids} establishes the levelwise \kl[device equivalence]{equivalence}
of three different alternation hierarchies on \kl{grids},
and may thus be interesting on its own.
On the other hand,
\cref{sec:encodings} is dedicated to encoding functions,
which constitute the more technical part of our demonstration.

\section{Preliminaries}
\label{sec:alternation-preliminaries}

Assume we are given some set of \kl{formulas} $\FormulaSet$,
referred to as \Intro{kernel},
which is free of \kl{set quantifiers} and closed under negation
(e.g., $\MLg$).
Then,
for $\Level≥0$, the class \Intro*{$\SigmaMSO{\Level}(\FormulaSet)$} consists of those \kl{formulas}
that one can construct by taking a member of $\FormulaSet$
and prepending to it at most~$\Level$ consecutive blocks of \kl{set quantifiers},
alternating between existential and universal blocks,
such that the first block is existential.
Reformulating this
solely in terms of existential \kl[set quantifiers]{quantifiers} and negations,
we get
\begin{align*}
  \SigmaMSO{0}(\FormulaSet)   &\defeq \FormulaSet \quad \text{and} \\
  \SigmaMSO{\Level+1}(\FormulaSet) &\defeq \lrsetbuilder{\Exists{\SetSymbol}}{\SetSymbol∈\SetSymbolSet}^*\!·\lrsetbuilder{\NOT\Formula}{\Formula∈\SigmaMSO{\Level}(\FormulaSet)},
\end{align*}
where the second line uses set concatenation and the Kleene star.
We define \Intro*{$\PiMSO{\Level}(\FormulaSet)$} as the corresponding dual class,
i.e., the set of all negations of \kl{formulas} in $\SigmaMSO{\Level}(\FormulaSet)$.
Generalizing this to arbitrary Boolean combinations,
let \Intro*{$\BC\SigmaMSO{\Level}(\FormulaSet)$} denote
the smallest superclass of $\SigmaMSO{\Level}(\FormulaSet)$
that is closed under negation and disjunction.

The \kl{formulas} in $\SigmaMSO{\Level}(\FormulaSet)$ and $\PiMSO{\Level}(\FormulaSet)$
are said to be in \Intro{prenex normal form}
with respect to the \kl{kernel} $\FormulaSet$.
It is well known that every \kl[class-formula]{$\MSOL$-formula}
can be transformed into \kl{prenex normal form}
with \kl{kernel} class $\FOL$.
This is based on the observation that \kl{first-order quantifiers}
can be replaced by \kl[set quantifiers]{second-order ones}.
Using the construction of
\cref{ex:uniqueness} in \cref{sec:example-formulas},
it is not difficult to see that
the analogue holds for $\MSO(\ML)$, $\MSO(\dML)$, $\MSO(\MLg)$ and $\MSO(\dMLg)$
with respect to their corresponding \kl{kernel} classes.
A more elaborate explanation can be found in \cite[Prp.~3]{DBLP:journals/jphil/Cate06}.

For the sake of clarity,
we break with the tradition of implicit quantification
that is customary in modal logic.
Instead of evaluating \kl[class-formula]{$\MSO(\ML)$-formulas} on \kl[pointed structure]{non-pointed structures}
by means of “hidden” universal quantification,
we shall explicitly put a \kl{global box} in front of our \kl{formulas}.
This leads to the class
\Phantomintro{\GBX}
\begin{equation*}
  \reintro*{\GBX\SigmaMSO{\Level}(\ML)} \defeq \set{\gbx}·\SigmaMSO{\Level}(\ML).
\end{equation*}
Analogously, we also define \reintro*{$\GBX\PiMSO{\Level}(\ML)$}.

All of our results will be stated in terms of the semantic classes
that one obtains by evaluating the preceding \kl{formula} classes
on some set of \kl{structures} $\StructClass$.
On the semantic side,
we will additionally consider the class
\Phantomintro{\DeltaMSO}
\begin{equation*}
  \semF{\reintro*{\DeltaMSO{\Level}}(\FormulaSet)}[\StructClass] \defeq \semF{\SigmaMSO{\Level}(\FormulaSet)}[\StructClass] ∩ \semF{\PiMSO{\Level}(\FormulaSet)}[\StructClass].
\end{equation*}
Since it is not based on any syntactic counterpart,
there is no meaning attributed to
the notation $\DeltaMSO{\Level}(\FormulaSet)$ by itself (without the brackets).

\section{Separation results}
\label{sec:results}

\begin{table*}[tp]
  \centering
  \newcommand*{\seplinespace}{\addlinespace[2\defaultaddspace]}
  \widefloat{
    \begin{tabular}{lll@{\hspace{0ex}}r@{\hspace{3ex}}l}
      \toprule
      \textit{Separation result} & \textit{Kernel} & \textit{Structures} & \textit{Levels} & \textit{Theorem} \\
          & Class $\FormulaSet[1]$ & Class $\StructClass$ & $\Level≥{·}$\, \\
      \midrule\addlinespace
      $\semF{\DeltaMSO{\Level+1}(\FormulaSet[1])}[\StructClass]⊈\semF{\BC\SigmaMSO{\Level}(\FormulaSet[1])}[\StructClass]$
          & $\FOL$         & $\GRID$, $\DIGRAPH$, $\GRAPH$ & $1$ & \ttheorem{MST}{thm:DB-FO} $\bast$ \\
          & $\dMLg$, $\MLg$ & $\GRID$, $\DIGRAPH$, $\GRAPH[1][1]$ & $1$ & \ttheorem{R}{thm:DB-HBG-HG} \\\seplinespace
      $\semF{\SigmaMSO{\Level}(\FormulaSet[1])}[\StructClass]\incomparable\semF{\PiMSO{\Level}(\FormulaSet[1])}[\StructClass]$
          & $\FOL$         & $\GRID$, $\DIGRAPH$, $\GRAPH$ & $1$ & \ttheorem{MST}{thm:SP-FO} $\bast$ \\
          & $\dMLg$, $\MLg$ & $\GRID$, $\DIGRAPH$, $\GRAPH[1][1]$ & $1$ & \ttheorem{R}{thm:SP-HBG-HG} \\
          & $\ML$          & $\pDIGRAPH$ & $1$ & \ttheorem{R}{thm:SP-H} \\\seplinespace
      $\semF{\GBX\SigmaMSO{\Level}(\FormulaSet[1])}[\StructClass]⊈\semF{\GBX\PiMSO{\Level}(\FormulaSet[1])}[\StructClass]$
          & $\ML$          & $\DIGRAPH$ & $2$ & \ttheorem{R}{thm:gSP-H} \\
      \addlinespace[1.5\defaultaddspace]\bottomrule
    \end{tabular}
  }
  \caption{The specific separation results of \cref{thm:separation-MST,thm:separation-R}.
    \Cref{thm:separation-MST} (marked by asterisks)
    is due to Matz, Schweikardt and Thomas.}
  \label{tab:results}
\end{table*}

With the notation in place,
we are ready to formally enunciate the main \lcnamecref{thm:separation-R},
whose complete proof will be the subject of the remainder of this \lcnamecref{ch:alternation}.
It is an extension to \kl[modal logic]{modal} \kl{kernel} \kl{formulas}
of the following result of Matz, Schweikardt and Thomas,
obtained by combining
\mbox{\cite[Thm.~1]{DBLP:journals/iandc/MatzST02}} and \mbox{\cite[Thm.~2.26]{DBLP:journals/tcs/Matz02}}\footnote{
  \cite[Thm.~2.26]{DBLP:journals/tcs/Matz02} states that
  $\semF{\SigmaMSO{\Level}(\FOL)}[\GRID]⊉\semF{\PiMSO{\Level}(\FOL)}[\GRID]$,
  which, by duality, also implies
  $\semF{\SigmaMSO{\Level}(\FOL)}[\GRID]⊈\semF{\PiMSO{\Level}(\FOL)}[\GRID]$.
}:

\begin{theorem}[Matz, Schweikardt, Thomas]
  \label{thm:separation-MST}
  The \kl{set quantifier} alternation hierarchy of $\MSOL$ is strict
  over the classes of \kl{grids}, \kl{digraphs} and \kl{undirected graphs}.

  \emph{A more precise statement of this theorem, referred to as
    \cref{thm:separation-MST}~\ref{thm:DB-FO}~and~\ref{thm:SP-FO},
    is given in \cref{tab:results}.}
\end{theorem}

Roughly speaking,
the extension provided in the present \lcnamecref{ch:alternation} tells us
that the preceding separations are largely maintained
if we replace the \kl[first-order logic]{first-order} \kl{kernel} by certain classes
of \kl[modal logic]{modal} \kl{formulas}.
To facilitate comparisons,
the formal statements of both \lcnamecrefs{thm:separation-MST}
are presented together in the same \lcnamecref{tab:results}.

\begin{theorem}[Main Results]
  \label{thm:separation-R}
  The \kl{set quantifier} alternation hierarchies of $\MSO(\dMLg)$ and $\MSO(\MLg)$ are strict
  over the classes of \kl{grids}, \kl{digraphs} and $1$-bit \kl{labeled} \kl{undirected graphs}.

  Furthermore, the corresponding hierarchies of $\MSO(\ML)$ and~$\GBX\+\MSO(\ML)$ are (mostly) strict
  over the classes of \kl{pointed digraphs} and \kl{digraphs}, respectively.

  \emph{A more precise statement of this theorem, referred to as
    \cref{thm:separation-R}~\ref{thm:DB-HBG-HG},~\ref{thm:SP-HBG-HG},~\ref{thm:SP-H}~and~\ref{thm:gSP-H},
    is given in \cref{tab:results}.}
\end{theorem}

By basic properties of predicate logic\marginnote{
  In particular, the inclusion
  of\, $\semF{\BC\SigmaMSO{\Level}(\FormulaSet[1])}[\StructClass]$
  in~$\semF{\DeltaMSO{\Level+1}(\FormulaSet[1])}[\StructClass]$
  follows from the fact that,
  when transforming a Boolean combination
  of \kl[class-formula]{$\SigmaMSO{\Level}(\FormulaSet[1])$-formulas}
  into \kl{prenex normal form},
  one is free to choose whether the resulting \kl{formula}
  (with up to $\Level+1$ \kl[set quantifier]{quantifier} alternations)
  should start with an existential or a universal \kl[set quantifier]{quantifier}.}
and the transitivity of set inclusion,
it is easy to infer from \cref{thm:separation-R}
the hierarchy diagrams represented in
\cref{fig:hierarchy-delta,fig:hierarchy-nodelta}.

If we take into account all the depicted relations,
the diagram in \cref{fig:hierarchy-delta}
is the same as in \cite{DBLP:journals/iandc/MatzST02} and \cite{DBLP:journals/tcs/Matz02}.
Hence,
when switching to one of the \kl[modal logic]{modal} \kl{kernels}
that include \kl{global modalities},
i.e., $\dMLg$ or $\MLg$,
the separations of \cref{thm:separation-MST}
are completely preserved on \kl{grids} and \kl{digraphs}.
Our proof method also allows us to
easily transfer this result to \kl{undirected graphs},
as long as we admit that
the vertices may be \kl{labeled} with at least one bit.
Additional work would be required to eliminate this condition.

\begin{SCfigure}
  \centering
  \input{fig/hierarchy-delta.tex}
  \hspace{1.5ex}
  \caption{The \kl{set quantifier} alternation hierarchies established by
    \cref{thm:separation-R}~\ref{thm:DB-HBG-HG},~\ref{thm:SP-HBG-HG}~and~\ref{thm:SP-H}.
    If we include the noninclusion in parentheses,
    this diagram holds for
    $\FormulaSet[1]{\:∈\,}\set{\dMLg,\MLg}$\, and\, $\StructClass{\:∈\,}\set{\GRID,\DIGRAPH,\GRAPH[1][1]}$.
    If we ignore that noninclusion,
    it is also verified for
    $\FormulaSet[1] = \ML$\, and\, $\StructClass = \pDIGRAPH$.
    In both cases, we assume $\Level≥1$.}
  \label{fig:hierarchy-delta}
\end{SCfigure}

\begin{SCfigure}
  \centering
  \input{fig/hierarchy-nodelta.tex}
  \caption{The \kl{set quantifier} alternation hierarchy implied by
    \cref{thm:separation-R}~\ref{thm:gSP-H}\,
    for $\FormulaSet[1] = \ML$,\, $\StructClass = \DIGRAPH$,\, and $\Level≥2$.}
  \label{fig:hierarchy-nodelta}
\end{SCfigure}

As a spin-off,
\cref{thm:separation-R} also provides an extension
of some of these separations to $\ML$,
a \kl{kernel} class without \kl{global modalities}.
Following \cite{DBLP:conf/aiml/Kuusisto08,DBLP:journals/apal/Kuusisto15},
we consider the alternation hierarchies of both $\MSO(\ML)$ and $\GBX\+\MSO(\ML)$.
For the former,
which is evaluated on \kl{pointed digraphs},
\cref{fig:hierarchy-delta} gives a detailed picture,
leaving open only whether the inclusion
$\semF{\BC\SigmaMSO{\Level}(\FormulaSet[1])}[\StructClass] ⊆ \semF{\DeltaMSO{\Level+1}(\FormulaSet[1])}[\StructClass]$ is proper.
Inferring the strictness of this inclusion from the preceding results
does not seem very difficult,
but would call for a generalization of our framework.
In contrast,
the second hierarchy based on $\ML$ is arguably less natural,
since every \kl[class-formula]{$\GBX\+\MSO(\ML)$-formula} is prefixed by a \kl{global box},
regardless of the occurring \kl{set quantifiers}.
This creates a certain asymmetry
between the $\SigmaMSO{\Level}$-{} and $\PiMSO{\Level}$-levels,
which becomes apparent when considering
the missing relations in \cref{fig:hierarchy-nodelta}.
Unlike for the other hierarchies,
one cannot simply argue by duality
to deduce from
$\semF{\GBX\SigmaMSO{\Level}(\FormulaSet[1])}[\StructClass]⊈\semF{\GBX\PiMSO{\Level}(\FormulaSet[1])}[\StructClass]$
that the converse noninclusion also holds.
Nevertheless,
the presented result is strong enough to answer
the specific strictness question mentioned in~\cite{DBLP:conf/aiml/Kuusisto08}:
For arbitrarily high~$\Level$, we~have
\begin{equation*}
  \semF{\GBX\SigmaMSO{\Level}(\ML)}[\DIGRAPH]⊉\semF{\GBX\SigmaMSO{\Level+1}(\ML)}[\DIGRAPH] \+.
\end{equation*}

\section{Top-level proofs}
\label{sec:proofs}

In accordance with our top-down approach,
the present \lcnamecref{sec:proofs} already provides the proof
of our main \lcnamecref{thm:separation-R},
where everything comes together.
It therefore acts as a gateway to the \lcnamecrefs{sec:grids}
with the technical parts,
especially \cref{sec:encodings}.

\subsection{Figurative inclusions}

First of all,
we need to introduce the primary tool with which
we will transfer separation results from one setting to another.
It can be seen as an abstraction of the \emph{strong first-order reductions}
used in \cite{DBLP:journals/iandc/MatzST02}.
Unlike the latter,
it is formulated independently of any logical language,
which allows us to postpone the technical details to the end of the \lcnamecref{ch:alternation}.

\begin{definition}[Figurative Inclusion]
  \NoMathBreak
  Consider two sets $\StructClass[1]$ and $\StructClass[2]$
  and a partial \emph{injective} function $\Encoding\colon\StructClass[1]\pto\StructClass[2]$.
  For any two families of subsets $\StructLangClass[1]⊆\powerset{\StructClass[1]}$ and $\StructLangClass[2]⊆\powerset{\StructClass[2]}$\!,
  we say that $\StructLangClass[1]$ is \Intro{forward included} in $\StructLangClass[2]$
  \Intro{figuro} $\Encoding$, and write \Intro*{$\StructLangClass[1] \figsubeq{\Encoding} \StructLangClass[2]$},
  if for every set $\StructLang[1]∈\StructLangClass[1]$, there is a set $\StructLang[2]∈\StructLangClass[2]$
  such that $\Encoding(\StructLang[1]) = \StructLang[2] ∩ \Encoding(\StructClass[1])$.
\end{definition}

Figuratively speaking,
the partial bijection $\Encoding$ creates a tunnel between $\StructClass[1]$ and $\StructClass[2]$,
and all the sets in $\StructLangClass[1]$ and $\StructLangClass[2]$ are cropped
to fit through that tunnel.
Two original sets are considered to be equal
if their cropped versions are mapped onto each other by~$\Encoding$.

\marginnote{
  We denote the inverse function of $\Encoding$
  by \Intro*{$\invf{\Encoding}$} \\
  and the identity \\
  function on $\StructClass$
  by \Intro*{$\id[\StructClass]$}.
}
We also define the shorthands $\figsupeq{\Encoding}$\, and $\figeq{\Encoding}$\,
as natural extensions of the previous notation:
\Intro*{$\StructLangClass[1]\figsupeq{\Encoding}\StructLangClass[2]$},
which is defined as $\StructLangClass[2]\figsubeq{\invf{\Encoding}}\StructLangClass[1]$,
means that $\StructLangClass[2]$ is \Intro{backward included} in $\StructLangClass[1]$ \reintro{figuro} $\Encoding$,
and \Intro*{$\StructLangClass[1]\figeq{\Encoding}\StructLangClass[2]$},
an abbreviation for the conjunction of
$\StructLangClass[1]\figsubeq{\Encoding}\StructLangClass[2]$ and $\StructLangClass[1]\figsupeq{\Encoding}\StructLangClass[2]$,
states that $\StructLangClass[1]$ is \Intro{forward equal} to $\StructLangClass[2]$ \reintro{figuro} $\Encoding$.
All of these relations are referred to as \Intro{figurative inclusions}.

Note that ordinary inclusion is
a special case of \kl{figurative inclusion}, i.e.,
for $\StructClass[1] = \StructClass[2]$,
\begin{equation*}
  \StructLangClass[1] ⊆ \StructLangClass[2] \quad \text{if and only if} \quad \StructLangClass[1] \figsubeq{\id[\StructClass[1]]} \StructLangClass[2].
\end{equation*}
Furthermore, \kl{figurative inclusion} is transitive in the sense that
\begin{equation*}
  \StructLangClass[1] \,\figsubeq{\Encoding[1]}\, \StructLangClass[2] \,\figsubeq{\Encoding[2]}\, \StructLangClass[3]
  \quad \text{implies} \quad
  \StructLangClass[1] \,\figsubeq{\Encoding[2]∘\Encoding[1]}\, \StructLangClass[3].
\end{equation*}
(This depends crucially on the fact that $\Encoding[2]$ is injective.)
\begin{proof}
  Consider three sets $\StructClass[1]$, $\StructClass[2]$ and $\StructClass[3]$,
  two partial injective functions $\Encoding[1]\colon\StructClass[1]\pto\StructClass[2]$ and $\Encoding[2]\colon\StructClass[2]\pto\StructClass[3]$,
  and three families of subsets $\StructLangClass[1]⊆\powerset{\StructClass[1]}$,\, $\StructLangClass[2]⊆\powerset{\StructClass[2]}$ and $\StructLangClass[3]⊆\powerset{\StructClass[3]}$.
  Assume that we have $\StructLangClass[1] \,\figsubeq{\Encoding[1]}\, \StructLangClass[2] \,\figsubeq{\Encoding[2]}\, \StructLangClass[3]$.
  Choose an arbitrary set $\StructLang[1]∈\StructLangClass[1]$.
  Since $\StructLangClass[1] \figsubeq{\Encoding[1]} \StructLangClass[2]$,
  there must be a set $\StructLang[2]∈\StructLangClass[2]$ such that $\Encoding[1](\StructLang[1]) = \StructLang[2] ∩ \Encoding[1](\StructClass[1])$.
  Furthermore, as $\StructLangClass[2] \figsubeq{\Encoding[2]} \StructLangClass[3]$,
  there is also a set $\StructLang[3]∈\StructLangClass[3]$ such that $\Encoding[2](\StructLang[2]) = \StructLang[3] ∩ \Encoding[2](\StructClass[2])$.
  Hence,
  \begin{align*}
    (\Encoding[2]∘\Encoding[1])(\StructLang[1]) &= \Encoding[2](\StructLang[2] ∩ \Encoding[1](\StructClass[1])) \\
             &= \Encoding[2](\StructLang[2]) ∩ (\Encoding[2]∘\Encoding[1])(\StructClass[1]) \tag{$*$}\label{eq:distributivity} \\
             &= \StructLang[3] ∩ \Encoding[2](\StructClass[2]) ∩ (\Encoding[2]∘\Encoding[1])(\StructClass[1]) \\
             &= \StructLang[3] ∩ (\Encoding[2]∘\Encoding[1])(\StructClass[1]).
  \end{align*}
  Equality~\eqref{eq:distributivity} holds because $\Encoding[2]$ is injective.
  Since the choice of~$\StructLang[1]$ was arbitrary,
  there is such an $\StructLang[3]∈\StructLangClass[3]$ for every $\StructLang[1]∈\StructLangClass[1]$,
  and thus $\StructLangClass[1] \figsubeq{\Encoding[2]∘\Encoding[1]} \StructLangClass[3]$.
\end{proof}

In our specific context,
given a noninclusion $\semF{\FormulaSet[1]_2}[\StructClass[1]] ⊈ \semF{\FormulaSet[1]_1}[\StructClass[1]]$,
we shall use the concept of \kl{figurative inclusion} to infer from it
another noninclusion $\semF{\FormulaSet[2]_2}[\StructClass[2]] ⊈ \semF{\FormulaSet[2]_1}[\StructClass[2]]$.
Here,
$\FormulaSet[1]_1,\FormulaSet[1]_2,\FormulaSet[2]_1,\FormulaSet[2]_2$ and $\StructClass[1],\StructClass[2]$
refer to some classes of \kl{formulas} and \kl{structures}, respectively.
The key part of the argument will be
to construct an appropriate encoding function $\Encoding\colon\StructClass[1]\to\StructClass[2]$,
in order to apply the following \lcnamecref{lem:separation-transfer}.~

\begin{lemma}
  \label{lem:separation-transfer}
  \NoMathBreak
  Let $\StructLangClass[1]_1,\StructLangClass[1]_2⊆\powerset{\StructClass[1]}$ and $\StructLangClass[2]_1,\StructLangClass[2]_2⊆\powerset{\StructClass[2]}$ be families of subsets of some sets $\StructClass[1]$ and $\StructClass[2]$.
  If there is a \emph{total \mbox{injective}} function $\Encoding\colon\StructClass[1]\to\StructClass[2]$
  such that $\StructLangClass[1]_2 \figsubeq{\Encoding} \StructLangClass[2]_2$ and $\StructLangClass[1]_1 \figsupeq{\Encoding} \StructLangClass[2]_1$,
  then
  \begin{equation*}
    \StructLangClass[1]_2 ⊈ \StructLangClass[1]_1 \quad \text{implies} \quad \StructLangClass[2]_2 ⊈ \StructLangClass[2]_1. \qedhere
  \end{equation*}
\end{lemma}
\begin{proof}
  \NoMathBreak
  To show the contrapositive,
  let us suppose that $\StructLangClass[2]_2 ⊆ \StructLangClass[2]_1$,
  or, equivalently, $\StructLangClass[2]_2 \figsubeq{\id[\StructClass[2]]} \StructLangClass[2]_1$.
  Then the chain of \kl{figurative inclusions}
  \begin{equation*}
    \StructLangClass[1]_2 \:\figsubeq{\Encoding}\: \StructLangClass[2]_2 \:\figsubeq{\id[\StructClass[2]]}\: \StructLangClass[2]_1 \:\figsubeq{\invf{\Encoding}}\: \StructLangClass[1]_1
  \end{equation*}
  yields $\StructLangClass[1]_2 \figsubeq{\id[\StructClass[1]]} \StructLangClass[1]_1$,
  since $(\invf{\Encoding}\!∘\id[\StructClass[2]]∘\Encoding) = \id[\StructClass[1]]$.
  (This depends on $\Encoding$ being total and injective.)
  Consequently, we have $\StructLangClass[1]_2 ⊆ \StructLangClass[1]_1$.
\end{proof}

In some cases,
we can combine two given \kl{figurative inclusions} in order to obtain a new one
that relates the corresponding intersection classes.
This property will be very useful for establishing \kl{figurative inclusions}
between classes of the form $\semF{\DeltaMSO{\Level}(\FormulaSet[1])}[\StructClass[1]]$.

\begin{lemma}
  \label{lem:figsubeq-intersection}
  \NoMathBreak
  Consider two sets $\StructClass[1]$ and $\StructClass[2]$,
  a partial injective function $\Encoding\colon\StructClass[1]\pto\StructClass[2]$,
  and four families of subsets $\StructLangClass[1]_1,\StructLangClass[1]_2⊆\powerset{\StructClass[1]}$ and $\StructLangClass[2]_1,\StructLangClass[2]_2⊆\powerset{\StructClass[2]}$.
  If $\Encoding(\StructClass[1])$ is a member of $\StructLangClass[2]_1∩\StructLangClass[2]_2$,
  and $\StructLangClass[2]_1$, $\StructLangClass[2]_2$ are both closed under intersection, then
  \begin{gather*}
   \StructLangClass[1]_1 \figsubeq{\Encoding} \StructLangClass[2]_1 \quad \text{and} \quad \StructLangClass[1]_2 \figsubeq{\Encoding} \StructLangClass[2]_2 \quad
   \quad \text{imply} \quad \StructLangClass[1]_1∩\StructLangClass[1]_2 \,\figsubeq{\Encoding}\, \StructLangClass[2]_1∩\StructLangClass[2]_2.
   \qedhere
  \end{gather*}
\end{lemma}
\begin{proof}
  Let $\StructLang[1]$ be any set in $\StructLangClass[1]_1∩\StructLangClass[1]_2$.
  Since $\StructLangClass[1]_1 \figsubeq{\Encoding} \StructLangClass[2]_1$,
  there is, by definition, a set $\StructLang[2]$ in $\StructLangClass[2]_1$
  such that $\Encoding(\StructLang[1]) = \StructLang[2] ∩ \Encoding(\StructClass[1])$.
  Furthermore, we also know that $\Encoding(\StructClass[1])$ lies in $\StructLangClass[2]_1$,
  and that the latter is closed under intersection.
  Hence, $\Encoding(\StructLang[1])∈\StructLangClass[2]_1$.
  Analogously, we also get that $\Encoding(\StructLang[1])∈\StructLangClass[2]_2$.
  Finally, knowing that
  for all $\StructLang[1]$ in $\StructLangClass[1]_1∩\StructLangClass[1]_2$,
  $\Encoding(\StructLang[1])$ lies in $\StructLangClass[2]_1∩\StructLangClass[2]_2$,
  we obviously have a sufficient condition for
  $\StructLangClass[1]_1∩\StructLangClass[1]_2 \figsubeq{\Encoding} \StructLangClass[2]_1∩\StructLangClass[2]_2$.
\end{proof}

\subsection{Proving the main theorem}
\label{ssec:main-proof}

We are now ready to give the central proof of this \lcnamecref{ch:alternation}.
Although it makes references to many statements of \cref{sec:grids,sec:encodings},
it is formulated in a way that can be understood without having read
anything beyond this point.

\begin{proof}[Proof of \cref{thm:separation-R}]
  The basis of our proof shall be laid in \cref{sec:grids},
  where the case $\BitCount=0$ of \cref{thm:levelwise-equivalence}
  will state the following:
  When restricted to the class of \kl{grids},
  the \kl{set quantifier} alternation hierarchies of
  $\MSOL$, $\MSO(\dMLg)$ and $\MSO(\MLg)$ are \kl[device equivalent]{equivalent}.
  More precisely,
  for every $\Level≥1$ and
  $\FormulaClassOperator ∈ \set{\SigmaMSO{\Level},\,\PiMSO{\Level},\,\BC\SigmaMSO{\Level},\,\DeltaMSO{\Level}}$,
  it holds that
  \begin{equation*}
    \semF{\+\FormulaClassOperator\+(\FOL)}[\GRID] \:=\; \semF{\+\FormulaClassOperator\+(\dMLg)}[\GRID] \:=\; \semF{\+\FormulaClassOperator\+(\MLg)}[\GRID] \+.
  \end{equation*}
  Hence,
  if we consider only the case $\StructClass=\GRID$,
  the separation results for the \kl{kernel} class $\FOL$ stated in
  \cref{thm:separation-MST}~\ref{thm:DB-FO}~and~\ref{thm:SP-FO}
  immediately imply those for $\dMLg$ and $\MLg$ in
  \cref{thm:separation-R}~\ref{thm:DB-HBG-HG}~and~\ref{thm:SP-HBG-HG}.

  The remainder of the proof now consists of
  establishing suitable \kl{figurative inclusions},
  in order to transfer these results
  to other classes of \kl{structures} and,
  to some extent,
  to weaker classes of \kl{kernel} \kl{formulas}.
  For this purpose,
  we shall introduce in \cref{sec:encodings} a notion of translatability
  between two classes of \kl{kernel} \kl{formulas} $\FormulaSet[1]$ and $\FormulaSet[2]$,
  with respect to a given total injective function $\Encoding$
  that encodes \kl{structures} from a class~$\StructClass[1]$ into \kl{structures} of some class~$\StructClass[2]$.
  As will be shown in \cref{lem:translation-inclusion},
  bidirectional translatability implies
  \begin{equation}
    \label{eqn:levelwise-figeq}
    \semF{\+\FormulaClassOperator\+(\FormulaSet[1])}[\StructClass[1]] \;\figeq{\Encoding}\; \semF{\+\FormulaClassOperator\+(\FormulaSet[2])}[\StructClass[2]]
    \tag{$\ast$}
  \end{equation}
  for all $\FormulaClassOperator ∈ \set{\SigmaMSO{\Level},\,\PiMSO{\Level},\,\BC\SigmaMSO{\Level}}$ with $\Level≥0$.
  If we can additionally show that $\Encoding(\StructClass[1])$ is (at most)
  $\DeltaMSO{2}(\FormulaSet[2])$-\kl{definable} over $\StructClass[2]$,
  then, by \cref{lem:figsubeq-intersection},
  the \kl{figurative equality} \eqref{eqn:levelwise-figeq}
  also holds for $\FormulaClassOperator = \DeltaMSO{\Level+1}$ with $\Level≥1$.
  Note that the backward part “$\figsupeq{\Encoding}$\!” is always true,
  since $\invf{\Encoding}(\StructClass[2])$ is trivially $\DeltaMSO{2}(\FormulaSet[1])$-\kl{definable} over $\StructClass[1]$.

  The groundwork being in place,
  we proceed by applying \cref{lem:separation-transfer} as follows:
  \begin{itemize}[nosep]
  \item If we have established \eqref{eqn:levelwise-figeq} for $\FormulaClassOperator ∈ \set{\SigmaMSO{\Level},\,\PiMSO{\Level}}$,
    then we can transfer the separation
    \begin{equation*}
      \label{eqn:separation-SP}
      \semF{\SigmaMSO{\Level}(\FormulaSet[1])}[\StructClass[1]]\incomparable\semF{\PiMSO{\Level}(\FormulaSet[1])}[\StructClass[1]]
      \tag{1}
    \end{equation*}
    to the \kl{kernel} class $\FormulaSet[2]$ evaluated on the class of \kl{structures}~$\StructClass[2]$.
  \item Similarly,
    if \eqref{eqn:levelwise-figeq} holds
    for $\FormulaClassOperator ∈ \set{\BC\SigmaMSO{\Level},\,\DeltaMSO{\Level+1}}$,
    then
    \begin{equation*}
      \label{eqn:separation-DB}
      \semF{\DeltaMSO{\Level+1}(\FormulaSet[1])}[\StructClass[1]]⊈\semF{\BC\SigmaMSO{\Level}(\FormulaSet[1])}[\StructClass[1]]
      \tag{2}
    \end{equation*}
    can also be transferred to $\FormulaSet[2]$ on $\StructClass[2]$.
  \end{itemize}

  It remains to provide concrete \kl{figurative inclusions}
  to prove the different parts of \cref{thm:separation-R}.

  \proofparagraph{\ref{thm:DB-HBG-HG}\:\!,~\ref{thm:SP-HBG-HG}}
  The first two parts are treated in parallel.
  We start by transferring~\eqref{eqn:separation-SP} and~\eqref{eqn:separation-DB}
  from \kl{grids} to \kl{digraphs}, for the \kl{kernel} class $\dMLg$,
  taking a detour via $2$-relational \kl{digraphs},
  and then via $2$-bit \kl{labeled} ones.
  For all $\FormulaClassOperator ∈ \set{\SigmaMSO{\Level},\,\PiMSO{\Level},\,\BC\SigmaMSO{\Level},\,\DeltaMSO{\Level+1}}$ with $\Level≥1$,
  we get
  \begin{alignat*}{2}
    \semF{\+\FormulaClassOperator\+(\dMLg)}[\GRID] \;&\figeq{\id[\+\GRID]} & \;&\semF{\+\FormulaClassOperator\+(\dMLg)}[{\DIGRAPH[0][2]}] \\
                               &\figeq{\Encoding_1}          &   &\semF{\+\FormulaClassOperator\+(\dMLg)}[{\DIGRAPH[2][1]}] \\
                               &\figeq{\Encoding_2}          &   &\semF{\+\FormulaClassOperator\+(\dMLg)}[\DIGRAPH] \+.
  \end{alignat*}
  The first line is trivial for $\FormulaClassOperator ∈ \set{\SigmaMSO{\Level},\,\PiMSO{\Level},\,\BC\SigmaMSO{\Level}}$,
  since $\GRID ⊆ \DIGRAPH[0][2]$.
  It also holds for $\FormulaClassOperator = \DeltaMSO{\Level+1}$ because
  $\GRID$ is $\PiMSO{1}(\dMLg)$-\kl{definable} over $\DIGRAPH[0][2]$,
  as shall be demonstrated in \cref{prp:grid-definability}.
  The other two lines rely on the existence of
  adequate injective functions $\Encoding_1$ and $\Encoding_2$
  that allow us to apply
  \crefnosort{lem:translation-inclusion,lem:figsubeq-intersection}
  in the way explained above.
  They will be provided by
  \cref{prp:multirelational-labeled,prp:labeled-unlabeled},
  respectively.

  We proceed in a similar way
  to transfer \eqref{eqn:separation-SP} and \eqref{eqn:separation-DB}
  from $\dMLg$ to $\MLg$ on \kl{digraphs}:
  \begin{alignat*}{2}
    \semF{\+\FormulaClassOperator\+(\dMLg)}[\DIGRAPH] \,&\figeq{\Encoding_3} & \;\,&\semF{\+\FormulaClassOperator\+(\MLg)}[{\DIGRAPH[0][2]}] \\
                                  &\figeq{\Encoding_1} &     &\semF{\+\FormulaClassOperator\+(\MLg)}[{\DIGRAPH[2][1]}] \\
                                  &\figeq{\Encoding_2} &     &\semF{\+\FormulaClassOperator\+(\MLg)}[\DIGRAPH] \+,
  \end{alignat*}
  for $\FormulaClassOperator ∈ \set{\SigmaMSO{\Level},\,\PiMSO{\Level},\,\BC\SigmaMSO{\Level},\,\DeltaMSO{\Level+1}}$ with $\Level≥1$.
  The very simple encoding function $\Encoding_3$,
  which lets us eliminate \kl{backward modalities}
  and again use \crefnosort{lem:translation-inclusion,lem:figsubeq-intersection},
  will be supplied by \cref{prp:hbg-to-hg}.
  The encodings $\Encoding_1$ and $\Encoding_2$ are the same as before,
  because the properties asserted by
  \cref{prp:multirelational-labeled,prp:labeled-unlabeled}
  hold for both $\MLg$ and $\dMLg$ as \kl{kernel} classes.
  Incidentally, this means
  we could transfer \eqref{eqn:separation-SP}
  directly from $\MLg$ on \kl{grids} to $\MLg$ on \kl{digraphs},
  without even mentioning $\dMLg$.

  To show that \eqref{eqn:separation-SP} and \eqref{eqn:separation-DB}
  are also valid for $\MLg$ on $1$-bit \kl{labeled} \kl{undirected graphs},
  we establish
  \begin{equation*}
    \semF{\+\FormulaClassOperator\+(\dMLg)}[\DIGRAPH] \, \figeq{\Encoding_4} \, \semF{\+\FormulaClassOperator\+(\MLg)}[{\GRAPH[1][1]}] \+,
  \end{equation*}
  again for all $\FormulaClassOperator ∈ \set{\SigmaMSO{\Level},\,\PiMSO{\Level},\,\BC\SigmaMSO{\Level},\,\DeltaMSO{\Level+1}}$ with $\Level≥1$.
  The appropriate encoding $\Encoding_4$ shall be constructed in \cref{prp:digraph-1bitgraph}.
  Since \kl{backward modalities} do not offer any additional expressive power
  on \kl{undirected graphs},
  the separations we obtain also hold for the \kl{kernel} $\dMLg$.

  \proofparagraph{\ref{thm:SP-H}}
  Next,
  to transfer \eqref{eqn:separation-SP}
  from $\MLg$ on \kl{digraphs} to $\ML$ on \kl{pointed digraphs},
  we show that, for $\FormulaClassOperator ∈ \set{\SigmaMSO{\Level},\,\PiMSO{\Level}}$ with $\Level≥1$,
  we have
  \begin{equation*}
    \semF{\+\FormulaClassOperator\+(\MLg)}[\DIGRAPH] \:\figeq{\Encoding_5}\; \semF{\+\FormulaClassOperator\+(\ML)}[\pDIGRAPH] \+.
  \end{equation*}
  The injective function $\Encoding_5$,
  which satisfies the translatability property
  required to obtain this \kl{figurative equality} via \cref{lem:translation-inclusion},
  will be provided by \cref{prp:digraph-pdigraph}.
  Its image $\Encoding_5(\DIGRAPH)$ is not $\MSO(\ML)$-\kl{definable},
  for the simple reason
  that an \kl[class-formula]{$\MSO(\ML)$-formula} is unable to distinguish between two \kl{structures}
  that are isomorphic when restricted to
  the \kl{connected} component containing the position marker $\PosSymbol$.
  Hence, we cannot merely apply \cref{lem:figsubeq-intersection}
  to show~\eqref{eqn:separation-DB}.
  Our approach would have to be refined
  to take into account equivalence classes of \kl{structures},
  which we shall not do in this thesis.

  \proofparagraph{\ref{thm:gSP-H}}
  Finally,
  \cref{prp:digraph-pdigraph} will also state that
  $\Encoding_5$ can be converted into an encoding $\Encoding_5'$, from $\DIGRAPH$ back into $\DIGRAPH$,
  that satisfies the following \kl{figurative inclusions} for all $\Level≥2$:
  \begin{alignat*}{2}
    \semF{\SigmaMSO{\Level}(\MLg)}[\DIGRAPH] &\,\figsubeq{\Encoding_5'}\;\, & &\semF{\GBX\SigmaMSO{\Level}(\ML)}[\DIGRAPH] \,, \\
    \semF{\PiMSO{\Level}(\MLg)}[\DIGRAPH] &\,\figeq{\Encoding_5'}\;\,    & &\semF{\GBX\PiMSO{\Level}(\ML)}[\DIGRAPH] \,.
  \end{alignat*}
  Using \eqref{eqn:separation-SP} for $\MLg$ on \kl{digraphs},
  and applying \cref{lem:separation-transfer},
  we can infer from this that
  \begin{equation*}
    \semF{\GBX\SigmaMSO{\Level}(\ML)}[\DIGRAPH]⊈\semF{\GBX\PiMSO{\Level}(\ML)}[\DIGRAPH] \+.
    \qedhere
  \end{equation*}
\end{proof}

\section{Grids as a starting point}
\label{sec:grids}

In this \lcnamecref{sec:grids},
we establish that the \kl{set quantifier} alternation hierarchies of
$\MSOL$, $\MSO(\dMLg)$ and $\MSO(\MLg)$ are \kl[device equivalent]{equivalent} on \kl{labeled} \kl{grids}.
In addition,
we give a \kl[class-formula]{$\eclF{\PiMSO{1}(\dMLg)}$-formula}
that characterizes the class of \kl{grids}.

\subsection{The standard translation}

Our first building block is a well-known property of modal logic,
which holds even if we do not confine ourselves to the setting of \kl{grids}.

\begin{proposition}
  \label{prp:standard-translation}
  For every \kl[class-formula]{$\dMLg$-formula},
  there is an \kl[device equivalent]{equivalent} \kl[class-formula]{$\FOL$-formula}, i.e.,
  \begin{equation*}
    \semF{\dMLg} \;⊆\; \semF{\FOL} \+.
    \qedhere
  \end{equation*}
\end{proposition}
\begin{proof}
  Given an \kl[class-formula]{$\dMLg$-formula} $\Formula[1]$,
  we have to construct an \kl[class-formula]{$\FOL$-formula} $\Formula[2]_{\Formula[1]}$
  such that $\Structure \Models \Formula[1]$ if and only if $\Structure \Models \Formula[2]_{\Formula[1]}$,
  for every \kl{structure} $\Structure$.
  This is simply a matter of transcribing the semantics of $\dMLg$
  given in \cref{tab:syntax-semantics}
  to the language of \kl{first-order logic},
  a method known as the \emph{standard translation} in modal logic
  (see, e.g., \cite[Def.~2.45]{BlackburnRV02}).
  The following table gives a recursive specification of this translation. \\
  \begin{center}
    \begin{tabular}{ll}
      \toprule
      $\Formula[1]∈\dMLg$
        & \kl[device equivalent]{Equivalent} \kl{formula}\, $\Formula[2]_{\Formula[1]}\!∈\FOL$ \\
      \midrule
      $\PosIs{\NodeSymbol}$
        & $\PosSymbol \Equals \NodeSymbol$
        \\\addlinespace
      $\PosIn{\SetSymbol}$
        & $\InSet{\SetSymbol}{\PosSymbol}$
        \\\addlinespace
      $\NOT \Formula[1]_1$
        & $\NOT \Formula[2]_{\Formula[1]_1}$
        \\\addlinespace
      $\Formula[1]_1 \OR \Formula[1]_2$
        & $\Formula[2]_{\Formula[1]_1} \OR \Formula[2]_{\Formula[1]_2}$
        \\\addlinespace
      $\dm[\RelSymbol](\Formula[1]_1,…,\Formula[1]_{\Arity})$
        & $\Exists{\NodeSymbol_1,…,\NodeSymbol_{\Arity}}
           \bigl(\,
           \InRel{\RelSymbol}{\PosSymbol,\NodeSymbol_1,…,\NodeSymbol_{\Arity}} \;\AND{}$
           \hspace*{\fill} $\bigAND_{1≤i≤\Arity}\subst{\Formula[2]_{\Formula[1]_i}}{\PosSymbol}{\NodeSymbol_i}
           \,\bigr)$
        \\\addlinespace
      $\bdm[\RelSymbol](\Formula[1]_1,…,\Formula[1]_{\Arity})$
        & \mbox{as above, except $\InRel{\RelSymbol}{\NodeSymbol_{\Arity},…,\NodeSymbol_1,\PosSymbol}$}
        \\\addlinespace
      $\gdm \Formula[1]_1$
        & $\Exists{\PosSymbol}\+ \Formula[2]_{\Formula[1]_1}$
        \\\addlinespace
      \bottomrule
    \end{tabular}
  \end{center}
  \vspace{\belowdisplayskip}
  Here, $\NodeSymbol∈\NodeSymbolSet$,\;
  $\SetSymbol∈\SetSymbolSet$,\; $\RelSymbol∈\RelSymbolSet{\Arity+1}$,\;
  $\Formula[1]_1,…,\Formula[1]_{\Arity}∈\dMLg$, for $\Arity≥1$,
  and $\NodeSymbol_1,…,\NodeSymbol_{\Arity}$ are \kl{node symbols},
  chosen such that $\NodeSymbol_i ∉ \free(\Formula[2]_{\Formula[1]_i})$.
  The notation $\subst{\Formula[2]_{\Formula[1]_i}}{\PosSymbol}{\NodeSymbol_i}$ designates the formula obtained by
  substituting each \kl{free} occurrence of~$\PosSymbol$ in $\Formula[2]_{\Formula[1]_i}$ by $\NodeSymbol_i$.
\end{proof}

\subsection{A detour through tiling systems}

By restricting our focus to the class of \kl{labeled} \kl{grids},
we can take advantage of a well-studied automaton model
introduced by Giammarresi and Restivo
in~\cite{DBLP:journals/ijprai/GiammarresiR92},
which is closely related to $\MSOL$.
A “machine” in this model,
called a \Intro{tiling system},
is defined as a tuple $\TilingSystem = \tuple{\Alphabet,\TStateSet,\TTileSet}$,\, where
\begin{itemize}
\item $\Alphabet = \Boolean^{\BitCount}$ is seen as an alphabet, with $\BitCount≥0$,
\item $\TStateSet$ is a finite set of sates,\, and
\item $\TTileSet ⊆ \bigl((\Alphabet×\TStateSet)∪\set{\#}\bigr)^{4\mathstrut}$
  is a set of $2{×}2$-tiles
  that may use a fresh letter $\#$ not contained in $(\Alphabet×\TStateSet)$.
  \qedhere
\end{itemize}
For a fixed number of bits $\BitCount$,
we denote by \Intro*{$\TS_{\BitCount}$} the set of all \kl{tiling systems} with alphabet $\Alphabet = \Boolean^{\BitCount}$.

Given a $\BitCount$-bit \kl{labeled} \kl{grid} $\Grid$,
a \kl{tiling system} $\TilingSystem∈\TS_{\BitCount}$ operates similarly to a nondeterministic finite automaton
generalized to two dimensions.
A run of~$\TilingSystem$ on $\Grid$ is an extended \kl{labeled} \kl{grid} $\Grid^\#$\!,
obtained by nondeterministically labeling each cell of~$\Grid$
with some state $\TState∈\TStateSet$
and surrounding the entire \kl{grid} with a border
consisting of new $\#$-\kl{labeled} cells.
We consider $\Grid^\#$ to be a valid run
if each of its $2{×}2$-\kl[grids]{subgrids} can be identified with some tile in $\TTileSet$.
The set recognized by $\TilingSystem$
consists precisely of those \kl{labeled} \kl{grids}
for which such a run exists.
By analogy with our existing notation,
we write \Intro*{$\semT{\TS_{\BitCount}}[\GRID[\BitCount]]$}
for the class formed by
the sets of~$\BitCount$-bit \kl{labeled} \kl{grids} that are recognized by some \kl{tiling system} in $\TS_{\BitCount}$.

Exploiting a locality property of \kl{first-order logic},
Giammarresi, Restivo, Seibert and Thomas
have shown in \cite{DBLP:journals/iandc/GiammarresiRST96}
that \kl{tiling systems} capture precisely
the existential fragment of $\MSOL$ on \kl{labeled} \kl{grids}:

\begin{theorem}[Giammarresi, Restivo, Seibert, Thomas]
  \label{thm:equivalence-ts-emso}
  For arbitrary $\BitCount≥0$, a set of $\BitCount$-bit \kl{labeled} \kl{grids}
  is \mbox{$\TS$-recognizable} if and only if it is $\SigmaMSO{1}(\FOL)$-\kl{definable}
  over $\GRID[\BitCount]$, i.e.,
  \begin{equation*}
    \semT{\TS_{\BitCount}}[\GRID[\BitCount]] \;=\; \semF{\SigmaMSO{1}(\FOL)}[\GRID[\BitCount]] \+.
    \qedhere
  \end{equation*}
\end{theorem}

The preceding result is extremely useful for our purposes,
because, from the perspective of \kl{modal logic},
it provides a much easier access to $\MSOL$.
This brings us to the following \lcnamecref{prp:inclusion-ts-ehgs}.

\begin{proposition}
  \label{prp:inclusion-ts-ehgs}
  For arbitrary $\BitCount≥0$, if a set of $\BitCount$-bit \kl{labeled} \kl{grids}
  is \mbox{$\TS$-recognizable}, then it is also $\SigmaMSO{1}(\MLg)$-\kl{definable}
  over $\GRID[\BitCount]$, i.e.,
  \begin{equation*}
    \semT{\TS_{\BitCount}}[\GRID[\BitCount]] \;⊆\; \semF{\SigmaMSO{1}(\MLg)}[\GRID[\BitCount]] \+.
    \qedhere
  \end{equation*}
\end{proposition}
\begin{proof}
  \newcommand*{\Middle}{{\Acronym{m}}}
  \newcommand*{\Top}{{\Acronym{t}}}
  \newcommand*{\Bottom}{{\Acronym{b}}}
  \newcommand*{\Left}{{\Acronym{l}}}
  \newcommand*{\Right}{{\Acronym{r}}}
  \newcommand*{\border}{{\operatorname{border}}}
  Let $\TilingSystem = \tuple{\Alphabet,\TStateSet,\TTileSet}$ be a \kl{tiling system}
  with alphabet $\Alphabet = \Boolean^{\BitCount}$.
  We have to construct a $\SigmaMSO{1}(\MLg)$-\kl{sentence} $\Formula[1]_\TilingSystem$
  over the \kl{signature} $\set{\SetConstant_1,…,\SetConstant_{\BitCount},\RelSymbol_1,\RelSymbol_2}$,
  such that each \kl{labeled} \kl{grid} $\Grid∈\GRID[\BitCount]$ \kl{satisfies} $\Formula[1]_\TilingSystem$
  if and only if it is accepted by $\TilingSystem$.

  The idea is standard:
  We represent the states of~$\TilingSystem$ by additional \kl{set symbols} $\tuple{\SetVariable[1]_\TState}_{\TState∈\TStateSet}$,
  and our \kl{formula} asserts that there exists a corresponding partition of~$\NodeSet{\Grid}$
  into $\card{\TStateSet}$ subsets
  that represent a run $\Grid^\#$ of~$\TilingSystem$ on~$\Grid$.
  To verify that it is indeed a valid run,
  we have to check that each $2{×}2$-\kl[grid]{subgrid} of~$\Grid^\#$
  corresponds to some tile
  \begin{equation*}
    \TTile =
    \begin{bmatrix*}[l]
      \TTile_1 & \TTile_2 \\
      \TTile_3 & \TTile_4
    \end{bmatrix*}
  \end{equation*}
  in $\TTileSet$.
  If the entry $\TTile_1$ is different from $\#$,
  we can easily write down an \kl[class-formula]{$\ML$-formula} $\Formula[1]_\TTile$
  that checks at a given position $\Node∈\NodeSet{\Grid}$,
  whether the $2{×}2$-\kl[grid]{subgrid} of~$\Grid^\#$ with upper-left corner $\Node$
  matches $\TTile$.
  Here, $\TTile_1$ is chosen as the representative entry of~$\TTile$,
  because the upper-left corner of the tile can “see” the other \kl{nodes}
  by following the directed $\RelSymbol_1$-{} and $\RelSymbol_2$-\kl{edges}.
  Otherwise, if $\TTile_1$ is equal to $\#$, there is no such \kl{node}~$\Node$,
  since $\Grid$ does not contain special border \kl{nodes}.
  However, we can always choose some other entry $\TTile_i$, different from $\#$,
  to be the representative of~$\TTile$,
  and write a \kl{formula} $\Formula[1]_\TTile$ describing the tile
  from the point of view of a \kl{node} corresponding to $\TTile_i$.
  This choice is never arbitrary,
  because the representative must be able to “see”
  the other \mbox{non-$\#$} entries of the tile.
  Consequently,
  we divide $\TTileSet$ into four disjoint sets $\TTileSet_1$, $\TTileSet_2$, $\TTileSet_3$, $\TTileSet_4$,
  such that $\TTileSet_i$ contains those tiles $\TTile$ that are represented by their entry $\TTile_i$.
  In order to facilitate the subsequent formalization,
  we further subdivide each set into partitions
  according to the $\#$-borders that occur within the tiles:
  $\TTileSet_\Middle$ contains the “middle tiles” (all entries different from~$\#$),
  $\TTileSet_\Left$ the “left tiles” (with $\TTile_1$ and $\TTile_3$ equal to~$\#$),
  $\TTileSet_{\Bottom\Right}$ the “bottom-right tiles”,
  and so forth~…
  Altogether,
  $\TTileSet$ is partitioned into nine subsets, grouped into four types:
  \begin{alignat*}{2}
    \TTileSet_1 &= \TTileSet_\Middle \dcup \TTileSet_\Bottom \dcup \TTileSet_\Right \dcup \TTileSet_{\Bottom\Right} \qquad& \TTileSet_2 &= \TTileSet_\Left \dcup \TTileSet_{\Bottom\Left} \\
    \TTileSet_3 &= \TTileSet_\Top \dcup \TTileSet_{\Top\Right}                                       & \TTileSet_4 &= \TTileSet_{\Top\Left}
  \end{alignat*}

  We now construct the \kl{formula} $\Formula[1]_\TilingSystem$ in a bottom-up manner,
  starting with a \kl[formula]{subformula} $\Formula[1]_{\TTile_i}$
  for each entry $\TTile_i$ other than $\#$, for every tile $\TTile∈\TTileSet$.
  Letting $\TTile_i$ be equal to $\tuple{\Letter,\TState}∈\Alphabet×\TStateSet$, with $\Letter = \Letter_1\dots\Letter_\BitCount$,
  the \kl{formula} $\Formula[1]_{\TTile_i}$ checks at a given position $\Node∈\NodeSet{\Grid}$
  if the \kl{labeling} of~$\Node$ matches~$\TTile_i$.
  \begin{equation*}
    \Formula[1]_{\TTile_i} \+=\; \smashoperator{\bigAND_{\Letter_j=1}}\PosIn{\SetConstant_j} \,\AND\,
                 \smashoperator{\bigAND_{\Letter_j=0}}\NOT \PosIn{\SetConstant_j} \,\AND\,
                 \PosIn{\SetVariable[1]_\TState} \,\AND\,
                 \smashoperator{\bigAND_{\TState'≠\TState}}\NOT \PosIn{\SetVariable[1]_{\TState'}}
  \end{equation*}

  Building on this, we can define for each tile $\TTile∈\TTileSet$
  the \kl{formula} $\Formula[1]_\TTile$ mentioned above.
  Since $\MLg$ does not have \kl{backward modalities},
  there is a certain asymmetry between tiles in $\TTileSet_1$,
  where the representative can “see” the entire $2{×}2$-\kl[grid]{subgrid},
  and the remaining tiles,
  where the representative must “know”
  that it lies in the leftmost column or the uppermost row of the \kl{grid} $\Grid$.
  We shall address this issue shortly,
  and just assume that information not accessible to the representative
  is verified by another part of the ultimate \kl{formula} $\Formula[1]_\TilingSystem$.
  For tiles in $\TTileSet_\Middle$, $\TTileSet_{\Bottom\Right}$, $\TTileSet_\Left$, $\TTileSet_{\Top\Left}$,
  the definitions of~$\Formula[1]_\TTile$ are given in the following table.
  For tiles in $\TTileSet_\Bottom$, $\TTileSet_\Right$, $\TTileSet_{\Bottom\Left}$, $\TTileSet_\Top$, $\TTileSet_{\Top\Right}$,
  the method is completely analogous. \\
  \begin{center}
    \begin{tabular}{rl}
      \toprule
      $\swl{\;\TTile}{
        \bigl[ \begin{smallmatrix*}[l]
          \TTile_1 & \TTile_2 \\
          \TTile_3 & \TTile_4
        \end{smallmatrix*} \bigr]
      }$
      & $\Formula[1]_\TTile$ \\
      \midrule\addlinespace
      $\swl{\TTileSet_\Middle}{\TTileSet_{\Bottom\Right}} ∋
        \bigl[ \begin{smallmatrix*}[l]
          \TTile_1 & \TTile_2 \\
          \TTile_3 & \TTile_4
        \end{smallmatrix*} \bigr]$
      & $\Formula[1]_{\TTile_1} {\+\AND\,} \dm[2]\Formula[1]_{\TTile_2} {\+\AND\,} \dm[1]\Formula[1]_{\TTile_3} {\+\AND\,} \dm[1]\dm[2]\Formula[1]_{\TTile_4}$ \\\addlinespace
      $\TTileSet_{\Bottom\Right} ∋
        \bigl[ \begin{smallmatrix*}[l]
          \TTile_1 & \# \\
          \# & \swl{\scriptstyle\#}{\TTile_4}
        \end{smallmatrix*} \bigr]$
      & $\Formula[1]_{\TTile_1} \AND\, \bx[2]\False \+\AND\, \bx[1]\False$ \\\addlinespace
      $\swl{\TTileSet_\Left}{\TTileSet_{\Bottom\Right}} ∋
        \bigl[ \begin{smallmatrix*}[l]
          \# & \TTile_2 \\
          \swl{\scriptstyle\#}{\TTile_3} & \TTile_4
        \end{smallmatrix*} \bigr]$
      & $\Formula[1]_{\TTile_2} \AND \dm[1]\Formula[1]_{\TTile_4}$ \\\addlinespace
      $\TTileSet_{\Top\Left} ∋
        \bigl[ \begin{smallmatrix*}[l]
          \# & \# \\
          \swl{\scriptstyle\#}{\TTile_3} & \TTile_4
        \end{smallmatrix*} \bigr]$
      & $\Formula[1]_{\TTile_4}$ \\\addlinespace
      \bottomrule
    \end{tabular}
  \end{center}
  \vspace{\belowdisplayskip}

  It remains to mark the top and left borders of~$\Grid$,
  using two additional predicates $\SetVariable[2]_\Top$ and $\SetVariable[2]_\Left$,
  over which we will \kl[set quantifier]{quantify} existentially.
  To this end,
  we write an \kl[class-formula]{$\MLg$-formula} $\Formula[1]_\border$, checking
  that top [resp.~left] \kl{nodes}
  have no $\RelSymbol_1$-{} [resp.~$\RelSymbol_2$-] predecessor,
  that there is a top-left \kl{node},
  and that being top [resp.~left] is passed on to the $\RelSymbol_2$-{} [resp.~$\RelSymbol_1$-] successor,
  if it exists.
  \begin{equation*}
    \begin{aligned}
      \Formula[1]_\border \+=\; \NOT&\gdm(\dm[1]\PosIn{\SetVariable[2]_\Top}\OR\dm[2]\PosIn{\SetVariable[2]_\Left}) \;\AND\; \gdm(\PosIn{\SetVariable[2]_\Top}\AND \PosIn{\SetVariable[2]_\Left}) \;\AND{} \\
                     &\+\+\gbx\Bigl((\PosIn{\SetVariable[2]_\Top}{\,\IMP\,}\bx[2]\PosIn{\SetVariable[2]_\Top}) \AND (\PosIn{\SetVariable[2]_\Left}{\,\IMP\,}\bx[1]\PosIn{\SetVariable[2]_\Left})\Bigr)
    \end{aligned}
  \end{equation*}

  Finally,
  we can put everything together
  to describe the acceptance condition of~$\TilingSystem$.
  Every \kl{node} $\Node∈\NodeSet{\Grid}$ has to ensure
  that it corresponds to the upper-left corner of some tile in $\TTileSet_1$.
  Furthermore, \kl{nodes} in the leftmost column or uppermost row of~$\Grid$
  must additionally check that the assignment of states is compatible
  with the tiles in $\TTileSet_2$, $\TTileSet_3$, $\TTileSet_4$.
  This leads to the desired \kl{formula}~$\Formula[1]_\TilingSystem$:
  \begin{align*}
    \Exists{\tuple{\SetVariable[1]_\TState}_{\TState∈\TStateSet},\,\SetVariable[2]_\Top,\,\SetVariable[2]_\Left}\biggl(\:
    &\Formula[1]_\border \;\AND{} \\
    &\,\gbx\bigl(\,\smashoperator{\bigOR_{\TTile ∈ \TTileSet_1}}\!\Formula[1]_\TTile\bigr)
     \;\AND\; \gbx\bigl(\PosIn{\SetVariable[2]_\Left}\IMP\smashoperator{\bigOR_{\TTile ∈ \TTileSet_2}}\!\Formula[1]_\TTile\bigr) \;\AND{} \\[-1ex]
    &\,\gbx\bigl(\PosIn{\SetVariable[2]_\Top}\IMP\smashoperator{\bigOR_{\TTile ∈ \TTileSet_3}}\!\Formula[1]_\TTile\bigr)
     \;\AND\; \gbx\bigl(\PosIn{\SetVariable[2]_\Top} \AND \PosIn{\SetVariable[2]_\Left}\IMP\smashoperator{\bigOR_{\TTile ∈ \TTileSet_4}}\!\Formula[1]_\TTile\bigr)\:\biggr)
  \end{align*}
  Note that we do not need a separate \kl[formula]{subformula} to check that the
  \kl{interpretations} of $\tuple{\SetVariable[1]_\TState}_{\TState∈\TStateSet}$ form a partition of~$\NodeSet{\Grid}$,
  since this is already done implicitly in the conjunct $\gbx(\bigOR_{\!\TTile ∈ \TTileSet_1}\!\Formula[1]_\TTile)$.
\end{proof}

\subsection{Equivalent hierarchies on grids}

We now have all we need
to prove the levelwise \kl[device equivalence]{equivalence} of
$\MSOL$, $\MSO(\dMLg)$ and $\MSO(\MLg)$ on \kl{labeled} \kl{grids}.

\pagebreak[3]
\begin{theorem}
  \label{thm:levelwise-equivalence}
  Let $\BitCount≥0$,\: $\Level≥1$ and
  $\FormulaClassOperator ∈ \set{\SigmaMSO{\Level},\,\PiMSO{\Level},\,\BC\SigmaMSO{\Level},\,\DeltaMSO{\Level}}$.
  When restricted to the class of $\BitCount$-bit \kl{labeled} \kl{grids},
  $\FormulaClassOperator\+(\FOL)$, $\FormulaClassOperator\+(\dMLg)$ and $\FormulaClassOperator\+(\MLg)$
  are \kl[device equivalent]{equivalent},\, i.e.,
  \begin{align*}
    \semF{\+\FormulaClassOperator\+(\FOL)}[{\GRID[\BitCount]}] \:&=\; \semF{\+\FormulaClassOperator\+(\dMLg)}[{\GRID[\BitCount]}] \\
                                   &=\; \semF{\+\FormulaClassOperator\+(\MLg)}[{\GRID[\BitCount]}] \+.
    \qedhere
  \end{align*}
\end{theorem}
\begin{proof}
  First,
  we show that the claim holds for the case $\FormulaClassOperator = \SigmaMSO{1}$
  (with arbitrary $\BitCount≥0$).
  This can be seen from the following circular chain of inclusions:
  \begin{align}
    \semF{\SigmaMSO{1}(\MLg)}[{\GRID[\BitCount]}] \:&⊆\; \semF{\SigmaMSO{1}(\dMLg)}[{\GRID[\BitCount]}] \label{eq:ehgs-ehbgs} \tag{a} \\
                                    &⊆\; \semF{\SigmaMSO{1}(\FOL)}[{\GRID[\BitCount]}] \label{eq:ehbgs-emso} \tag{b} \\
                                    &⊆\; \semT{\TS_{\BitCount}}[\GRID[\BitCount]] \label{eq:emso-ts} \tag{c} \\
                                    &⊆\; \semF{\SigmaMSO{1}(\MLg)}[{\GRID[\BitCount]}] \label{eq:ts-ehgs} \tag{d}
  \end{align}
  \begin{enumerate}
  \item[(\ref{eq:ehgs-ehbgs})]
    The first inclusion follows from the fact that
    $\SigmaMSO{1}(\MLg)$ is a syntactic fragment of $\SigmaMSO{1}(\dMLg)$.
  \item[(\ref{eq:ehbgs-emso})]
    For the second inclusion,
    consider any \kl[class-formula]{$\SigmaMSO{1}(\dMLg)$-formula}
    $\widehat{\Formula[1]} = \Exists{\SetVariable_1,…,\SetVariable_n}(\Formula[1])$,
    where $\SetVariable_1,…,\SetVariable_n$ are \kl{set symbols} and $\Formula[1]$ is
    an \kl[class-formula]{$\dMLg$-formula}.
    By \cref{prp:standard-translation},
    we can replace $\Formula[1]$ in $\widehat{\Formula[1]}$ by an \kl[device equivalent]{equivalent} $\FOL$-formula $\Formula[2]_{\Formula[1]}$.
    This results in the \kl[class-formula]{$\SigmaMSO{1}(\FOL)$-formula}
    $\Formula[2]_{\widehat{\Formula[1]}} = \Exists{\SetVariable_1,…,\SetVariable_n}(\Formula[2]_{\Formula[1]})$,
    which is \kl[device equivalent]{equivalent} to $\widehat{\Formula[1]}$ on arbitrary \kl{structures},
    and thus, in particular, on $\BitCount$-bit \kl{labeled} \kl{grids}.
  \item[(\ref{eq:emso-ts})]
    The translation from $\SigmaMSO{1}(\FOL)$ on \kl{labeled} \kl{grids} to \kl{tiling systems}
    corresponds to the more challenging direction of \cref{thm:equivalence-ts-emso},
    which is the main result of \cite{DBLP:journals/iandc/GiammarresiRST96}.
  \item[(\ref{eq:ts-ehgs})]
    The last inclusion is given by \cref{prp:inclusion-ts-ehgs}.
  \end{enumerate}

  The general version of the \lcnamecref{thm:levelwise-equivalence}
  can now be obtained by induction on $\Level$.
  This is straightforward,
  because the classes
  $\PiMSO{\Level}(\FormulaSet[1])$,\, $\BC\SigmaMSO{\Level}(\FormulaSet[1])$ and $\SigmaMSO{\Level+1}(\FormulaSet[1])$
  are defined syntactically in terms of $\SigmaMSO{\Level}(\FormulaSet[1])$,
  for any set of \kl{kernel} \kl{formulas} $\FormulaSet[1]$
  (see \cref{sec:alternation-preliminaries}),
  and if the claim holds for $\FormulaClassOperator ∈ \set{\SigmaMSO{\Level},\,\PiMSO{\Level}}$,
  then it also holds for the intersection classes of the form $\semF{\DeltaMSO{\Level}(\FormulaSet[1])}[{\GRID[\BitCount]}]$.
\end{proof}

\subsection{A logical characterization of grids}

We conclude this \lcnamecref{sec:grids} by showing
that a single layer of universal \kl{set quantifiers} is enough
to describe \kl{grids} in $\MSO(\dMLg)$.

\begin{proposition}
  \label{prp:grid-definability}
  The set of all \kl{grids} is $\PiMSO{1}(\dMLg)$-\kl{definable}
  over \mbox{$2$-relational} \kl{digraphs}, i.e.,
  \begin{equation*}
    \GRID \;∈\; \semF{\PiMSO{1}(\dMLg)}[{\DIGRAPH[0][2]}] \+.
    \qedhere
  \end{equation*}
\end{proposition}
\begin{proof}
  \newcommand*{\Source}{{\operatorname{src}}}
  In the course of this proof,
  we give a list of properties,
  items~\labelcref{itm:injective-function,itm:reach-sink,itm:unique-source,itm:horiz-sources-linked,itm:down-right-diagonal,itm:down-right-commute},
  which are obviously necessary
  for a $2$-relational \kl{digraph} $\Digraph$ to be a \kl{grid},
  and show how to express them as \kl[class-formula]{$\eclF{\PiMSO{1}(\dMLg)}$-formulas}.
  We argue that the conjunction of all of these properties
  also constitutes a sufficient condition for being a \kl{grid},
  which immediately provides us with the required \kl{formula},
  since $\semF{\PiMSO{1}(\dMLg)}$ is closed under intersection.
  \begin{enumerate}
  \item \label{itm:injective-function}
    For each \kl{relation symbol} $\RelSymbol ∈ \set{\RelSymbol_1,\RelSymbol_2}$,
    every \kl{node} has at most one $\RelSymbol$-predecessor and at most one $\RelSymbol$-successor;
    in other words,
    $\inp{\RelSymbol_1}{\Digraph}$ and $\inp{\RelSymbol_2}{\Digraph}$ are partial injective functions.
    \begin{equation*}
      \smashoperator[r]{\bigAND_{\RelSymbol\,∈\,\set{\RelSymbol_1,\,\invr{\RelSymbol_1}\!,\,\RelSymbol_2,\,\invr{\RelSymbol_2}}}} \:
      \Forall{\SetVariable}\+\gbx \bigl( \dm[\RelSymbol]\PosIn{\SetVariable} \IMP \bx[\RelSymbol]\PosIn{\SetVariable} \bigr)
    \end{equation*}
  \item \label{itm:reach-sink}
    Again considering each $\RelSymbol ∈ \set{\RelSymbol_1,\RelSymbol_2}$ separately,
    there is a directed $\RelSymbol$-path from every \kl{node} to an $\RelSymbol$-\kl{sink},
    i.e., to some \kl{node} without $\RelSymbol$-successor.
    \begin{equation*}
      \smashoperator[r]{\bigAND_{\RelSymbol\,∈\,\set{\RelSymbol_1,\,\RelSymbol_2}}} \,
      \Forall{\SetVariable} \Bigl(
      \gdm \PosIn{\SetVariable} {\,\AND\,} \gbx\bigl( \PosIn{\SetVariable} {\,\IMP\,} \bx[\RelSymbol]\PosIn{\SetVariable} \bigr)
      \IMP \gdm\bigl( \PosIn{\SetVariable} {\,\AND\,} \bx[\RelSymbol]\False \bigr)
      \Bigr)
    \end{equation*}
  \end{enumerate}
  Taken together,
  properties~\labelcref{itm:injective-function,itm:reach-sink}
  state that $\inp{\RelSymbol_1}{\Digraph}$ and $\inp{\RelSymbol_2}{\Digraph}$
  each form a collection of directed, acyclic, pairwise vertex-disjoint paths.
  Let us refer to the first \kl{nodes} of those paths
  as $\RelSymbol_1$-{} and $\RelSymbol_2$-\kl{sources}, respectively.
  \begin{enumerate}[resume*]
  \item \label{itm:unique-source}
    There is precisely one \kl{node} that is both an $\RelSymbol_1$-{} and an $\RelSymbol_2$-\kl{source}.
    \begin{equation*}
      \totone\bigl(\+ \bbx[1]\False \AND \bbx[2]\False \bigr)
    \end{equation*}
    (Here, $\totone$ is the schema from
    \cref{ex:uniqueness} in \cref{sec:example-formulas}.)
  \item \label{itm:horiz-sources-linked}
    The $\RelSymbol_1$-predecessors and $\RelSymbol_1$-successors of $\RelSymbol_2$-\kl{sources}
    must be $\RelSymbol_2$-\kl{sources} themselves.
    \begin{equation*}
      \gbx \bigl(\+ \bbx[2]\False \+\IMP\+ \bbx[1]\bbx[2]\False \AND \bx[1]\bbx[2]\False \bigr)
    \end{equation*}
  \end{enumerate}
  By adding~\labelcref{itm:unique-source,itm:horiz-sources-linked}
  to our list of conditions,
  we ensure that there is an $\RelSymbol_1$-path consisting precisely of the $\RelSymbol_2$-\kl{sources},
  thereby also forcing the \kl{digraph} $\Digraph$ to be \kl{connected}.
  \begin{enumerate}[resume*]
  \item \label{itm:down-right-diagonal}
    If a \kl{node} has both an $\RelSymbol_1$- and an $\RelSymbol_2$-successor,
    then it also has a descendant
    reachable by first taking an $\RelSymbol_1$-\kl{edge} and then an $\RelSymbol_2$-\kl{edge}.
    \begin{equation*}
      \gbx \bigl( \dm[1]\True \AND \dm[2]\True \+\IMP\+ \dm[1]\dm[2]\True \bigr)
    \end{equation*}
  \item \label{itm:down-right-commute}
    The relations $\inp{\RelSymbol_1}{\Digraph}$ and $\inp{\RelSymbol_2}{\Digraph}$ commute.
    This means that following an $\RelSymbol_1$-\kl{edge} and then an $\RelSymbol_2$-\kl{edge}
    leads to the same \kl{node} as
    first taking an $\RelSymbol_2$-\kl{edge} and then an $\RelSymbol_1$-\kl{edge}.
    \begin{equation*}
      \Forall{\SetVariable}\+\gbx \bigl( \dm[1]\dm[2]\PosIn{\SetVariable} \IFF \dm[2]\dm[1]\PosIn{\SetVariable} \bigr)
    \end{equation*}
  \end{enumerate}
  Considered in conjunction with
  condition~\labelcref{itm:injective-function},
  there are only two ways to
  satisfy~\labelcref{itm:down-right-diagonal,itm:down-right-commute}
  from the point of view of
  two \kl{nodes} $\Node[1],\Node[2]∈\NodeSet{\Digraph}$ that are connected by an $\RelSymbol_1$-\kl{edge}
  from $\Node[1]$ to $\Node[2]$:
  either both \kl{nodes} are $\RelSymbol_2$-\kl{sinks},
  or they have $\RelSymbol_2$-successors $\Node[1]'$ and $\Node[2]'$, respectively,
  with an $\RelSymbol_1$-\kl{edge} from $\Node[1]'$ to $\Node[2]'$.
  Moreover,
  $\Node[2]'$ only possesses an $\RelSymbol_1$-successor if $\Node[2]$ does.
  Now, imagine we start from the left border,
  i.e., from the $\RelSymbol_1$-path that consists of all the $\RelSymbol_2$-\kl{sources},
  which is provided by
  properties~\labelcref{itm:injective-function,itm:reach-sink,itm:unique-source,itm:horiz-sources-linked},
  and iteratively enforce the requirements just mentioned.
  Then, in doing so, we propagate the grid topology through the entire \kl{digraph}.
  More specifically,
  the additional requirements of~\labelcref{itm:down-right-diagonal,itm:down-right-commute}
  entail that all the $\RelSymbol_2$-paths have the same length,
  and that the \kl{nodes} lying at a fixed (horizontal) position of those $\RelSymbol_2$-paths
  constitute an independent $\RelSymbol_1$-path,
  ordered in the same way as their respective $\RelSymbol_2$-predecessors.
\end{proof}

\section{A toolbox of encodings}
\label{sec:encodings}

In this \lcnamecref{sec:encodings},
we provide all the encoding functions
used in the proof of \cref{thm:separation-R} (see \cref{ssec:main-proof}),
and show that they satisfy suitable translatability properties,
allowing us to establish the required \kl{figurative inclusions}.
With a view to modularity and reusability,
some of our constructions are more general than needed.

Given a set of \kl{symbols} $\Signature$,
the extension $\Signature ∪ \set{\PosSymbol}$
will be abbreviated to~\Intro*{$\enrichedsig{\Signature}{\PosSymbol}$}.

\subsection{Encodings that allow for translation}

We shall only consider encoding functions
that are linear in the following sense:

\begin{definition}[Linear Encoding]
  \label{def:linear-encoding}
  \NoMathBreak
  Let $\StructClass[1]$, $\StructClass[2]$ be two classes of \kl{structures},
  and $\EncodingScale$, $\EncodingOffset$ be integers
  such that $1≤\EncodingScale≤\EncodingOffset$.
  A \Intro{linear encoding} from $\StructClass[1]$
  into $\StructClass[2]$
  with parameters $\EncodingScale$, $\EncodingOffset$
  is a total injective function
  $\Encoding\colon\StructClass[1]\to\StructClass[2]$
  that assigns to each \kl{structure} $\Structure[1]∈\StructClass[1]$
  a \kl{structure} $\Encoding(\Structure[1])∈\StructClass[2]$,
  whose \kl{domain} is composed of~$\EncodingScale$ disjoint copies
  of the \kl{domain} of $\Structure[1]$
  and $\EncodingOffset-\EncodingScale$ additional \kl{nodes}, i.e.,
  \begin{equation*}
    \NodeSet{\Encoding(\Structure[1])} =
    \bigl(\range[1]{\EncodingScale}×\NodeSet{\Structure[1]}\bigr) \;∪\;
    \lrange[\EncodingScale]{\EncodingOffset}.
    \qedhere
  \end{equation*}
\end{definition}

Given such a \kl{linear encoding} $\Encoding$
and some \kl[class-formula]{$\dMLg$-formula} $\Formula[1]$,
we want to be able to construct a new \kl{formula} $\Formula[2]_{\Formula[1]}$,
such that
evaluating
$\Formula[1]$ on $\StructClass[1]$
is equivalent to evaluating
$\Formula[2]_{\Formula[1]}$ on $\Encoding(\StructClass[1])$.
Conversely,
we also desire a way of constructing a \kl{formula} $\Formula[1]_{\Formula[2]}$
that is equivalent on $\StructClass[1]$
to a given \kl{formula} $\Formula[2]$ on $\Encoding(\StructClass[1])$.
The following two \lcnamecrefs{def:forward-translation}
formalize this translatability property for both directions.
We then show in \cref{lem:translation-inclusion}
that they adequately capture our intended meaning.
Although the underlying idea is very simple,
the presentation is a bit lengthy
because we have to exhaustively cover the structure
of \kl[class-formula]{$\dMLg$-formulas}.

\begin{definition}[Forward Translation]
  \label{def:forward-translation}
  \NoMathBreak
  Consider two classes of \kl{structures} $\StructClass[1]$ and $\StructClass[2]$
  over \kl{signatures} $\Signature[1]$ and $\Signature[2]$, respectively,
  two classes of \kl{formulas}
  $\FormulaSet[1],\FormulaSet[2]∈\set{\ML,\,\dML,\,\MLg,\,\dMLg}$,
  and a \kl{linear encoding}
  $\Encoding\colon\StructClass[1]\to\StructClass[2]$.
  We say that $\Encoding$
  allows for \Intro{forward translation} from $\FormulaSet[1]$ to $\FormulaSet[2]$
  if the following properties are satisfied:
  \begin{enumerate}
  \item \label{itm:fwd-set}
    For each \kl{node symbol} or \kl{set symbol} $\SetConstant[1]$
    in $\Signature[1]$,
    there is a $\FormulaSet[2]$-\kl{sentence} $\Formula[2]_{\SetConstant[1]}$
    over $\enrichedsig{\Signature[2]}{\PosSymbol}$,
    such that
    \begin{equation*}
      \ver{\Structure[1]}{\PosSymbol}{\Node[1]}
      \,\Models\, \PosIn{\SetConstant[1]}
      \quad\text{iff}\quad
      \ver{\Encoding(\Structure[1])}{\PosSymbol}{\tuple{1,\Node[1]}}
      \,\Models\,
      \Formula[2]_{\SetConstant[1]}\+,
    \end{equation*}
    for all $\Structure[1]∈\StructClass[1]$
    and $\Node[1]∈\NodeSet{\Structure[1]}$.
  \item \label{itm:fwd-relation}
    For each \kl{relation symbol} $\RelSymbol[1]$ in $\Signature[1]$
    of \kl{arity} $\Arity+1≥2$,
    there is a $\FormulaSet[2]$-\kl{sentence} $\Formula[2]_{\RelSymbol[1]}$
    over $\enrichedsig{\Signature[2]}{\PosSymbol}$
    enriched with additional \kl{set symbols} $\tuple{\SetVariable[2]_i}_{1≤i≤\Arity}$,
    such that
    \begin{align*}
      &\bigver{\Structure[1]}
              {\PosSymbol,\tuple{\SetVariable[1]_i}_{i≤\Arity}}
              {\Node[1],\tuple{\NodeSubset[1]_i}_{i≤\Arity}}
       \,\Models\,
       \dm[\RelSymbol[1]]\tuple{\PosIn{\SetVariable[1]_i}}_{i≤\Arity} \\
      &\quad\text{if and only if} \\
      &\bigver{\Encoding(\Structure[1])}
              {\PosSymbol,\tuple{\SetVariable[2]_i}_{i≤\Arity}}
              {\tuple{1,\Node[1]},\tuple{\NodeSubset[2]_i}_{i≤\Arity}}
       \,\Models\,
       \Formula[2]_{\RelSymbol[1]}\+, \\[0.5ex]
      &\quad
       \text{
         assuming
         $\NodeSubset[1]_i,\NodeSubset[2]_i$
         satisfy
         $\Node[1]'{∈\,}\NodeSubset[1]_i⇔\tuple{1,\Node[1]'}{\,∈\,}\NodeSubset[2]_i$,
       }
    \end{align*}
    for all $\Structure[1]∈\StructClass[1]$,\;
    $\Node[1]∈\NodeSet{\Structure[1]}$,
    sets $\tuple{\NodeSubset[1]_i}_{1≤i≤\Arity}⊆\NodeSet{\Structure[1]}$
    and $\tuple{\NodeSubset[2]_i}_{1≤i≤\Arity}⊆\NodeSet{\Encoding(\Structure[1])}$,
    and \kl{set symbols} $\tuple{\SetVariable[1]_i}_{1≤i≤\Arity}$.
  \item \label{itm:fwd-backward}
    If $\FormulaSet[1]$ includes \kl{backward modalities},
    then for each \kl{relation symbol} $\RelSymbol[1]$ in $\Signature[1]$
    of \kl{arity} at least $2$,
    there is a \kl[class-formula]{$\FormulaSet[2]$-formula} $\Formula[2]_{\invr{\RelSymbol[1]}}$
    that satisfies
    the property of item~\ref{itm:fwd-relation}
    for $\invr{\RelSymbol[1]}$ instead of $\RelSymbol[1]$.
  \item \label{itm:fwd-global}
    If $\FormulaSet[1]$ includes \kl{global modalities},
    then there is a \kl[class-formula]{$\FormulaSet[2]$-formula} $\Formula[2]_{\GlobalRelSymbol}$
    that satisfies
    the property of item~\ref{itm:fwd-relation}
    for~$\GlobalRelSymbol$ instead of $\RelSymbol[1]$ and $\Arity=1$.
  \item \label{itm:fwd-initial}
    \newcommand*{\XgdmX}{%
      \NoHeight{\scalebox{0.8}{%
          $\displaystyle\frac{\PosIn{\SetVariable[1]}}{\gdm \PosIn{\SetVariable[1]} \vphantom{\big(}}$%
      }}%
    }
    There is a $\FormulaSet[2]$-\kl{sentence} $\Formula[2]_\ini$ over $\Signature[2]$
    enriched with an additional \kl{set symbol} $\SetVariable[2]$, such that
    \begin{align*}
      &\ver{\Structure[1]}{\SetVariable[1]}{\NodeSubset[1]}
       \,\Models\,
       \XgdmX
       \quad\text{iff}\quad
       \ver{\Encoding(\Structure[1])}{\SetVariable[2]}{\NodeSubset[2]}
       \,\Models\,
       \Formula[2]_\ini, \\[1ex]
      &\quad\text{
        assuming
        $\NodeSubset[1],\NodeSubset[2]$
        satisfy
        $\Node[1]∈\NodeSubset[1]⇔\tuple{1,\Node[1]}∈\NodeSubset[2]$,
       } \\
      &\quad\text{
        where
        $\XgdmX$ is \,$\PosIn{\SetVariable[1]}$\,
        if $\PosSymbol∈\Signature[1]$,\,
        and $\gdm \PosIn{\SetVariable[1]}$ otherwise,
       }
    \end{align*}
    for all $\Structure[1]∈\StructClass[1]$,\,
    $\NodeSubset[1]⊆\NodeSet{\Structure[1]}$,\,
    $\NodeSubset[2]⊆\NodeSet{\Encoding(\Structure[1])}$
    and $\SetVariable[1]∈\RelSymbolSet{1}$.
    \qedhere
  \end{enumerate}
\end{definition}

\begin{definition}[Backward Translation]
  \label{def:backward-translation}
  \NoMathBreak
  Consider two classes of \kl{structures} $\StructClass[1]$ and $\StructClass[2]$
  over \kl{signatures} $\Signature[1]$ and $\Signature[2]$, respectively,
  two classes of \kl{formulas} $\FormulaSet[1],\FormulaSet[2]∈\set{\ML,\,\dML,\,\MLg,\,\dMLg}$,
  and a \kl{linear encoding} $\Encoding\colon\StructClass[1]\to\StructClass[2]$
  with parameters $\EncodingScale$, $\EncodingOffset$.
  We say that $\Encoding$
  allows for \Intro{backward translation} from $\FormulaSet[2]$ to $\FormulaSet[1]$
  if the following properties are satisfied:
  \begin{enumerate}
  \item \label{itm:bwd-set}
    For each \kl{node symbol} or \kl{set symbol} $\SetConstant[2]$ in $\Signature[2]$ and all $h∈\range{\EncodingOffset}$,
    there is a $\FormulaSet[1]$-\kl{sentence} $\Formula[1]_{\SetConstant[2]}^h$ over $\enrichedsig{\Signature[1]}{\PosSymbol}$, such that
    \begin{align*}
      &\ver{\Structure[1]}{\PosSymbol}{\Node[1]} \,\Models\, \Formula[1]_{\SetConstant[2]}^h \quad\text{iff}\quad \ver{\Encoding(\Structure[1])}{\PosSymbol}{\Node[2]} \,\Models\, \PosIn{\SetConstant[2]}, \\[0.5ex]
      &\quad\text{where $\Node[2]$ is $\tuple{h,\Node[1]}$ if $h≤\EncodingScale$, and $h$ otherwise,}
    \end{align*}
    for all $\Structure[1]∈\StructClass[1]$ and $\Node[1]∈\NodeSet{\Structure[1]}$.
  \item \label{itm:bwd-relation}
    For each \kl{relation symbol} $\RelSymbol[2]$ in $\Signature[2]$ of \kl{arity} $\Arity+1≥2$, and all $h∈\range{\EncodingOffset}$,
    there is a $\FormulaSet[1]$-\kl{sentence} $\Formula[1]_{\RelSymbol[2]}^h$ over $\enrichedsig{\Signature[1]}{\PosSymbol}$ \strut
    enriched with additional \kl{set symbols} $\tuple{\SetVariable[1]_i^j}_{1≤i≤\Arity}^{1≤j≤\EncodingOffset}$,\, such that
    \begin{align*}
      &\bigver{\Structure[1]}{\PosSymbol,\tuple{\SetVariable[1]_i^j}_{i≤\Arity}^{j≤\EncodingOffset}}{\Node[1],\tuple{\NodeSubset[1]_i^j}_{i≤\Arity}^{j≤\EncodingOffset}} \,\Models\, \Formula[1]_{\RelSymbol[2]}^h \\
      &\quad\text{if and only if} \\
      &\bigver{\Encoding(\Structure[1])}{\PosSymbol,\tuple{\SetVariable[2]_i}_{i≤\Arity}}{\Node[2], \tuple{\NodeSubset[2]_i}_{i≤\Arity}} \,\Models\, \dm[\RelSymbol[2]]\tuple{\PosIn{\SetVariable[2]_i}}_{i≤\Arity}, \\[1ex]
      &\quad\text{where $\Node[2]$ is $\tuple{h,\Node[1]}$ if $h≤\EncodingScale$, otherwise $h$,\, and} \\
      &\quad \NodeSubset[2]_i =\; \smashoperator{\bigcup_{1≤j≤\EncodingScale}}\,\bigl(\set{j}{×}\NodeSubset[1]_i^j\bigr)
      \,∪\, \,\smashoperator{\bigcup_{\EncodingScale<j≤\EncodingOffset}}\, \setbuilder{j}{\NodeSubset[1]_i^j=\NodeSet{\Structure[1]}},
    \end{align*}
    for all $\Structure[1]∈\StructClass[1]$, \kl{nodes} $\Node[1]∈\NodeSet{\Structure[1]}$, sets $\tuple{\NodeSubset[1]_i^j}_{1≤i≤\Arity}^{1≤j≤\EncodingScale}⊆\NodeSet{\Structure[1]}$\,
    and $\tuple{\NodeSubset[1]_i^j}_{1≤i≤\Arity}^{\EncodingScale<j≤\EncodingOffset}∈\set{\EmptySet,\NodeSet{\Structure[1]}}$,
    and \kl{set symbols} $\tuple{\SetVariable[2]_i}_{1≤i≤\Arity}$.
  \item \label{itm:bwd-backward}
    If $\FormulaSet[2]$ includes \kl{backward modalities},
    then for each \kl{relation symbol} $\RelSymbol[2]$ in $\Signature[2]$ of \kl{arity} at least $2$, and all $h∈\range{\EncodingOffset}$,
    there is a \kl[class-formula]{$\FormulaSet[1]$-formula} $\Formula[1]_{\invr{\RelSymbol[2]}}^h$ that satisfies
    the property of item~\ref{itm:bwd-relation} for $\invr{\RelSymbol[2]}$ instead of $\RelSymbol[2]$.
  \item \label{itm:bwd-global}
    If $\FormulaSet[2]$ includes \kl{global modalities}, then for all $h∈\range{\EncodingOffset}$,
    there is a \kl[class-formula]{$\FormulaSet[1]$-formula} $\Formula[1]_{\GlobalRelSymbol}^h$ that satisfies
    the property of item~\ref{itm:bwd-relation} for $\GlobalRelSymbol$ instead of $\RelSymbol[2]$ and $\Arity=1$.
  \item \label{itm:bwd-initial}
    \newcommand*{\YgdmY}{\NoHeight{\scalebox{0.8}{$\displaystyle\frac{\PosIn{\SetVariable[2]}}{\gdm \PosIn{\SetVariable[2]} \vphantom{\big(}}$}}}
    There is a $\FormulaSet[1]$-\kl{sentence} $\Formula[1]_\ini$ over $\Signature[1]$
    enriched with additional \kl{set symbols} $\tuple{\SetVariable[1]^j}^{1≤j≤\EncodingOffset}$, such that
    \begin{align*}
      &\bigver{\Structure[1]}{\tuple{\SetVariable[1]^j}^{j≤\EncodingOffset}}{\tuple{\NodeSubset[1]^j}^{j≤\EncodingOffset}} \,\Models\, \Formula[1]_\ini \\
      &\quad\text{if and only if} \\
      &\ver{\Encoding(\Structure[1])}{\SetVariable[2]}{\NodeSubset[2]} \,\Models\, \YgdmY \,, \\[1ex]
      &\quad\text{where $\YgdmY$ is \,$\PosIn{\SetVariable[2]}$\, if $\PosSymbol∈\Signature[2]$,\, otherwise $\gdm \PosIn{\SetVariable[2]}$, and} \\[0.5ex]
      &\quad \NodeSubset[2] =\; \smashoperator{\bigcup_{1≤j≤\EncodingScale}}\,\bigl(\set{j}{×}\NodeSubset[1]^j\bigr)
      \,∪\, \,\smashoperator{\bigcup_{\EncodingScale<j≤\EncodingOffset}}\, \setbuilder{j}{\NodeSubset[1]^j=\NodeSet{\Structure[1]}},
    \end{align*}
    for all \kl{structures} $\Structure[1]∈\StructClass[1]$,\, sets $\tuple{\NodeSubset[1]^j}^{1≤j≤\EncodingScale}⊆\NodeSet{\Structure[1]}$
    and $\tuple{\NodeSubset[1]^j}^{\EncodingScale<j≤\EncodingOffset}∈\set{\EmptySet,\NodeSet{\Structure[1]}}$,
    and $\SetVariable[2]∈\RelSymbolSet{1}$.
    \qedhere
  \end{enumerate}
\end{definition}

To simplify matters slightly,
we shall say that a \kl{linear encoding} $\Encoding$
allows for \Intro{bidirectional translation} between $\FormulaSet[1]$ and $\FormulaSet[2]$,
if it allows for both
forward  translation from $\FormulaSet[1]$ to $\FormulaSet[2]$ and
\kl{backward translation} from $\FormulaSet[2]$ to $\FormulaSet[1]$.
Furthermore,
in case $\FormulaSet[1]=\FormulaSet[2]$,
we may say “within $\FormulaSet[1]$” instead of “between $\FormulaSet[1]$ and $\FormulaSet[1]$”.

Let us now prove
that our notion of translatability is indeed sufficient
to imply \kl{figurative inclusion} on the semantic side,
even if we bring \kl{set quantifiers} into~play.

\begin{lemma}
  \label{lem:translation-inclusion}
  Consider two classes of \kl{structures} $\StructClass[1]$ and $\StructClass[2]$,
  a \kl{linear encoding} $\Encoding\colon\StructClass[1]\to\StructClass[2]$,
  two classes of \kl{formulas} $\FormulaSet[1],\FormulaSet[2]∈\set{\ML,\,\dML,\,\MLg,\,\dMLg}$,
  and let $\FormulaClassOperator ∈ \set{\SigmaMSO{\Level},\,\PiMSO{\Level},\,\BC\SigmaMSO{\Level}}$,\,
  for some arbitrary $\Level≥0$.
  \begin{enumerate}
  \item \label{itm:forward-inclusion}
    If $\Encoding$ allows for \kl{forward translation} from $\FormulaSet[1]$ to $\FormulaSet[2]$,
    then we have
    \begin{equation*}
      \semF{\+\FormulaClassOperator\+(\FormulaSet[1])}[\StructClass[1]] \;\figsubeq{\Encoding}\; \semF{\+\FormulaClassOperator\+(\FormulaSet[2])}[\StructClass[2]]\+.
    \end{equation*}
  \item \label{itm:backward-inclusion}
    Similarly,
    if $\Encoding$ allows for \kl{backward translation} from $\FormulaSet[2]$ to $\FormulaSet[1]$,
    then we have
    \begin{equation*}
      \semF{\+\FormulaClassOperator\+(\FormulaSet[1])}[\StructClass[1]] \;\figsupeq{\Encoding}\; \semF{\+\FormulaClassOperator\+(\FormulaSet[2])}[\StructClass[2]]\+.
      \qedhere
    \end{equation*}
  \end{enumerate}
\end{lemma}
\begin{proof}
  Let $\Signature[1]$ and $\Signature[2]$ be the \kl{signatures} underlying $\StructClass[1]$ and $\StructClass[2]$, respectively.
  Parts \ref{itm:forward-inclusion} and \ref{itm:backward-inclusion}
  of the \lcnamecref{lem:translation-inclusion}
  are treated separately in the following proof.

  In several places,
  given some \kl[class-formula]{$\MSO(\dMLg)$-formula} $\Formula[1]$,
  the need will arise to
  substitute newly created \kl[class-formula]{$\dMLg$-formulas}
  $\Formula[1]_1,…,\Formula[1]_k$ for \kl{set symbols} $\SetVariable[1]_1,…,\SetVariable[1]_k$.
  We shall write \Intro*{$\subst{\Formula[1]}{\tuple{\SetVariable[1]_i}_{i≤k}}{\tuple{\Formula[1]_i}_{i≤k}}$}
  to denote the \kl[class-formula]{$\MSO(\dMLg)$-formula}
  that one obtains by simultaneously replacing
  every \kl{free} occurrence of each $\SetVariable[1]_i$ in $\Formula[1]$ by the \kl{formula}~$\Formula[1]_i$.

  \proofsection{\ref{itm:forward-inclusion}.}
  For every $\FormulaClassOperator\+(\FormulaSet[1])$-\kl{sentence} $\Formula[1]$ over $\Signature[1]$,
  we must construct a $\FormulaClassOperator\+(\FormulaSet[2])$-\kl{sentence} $\Formula[2]_{\Formula[1]}$ over $\Signature[2]$,
  such that $\Formula[2]_{\Formula[1]}$ says about $\Encoding(\Structure[1])$ the same as $\Formula[1]$ says about $\Structure[1]$,
  for all \kl{structures} $\Structure[1]∈\StructClass[1]$.

  We start by focusing on the \kl{kernel} classes $\FormulaSet[1],\FormulaSet[2]$,
  and show the following by induction on the structure of \kl[class-formula]{$\FormulaSet[1]$-formulas}:
  For every $\FormulaSet[1]$-\kl{sentence} $\Formula[1]$ over \mbox{$\enrichedsig{\Signature[1]}{\PosSymbol}∪\ExtraSetVariableSet$},
  with $\ExtraSetVariableSet=\set{\SetVariable[3]_1,…,\SetVariable[3]_{\ESVCount}}$ being any collection of \kl{set symbols} disjoint from $\Signature[1]$ and $\Signature[2]$
  (i.e., \kl{free} \kl{set variables}),
  there is a $\FormulaSet[2]$-\kl{sentence} $\Formula[2]_{\Formula[1]}^*$ over $\enrichedsig{\Signature[2]}{\PosSymbol}∪\ExtraSetVariableSet$ such that
  \begin{align*}
    &\bigver{\Structure[1]}{\PosSymbol,\tuple{\SetVariable[3]_{\ESVIndex}}_{\ESVIndex≤\ESVCount}}{\Node[1],\tuple{\NodeSubset[1]_{\ESVIndex}}_{\ESVIndex≤\ESVCount}} \,\Models\, \Formula[1] \\
    &\quad\text{if and only if} \\
    &\bigver{\Encoding(\Structure[1])}{\PosSymbol,\tuple{\SetVariable[3]_{\ESVIndex}}_{\ESVIndex≤\ESVCount}}{\tuple{1,\Node[1]},\tuple{\NodeSubset[2]_{\ESVIndex}}_{\ESVIndex≤\ESVCount}} \,\Models\, \Formula[2]_{\Formula[1]}^*\+, \\[0.5ex]
    &\quad\text{assuming $\NodeSubset[1]_{\ESVIndex},\NodeSubset[2]_{\ESVIndex}$ satisfy $\Node[1]'{∈\,}\NodeSubset[1]_{\ESVIndex}⇔\tuple{1,\Node[1]'}{\,∈\,}\NodeSubset[2]_{\ESVIndex}$,}
  \end{align*}
  for all \kl{structures} $\Structure[1]∈\StructClass[1]$, \kl{nodes} $\Node[1]∈\NodeSet{\Structure[1]}$,
  and sets $\tuple{\NodeSubset[1]_{\ESVIndex}}_{1≤\ESVIndex≤\ESVCount}⊆\NodeSet{\Structure[1]}$ and $\tuple{\NodeSubset[2]_{\ESVIndex}}_{1≤\ESVIndex≤\ESVCount}⊆\NodeSet{\Encoding(\Structure[1])}$.
  \begin{itemize}
  \item If $\Formula[1] = \PosIs{\PosSymbol}$ or $\Formula[1] = \PosIn{\SetVariable[3]}$, for some $\SetVariable[3]∈\ExtraSetVariableSet$,
    it suffices to set $\Formula[2]_{\Formula[1]}^* = \Formula[1]$.
  \item If $\Formula[1] = \PosIn{\SetConstant[1]}$, for some \kl{node symbol} or \kl{set symbol} $\SetConstant[1]$ in $\Signature[1]$,
    we exploit that $\Encoding$ allows for \kl{forward translation} from $\FormulaSet[1]$ to $\FormulaSet[2]$,
    and choose $\Formula[2]_{\Formula[1]}^* = \Formula[2]_{\SetConstant[1]}$.
    Here, $\Formula[2]_{\SetConstant[1]}$ is the \kl{formula} postulated by
    \cref{def:forward-translation}~\ref{itm:fwd-set};
    it fulfills the induction hypothesis,
    since adding \kl{interpretations} of the \kl{symbols} $\SetVariable[3]_1,…,\SetVariable[3]_{\ESVCount}$
    to a \kl{structure}
    has no influence on whether or not that \kl{structure} \kl{satisfies}
    a \kl{sentence} over a \kl{signature} that does not contain these \kl{symbols}.
  \item If $\Formula[1] = \NOT \Formula[1]_1$ \,or\, $\Formula[1] = \Formula[1]_1 \OR \Formula[1]_2$,
    where $\Formula[1]_1$ and $\Formula[1]_2$ are \kl{formulas} that satisfy the induction hypothesis,
    we set $\Formula[2]_{\Formula[1]}^* = \NOT \Formula[2]_{\Formula[1]_1}^*$ \,or\, $\Formula[2]_{\Formula[1]}^* = \Formula[2]_{\Formula[1]_1}^* \OR \Formula[2]_{\Formula[1]_2}^*$,\, respectively.
  \item If $\Formula[1] = \dm[\RelSymbol[1]]\tuple{\Formula[1]_i}_{i≤\Arity}$,
    where $\RelSymbol[1]$ is a \kl{relation symbol} in $\Signature[1]$ of \kl{arity} $\Arity+1≥2$,
    and $\tuple{\Formula[1]_i}_{i≤\Arity}$ are $\FormulaSet[1]$-\kl{sentences} over $\enrichedsig{\Signature[1]}{\PosSymbol}∪\ExtraSetVariableSet$
    satisfying the induction hypothesis,
    we again use the fact that $\Encoding$ allows for \kl{forward translation} from $\FormulaSet[1]$ to $\FormulaSet[2]$.
    The desired \kl{formula} $\Formula[2]_{\Formula[1]}^*$ is obtained by substituting $\tuple{\Formula[2]_{\Formula[1]_i}^*}_{i≤\Arity}$
    for the \kl{symbols} $\tuple{\SetVariable[2]_i}_{i≤\Arity}$ in the \kl{formula} $\Formula[2]_{\RelSymbol[1]}$,
    whose existence is asserted by
    \cref{def:forward-translation}~\ref{itm:fwd-relation}, i.e.,
    \begin{equation*}
      \Formula[2]_{\Formula[1]}^* = \bigsubst{\Formula[2]_{\RelSymbol[1]}}{\tuple{\SetVariable[2]_i}_{i≤\Arity}}{\tuple{\Formula[2]_{\Formula[1]_i}^*}_{i≤\Arity}}.
    \end{equation*}
    For any integer $i \in \range[1]{\Arity}$,
    let $\NodeSubset[1]'_i$ be the set of \kl{nodes} $\Node[1]'∈\NodeSet{\Structure[1]}$
    that \kl{satisfy} $\Formula[1]_i$ in $\bigver{\Structure[1]}{\tuple{\SetVariable[3]_{\ESVIndex}}_{\ESVIndex≤\ESVCount}}{\tuple{\NodeSubset[1]_{\ESVIndex}}_{\ESVIndex≤\ESVCount}}$,
    and let $\NodeSubset[2]'_i$ be the set of \kl{nodes} $\Node[2]'∈\NodeSet{\Encoding(\Structure[1])}$
    that \kl{satisfy} $\Formula[2]_{\Formula[1]_i}$ in $\bigver{\Encoding(\Structure[1])}{\tuple{\SetVariable[3]_{\ESVIndex}}_{\ESVIndex≤\ESVCount}}{\tuple{\NodeSubset[2]_{\ESVIndex}}_{\ESVIndex≤\ESVCount}}$.
    By induction hypothesis, we are guaranteed that
    all the sets $\NodeSubset[1]'_i$, $\NodeSubset[2]'_i$ are such that
    a \kl{node} $\Node[1]'$ lies in $\NodeSubset[1]'_i$ if and only if $\tuple{1,\Node[1]'}$ lies in $\NodeSubset[2]'_i$.
    Thus, we have
    \begin{equation*}
      \begin{aligned}
        &\bigver{\Structure[1]}{\PosSymbol,\tuple{\SetVariable[3]_{\ESVIndex}}_{\ESVIndex≤\ESVCount}}{\Node[1],\tuple{\NodeSubset[1]_{\ESVIndex}}_{\ESVIndex≤\ESVCount}} \,\Models\, \Formula[1] \\[0.5ex]
        \text{iff}\quad
        &\bigver{\Structure[1]}{\PosSymbol,\tuple{\SetVariable[1]_i}_{i≤\Arity}}{\Node[1],\tuple{\NodeSubset[1]'_i}_{i≤\Arity}} \,\Models\, \dm[\RelSymbol[1]]\tuple{\PosIn{\SetVariable[1]_i}}_{i≤\Arity} \\[0.5ex]
        \text{iff}\quad
        &\bigver{\Encoding(\Structure[1])}{\PosSymbol,\tuple{\SetVariable[2]_i}_{i≤\Arity}\:\!}{\tuple{1,\Node[1]},\tuple{\NodeSubset[2]'_i}_{i≤\Arity}} \,\Models\, \Formula[2]_{\RelSymbol[1]} \\[0.5ex]
        \text{iff}\quad
        &\bigver{\Encoding(\Structure[1])}{\PosSymbol,\tuple{\SetVariable[3]_{\ESVIndex}}_{\ESVIndex≤\ESVCount}}{\tuple{1,\Node[1]},\tuple{\NodeSubset[2]_{\ESVIndex}}_{\ESVIndex≤\ESVCount}} \,\Models\, \Formula[2]_{\Formula[1]}^*\+.
      \end{aligned}
    \end{equation*}
  \item If $\Formula[1] = \bdm[\RelSymbol[1]]\tuple{\Formula[1]_i}_{i≤\Arity}$,
    assuming $\FormulaSet[1]$ incorporates \kl{backward modalities},
    we obtain $\Formula[2]_{\Formula[1]}^*$ by applying the same argument as in the previous case,
    but this time considering $\invr{\RelSymbol[1]}$ instead of $\RelSymbol[1]$ and invoking
    \cref{def:forward-translation}~\ref{itm:fwd-backward}.
  \item If $\Formula[1] = \gdm \Formula[1]_1$,
    supposing $\FormulaSet[1]$ includes \kl{global modalities},
    we again follow the same line of reasoning as in the case $\Formula[1] = \dm[\RelSymbol[1]]\tuple{\Formula[1]_i}_{i≤\Arity}$,
    referring to \cref{def:forward-translation}~\ref{itm:fwd-global}
    and using~$\GlobalRelSymbol$ instead of $\RelSymbol[1]$, with $\Arity=1$.
  \end{itemize}

  Now we can consider the case where
  the \kl{position symbol}~$\PosSymbol$ is not (re)mapped,
  and then look beyond the \kl{kernel} classes
  to finally deal with \kl{set quantifiers}.
  Arguing once more by structural induction,
  we extend the preceding claim as follows:
  For every $\FormulaClassOperator\+(\FormulaSet[1])$-\kl{sentence} $\Formula[1]$ over $\Signature[1]∪\ExtraSetVariableSet$,
  with $\ExtraSetVariableSet=\set{\SetVariable[3]_1,…,\SetVariable[3]_{\ESVCount}}$ as before (possibly empty),
  there is a $\FormulaClassOperator\+(\FormulaSet[2])$-\kl{sentence} $\Formula[2]_{\Formula[1]}$ over $\Signature[2]∪\ExtraSetVariableSet$ such that
  \begin{align*}
    &\bigver{\Structure[1]}{\tuple{\SetVariable[3]_{\ESVIndex}}_{\ESVIndex≤\ESVCount}}{\tuple{\NodeSubset[1]_{\ESVIndex}}_{\ESVIndex≤\ESVCount}} \,\Models\, \Formula[1] \\
    &\quad\text{if and only if} \\
    &\bigver{\Encoding(\Structure[1])}{\tuple{\SetVariable[3]_{\ESVIndex}}_{\ESVIndex≤\ESVCount}}{\tuple{\NodeSubset[2]_{\ESVIndex}}_{\ESVIndex≤\ESVCount}} \,\Models\, \Formula[2]_{\Formula[1]}\+, \\[0.5ex]
    &\quad\text{assuming $\NodeSubset[1]_{\ESVIndex},\NodeSubset[2]_{\ESVIndex}$ satisfy $\Node[1]{\,∈\,}\NodeSubset[1]_{\ESVIndex}⇔\tuple{1,\Node[1]}{\,∈\,}\NodeSubset[2]_{\ESVIndex}$,}
  \end{align*}
  for all $\Structure[1]∈\StructClass[1]$,\,
  $\tuple{\NodeSubset[1]_{\ESVIndex}}_{1≤\ESVIndex≤\ESVCount}⊆\NodeSet{\Structure[1]}$ and $\tuple{\NodeSubset[2]_{\ESVIndex}}_{1≤\ESVIndex≤\ESVCount}⊆\NodeSet{\Encoding(\Structure[1])}$.
  \begin{itemize}
  \item If $\Formula[1]$ lies in the \kl{kernel} $\FormulaSet[1]$,
    we make use of the claim just proven,
    together with the \kl{formula} $\Formula[2]_\ini$ described in
    \cref{def:forward-translation}~\ref{itm:fwd-initial}.
    We set\, $\Formula[2]_{\Formula[1]} = \subst{\Formula[2]_\ini}{\SetVariable[2]}{\Formula[2]_{\Formula[1]}^*}$.
    \begin{itemize}
    \item If $\PosSymbol$ belongs to $\Signature[1]$,
      the asserted property of $\Formula[2]_\ini$ guarantees
      that $\Formula[1]$ holds at the initial position $\inp{\PosSymbol}{\Structure[1]}$
      in the \kl[extended variant]{$\ExtraSetVariableSet$-extended variant} of $\Structure[1]$
      if and only if $\Formula[2]_{\Formula[1]}$ is \kl{satisfied} by the \kl[extended variant]{$\ExtraSetVariableSet$-extended variant} of $\Encoding(\Structure[1])$.
    \item Otherwise, $\PosSymbol$ cannot be \kl{free} in $\Formula[1]$,
      since $\Formula[1]$ is a \kl{sentence} over $\Signature[1]∪\ExtraSetVariableSet$,
      which also implies that $\FormulaSet[1]$ incorporates \kl{global modalities}.
      It follows that $\Formula[1]$ is \kl[device equivalent]{equivalent} to $\gdm \Formula[1]$.
      Again applying the definition of $\Formula[2]_\ini$,
      we obtain that the \kl[extended variant]{$\ExtraSetVariableSet$-extended variant} of $\Structure[1]$ \kl{satisfies} $\gdm \Formula[1]$,
      and thus $\Formula[1]$,
      if and only if the \kl[extended variant]{$\ExtraSetVariableSet$-extended variant} of $\Encoding(\Structure[1])$ \kl{satisfies} $\Formula[2]_{\Formula[1]}$.
    \end{itemize}
  \item If $\Formula[1]$ is a Boolean combination of \kl{formulas}
    that satisfy the induction hypothesis,
    the translation is straightforward,
    just as in the previous part of the proof.
  \item If $\Formula[1] = \Exists{\SetVariable[3]_{\ESVCount+1}}\Formula[1]_1$,
    where $\Formula[1]_1$ is a $\FormulaClassOperator\+(\FormulaSet[1])$-\kl{sentence} over \mbox{$\Signature[1]∪\set{\SetVariable[3]_1,…,\SetVariable[3]_{\ESVCount+1}}$}
    that satisfies the hypothesis,
    we choose $\Formula[2]_{\Formula[1]} = \Exists{\SetVariable[3]_{\ESVCount+1}}\Formula[2]_{\Formula[1]_1}$.
    To justify this choice,
    let $\Structure[1]'$ and $\Encoding(\Structure[1])'$ denote the \kl[extended variant]{$\ExtraSetVariableSet$-extended variants} of $\Structure[1]$ and $\Encoding(\Structure[1])$, respectively.
    We get the following by induction:
    \begin{itemize}
    \item If choosing $\map{\SetVariable[3]_{\ESVCount+1}}{\NodeSubset[1]_{\ESVCount+1}}$ leads to \kl[satisfy]{satisfaction} of $\Formula[1]_1$ in $\Structure[1]'$,
      then choosing $\map{\SetVariable[3]_{\ESVCount+1}}{\set{1}\+{×}\+\NodeSubset[1]_{\ESVCount+1}}$ does the same for $\Formula[2]_{\Formula[1]_1}$ in $\Encoding(\Structure[1])'$.
    \item Conversely, if $\map{\SetVariable[3]_{\ESVCount+1}}{\NodeSubset[2]_{\ESVCount+1}}$ is a \kl[satisfy]{satisfying} choice for $\Formula[2]_{\Formula[1]_1}$ in $\Encoding(\Structure[1])'$,
      then so is $\map{\SetVariable[3]_{\ESVCount+1}}{{\setbuilder{\Node[1]}{\tuple{1,\Node[1]}∈\NodeSubset[2]_{\ESVCount+1}}}}$ for $\Formula[1]_1$ in $\Structure[1]'$.
    \end{itemize}
  \end{itemize}

  \proofsection{\ref{itm:backward-inclusion}.}
  The proof of the reverse direction of the \lcnamecref{lem:translation-inclusion}
  is very similar to the previous one, but a bit more cumbersome,
  because each \kl{node} of a \kl{structure} $\Structure[1]$ has to play the role
  of several different \kl{nodes} in $\Encoding(\Structure[1])$.
  Given any $\FormulaClassOperator\+(\FormulaSet[2])$-\kl{sentence} $\Formula[2]$ over $\Signature[2]$,
  we need to construct a $\FormulaClassOperator\+(\FormulaSet[1])$-\kl{sentence} $\Formula[1]_{\Formula[2]}$ over $\Signature[1]$,
  such that evaluating $\Formula[1]_{\Formula[2]}$ on $\Structure[1]$ is equivalent to evaluating $\Formula[2]$ on $\Encoding(\Structure[1])$,
  for all $\Structure[1]∈\StructClass[1]$.
  For the remainder of this proof,
  let $\EncodingScale$, $\EncodingOffset$ be the parameters of the \kl{linear encoding} $\Encoding$.

  Again, we first deal with the \kl{kernel} classes $\FormulaSet[1],\FormulaSet[2]$,
  and show the following claim by induction on the structure of \kl[class-formula]{$\FormulaSet[2]$-formulas}:
  For every $\FormulaSet[2]$-\kl{sentence} $\Formula[2]$ over \mbox{$\enrichedsig{\Signature[2]}{\PosSymbol}∪\ExtraSetVariableSet$} and all \mbox{$h∈\range{\EncodingOffset}$},
  with $\ExtraSetVariableSet=\set{\SetVariable[3]_1,…,\SetVariable[3]_{\ESVCount}}⊆\RelSymbolSet{1}{\setminus}\Signature[2]$,
  there is a $\FormulaSet[1]$-\kl{sentence} $\Formula[1]_{\Formula[2]}^h$ over \mbox{$\enrichedsig{\Signature[1]}{\PosSymbol}∪\tilde{\ExtraSetVariableSet}$},
  with $\tilde{\ExtraSetVariableSet}=\set{\SetVariable[3]_1^1,…,\SetVariable[3]_{\ESVCount}^{\EncodingOffset}}⊆\RelSymbolSet{1}{\setminus}\Signature[1]$,\,
  such that
  \begin{align*}
    &\bigver{\Structure[1]}{\PosSymbol,\tuple{\SetVariable[3]_{\ESVIndex}^j}_{\ESVIndex≤\ESVCount}^{j≤\EncodingOffset}}{\Node[1],\tuple{\NodeSubset[1]_{\ESVIndex}^j}_{\ESVIndex≤\ESVCount}^{j≤\EncodingOffset}} \,\Models\, \Formula[1]_{\Formula[2]}^h \\
    &\quad\text{if and only if} \\
    &\bigver{\Encoding(\Structure[1])}{\PosSymbol,\tuple{\SetVariable[3]_{\ESVIndex}}_{\ESVIndex≤\ESVCount}}{\Node[2],\tuple{\NodeSubset[2]_{\ESVIndex}}_{\ESVIndex≤\ESVCount}} \,\Models\, \Formula[2], \\[1ex]
    &\quad\text{where $\Node[2]$ is $\tuple{h,\Node[1]}$ if $h≤\EncodingScale$, otherwise $h$,\, and} \\
    &\quad \NodeSubset[2]_{\ESVIndex} =\; \smashoperator{\bigcup_{1≤j≤\EncodingScale}}\,\bigl(\set{j}{×}\NodeSubset[1]_{\ESVIndex}^j\bigr)
    \,∪\, \,\smashoperator{\bigcup_{\EncodingScale<j≤\EncodingOffset}}\, \setbuilder{j}{\NodeSubset[1]_{\ESVIndex}^j=\NodeSet{\Structure[1]}},
  \end{align*}
  for all $\Structure[1]∈\StructClass[1]$,\, $\Node[1]∈\NodeSet{\Structure[1]}$, and sets $\tuple{\NodeSubset[1]_{\ESVIndex}^j}_{1≤\ESVIndex≤\ESVCount}^{1≤j≤\EncodingScale}⊆\NodeSet{\Structure[1]}$
  and $\tuple{\NodeSubset[1]_{\ESVIndex}^j}_{1≤\ESVIndex≤\ESVCount}^{\EncodingScale<j≤\EncodingOffset}∈\set{\EmptySet,\NodeSet{\Structure[1]}}$.
  \begin{itemize}
  \item If $\Formula[2] = \PosIs{\PosSymbol}$, it suffices to set $\Formula[1]_{\Formula[2]}^h = \PosIs{\PosSymbol}$.
  \item If $\Formula[2] = \PosIn{\SetVariable[3]_{\ESVIndex}}$, for some $\SetVariable[3]_{\ESVIndex}∈\ExtraSetVariableSet$,
    the translation is given by $\Formula[1]_{\Formula[2]}^h = \PosIn{\SetVariable[3]_{\ESVIndex}^h}$.
  \item If $\Formula[2] = \PosIn{\SetConstant[2]}$, for some \kl{node symbol} or \kl{set symbol} $\SetConstant[2]$ in $\Signature[2]$,
    we use the fact that $\Encoding$ allows for \kl{backward translation} from $\FormulaSet[2]$ to $\FormulaSet[1]$,
    and choose $\Formula[1]_{\Formula[2]}^h$ to be the \kl{formula} $\Formula[1]_{\SetConstant[2]}^h$,
    which is provided by \cref{def:backward-translation}~\ref{itm:bwd-set}.
    The definition asserts that
    this \kl{formula} fulfills the induction hypothesis for the case
    where $\Structure[1]$ and $\Encoding(\Structure[1])$ are not extended using additional \kl{set symbols}
    from $\tilde{\ExtraSetVariableSet}$ and $\ExtraSetVariableSet$.
    But since these \kl{symbols} do not occur \kl{freely} in $\Formula[1]_{\SetConstant[2]}^h$ and $\PosIn{\SetConstant[2]}$,
    their \kl{interpretations} do not influence the evaluation of the \kl{formulas}.
  \item If $\Formula[2] = \NOT \Formula[2]_1$ \,or\, $\Formula[2] = \Formula[2]_1 \OR \Formula[2]_2$,
    where $\Formula[2]_1$ and $\Formula[2]_2$ are \kl{formulas} that satisfy the induction hypothesis,
    we set $\Formula[1]_{\Formula[2]}^h = \NOT \Formula[1]_{\Formula[2]_1}^h$ \,or\, $\Formula[1]_{\Formula[2]}^h = \Formula[1]_{\Formula[2]_1}^h \OR \Formula[1]_{\Formula[2]_2}^h$,\, respectively.
  \item If $\Formula[2] = \dm[\RelSymbol[2]]\tuple{\Formula[2]_i}_{i≤\Arity}$,
    where $\RelSymbol[2]$ is a \kl{relation symbol} in $\Signature[2]$ of \kl{arity} $\Arity+1≥2$,
    and $\tuple{\Formula[2]_i}_{i≤\Arity}$ are $\FormulaSet[2]$-\kl{sentences} over $\enrichedsig{\Signature[2]}{\PosSymbol}∪\ExtraSetVariableSet$
    satisfying the hypothesis,
    we again rely on the premise
    that $\Encoding$ allows for \kl{backward translation} from $\FormulaSet[2]$ to $\FormulaSet[1]$.
    We construct $\Formula[1]_{\Formula[2]}^h$ by plugging the \kl{formulas} $\tuple{\Formula[1]_{\Formula[2]_i}^j}_{i≤\Arity}^{j≤\EncodingOffset}$
    provided by induction into the \kl{formula} $\Formula[1]_{\RelSymbol[2]}^h$ of
    \cref{def:backward-translation}~\ref{itm:bwd-relation} as follows:
    \begin{equation*}
      \Formula[1]_{\Formula[2]}^h = \bigsubst{\Formula[1]_{\RelSymbol[2]}^h}{\tuple{\SetVariable[1]_i^j}_{i≤\Arity}^{j≤\EncodingOffset}}{\tuple{\Formula[1]_{\Formula[2]_i}^j}_{i≤\Arity}^{j≤\EncodingOffset}}.
    \end{equation*}
    For $1≤i≤\Arity$ and $1≤j≤\EncodingOffset$, let $\NodeSubset[1]_i^{j\prime}$ be the set of \kl{nodes} $\Node[1]'∈\NodeSet{\Structure[1]}$
    that \kl{satisfy} $\Formula[1]_{\Formula[2]_i}^j$ in $\bigver{\Structure[1]}{\tuple{\SetVariable[3]_{\ESVIndex}^j}_{\ESVIndex≤\ESVCount}^{j≤\EncodingOffset}}{\tuple{\NodeSubset[1]_{\ESVIndex}^j}_{\ESVIndex≤\ESVCount}^{j≤\EncodingOffset}}$,
    and let $\NodeSubset[2]'_i$ be the set of \kl{nodes} $\Node[2]'∈\NodeSet{\Encoding(\Structure[1])}$
    that \kl{satisfy} $\Formula[2]_i$ in $\bigver{\Encoding(\Structure[1])}{\tuple{\SetVariable[3]_{\ESVIndex}}_{\ESVIndex≤\ESVCount}}{\tuple{\NodeSubset[2]_{\ESVIndex}}_{\ESVIndex≤\ESVCount}}$.
    The induction hypothesis ensures that
    \begin{equation*}
      \NodeSubset[2]'_i =\; \smashoperator{\bigcup_{1≤j≤\EncodingScale}}\,\bigl(\set{j}{×}\NodeSubset[1]_i^{j\prime}\bigr)
      \,∪\, \,\smashoperator{\bigcup_{\EncodingScale<j≤\EncodingOffset}}\, \setbuilder{j}{\NodeSubset[1]_i^{j\prime}=\NodeSet{\Structure[1]}}.
    \end{equation*}
    Hence, we obtain the required equivalence as follows:
    \begin{equation*}
      \begin{aligned}
        &\bigver{\Structure[1]}{\PosSymbol,\tuple{\SetVariable[3]_{\ESVIndex}^j}_{\ESVIndex≤\ESVCount}^{j≤\EncodingOffset}}{\Node[1],\tuple{\NodeSubset[1]_{\ESVIndex}^j}_{\ESVIndex≤\ESVCount}^{j≤\EncodingOffset}} \,\Models\, \Formula[1]_{\Formula[2]}^h \\[0.5ex]
        \text{iff}\quad
        &\bigver{\Structure[1]}{\PosSymbol,\tuple{\SetVariable[1]_i^j}_{i≤\Arity}^{j≤\EncodingOffset}\:\!}{\Node[1],\tuple{\NodeSubset[1]_i^{j\prime}}_{i≤\Arity}^{j≤\EncodingOffset}\!} \,\Models\, \Formula[1]_{\RelSymbol[2]}^h \\[0.5ex]
        \text{iff}\quad
        &\bigver{\Encoding(\Structure[1])}{\PosSymbol,\tuple{\SetVariable[2]_i}_{i≤\Arity}\:\!}{\Node[2],\tuple{\NodeSubset[2]'_i}_{i≤\Arity}} \,\Models\, \dm[\RelSymbol[2]]\tuple{\PosIn{\SetVariable[2]_i}}_{i≤\Arity} \\[0.5ex]
        \text{iff}\quad
        &\bigver{\Encoding(\Structure[1])}{\PosSymbol,\tuple{\SetVariable[3]_{\ESVIndex}}_{\ESVIndex≤\ESVCount}}{\Node[2],\tuple{\NodeSubset[2]_{\ESVIndex}}_{\ESVIndex≤\ESVCount}} \,\Models\, \Formula[2].
      \end{aligned}
    \end{equation*}
  \item If $\Formula[2] = \bdm[\RelSymbol[2]]\tuple{\Formula[2]_i}_{i≤\Arity}$,
    supposing $\FormulaSet[2]$ includes \kl{backward modalities},
    we construct $\Formula[1]_{\Formula[2]}^h$ using the same approach as in the previous case,
    the only difference being that we consider $\invr{\RelSymbol[2]}$ instead of $\RelSymbol[2]$ and invoke
    \cref{def:backward-translation}~\ref{itm:bwd-backward} instead of
    \ref{def:backward-translation}~\ref{itm:bwd-relation}.
  \item If $\Formula[2] = \gdm \Formula[2]_1$, in case $\FormulaSet[2]$ includes \kl{global modalities},
    we again proceed as for the case $\Formula[2] = \dm[\RelSymbol[2]]\tuple{\Formula[2]_i}_{i≤\Arity}$,
    this time using~$\GlobalRelSymbol$ instead of $\RelSymbol[2]$, with $\Arity=1$,
    and referring to
    \cref{def:backward-translation}~\ref{itm:bwd-global}.
  \end{itemize}

  Similarly to the proof of part~\ref{itm:forward-inclusion},
  we now extend the previous property to cover
  \kl{formulas} with \kl{set quantifiers},
  evaluated on \kl{structures} that may \kl{interpret} the \kl{position symbol} $\PosSymbol$ arbitrarily.
  Our induction hypothesis is the following:
  For every $\FormulaClassOperator\+(\FormulaSet[2])$-\kl{sentence} $\Formula[2]$ over \mbox{$\Signature[2]∪\ExtraSetVariableSet$},
  with $\ExtraSetVariableSet=\set{\SetVariable[3]_1,…,\SetVariable[3]_{\ESVCount}}⊆\RelSymbolSet{1}{\setminus}\Signature[2]$ (possibly empty),
  there is a $\FormulaClassOperator\+(\FormulaSet[1])$-\kl{sentence} $\Formula[1]_{\Formula[2]}$ over \mbox{$\Signature[1]∪\tilde{\ExtraSetVariableSet}$},
  with $\tilde{\ExtraSetVariableSet}=\set{\SetVariable[3]_1^1,…,\SetVariable[3]_{\ESVCount}^{\EncodingOffset}}⊆\RelSymbolSet{1}{\setminus}\Signature[1]$,\,
  such that
  \begin{align*}
    &\bigver{\Structure[1]}{\tuple{\SetVariable[3]_{\ESVIndex}^j}_{\ESVIndex≤\ESVCount}^{j≤\EncodingOffset}}{\tuple{\NodeSubset[1]_{\ESVIndex}^j}_{\ESVIndex≤\ESVCount}^{j≤\EncodingOffset}} \,\Models\, \Formula[1]_{\Formula[2]} \\
    &\quad\text{if and only if} \\
    &\bigver{\Encoding(\Structure[1])}{\tuple{\SetVariable[3]_{\ESVIndex}}_{\ESVIndex≤\ESVCount}}{\tuple{\NodeSubset[2]_{\ESVIndex}}_{\ESVIndex≤\ESVCount}} \,\Models\, \Formula[2], \quad\text{where} \\[1ex]
    &\quad \NodeSubset[2]_{\ESVIndex} =\; \smashoperator{\bigcup_{1≤j≤\EncodingScale}}\,\bigl(\set{j}{×}\NodeSubset[1]_{\ESVIndex}^j\bigr)
    \,∪\, \,\smashoperator{\bigcup_{\EncodingScale<j≤\EncodingOffset}}\, \setbuilder{j}{\NodeSubset[1]_{\ESVIndex}^j=\NodeSet{\Structure[1]}},
  \end{align*}
  for all \kl{structures} $\Structure[1]∈\StructClass[1]$, and sets $\tuple{\NodeSubset[1]_{\ESVIndex}^j}_{1≤\ESVIndex≤\ESVCount}^{1≤j≤\EncodingScale}⊆\NodeSet{\Structure[1]}$
  and $\tuple{\NodeSubset[1]_{\ESVIndex}^j}_{1≤\ESVIndex≤\ESVCount}^{\EncodingScale<j≤\EncodingOffset}∈\set{\EmptySet,\NodeSet{\Structure[1]}}$.
  \begin{itemize}
  \item If $\Formula[2]$ belongs to the \kl{kernel} class $\FormulaSet[2]$,
    we apply the claim just proven,
    and construct $\Formula[1]_{\Formula[2]}$ by substituting into the \kl{formula} $\Formula[1]_\ini$ provided by
    \cref{def:backward-translation}~\ref{itm:bwd-initial}:
    \mbox{$\Formula[1]_{\Formula[2]} = \bigsubst{\Formula[1]_\ini}{\tuple{\SetVariable[1]^j}^{j≤\EncodingOffset}}{\tuple{\Formula[1]_{\Formula[2]}^j}^{j≤\EncodingOffset}}$}.
    Proceeding analogously to the proof of part~\ref{itm:forward-inclusion},
    we have to distinguish whether or not the \kl{position symbol} $\PosSymbol$ belongs to $\Signature[2]$.
    (If it does not, $\Formula[2]$ is necessarily \kl[device equivalent]{equivalent} to $\gdm \Formula[2]$.)
    In both cases, the definition of $\Formula[1]_\ini$ guarantees
    that the \kl[extended variant]{$\tilde{\ExtraSetVariableSet}$-extended variant} of $\Structure[1]$ \kl{satisfies} $\Formula[1]_{\Formula[2]}$
    if and only if the \kl[extended variant]{$\ExtraSetVariableSet$-extended variant} of $\Encoding(\Structure[1])$ \kl{satisfies} $\Formula[2]$.
  \item If $\Formula[2]$ is a Boolean combination of \kl[formulas]{subformulas}
    that satisfy the induction hypothesis,
    then $\Formula[1]_{\Formula[2]}$ is simply the corresponding Boolean combination
    of the translated \kl[formulas]{subformulas}.
  \item If $\Formula[2] = \Exists{\SetVariable[3]_{\ESVCount+1}}\Formula[2]_1$,
    where $\Formula[2]_1$ is a $\FormulaClassOperator\+(\FormulaSet[2])$-\kl{sentence} over \mbox{$\Signature[2]∪\set{\SetVariable[3]_1,…,\SetVariable[3]_{\ESVCount+1}}$}
    that satisfies the induction hypothesis,
    we choose $\Formula[1]_{\Formula[2]}$ to be the \kl{formula}
    \begin{equation*}
      \Exists{\tuple{\SetVariable[3]_{\ESVCount+1}^j}^{j≤\EncodingScale}}
      \Bigl(\,
      \smashoperator{\bigOR_{N\,⊆\;\lrange[\EncodingScale\:\!\!]{\:\!\!\EncodingOffset} \vphantom{\big(}}}
      \bigsubst{\Formula[1]_{\Formula[2]_1}}{\bigtuple{\SetVariable[3]_{\ESVCount+1}^j}^{j>\EncodingScale}\!}{\!\bigtuple{N(j)}^{j>\EncodingScale}}
      \,\Bigr),
    \end{equation*}
    with $N(j)=\True$ if $j∈N$, and $N(j)=\False$ otherwise.
    For each set $N⊆\lrange[\EncodingScale]{\EncodingOffset}$,
    let $\Formula[1]_{\Formula[2]_1}^N$ denote the disjunct corresponding to $N$ in the \kl{formula} above.
    By induction, we have the following equivalence:
    the \kl{interpretation} map
    $\map{\tuple{\SetVariable[3]_{\ESVCount+1}^j}^{j≤\EncodingScale}}{\tuple{\NodeSubset[1]_{\ESVCount+1}^j}^{j≤\EncodingScale}}$
    leads to \kl[satisfy]{satisfaction} of $\Formula[1]_{\Formula[2]_1}^N$ in the \kl[extended variant]{$\tilde{\ExtraSetVariableSet}$-extended variant} of $\Structure[1]$
    if and only if
    \begin{equation*}
      \SetVariable[3]_{\ESVCount+1} \,\mapsto\,\;
      \smashoperator{\bigcup_{1≤j≤\EncodingScale}}\,\bigl(\set{j}{×}\NodeSubset[1]_{\ESVCount+1}^j\bigr) \,∪\, \,N
    \end{equation*}
    is a \kl[satisfy]{satisfying} choice for $\Formula[2]_1$ in the \kl[extended variant]{$\ExtraSetVariableSet$-extended variant} of $\Encoding(\Structure[1])$.
    \qedhere
  \end{itemize}
\end{proof}

\subsection{Getting rid of multiple edge relations}

We now show how to encode a multi-relational \kl{digraph} into a $1$-relational one,
by inserting additional \kl{labeled} \kl{nodes}
that represent the different \kl{edge relations}.

\begin{proposition}
  \label{prp:multirelational-labeled}
  \NoMathBreak
  For all $\BitCount,\RelCount≥0$ and $\FormulaSet[1]∈\set{\MLg,\dMLg}$,
  there is a \kl{linear encoding} $\Encoding$ from $\DIGRAPH[\BitCount][\RelCount]$ into $\DIGRAPH[\BitCount+\RelCount][1]$
  that allows for \kl{bidirectional translation} within $\FormulaSet[1]$.

  Moreover, $\Encoding(\DIGRAPH[\BitCount][\RelCount])$ is $\PiMSO{1}(\MLg)$-\kl{definable} over $\DIGRAPH[\BitCount+\RelCount][1]$.
\end{proposition}
\begin{proof}
  We choose $\Encoding$ to be the \kl{linear encoding}
  that assigns to each $\BitCount$-bit \kl{labeled}, $\RelCount$-relational \kl{digraph} $\Digraph[1]$
  the \mbox{$(\BitCount+\RelCount)$-bit} \kl{labeled} ($1$-relational) \kl{digraph} $\Encoding(\Digraph[1])=\Digraph[2]$
  with \kl{domain} $\range{\RelCount+1}×\NodeSet{\Digraph[1]}$,
  \kl{labeling sets} $\inp{\SetConstant_i}{\Digraph[2]}=\set{1}×\inp{\SetConstant_i}{\Digraph[1]}$\!, for $1≤i≤\BitCount$,
  and $\inp{\SetConstant_{\BitCount+i}}{\Digraph[2]}=\set{i+1}×\NodeSet{\Digraph[1]}$, for $1≤i≤\RelCount$,
  and \kl{edge relation}
  \begin{equation*}
    \begin{aligned}
      \inp{\RelSymbol}{\Digraph[2]} = \; \smashoperator{\bigcup_{1≤i≤\RelCount}} \:\,
      \Bigl(\,&\bigsetbuilder{\bigtuple{\tuple{1,\Node[1]},\tuple{i+1,\Node[1]}}}{\Node[1]∈\NodeSet{\Digraph[1]}} \;\; ∪ \\[-1.6ex]
              &\bigsetbuilder{\bigtuple{\tuple{i+1,\Node[1]},\tuple{1,\Node[1]}}}{\Node[1]∈\NodeSet{\Digraph[1]}} \;\; ∪ \\[-0.6ex]
              &\bigsetbuilder{\bigtuple{\tuple{i+1,\Node[1]},\tuple{i+1,\Node[1]'}}}{\tuple{\Node[1],\Node[1]'}∈\inp{\RelSymbol_i}{\Digraph[1]}}\,\Bigr).
    \end{aligned}
  \end{equation*}
  That is,
  for each \kl{node} $\Node[1]∈\NodeSet{\Digraph[1]}$ and for $1≤i≤\RelCount$,
  we introduce an additional \kl{node}
  representing the “$\RelSymbol_i$-port” of $\Node[1]$,
  and connect everything accordingly.

  Our \kl{forward translation},
  from $\FormulaSet[1]$ on $\DIGRAPH[\BitCount][\RelCount]$ to $\FormulaSet[1]$ on $\Encoding(\DIGRAPH[\BitCount][\RelCount])$,
  is given by
  \begin{equation*}
    \begin{aligned}
      &\begin{aligned}
        \Formula[2]_{\SetConstant_i}       &= \PosIn{\SetConstant_i} \quad\text{for $1≤i≤\BitCount$}, \\
        \Formula[2]_{\RelSymbol_i}       &= \dm\bigl(\Formula[2]_{i+1} \AND \dm\dm(\Formula[2]_1 \AND \PosIn{\SetVariable[2]})\bigr) \quad\text{for $1≤i≤\RelCount$}, \\
        \Formula[2]_{\invr{\RelSymbol_i}} &= \dm\bigl(\Formula[2]_{i+1} \AND \bdm\dm(\Formula[2]_1 \AND \PosIn{\SetVariable[2]})\bigr) \quad\text{for $1≤i≤\RelCount$}, \\
        \Formula[2]_{\GlobalRelSymbol}     &= \gdm(\Formula[2]_1 \AND \PosIn{\SetVariable[2]}),\\
        \Formula[2]_\ini        &= \Formula[2]_{\GlobalRelSymbol}\+,
      \end{aligned} \\[1ex]
      &\begin{alignedat}{3}
        \text{where}&\; &    \Formula[2]_1 &= \textstyle \NOT\bigOR_{1≤i≤\RelCount}(\Formula[2]_{i+1}) & \qquad&\text{(“regular”)}, \\
                    &   & \Formula[2]_{i+1} &= \PosIn{\SetConstant_{\BitCount+i}} \quad\text{for $1≤i≤\RelCount$}    &       &\text{(“$\RelSymbol_i$-port”)}.
      \end{alignedat}
    \end{aligned}
  \end{equation*}

  Our translation in the other direction,
  from $\FormulaSet[1]$ on $\Encoding(\DIGRAPH[\BitCount][\RelCount])$ to $\FormulaSet[1]$ on $\DIGRAPH[\BitCount][\RelCount]$,
  is given by
  \begin{align*}
    \Formula[1]_{\SetConstant_i}^{h+1}   &= \begin{cases*}
                       \PosIn{\SetConstant_i} & for $h=0$ and $1≤i≤\BitCount$, \\
                       \False   & for $h=0$ and $\BitCount+1≤i≤\BitCount+\RelCount$, \\
                       \True   & for $1≤h≤\RelCount$ and $i=\BitCount+h$, \\
                       \False   & for $1≤h≤\RelCount$ and $i≠\BitCount+h$,
                     \end{cases*} \\
    \Formula[1]_{\RelSymbol}^{h+1}      &= \begin{cases*}
                       \bigOR_{1≤i≤\RelCount}\PosIn{\SetVariable[1]^{i+1}} & for $h=0$, \\
                       \PosIn{\SetVariable[1]^1} \OR \dm[h]\PosIn{\SetVariable[1]^{h+1}}   & for $1≤h≤\RelCount$,
                     \end{cases*} \\
    \Formula[1]_{\invr{\RelSymbol}}^{h+1}   &= \begin{cases*}
                       \bigOR_{1≤i≤\RelCount}\PosIn{\SetVariable[1]^{i+1}} & for $h=0$, \\
                       \PosIn{\SetVariable[1]^1} \OR \bdm[h]\PosIn{\SetVariable[1]^{h+1}}  & for $1≤h≤\RelCount$,
                     \end{cases*} \\
    \Formula[1]_{\GlobalRelSymbol}^{h+1} &= \gdm(\PosIn{\SetVariable[1]^1} \OR … \OR \PosIn{\SetVariable[1]^{\RelCount+1}}) \quad\;\,\text{for $0≤h≤\RelCount$}, \\
    \Formula[1]_\ini         &= \Formula[1]_{\GlobalRelSymbol}^1.
  \end{align*}

  We can characterize $\Encoding(\DIGRAPH[\BitCount][\RelCount])$ over the class $\DIGRAPH[\BitCount+\RelCount][1]$
  by the conjunction of the following $\PiMSO{1}(\MLg)$-\kl{definable} properties,
  using our helper \kl{formulas} $\tuple{\Formula[2]_i}_{1≤i≤\RelCount+1}$
  from the \kl{forward translation}.
  \begin{itemize}
  \item A “port” that corresponds to a relation symbol $\RelSymbol_i$
    may not be associated with any other relation symbol $\RelSymbol_j$,
    nor be \kl{labeled} with predicates $\tuple{\SetConstant_j}_{1≤j≤\BitCount}$.
    \begin{equation*}
      \smashoperator{\bigAND_{1≤i≤\RelCount}} \gbx\Bigl(
      \Formula[2]_{i+1} \,\IMP\;
      \NOT\smashoperator{\bigOR_{1≤j≤\RelCount,\;j≠i}}(\Formula[2]_{j+1}) \+\AND\,
      \NOT\smashoperator{\bigOR_{1≤j≤\BitCount}}(\PosIn{\SetConstant_j})
      \Bigr)
    \end{equation*}
  \item Every “regular \kl{node}” is connected to its $\RelCount$ different “ports”,
    and nothing else.
    The uniqueness of each “$\RelSymbol_i$-port” can be expressed by
    the \kl[class-formula]{$\eclF{\PiMSO{1}(\MLg)}$-formula} $\seeone(\Formula[2]_{i+1})$,
    using the construction from
    \cref{ex:uniqueness} in \cref{sec:example-formulas}.
    \begin{equation*}
      \gbx\Bigl(
      \Formula[2]_1 \,\IMP\;
      \NOT\dm \Formula[2]_1 \+\AND\+
      \smashoperator{\bigAND_{1≤i≤\RelCount}}\seeone(\Formula[2]_{i+1})
      \Bigr)
    \end{equation*}
  \item Similarly,
    each “port” is connected to precisely one “regular \kl{node}”
    and to an arbitrary number of “ports” of the same \kl{relation symbol},
    but not to any other ones.
    \begin{equation*}
      \smashoperator{\bigAND_{1≤i≤\RelCount}} \gbx\Bigl(
      \Formula[2]_{i+1} \,\IMP\;
      \seeone(\Formula[2]_1) \+\AND\,
      \NOT\smashoperator{\bigOR_{1≤j≤\RelCount,\;j≠i}}\dm \Formula[2]_{j+1}
      \Bigr)
    \end{equation*}
  \item Finally,
    the links between “regular \kl{nodes}” and “ports”
    have to be bidirectional:
    for each \kl{edge} from a \kl{node} of one type
    to a \kl{node} of a different type,
    the corresponding inverse \kl{edge} must also exist.
    \begin{equation*}
      \smashoperator{\bigAND_{1≤i≤\RelCount+1}}
      \Forall{\SetVariable} \gbx\bigl( \Formula[2]_i \AND\+ \PosIn{\SetVariable} \:\IMP\: \bx(\NOT \Formula[2]_i \IMP \dm \PosIn{\SetVariable}) \bigr)
    \end{equation*}
    Note that,
    in combination with the previous properties,
    this ensures that we have the same total number of \kl{nodes}
    for each type $i ∈ \range[1]{\RelCount+1}$.
    \qedhere
  \end{itemize}
\end{proof}

\subsection{Getting rid of vertex labels}

Being able to eliminate multiple \kl{edge relations}
at the cost of additional \kl{labeling sets}
(see \cref{prp:multirelational-labeled}),
our natural next step is to encode \kl{labeled} \kl{digraphs} into \kl[labeled]{unlabeled} ones.

\begin{proposition}
  \label{prp:labeled-unlabeled}
  For all $\BitCount≥1$ and $\FormulaSet[1]∈\set{\MLg,\dMLg}$,
  there is a \kl{linear encoding} $\Encoding$ from $\DIGRAPH[\BitCount][1]$ into $\DIGRAPH$
  that allows for \kl{bidirectional translation} within $\FormulaSet[1]$.

  Moreover, $\Encoding(\DIGRAPH[\BitCount][1])$ is $\PiMSO{1}(\MLg)$-\kl{definable} over $\DIGRAPH$.
\end{proposition}
\begin{proof}
  We construct the \kl{linear encoding} $\Encoding$
  that assigns to each $\BitCount$-bit \kl{labeled} \kl{digraph} $\Digraph[1]$
  the (\kl[labeled]{unlabeled}) \kl{digraph} $\Encoding(\Digraph[1])=\Digraph[2]$
  with \kl{domain} $(\set{1}×\NodeSet{\Digraph[1]}) ∪ \range[2]{\BitCount+3}$
  and \kl{edge relation}
  \begin{align*}
    \inp{\RelSymbol}{\Digraph[2]} =
    &\phantom{{}∪\:{}}\bigsetbuilder{\bigtuple{\tuple{1,\Node[1]},\tuple{1,\Node[1]'}}}{\tuple{\Node[1],\Node[1]'}∈\inp{\RelSymbol}{\Digraph[1]}} \\
    &∪\:\bigsetbuilder{\bigtuple{\tuple{1,\Node[1]},\,3}}{\Node[1]∈\NodeSet{\Digraph[1]}} \\
    &∪\:\:\smashoperator{\bigcup_{1≤i≤\BitCount}}\; \bigsetbuilder{\bigtuple{\tuple{1,\Node[1]},\,i+3}}{\Node[1]∈\inp{\SetConstant_i}{\Digraph[1]}} \\[-0.2ex]
    &∪\:\bigsetbuilder{\tuple{i+3,\:i+2}}{1≤i≤\BitCount} \\
    &∪\:\bigsetbuilder{\tuple{i+3,\:2}}{0≤i≤\BitCount}.
  \end{align*}
  The idea is to introduce a gadget that contains a separate \kl{node}
  for each \kl{labeling set} of the original \kl{digraph},
  and then connect the “regular \kl{nodes}” to this gadget
  in a way that corresponds to their respective \kl{labeling}.
  The gadget is easily identifiable because
  it contains the only \kl{node} in the \kl{digraph} that has no outgoing \kl{edge}
  (namely, \kl{node}~$2$).
  We ensure this by connecting all the “regular \kl{nodes}” to \kl{node}~$3$.

  Our \kl{forward translation},
  from $\FormulaSet[1]$ on $\DIGRAPH[\BitCount][1]$ to $\FormulaSet[1]$ on $\Encoding(\DIGRAPH[\BitCount][1])$,
  is given by
  \begin{gather*}
    \begin{aligned}
      \Formula[2]_{\SetConstant_i}   &= \dm \Formula[2]_{i+3} \quad\text{for $1≤i≤\BitCount$}, \\
      \Formula[2]_{\RelSymbol}      &= \dm(\Formula[2]_1 \AND \PosIn{\SetVariable[2]}), \\
      \Formula[2]_{\invr{\RelSymbol}}   &= \bdm \PosIn{\SetVariable[2]}, \\
      \Formula[2]_{\GlobalRelSymbol} &= \gdm(\Formula[2]_1 \AND \PosIn{\SetVariable[2]}),\\
      \Formula[2]_\ini    &= \Formula[2]_{\GlobalRelSymbol}\+,
    \end{aligned} \\[1ex]
    \begin{alignedat}{2}
      \text{where\!}& &    \Formula[2]_1 &= \textstyle \NOT\bigOR_{2≤i≤\BitCount+3}(\Formula[2]_i), \\
                    & &    \Formula[2]_2 &= \bx \False, \\
                    & &    \Formula[2]_3 &= \dm \Formula[2]_2 \AND \bx \Formula[2]_2, \\
                    & & \Formula[2]_{i+3} &= \dm \Formula[2]_2 \AND \dm \Formula[2]_{i+2} \AND \bx(\Formula[2]_2 \OR \Formula[2]_{i+2}) \quad\text{for $1≤i≤\BitCount$}.
    \end{alignedat}
  \end{gather*}

  Our translation in the other direction,
  from $\FormulaSet[1]$ on $\Encoding(\DIGRAPH[\BitCount][1])$ to $\FormulaSet[1]$ on $\DIGRAPH[\BitCount][1]$,
  is given by
  \begin{equation*}
    \begin{aligned}
      \Formula[1]_{\RelSymbol}^h       &= \begin{cases*}
                       \dm \PosIn{\SetVariable[1]^1} \,\OR\, \PosIn{\SetVariable[1]^3} \;\OR{} \\[-1ex]
                       \quad\bigOR_{1≤i≤\BitCount}(\PosIn{\SetConstant_i} \AND \PosIn{\SetVariable[1]^{i+3}}) \hspace{1ex}& for $h=1$, \\
                       \False            & for $h=2$, \\
                       \PosIn{\SetVariable[1]^2}          & for $h=3$, \\
                       \PosIn{\SetVariable[1]^2} \OR \PosIn{\SetVariable[1]^{h-1}} & for $4≤h≤\BitCount+3$,
                     \end{cases*} \\
      \Formula[1]_{\invr{\RelSymbol}}^h    &= \begin{cases*}
                       \bdm \PosIn{\SetVariable[1]^1}                    & for $h=1$, \\
                       \bigOR_{0≤i≤\BitCount}\PosIn{\SetVariable[1]^{i+3}}        & for $h=2$, \\
                       \gdm \PosIn{\SetVariable[1]^1} \OR \PosIn{\SetVariable[1]^{h+1}}           & for $h=3$, \\
                       \gdm(\PosIn{\SetConstant_{h-3}} \AND \PosIn{\SetVariable[1]^1}) \OR \PosIn{\SetVariable[1]^{h+1}} \hspace{-0.5ex}& for $4≤h≤\BitCount+2$, \\
                       \gdm(\PosIn{\SetConstant_{h-3}} \AND \PosIn{\SetVariable[1]^1})          & for $h=\BitCount+3$,
                     \end{cases*} \\
      \Formula[1]_{\GlobalRelSymbol}^h &= \gdm(\PosIn{\SetVariable[1]^1} \OR … \OR \PosIn{\SetVariable[1]^{\BitCount+3}}) \quad\+\text{for $1≤h≤\BitCount+3$}, \\
      \Formula[1]_\ini       &= \Formula[1]_{\GlobalRelSymbol}^1.
    \end{aligned}
  \end{equation*}

  Using the helper \kl{formulas} $\tuple{\Formula[2]_i}_{1≤i≤\BitCount+3}$ from the \kl{forward translation},
  we can characterize $\Encoding(\DIGRAPH[\BitCount][1])$ over $\DIGRAPH$ as
  \begin{equation*}
    \gdm \Formula[2]_1
    \,\AND\, \smashoperator{\bigAND_{2≤i≤\BitCount+3}}\totone(\Formula[2]_i)
    \,\AND\, \gbx(\Formula[2]_1\IMP \dm \Formula[2]_3 \AND \NOT\dm \Formula[2]_2).
  \end{equation*}
  Here, each $\eclF{\PiMSO{1}(\MLg)}$-subformula \,$\totone(\Formula[2]_i)$
  is obtained through
  the singleton construction from 
  \cref{ex:uniqueness} in \cref{sec:example-formulas}.
\end{proof}

\subsection{Getting rid of backward modalities}

For the sake of completeness,
we also describe the encoding
that lets us simulate \kl{backward modalities}
by means of an additional \kl{edge relation}.

\begin{proposition}
  \label{prp:hbg-to-hg}
  There is a \kl{linear encoding} $\Encoding$ from $\DIGRAPH$ into $\DIGRAPH[0][2]$
  that allows for \kl{bidirectional translation} \mbox{between} $\dMLg$ and $\MLg$.

  Moreover, $\Encoding(\DIGRAPH)$ is $\PiMSO{1}(\MLg)$-\kl{definable} over $\DIGRAPH[0][2]$.
\end{proposition}
\begin{proof}
  The encoding is straightforward:
  to each \kl{digraph} $\Digraph[1]$,
  we assign a copy $\Encoding(\Digraph[1]) = \Digraph[2]$
  that is enriched with a second \kl{edge relation},
  which coincides with the inverse of the first.
  Formally,
  $\NodeSet{\Digraph[2]} = \set{1} × \NodeSet{\Digraph[1]}$,
  \begin{align*}
    \inp{\RelSymbol_1}{\Digraph[2]} &= \bigsetbuilder{\bigtuple{\tuple{1,\Node[1]},\tuple{1,\Node[1]'}}}{\tuple{\Node[1],\Node[1]'}∈\inp{\RelSymbol}{\Digraph[1]}},
    \quad \text{and} \\
    \inp{\RelSymbol_2}{\Digraph[2]} &= \bigsetbuilder{\tuple{\Node[2]'\!,\Node[2]}}{\tuple{\Node[2],\Node[2]'}∈\inp{\RelSymbol_1}{\Digraph[2]}}.
  \end{align*}

  With this,
  in order to translate between $\dMLg$ on $\DIGRAPH$ and $\MLg$ on $\Encoding(\DIGRAPH)$,
  we merely have to replace \kl{backward modalities} by $\RelSymbol_2$-\kl{modalities},
  and vice versa.
  Hence, when we fix our \kl{forward translation},
  we choose $\Formula[2]_{\RelSymbol} = \dm[1] \PosIn{\SetVariable[2]}$ and $\Formula[2]_{\invr{\RelSymbol}} = \dm[2] \PosIn{\SetVariable[2]}$,
  and for the \kl{backward translation}
  we set $\Formula[1]_{\RelSymbol_1}^1 = \dm \PosIn{\SetVariable[1]^1}$ and $\Formula[1]_{\RelSymbol_2}^1 = \bdm \PosIn{\SetVariable[1]^1}$\!.

  To define $\Encoding(\DIGRAPH)$ over $\DIGRAPH[0][2]$,
  we can use the following \kl[class-formula]{$\PiMSO{1}(\MLg)$-formula}:
  \begin{equation*}
    \Forall{\SetVariable}\gbx\bigl(\PosIn{\SetVariable} \IMP \bx[1]\dm[2]\PosIn{\SetVariable} \AND \bx[2]\dm[1]\PosIn{\SetVariable}\bigr)
    \qedhere
  \end{equation*}
\end{proof}

\subsection{Getting rid of directed edges}

In order to encode a \kl{digraph} into an \kl{undirected graph},
we proceed in a similar manner to the elimination of multiple \kl{edge relations}
in \cref{prp:multirelational-labeled}.
Using an ad-hoc trick,
we can do this by introducing only one additional \kl{labeling set}.

\begin{proposition}
  \label{prp:digraph-1bitgraph}
  There is a \kl{linear encoding} $\Encoding$ from $\DIGRAPH$ into $\GRAPH[1][1]$
  that allows for \kl{bidirectional translation} \mbox{between} $\dMLg$ and $\MLg$.

  Moreover, $\Encoding(\DIGRAPH)$ is $\PiMSO{1}(\MLg)$-\kl{definable} over $\GRAPH[1][1]$.
\end{proposition}
\begin{proof}
  A suitable choice for $\Encoding$ is to take the function that
  assigns to every \kl{digraph} $\Digraph[1]$
  the \mbox{$1$-bit} \kl{labeled} \kl{undirected graph} $\Encoding(\Digraph[1])=\Graph[2]$
  with \kl{domain} $(\range{3}×\NodeSet{\Digraph[1]}) ∪ \range[4]{6}$,\,
  \kl{labeling set} $\inp{\SetConstant}{\Graph[2]} = \range[4]{6}$,
  and \kl{edge relation} $\inp{\RelSymbol}{\Graph[2]}=\bigsetbuilder{\tuple{\Node[2],\Node[2]'}}{\set{\Node[2],\Node[2]'}∈\UndirectedEdgeSet{\Graph[2]}}$,\,
  where
  \begin{align*}
    \UndirectedEdgeSet{\Graph[2]} =
    &\phantom{{}∪\:{}}\bigsetbuilder{\set{\tuple{1,\Node[1]},\tuple{2,\Node[1]}}}{\Node[1]∈\NodeSet{\Digraph[1]}} \\
    &∪\:\bigsetbuilder{\set{\tuple{1,\Node[1]},\tuple{3,\Node[1]}}}{\Node[1]∈\NodeSet{\Digraph[1]}} \\
    &∪\:\bigsetbuilder{\set{\tuple{2,\Node[1]},\tuple{3,\Node[1]'}}}{\tuple{\Node[1],\Node[1]'}∈\inp{\RelSymbol}{\Digraph[1]}} \\
    &∪\:\bigsetbuilder{\set{\tuple{2,\Node[1]},\,4}}{\Node[1]∈\NodeSet{\Digraph[1]}} \\
    &∪\:\bigsetbuilder{\set{\tuple{3,\Node[1]},\,5}}{\Node[1]∈\NodeSet{\Digraph[1]}} \\
    &∪\:\bigset{\set{5,\,6}}.
  \end{align*}
  The idea is that we connect each original \kl{node} ${\Node[1]∈\NodeSet{\Digraph[1]}}$
  to two new \kl{nodes},
  which represent the “outgoing port” and “incoming port” of $d$,
  and use undirected \kl{edges} between “ports”
  to simulate directed \kl{edges} between “regular \kl{nodes}”.
  In order to distinguish the different types of \kl{nodes},
  we connect them in different ways
  to the additional $\SetConstant$-\kl{labeled} \kl{nodes}.

  Our \kl{forward translation},
  from $\dMLg$ on $\DIGRAPH$ to $\MLg$ on $\Encoding(\DIGRAPH)$,
  is given by
  \begin{gather*}
    \begin{aligned}
      \Formula[2]_{\RelSymbol}      &= \dm\bigl(\Formula[2]_2 \AND \dm\dm(\Formula[2]_1 \AND \PosIn{\SetVariable[2]})\bigr), \\
      \Formula[2]_{\invr{\RelSymbol}}   &= \dm\bigl(\Formula[2]_3 \AND \dm\dm(\Formula[2]_1 \AND \PosIn{\SetVariable[2]})\bigr), \\
      \Formula[2]_{\GlobalRelSymbol} &= \gdm(\Formula[2]_1 \AND \PosIn{\SetVariable[2]}),\\
      \Formula[2]_\ini    &= \Formula[2]_{\GlobalRelSymbol}\+,
    \end{aligned} \\[1ex]
    \begin{alignedat}{2}
      \text{where} \quad \Formula[2]_1 &= \NOT(\Formula[2]_2 \OR … \OR \Formula[2]_6) & &\text{(“regular”)}, \\
                         \Formula[2]_2 &= \dm \Formula[2]_4          & &\text{(“outgoing”)}, \\
                         \Formula[2]_3 &= \NOT \PosIn{\SetConstant} \AND \dm \Formula[2]_5     & &\text{(“incoming”)}, \\
                         \Formula[2]_4 &= \PosIn{\SetConstant} \AND \NOT\dm \PosIn{\SetConstant} \AND \phantom{\NOT}\dm \NOT \PosIn{\SetConstant}, \!\! \\
                         \Formula[2]_5 &= \PosIn{\SetConstant} \AND \phantom{\NOT}\dm \PosIn{\SetConstant} \AND \phantom{\NOT}\dm \NOT \PosIn{\SetConstant}, \!\! \\
                         \Formula[2]_6 &= \PosIn{\SetConstant} \AND \phantom{\NOT}\dm \PosIn{\SetConstant} \AND \NOT\dm \NOT \PosIn{\SetConstant}. \!\!
    \end{alignedat}
  \end{gather*}

  Our \kl{backward translation},
  from $\MLg$ on $\Encoding(\DIGRAPH)$ to $\dMLg$ on $\DIGRAPH$,
  is given by
  \begin{align*}
    \Formula[1]_{\SetConstant}^h       &= \begin{cases*}
                       \False   & for $1≤h≤3$, \\
                       \True   & for $4≤h≤6$,
                   \end{cases*} \displaybreak[0]\\
    \Formula[1]_{\RelSymbol}^h       &= \begin{cases*}
                     \PosIn{\SetVariable[1]^2} \OR \PosIn{\SetVariable[1]^3}             & for $h=1$, \\
                     \PosIn{\SetVariable[1]^1} \OR \dm \PosIn{\SetVariable[1]^3} \OR \PosIn{\SetVariable[1]^4}   \hspace{0.4ex}& for $h=2$, \\
                     \PosIn{\SetVariable[1]^1} \OR \bdm \PosIn{\SetVariable[1]^2} \OR \PosIn{\SetVariable[1]^5}  & for $h=3$, \\
                     \gdm \PosIn{\SetVariable[1]^2}              & for $h=4$, \\
                     \gdm \PosIn{\SetVariable[1]^3} \OR \PosIn{\SetVariable[1]^6}        & for $h=5$, \\
                     \PosIn{\SetVariable[1]^5}                   & for $h=6$,
                   \end{cases*} \displaybreak[0]\\
    \Formula[1]_{\GlobalRelSymbol}^h &= \gdm(\PosIn{\SetVariable[1]^1} \OR … \OR \PosIn{\SetVariable[1]^6}) \quad\text{for $1≤h≤6$}, \\
    \Formula[1]_\ini       &= \Formula[1]_{\GlobalRelSymbol}^1.
  \end{align*}

  We can define $\Encoding(\DIGRAPH)$ over $\GRAPH[1][1]$
  with the following \kl[class-formula]{$\eclF{\PiMSO{1}(\MLg)}$-formula}.
  It makes use of
  our helper \kl{formulas} $\tuple{\Formula[2]_i}_{1≤i≤6}$ from the \kl{forward translation}
  and the constructions \,$\seeone(\Formula[2]_i)$\, and \,$\totone(\Formula[2]_i)$\,
  from \cref{ex:uniqueness} in \cref{sec:example-formulas}.
  \begin{equation*}
    \begin{aligned}
      &\smashoperator{\bigAND_{4≤i≤6}} \totone(\Formula[2]_i) \,\,\AND {} \\
      &\gbx\bigl( \Formula[2]_2 \,\IMP\, \seeone(\Formula[2]_1) \AND \bx(\Formula[2]_1 \OR \Formula[2]_3 \OR \Formula[2]_4) \bigr) \,\,\AND {} \\
      &\gbx\bigl( \Formula[2]_3 \,\IMP\, \seeone(\Formula[2]_1) \AND \bx(\Formula[2]_1 \OR \Formula[2]_2 \OR \Formula[2]_5) \bigr) \,\,\AND {} \\
      &\gbx\bigl( \Formula[2]_1 \,\IMP\, \seeone(\Formula[2]_2) \AND \seeone(\Formula[2]_3) \AND \bx(\Formula[2]_2 \OR \Formula[2]_3) \bigr)
    \end{aligned}
  \end{equation*}
  The first line states that the three $\SetConstant$-\kl{labeled} \kl{nodes} are unique,
  which forces $5$ and $6$ to be connected.
  The remaining lines ensure
  that each “port” is connected to exactly one “regular \kl{node}”,
  and, conversely,
  that every “regular \kl{node}” is linked to precisely
  one “outgoing port” and one “incoming port”.
  As a consequence,
  the number of “regular \kl{nodes}” must be the same
  as the number of “ports” of each type.
  Furthermore,
  the \kl{formula} restricts the types of \kl{neighbors} each \kl{node} can have,
  while the usage of the helper \kl{formulas} $\Formula[2]_2$ and $\Formula[2]_3$
  makes sure that the required connections to the $\SetConstant$-\kl{labeled} \kl{nodes}
  are established.
  Finally,
  the fact that $\Formula[2]_1$ characterizes
  the “regular \kl{nodes}” as the “remaining ones”
  guarantees that there are no unaccounted-for \kl{nodes}.
\end{proof}

\subsection{Getting rid of global modalities}

Our last encoding function lets us simulate \kl{global modalities}
by inserting a new \kl{node}
that is bidirectionally connected to all the “regular \kl{nodes}”.

\begin{proposition}
  \label{prp:digraph-pdigraph}
  There is a \kl{linear encoding} $\Encoding$ from $\DIGRAPH$ into $\pDIGRAPH$
  that allows for \kl{bidirectional translation} \mbox{between} $\MLg$ and $\ML$.

  Furthermore, $\Encoding$ can be easily adapted into
  a \kl{linear encoding} $\Encoding'$ from $\DIGRAPH$ into $\DIGRAPH$
  that satisfies the following \kl{figurative inclusions},
  for arbitrary $\Level≥2$:
  \begin{align*}
    \semF{\SigmaMSO{\Level}(\MLg)}[\DIGRAPH] &\,\figsubeq{\Encoding'}\;\, \semF{\GBX\SigmaMSO{\Level}(\ML)}[\DIGRAPH] \,, \\
    \semF{\PiMSO{\Level}(\MLg)}[\DIGRAPH] &\,\figeq{\Encoding'}\;\,    \semF{\GBX\PiMSO{\Level}(\ML)}[\DIGRAPH] \,.
    \qedhere
  \end{align*}
\end{proposition}
\begin{proof}
  We choose $\Encoding$ to be the \kl{linear encoding}
  that maps each \kl{digraph} $\Digraph[1]$
  to the \kl{pointed digraph} $\Encoding(\Digraph[1])=\PointedDigraph[2]$
  with \kl{domain} $(\set{1}×\NodeSet{\Digraph[1]}) ∪ \range[2]{3}$,\,
  position $\inp{\PosSymbol}{\PointedDigraph[2]} = 2$,
  and \kl[edge relation]{edge~relation}
  \begin{align*}
    \inp{\RelSymbol}{\PointedDigraph[2]} =
    &\phantom{{}∪\:{}}\bigsetbuilder{\bigtuple{\tuple{1,\Node[1]},\tuple{1,\Node[1]'}}}{\tuple{\Node[1],\Node[1]'}∈\inp{\RelSymbol}{\Digraph[1]}} \\
    &∪\:\bigsetbuilder{\bigtuple{\tuple{1,\Node[1]},\,2}}{\Node[1]∈\NodeSet{\Digraph[1]}} \\
    &∪\:\bigsetbuilder{\bigtuple{2,\,\tuple{1,\Node[1]}}}{\Node[1]∈\NodeSet{\Digraph[1]}} \\
    &∪\:\bigset{\tuple{2,\,3}}.
  \end{align*}
  One can distinguish \kl{node}~$2$ from the others
  because it is connected to~$3$,
  which is the only \kl{node} without any outgoing \kl{edge}.

  Our \kl{forward translation},
  from $\MLg$ on $\DIGRAPH$ to $\ML$ on $\Encoding(\DIGRAPH)$,
  is given by
  \begin{gather*}
    \begin{aligned}
      \Formula[2]_{\RelSymbol}      &= \dm(\Formula[2]_1 \AND \PosIn{\SetVariable[2]}), \\
      \Formula[2]_{\GlobalRelSymbol} &= \dm(\Formula[2]_2 \AND\+ \dm(\Formula[2]_1 \AND \PosIn{\SetVariable[2]})),\\
      \Formula[2]_\ini    &= \dm(\Formula[2]_1 \AND \PosIn{\SetVariable[2]}),
    \end{aligned} \\[1ex]
    \begin{gathered}
      \text{where}\quad \Formula[2]_1 = \dm\dm\bx \False \quad\text{and}\quad \Formula[2]_2 = \dm\bx \False.
    \end{gathered}
  \end{gather*}

  Our \kl{backward translation},
  from $\ML$ on $\Encoding(\DIGRAPH)$ to $\MLg$ on $\DIGRAPH$,
  is given by
  \begin{align*}
    \Formula[1]_{\RelSymbol}^h &= \begin{cases*}
               \dm \PosIn{\SetVariable[1]^1} \OR \PosIn{\SetVariable[1]^2}  & for $h=1$, \\
               \gdm \PosIn{\SetVariable[1]^1} \OR \PosIn{\SetVariable[1]^3} & for $h=2$, \\
               \False              & for $h=3$,
             \end{cases*} \\
    \Formula[1]_\ini &= \gdm \PosIn{\SetVariable[1]^2}.
  \end{align*}

  Turning to the second claim of the \lcnamecref{prp:digraph-pdigraph},
  we obtain $\Encoding'(\Digraph[1])$ by simply removing the position marker from $\Encoding(\Digraph[1])$,
  i.e., for every \kl{digraph} $\Digraph[1]$,\,
  $\Encoding'(\Digraph[1])$ is such that $\ver{\Encoding'(\Digraph[1])}{\PosSymbol}{2} = \Encoding(\Digraph[1])$.

  For the \kl[forward inclusion]{forward} \kl{figurative inclusions},
  let $\FormulaClassOperator ∈ \set{\SigmaMSO{\Level},\, \PiMSO{\Level}}$,\, for some arbitrary $\Level≥0$.
  By applying \cref{lem:translation-inclusion}~\ref{itm:forward-inclusion}
  on $\Encoding$, we get that
  for every $\FormulaClassOperator\+(\MLg)$-\kl{sentence} $\Formula[1]$ over $\set{\RelSymbol}$,
  there is a $\FormulaClassOperator\+(\ML)$-\kl{sentence} $\Formula[2]_{\Formula[1]}$ over $\set{\PosSymbol,\RelSymbol}$
  such that, for all $\Digraph[1]∈\DIGRAPH$,
  \begin{alignat*}{1}
    \Digraph[1] \,\Models\, \Formula[1] \quad&\text{iff}\quad \ver{\Encoding'(\Digraph[1])}{\PosSymbol}{2} \,\Models\, \Formula[2]_{\Formula[1]}\+, \\
                     &\text{iff}\quad \Encoding'(\Digraph[1]) \,\Models\, \gbx(\Formula[2]_2 \IMP \Formula[2]_{\Formula[1]}).
  \end{alignat*}
  Hence,\,
  $\semF{\+\FormulaClassOperator\+(\MLg)}[\DIGRAPH] \figsubeq{\Encoding'} \semF{\GBX\,\FormulaClassOperator\+(\ML)}[\DIGRAPH]\+$.

  For the \kl[backward inclusion]{backward} \kl{figurative inclusion},
  we require that $\Level≥2$.
  Slightly adapting the proof of
  \cref{lem:translation-inclusion}~\ref{itm:backward-inclusion}
  to discard the part where we make use of the \kl{formula} $\Formula[1]_\ini$ from
  \cref{def:backward-translation}~\ref{itm:bwd-initial}
  (incidentally allowing us to merge the two consecutive induction proofs),
  it is easy to show the following:
  Given $h∈\range{3}$ and any $\PiMSO{\Level}(\ML)$-\kl{sentence} $\Formula[2]$ over $\set{\PosSymbol,\RelSymbol}$,
  we can construct a $\PiMSO{\Level}(\MLg)$-\kl{sentence} $\Formula[1]_{\Formula[2]}^h$ over $\set{\PosSymbol,\RelSymbol}$
  such that, for all $\Digraph[1]∈\DIGRAPH$ and $\Node[1]∈\NodeSet{\Digraph[1]}$,
  \begin{align*}
    &\ver{\Digraph[1]}{\PosSymbol}{\Node[1]} \,\Models\, \Formula[1]_{\Formula[2]}^h \quad\text{iff}\quad \ver{\Encoding'(\Digraph[1])}{\PosSymbol}{\Node[2]} \,\Models\, \Formula[2], \\[0.5ex]
    &\quad\text{where $\Node[2]$ is $\tuple{h,\Node[1]}$ if $h=1$, and $h$ otherwise.}
  \end{align*}
  This immediately gives us a way of translating $\gbx \Formula[2]$:
  \begin{equation*}
    \Digraph[1] \,\Models\, \gbx(\Formula[1]_{\Formula[2]}^1 \AND \Formula[1]_{\Formula[2]}^2 \AND \Formula[1]_{\Formula[2]}^3) \quad\text{iff}\quad \Encoding'(\Digraph[1]) \,\Models\, \gbx \Formula[2].
  \end{equation*}
  The left-hand side \kl{sentence} can be transformed into \kl{prenex normal form}
  by simulating the \kl{global box} with a universal \kl{set quantifier}.
  Checking that a given set is \emph{not} a singleton can be done in $\SigmaMSO{1}(\MLg)$,
  since the negation is $\PiMSO{1}(\MLg)$-expressible
  (see \cref{ex:uniqueness} in \cref{sec:example-formulas}).
  Thus,
  the given \kl{formula} is \kl[device equivalent]{equivalent} to a \kl[class-formula]{$\PiMSO{\Level}(\MLg)$-formula},
  and we obtain that\,
  $\vphantom{\Bigl(}
  \semF{\PiMSO{\Level}(\MLg)}[\DIGRAPH] \figsupeq{\Encoding'} \semF{\GBX\PiMSO{\Level}(\ML)}[\DIGRAPH]\+$.
\end{proof}

%% file: fig/hierarchy-delta.tex
\begin{tikzpicture}[semithick,node distance=13ex]
  \node (b1) {$\semF{\BC\SigmaMSO{\Level}(\FormulaSet)}[\StructClass]$};
  \node (s1) [below left of=b1] {$\semF{\SigmaMSO{\Level}(\FormulaSet)}[\StructClass]$};
  \node (p1) [below right of=b1] {$\semF{\PiMSO{\Level}(\FormulaSet)}[\StructClass]$};
  \node (d2) [above of=b1,yshift=-4ex] {$\semF{\DeltaMSO{\Level+1}(\FormulaSet)}[\StructClass]$};
  \node (s2) [above left of=d2] {$\semF{\SigmaMSO{\Level+1}(\FormulaSet)}[\StructClass]$};
  \node (p2) [above right of=d2] {$\semF{\PiMSO{\Level+1}(\FormulaSet)}[\StructClass]$};
  \path[every node/.style={sloped,allow upside down,auto=false,inner sep=0.4ex}]
    (s2) edge node [anchor=south] {$⊈$}
              node [anchor=north] {$⊉$} (p2)
    (s2) edge node [anchor=north] {$\supsetneqq$} (d2)
    (d2) edge node [anchor=north] {$\subsetneqq$} (p2)
    (b1) edge node [anchor=south] {$⊆$}
              node [anchor=north] {$(⊉)$} (d2)
    (s1) edge node [anchor=south] {$\subsetneqq$} (b1)
    (b1) edge node [anchor=south] {$\supsetneqq$} (p1)
    (s1) edge node [anchor=south] {$⊈$}
              node [anchor=north] {$⊉$} (p1);
\end{tikzpicture}

%% file: fig/hierarchy-nodelta.tex
\begin{tikzpicture}[semithick,node distance=15ex]
  \node (s1) {$\semF{\GBX\SigmaMSO{\Level}(\FormulaSet)}[\StructClass]$};
  \node (p1) [right of=s1,xshift=5ex] {$\semF{\GBX\PiMSO{\Level}(\FormulaSet)}[\StructClass]$};
  \node (s2) [above of=s1] {$\semF{\GBX\SigmaMSO{\Level+1}(\FormulaSet)}[\StructClass]$};
  \node (p2) [above of=p1] {$\semF{\GBX\PiMSO{\Level+1}(\FormulaSet)}[\StructClass]$};
  \newlength{\labelshift}\setlength{\labelshift}{9ex}
  \path[every node/.style={sloped,allow upside down,auto=false,inner sep=0.4ex}]
    (s2) edge node [anchor=south] {$⊈$} (p2)
    (s1) edge node [anchor=south] {$\subsetneqq$} (s2)
    (s2) edge node [anchor=north] {\hspace{-1.1\labelshift} $\supsetneqq$} (p1)
    (s1) edge node [anchor=north] {\hspace{\labelshift} $⊆$} (p2)
    (p1) edge node [anchor=north] {$\subsetneqq$}(p2)
    (s1) edge node [anchor=south] {$⊈$} (p1);
\end{tikzpicture}

%% file: tex/perspectives.tex
\chapter{Perspectives}
\label{ch:perspectives}

Coming to the end of this thesis,
we discuss some ideas for future research.
They can be separated into two categories:
rather focused questions
that directly follow up on the results presented here,
and broader questions
that aim at the bigger picture.

\section{Focused questions}

Let us start with the topics directly related to this work,
following roughly the order of discussion in the document.

\subsection{Is there an alternation level that covers first-order logic?}

In \cref{ch:local},
we have related the classes of \kl{digraph languages}
\kl[global recognizable]{recognizable} by our three flavors of $\LDAg$'s
to those \kl{definable} in $\MSOL$, $\EMSOL$ and $\FOL$.
As shown in \cref{fig:venn-diagram} on page \pageref{fig:venn-diagram},
$\ALDAg$'s cover $\FOL$
(as a direct consequence of their equivalence to $\MSOL$),
whereas $\NLDAg$'s do not.
It is also easy to see
that every $\NLDAg$ is \kl[device equivalent]{equivalent}
to an $\NLDAg$ of \kl[automaton length]{length} $1$,
since each \kl{node} can simply guess
all of its nondeterministic transitions at once,
and then verify in one round of communication
that its own choices are consistent
with those of its \kl[incoming neighbors]{neighbors}.
Furthermore,
we know from \cref{ch:alternation}
(\cref{thm:separation-R} on page \pageref{thm:separation-R}) that
$\ALDAg$'s of \kl[automaton length]{length} $\Level + 1$
are strictly more expressive than
$\ALDAg$'s of \kl[automaton length]{length} $\Level$
(or equivalently,
that the \kl{set quantifier} alternation hierarchy of $\MSO(\bMLg)$
is strict over \kl{digraphs}).
This means that \kl[automaton length]{length}-restricted $\ALDAg$'s
form an infinite hierarchy of \kl[\LDAg]{automata} classes
between $\NLDAg$'s and $\ALDAg$'s.

Against this backdrop,
a natural question is whether there exists a bound $\Level$
such that $\ALDAg$'s of \kl[automaton length]{length} $\Level$
can \kl[global recognize]{recognize}
all \kl[digraph languages]{languages}
\kl{definable} in $\FOL$ on arbitrary \kl{digraphs}.
Note that this would also imply
that $\ALDAg$'s of \kl[automaton length]{length} $\Level + 1$
fully cover~$\EMSOL$.

When restricted to \kl{digraphs} of bounded degree,
the answer is positive.
This can be seen using Hanf's locality theorem,
which basically states that on \kl{digraphs} of bounded degree,
every \kl[class-formula]{$\FOL$-formula}
is \kl[device equivalent]{equivalent} to
a Boolean combination of conditions of the form
“$r$-sphere $\Digraph[2]$ occurs at least $n$ times”,
where an $r$-sphere is a \kl{pointed digraph}
that represents the $r$-neighborhood of its \kl{distinguished node}
(see, e.g., \cite[Thm~4.1]{DBLP:conf/dimacs/Thomas96}
or \cite[Thm~4.24]{DBLP:books/sp/Libkin04}).
Based on this characterization,
it is relatively easy to show
that any \kl[class-formula]{$\FOL$-formula}
can be translated to an $\ALDAg$ of \kl[automaton length]{length}~$3$.
However,
for \kl{digraphs} of unbounded degree,
the author does not know the answer to the above question.

\subsection{Does asynchrony entail quasi-acyclicity?}

As already mentioned in \cref{sec:result},
we make crucial use of \kl{quasi-acyclicity}
to prove the \kl[device equivalence]{equivalence}
of $\aQDA$'s and the \kl{backward $\mu$-fragment}
(i.e., \cref{thm:main-result} on page~\pageref{thm:main-result}).
It is however open
whether we really need to impose this condition
on our \kl{asynchronous automata}
in order to be able to convert them into \kl{formulas}
of $\SigmaMu{1}(\bML)$.
\kl[asynchronous automaton]{Asynchrony} is a very strong requirement,
and it might well be the case
that every \kl{asynchronous automaton} is in fact
\kl[device equivalent]{equivalent} to a \kl{quasi-acyclic} one.
Moreover,
if this assumption turned out to be true,
it would be interesting to know
if it extends to \kl{lossless-asynchronous automata}.

\subsection{Is asynchrony decidable?}

Another natural question
concerning \kl[asynchronous automaton]{asynchrony}
is whether there exists an algorithm
that decides if a given \kl{distributed automaton}
is \kl[asynchronous automaton]{asynchronous},
or alternatively,
if it is \kl[lossless-asynchronous automaton]{lossless-asynchronous}.
Even though we can effectively translate
from \kl{quasi-acyclic}
\kl[asynchronous automata]{(lossless-)asynchronous automata}
to the \kl{backward $\mu$-fragment}
(see \cref{prp:automata-to-logic}),
our translation procedure relies on the guarantee
that the given \kl[distributed automaton]{automaton}
is indeed an $\laQDA$.
From a practical perspective,
it would be advantageous
if the procedure could also check that its input is valid.
While \kl{quasi-acyclicity} can be easily verified,
\kl[asynchronous automaton]{(lossless-)asynchrony}
seems to present a more challenging problem.

\subsection{Are forgetful automata useful as tree automata?}

In \cref{ch:emptiness},
we have seen that
\kl{forgetful distributed automata}
are strictly more expressive than
classical \kl[pointed ordered ditree]{tree} automata
on \kl{ordered ditrees}
(\cref{prp:inclusion-tree-automata} on page \pageref{prp:inclusion-tree-automata}).
Moreover,
their \kl{emptiness problem} is decidable on arbitrary \kl{digraphs}
(\cref{thm:forgetful-emptiness} on page \pageref{thm:forgetful-emptiness}),
and since all \kl{distributed automata} satisfy a tree-model property
(see \cref{lem:tree-model-property} on page \pageref{lem:tree-model-property}),
it is straightforward to adapt our decision procedure
to the special case of \kl{ordered ditrees}.

This begs the question
whether \kl{forgetful automata}
could be of use in typical application areas
of \kl[pointed ordered ditree]{tree} automata,
such as program verification and processing of \textsc{xml}-like data.
A first step towards an answer would be
to investigate their closure properties
(which are probably not as nice as
those of \kl[pointed ordered ditree]{tree} automata)
and to precisely analyze the complexities of their decision problems.
Indeed,
the $\LOGSPACE$ complexity of the \kl{emptiness problem}
(stated in \cref{thm:forgetful-emptiness})
has to be revised for the case of \kl{ordered ditrees}
because \kl{distributed automata}
that operate exclusively on such \kl{structures}
do not have to deal with sets of \kl{states}
and thus can be represented more compactly;
this leads to higher computational complexity.

\subsection{How powerful are quasi-acyclic automata on dipaths?}

We have presented two different constructions
to prove the undecidability
of the \kl{emptiness problem} for \kl{distributed automata}.
The first
(\cref{thm:dipath-emptiness} on page \pageref{thm:dipath-emptiness})
uses the idea of exchanging space and time
to simulate a Turing machine by a \kl{distributed automaton}
that runs on a \kl{dipath}.
This simulation shows
that even the \kl{dipath-emptiness problem} is undecidable,
but it works only if the \kl{state} diagram
of the simulating \kl[distributed automaton]{automaton}
may contain cycles.
Our second approach
(\cref{thm:emptiness-quasi-acyclic} on page \pageref{thm:emptiness-quasi-acyclic})
shows that also for \kl{quasi-acyclic automata}
the \kl{emptiness problem} is undecidable,
but the construction is much more technical
and does not work if we restrict ourselves to \kl{dipaths}.

It is part of ongoing work to establish
a precise characterization of \kl{quasi-acyclic automata} on \kl{dipaths}
in terms of counter machines.
As a corollary,
this will yield a stronger undecidability result
that supersedes the two previous ones.

\section{Broader questions}

To conclude,
let us expand our focus
by suggesting possible extensions
and asking how the present work fits into the wider landscape
of \kl[digraph]{graph} automata and distributed computing.

\subsection{What about distributed automata on infinite digraphs?}

Although in this thesis
we have considered only finite \kl{structures},
this restriction is by no means necessary;
\kl{distributed automata} could also run on infinite \kl{digraphs},
and this would not even require changing their definition.
It is straightforward to see that
the \kl[device equivalence]{equivalence} of $\ALDAg$'s and $\MSOL$
established in \cref{ch:local}
immediately extends to the infinite setting
(simply by verifying that the given proofs remain applicable).
However,
this is not the case for all the results presented here.
In particular in \cref{ch:nonlocal},
we have relied on the fact that our \kl{digraphs} are finite
to prove the \kl[device equivalence]{equivalence} of
\kl{quasi-acyclic} \kl{asynchronous automata}
and the \kl[backward $\mu$-fragment]{backward \mbox{$\mu$-fragment}}
(see the proof of \cref{prp:logic-to-automata}
on page~\pageref{prp:logic-to-automata}).
It seems that a more powerful \kl{acceptance condition}
would be required
in order to get a corresponding
\kl[device equivalence]{equivalence} on infinite \kl{digraphs}.

For future research on \kl{distributed automata},
it would be worthwhile to
systematically consider both the finite and the infinite case.

\subsection{What is the overlap with cellular automata?}
\label{sec:cellular-automata}

Obviously,
\kl{distributed automata} are closely related to cellular automata.
The only noteworthy difference is that
\kl{distributed automata} can operate on arbitrary \kl{digraphs},
whereas cellular automata are usually confined to regular \kl{structures},
such as (doubly linked) \kl{grids} or \kl{dipaths}.
That is,
if we restrict ourselves to the appropriate classes of \kl{digraphs},
then the two models are exactly the same.
Furthermore,
there is a branch of research concerned with
cellular automata as \kl[pointed-digraph language]{language} acceptors
(see, e.g.,
\mbox{\cite[\S~6.5]{DBLP:series/sci/Kutrib08}} or
\cite{DBLP:reference/nc/Terrier12}).
In order not to “reinvent the wheel”,
it is thus important to relate
questions arising in the study of \kl{distributed automata}
to the existing body of knowledge in cellular automata theory.

An example where this was not done thoroughly enough
can be found in \cref{sec:space-time}.
As the author has been recently informed by N. Bacquey
(on \osf{\printdate{2017-10-20}}),
the idea of exchanging space and time
is well-known within the community of cellular automata.
It is for instance documented in
\cite[Fig.~9]{DBLP:reference/nc/Terrier12},
where it is employed to simulate
a real-time one-dimensional two-way cellular automaton
by a corresponding one-way automaton.
Although the presentation and purpose
differ considerably from those in the present work,
the technical construction is essentially the same.

\subsection{Can we characterize more powerful models?}

As mentioned at the beginning of \cref{ch:introduction},
the original motivation for this thesis
was to work toward
a descriptive complexity theory for distributed computing.
By focusing on \kl{distributed automata},
we were able to make some progress in that direction,
but the main challenge remains
to establish logical characterizations
of stronger models of computation,
powerful enough
to cover the kinds of algorithms
usually considered in distributed computing.
In order to be of practical interest,
such a characterization should be
in terms of \emph{finite} \kl{formulas},
just like the one provided by Fagin's theorem
for nondeterministic polynomial-time Turing machines.

There are several ideas “in the air”
on how one might characterize
distributed finite-state machines
equipped with unique identifiers,
or even distributed Turing machines
subject to certain time and space constraints.
As of the time of writing,
the author is not aware of any fully developed solution,
but new results should be expected in the next few years.

%% file: thesis.bbl
\newcommand{\etalchar}[1]{$^{#1}$}
\begin{thebibliography}{CDG{\etalchar{+}}08}
\providecommand{\bibitemdeclare}[2]{}
\providecommand{\surnamestart}{}
\providecommand{\surnameend}{}
\providecommand{\urlprefix}{Available at }
\providecommand{\url}[1]{\texttt{#1}}
\providecommand{\href}[2]{\texttt{#2}}
\providecommand{\urlalt}[2]{\href{#1}{#2}}
\providecommand{\doi}[1]{doi:\urlalt{http://dx.doi.org/#1}{#1}}
\providecommand{\bibinfo}[2]{#2}

\bibitemdeclare{article}{DBLP:journals/jacm/AlurM09}
\bibitem[AM09]{DBLP:journals/jacm/AlurM09}
\bibinfo{author}{Rajeev \surnamestart Alur\surnameend} \&
  \bibinfo{author}{P.~\surnamestart Madhusudan\surnameend}
  (\bibinfo{year}{2009}): \emph{\bibinfo{title}{Adding nesting structure to
  words}}.
\newblock {\slshape \bibinfo{journal}{J. {ACM}}}
  \bibinfo{volume}{56}(\bibinfo{number}{3}), pp. \bibinfo{pages}{16:1--16:43},
  \doi{10.1145/1516512.1516518}.

\bibitemdeclare{incollection}{BlackburnB07}
\bibitem[BB07]{BlackburnB07}
\bibinfo{author}{Patrick \surnamestart Blackburn\surnameend} \&
  \bibinfo{author}{Johan \surnamestart van Benthem\surnameend}
  (\bibinfo{year}{2007}): \emph{\bibinfo{title}{Modal logic: a semantic
  perspective}}.
\newblock In \bibinfo{editor}{Patrick \surnamestart Blackburn\surnameend},
  \bibinfo{editor}{Johan \surnamestart van Benthem\surnameend} \&
  \bibinfo{editor}{Frank \surnamestart Wolter\surnameend}, editors: {\slshape
  \bibinfo{booktitle}{Handbook of Modal Logic}}, {\slshape
  \bibinfo{series}{Studies in Logic and Practical
  Reasoning}}~\bibinfo{volume}{3}, \bibinfo{publisher}{Elsevier}, pp.
  \bibinfo{pages}{1--84}, \doi{10.1016/S1570-2464(07)80004-8}.

\bibitemdeclare{inproceedings}{DBLP:conf/stacs/BalliuDFO17}
\bibitem[BDFO17]{DBLP:conf/stacs/BalliuDFO17}
\bibinfo{author}{Alkida \surnamestart Balliu\surnameend},
  \bibinfo{author}{Gianlorenzo \surnamestart D'Angelo\surnameend},
  \bibinfo{author}{Pierre \surnamestart Fraigniaud\surnameend} \&
  \bibinfo{author}{Dennis \surnamestart Olivetti\surnameend}
  (\bibinfo{year}{2017}): \emph{\bibinfo{title}{What Can Be Verified Locally?}}
\newblock In \bibinfo{editor}{Heribert \surnamestart Vollmer\surnameend} \&
  \bibinfo{editor}{Brigitte \surnamestart Vall{\'{e}}e\surnameend}, editors:
  {\slshape \bibinfo{booktitle}{34th Symposium on Theoretical Aspects of
  Computer Science, {STACS} 2017, March 8-11, 2017, Hannover, Germany}},
  {\slshape \bibinfo{series}{LIPIcs}}~\bibinfo{volume}{66},
  \bibinfo{publisher}{Schloss Dagstuhl - Leibniz-Zentrum fuer Informatik}, pp.
  \bibinfo{pages}{8:1--8:13}, \doi{10.4230/LIPIcs.STACS.2017.8}.

\bibitemdeclare{book}{Benthem83}
\bibitem[Ben83]{Benthem83}
\bibinfo{author}{Johan \surnamestart van Benthem\surnameend}
  (\bibinfo{year}{1983}): \emph{\bibinfo{title}{Modal Logic and Classical
  Logic}}.
\newblock \bibinfo{series}{Indices: Monographs in Philosophical Logic and
  Formal Linguistics, Vol 3}, \bibinfo{publisher}{Bibliopolis}.

\bibitemdeclare{book}{BlackburnRV02}
\bibitem[BRV02]{BlackburnRV02}
\bibinfo{author}{Patrick \surnamestart Blackburn\surnameend},
  \bibinfo{author}{Maarten \surnamestart de~Rijke\surnameend} \&
  \bibinfo{author}{Yde \surnamestart Venema\surnameend} (\bibinfo{year}{2002}):
  \emph{\bibinfo{title}{Modal logic}}.
\newblock {\slshape \bibinfo{series}{Cambridge Tracts in Theoretical Computer
  Science}}~\bibinfo{volume}{53}, \bibinfo{publisher}{Cambridge University
  Press, Cambridge}, \doi{10.1017/CBO9781107050884}.

\bibitemdeclare{incollection}{BradfieldS07}
\bibitem[BS07]{BradfieldS07}
\bibinfo{author}{Julian \surnamestart Bradfield\surnameend} \&
  \bibinfo{author}{Colin \surnamestart Stirling\surnameend}
  (\bibinfo{year}{2007}): \emph{\bibinfo{title}{Modal mu-calculi}}.
\newblock In \bibinfo{editor}{Patrick \surnamestart Blackburn\surnameend},
  \bibinfo{editor}{Johan \surnamestart van Benthem\surnameend} \&
  \bibinfo{editor}{Frank \surnamestart Wolter\surnameend}, editors: {\slshape
  \bibinfo{booktitle}{Handbook of Modal Logic}}, {\slshape
  \bibinfo{series}{Studies in Logic and Practical
  Reasoning}}~\bibinfo{volume}{3}, \bibinfo{publisher}{Elsevier}, pp.
  \bibinfo{pages}{721--756}, \doi{10.1016/S1570-2464(07)80015-2}.

\bibitemdeclare{article}{Buchi60}
\bibitem[B{\"{u}}c60]{Buchi60}
\bibinfo{author}{J.~Richard \surnamestart B{\"{u}}chi\surnameend}
  (\bibinfo{year}{1960}): \emph{\bibinfo{title}{Weak Second-Order Arithmetic
  and Finite Automata}}.
\newblock {\slshape \bibinfo{journal}{Zeitschrift f\"ur Mathematische Logik und
  Grundlagen der Mathematik}} \bibinfo{volume}{6}, pp. \bibinfo{pages}{66--92},
  \doi{10.1002/malq.19600060105}.

\bibitemdeclare{article}{DBLP:journals/jphil/Cate06}
\bibitem[Cat06]{DBLP:journals/jphil/Cate06}
\bibinfo{author}{Balder \surnamestart ten Cate\surnameend}
  (\bibinfo{year}{2006}): \emph{\bibinfo{title}{Expressivity of Second Order
  Propositional Modal Logic}}.
\newblock {\slshape \bibinfo{journal}{J. Philosophical Logic}}
  \bibinfo{volume}{35}(\bibinfo{number}{2}), pp. \bibinfo{pages}{209--223},
  \doi{10.1007/s10992-005-9012-9}.

\bibitemdeclare{misc}{TATA2008}
\bibitem[CDG{\etalchar{+}}08]{TATA2008}
\bibinfo{author}{Hubert \surnamestart Comon\surnameend}, \bibinfo{author}{Max
  \surnamestart Dauchet\surnameend}, \bibinfo{author}{Rémi \surnamestart
  Gilleron\surnameend}, \bibinfo{author}{Florent \surnamestart
  Jacquemard\surnameend}, \bibinfo{author}{Christof \surnamestart
  Löding\surnameend}, \bibinfo{author}{Denis \surnamestart Lugiez\surnameend},
  \bibinfo{author}{Sophie \surnamestart Tison\surnameend} \&
  \bibinfo{author}{Marc \surnamestart Tommasi\surnameend}
  (\bibinfo{year}{2008}): \emph{\bibinfo{title}{Tree Automata Techniques and
  Applications}}.
\newblock \urlprefix\url{http://tata.gforge.inria.fr}.
\newblock \bibinfo{note}{Release November 18, 2008}.

\bibitemdeclare{book}{CourcelleE12}
\bibitem[CE12]{CourcelleE12}
\bibinfo{author}{Bruno \surnamestart Courcelle\surnameend} \&
  \bibinfo{author}{Joost \surnamestart Engelfriet\surnameend}
  (\bibinfo{year}{2012}): \emph{\bibinfo{title}{Graph Structure and Monadic
  Second-Order Logic: {A} Language-Theoretic Approach}}.
\newblock {\slshape \bibinfo{series}{Encyclopedia of mathematics and its
  applications}} \bibinfo{volume}{138}, \bibinfo{publisher}{Cambridge
  University Press}, \doi{10.1017/CBO9780511977619}.
\newblock \urlprefix\url{https://hal.archives-ouvertes.fr/hal-00646514}.

\bibitemdeclare{article}{DBLP:journals/tit/Chomsky56}
\bibitem[Cho56]{DBLP:journals/tit/Chomsky56}
\bibinfo{author}{Noam \surnamestart Chomsky\surnameend} (\bibinfo{year}{1956}):
  \emph{\bibinfo{title}{Three models for the description of language}}.
\newblock {\slshape \bibinfo{journal}{{IRE} Trans. Information Theory}}
  \bibinfo{volume}{2}(\bibinfo{number}{3}), pp. \bibinfo{pages}{113--124},
  \doi{10.1109/TIT.1956.1056813}.

\bibitemdeclare{article}{DBLP:journals/jacm/ChandraKS81}
\bibitem[CKS81]{DBLP:journals/jacm/ChandraKS81}
\bibinfo{author}{Ashok~K. \surnamestart Chandra\surnameend},
  \bibinfo{author}{Dexter \surnamestart Kozen\surnameend} \&
  \bibinfo{author}{Larry~J. \surnamestart Stockmeyer\surnameend}
  (\bibinfo{year}{1981}): \emph{\bibinfo{title}{Alternation}}.
\newblock {\slshape \bibinfo{journal}{J. {ACM}}}
  \bibinfo{volume}{28}(\bibinfo{number}{1}), pp. \bibinfo{pages}{114--133},
  \doi{10.1145/322234.322243}.

\bibitemdeclare{article}{DBLP:journals/iandc/Courcelle90}
\bibitem[Cou90]{DBLP:journals/iandc/Courcelle90}
\bibinfo{author}{Bruno \surnamestart Courcelle\surnameend}
  (\bibinfo{year}{1990}): \emph{\bibinfo{title}{The Monadic Second-Order Logic
  of Graphs. I. Recognizable Sets of Finite Graphs}}.
\newblock {\slshape \bibinfo{journal}{Inf. Comput.}}
  \bibinfo{volume}{85}(\bibinfo{number}{1}), pp. \bibinfo{pages}{12--75},
  \doi{10.1016/0890-5401(90)90043-H}.

\bibitemdeclare{incollection}{DiekertM97}
\bibitem[DM97]{DiekertM97}
\bibinfo{author}{Volker \surnamestart Diekert\surnameend} \&
  \bibinfo{author}{Yves \surnamestart M{\'e}tivier\surnameend}
  (\bibinfo{year}{1997}): \emph{\bibinfo{title}{Partial Commutation and
  Traces}}.
\newblock In \bibinfo{editor}{Grzegorz \surnamestart Rozenberg\surnameend} \&
  \bibinfo{editor}{Arto \surnamestart Salomaa\surnameend}, editors: {\slshape
  \bibinfo{booktitle}{Handbook of Formal Languages: Volume 3 Beyond Words}},
  \bibinfo{publisher}{Springer Berlin Heidelberg}, pp.
  \bibinfo{pages}{457--533}, \doi{10.1007/978-3-642-59126-6_8}.

\bibitemdeclare{article}{DBLP:journals/jcss/Doner70}
\bibitem[Don70]{DBLP:journals/jcss/Doner70}
\bibinfo{author}{John \surnamestart Doner\surnameend} (\bibinfo{year}{1970}):
  \emph{\bibinfo{title}{Tree Acceptors and Some of Their Applications}}.
\newblock {\slshape \bibinfo{journal}{J. Comput. Syst. Sci.}}
  \bibinfo{volume}{4}(\bibinfo{number}{5}), pp. \bibinfo{pages}{406--451},
  \doi{10.1016/S0022-0000(70)80041-1}.

\bibitemdeclare{article}{Elgot61}
\bibitem[Elg61]{Elgot61}
\bibinfo{author}{Calvin~C. \surnamestart Elgot\surnameend}
  (\bibinfo{year}{1961}): \emph{\bibinfo{title}{Decision Problems of Finite
  Automata Design and Related Arithmetics}}.
\newblock {\slshape \bibinfo{journal}{Transactions of the American Mathematical
  Society}} \bibinfo{volume}{98}(\bibinfo{number}{1}), pp.
  \bibinfo{pages}{21--51}, \doi{10.2307/1993511}.

\bibitemdeclare{inproceedings}{Fagin74}
\bibitem[Fag74]{Fagin74}
\bibinfo{author}{Ronald \surnamestart Fagin\surnameend} (\bibinfo{year}{1974}):
  \emph{\bibinfo{title}{Generalized First-Order Spectra and Polynomial-Time
  Recognizable Sets}}.
\newblock In \bibinfo{editor}{Richard~M. \surnamestart Karp\surnameend},
  editor: {\slshape \bibinfo{booktitle}{Complexity of Computation}}, {\slshape
  \bibinfo{series}{SIAM-AMS Proceedings}}~\bibinfo{volume}{7}, pp.
  \bibinfo{pages}{43--73}.
\newblock
  \urlprefix\url{http://www.almaden.ibm.com/cs/people/fagin/genspec.pdf}.

\bibitemdeclare{article}{DBLP:journals/mlq/Fagin75}
\bibitem[Fag75]{DBLP:journals/mlq/Fagin75}
\bibinfo{author}{Ronald \surnamestart Fagin\surnameend} (\bibinfo{year}{1975}):
  \emph{\bibinfo{title}{Monadic generalized spectra}}.
\newblock {\slshape \bibinfo{journal}{Math. Log. Q.}}
  \bibinfo{volume}{21}(\bibinfo{number}{1}), pp. \bibinfo{pages}{89--96},
  \doi{10.1002/malq.19750210112}.

\bibitemdeclare{article}{DBLP:journals/eatcs/FeuilloleyF16}
\bibitem[FF16]{DBLP:journals/eatcs/FeuilloleyF16}
\bibinfo{author}{Laurent \surnamestart Feuilloley\surnameend} \&
  \bibinfo{author}{Pierre \surnamestart Fraigniaud\surnameend}
  (\bibinfo{year}{2016}): \emph{\bibinfo{title}{Survey of Distributed
  Decision}}.
\newblock {\slshape \bibinfo{journal}{Bulletin of the {EATCS}}}
  \bibinfo{volume}{119}.
\newblock
  \urlprefix\url{http://bulletin.eatcs.org/index.php/beatcs/article/view/411/391}.

\bibitemdeclare{inproceedings}{DBLP:conf/icalp/FeuilloleyFH16}
\bibitem[FFH16]{DBLP:conf/icalp/FeuilloleyFH16}
\bibinfo{author}{Laurent \surnamestart Feuilloley\surnameend},
  \bibinfo{author}{Pierre \surnamestart Fraigniaud\surnameend} \&
  \bibinfo{author}{Juho \surnamestart Hirvonen\surnameend}
  (\bibinfo{year}{2016}): \emph{\bibinfo{title}{A Hierarchy of Local
  Decision}}.
\newblock In \bibinfo{editor}{Ioannis \surnamestart
  Chatzigiannakis\surnameend}, \bibinfo{editor}{Michael \surnamestart
  Mitzenmacher\surnameend}, \bibinfo{editor}{Yuval \surnamestart
  Rabani\surnameend} \& \bibinfo{editor}{Davide \surnamestart
  Sangiorgi\surnameend}, editors: {\slshape \bibinfo{booktitle}{43rd
  International Colloquium on Automata, Languages, and Programming, {ICALP}
  2016, July 11-15, 2016, Rome, Italy}}, {\slshape
  \bibinfo{series}{LIPIcs}}~\bibinfo{volume}{55}, \bibinfo{publisher}{Schloss
  Dagstuhl - Leibniz-Zentrum fuer Informatik}, pp.
  \bibinfo{pages}{118:1--118:15}, \doi{10.4230/LIPIcs.ICALP.2016.118}.

\bibitemdeclare{book}{DBLP:series/txtcs/GradelKLMSVVW07}
\bibitem[GKL{\etalchar{+}}07]{DBLP:series/txtcs/GradelKLMSVVW07}
\bibinfo{author}{Erich \surnamestart Gr{\"{a}}del\surnameend},
  \bibinfo{author}{Phokion~G. \surnamestart Kolaitis\surnameend},
  \bibinfo{author}{Leonid \surnamestart Libkin\surnameend},
  \bibinfo{author}{Maarten \surnamestart Marx\surnameend},
  \bibinfo{author}{Joel \surnamestart Spencer\surnameend},
  \bibinfo{author}{Moshe~Y. \surnamestart Vardi\surnameend},
  \bibinfo{author}{Yde \surnamestart Venema\surnameend} \&
  \bibinfo{author}{Scott \surnamestart Weinstein\surnameend}
  (\bibinfo{year}{2007}): \emph{\bibinfo{title}{Finite Model Theory and Its
  Applications}}.
\newblock \bibinfo{series}{Texts in Theoretical Computer Science. An {EATCS}
  Series}, \bibinfo{publisher}{Springer}, \doi{10.1007/3-540-68804-8}.

\bibitemdeclare{article}{DBLP:journals/ijprai/GiammarresiR92}
\bibitem[GR92]{DBLP:journals/ijprai/GiammarresiR92}
\bibinfo{author}{Dora \surnamestart Giammarresi\surnameend} \&
  \bibinfo{author}{Antonio \surnamestart Restivo\surnameend}
  (\bibinfo{year}{1992}): \emph{\bibinfo{title}{Recognizable Picture
  Languages}}.
\newblock {\slshape \bibinfo{journal}{{IJPRAI}}}
  \bibinfo{volume}{6}(\bibinfo{number}{2{\&}3}), pp. \bibinfo{pages}{241--256},
  \doi{10.1142/S021800149200014X}.

\bibitemdeclare{article}{DBLP:journals/iandc/GiammarresiRST96}
\bibitem[GRST96]{DBLP:journals/iandc/GiammarresiRST96}
\bibinfo{author}{Dora \surnamestart Giammarresi\surnameend},
  \bibinfo{author}{Antonio \surnamestart Restivo\surnameend},
  \bibinfo{author}{Sebastian \surnamestart Seibert\surnameend} \&
  \bibinfo{author}{Wolfgang \surnamestart Thomas\surnameend}
  (\bibinfo{year}{1996}): \emph{\bibinfo{title}{Monadic Second-Order Logic Over
  Rectangular Pictures and Recognizability by Tiling Systems}}.
\newblock {\slshape \bibinfo{journal}{Inf. Comput.}}
  \bibinfo{volume}{125}(\bibinfo{number}{1}), pp. \bibinfo{pages}{32--45},
  \doi{10.1006/inco.1996.0018}.

\bibitemdeclare{inproceedings}{DBLP:conf/podc/GoosS11}
\bibitem[GS11]{DBLP:conf/podc/GoosS11}
\bibinfo{author}{Mika \surnamestart G{\"{o}}{\"{o}}s\surnameend} \&
  \bibinfo{author}{Jukka \surnamestart Suomela\surnameend}
  (\bibinfo{year}{2011}): \emph{\bibinfo{title}{Locally checkable proofs}}.
\newblock In \bibinfo{editor}{Cyril \surnamestart Gavoille\surnameend} \&
  \bibinfo{editor}{Pierre \surnamestart Fraigniaud\surnameend}, editors:
  {\slshape \bibinfo{booktitle}{Proceedings of the 30th Annual {ACM} Symposium
  on Principles of Distributed Computing, {PODC} 2011, San Jose, CA, USA, June
  6-8, 2011}}, \bibinfo{publisher}{{ACM}}, pp. \bibinfo{pages}{159--168},
  \doi{10.1145/1993806.1993829}.

\bibitemdeclare{article}{DBLP:journals/toc/GoosS16}
\bibitem[GS16]{DBLP:journals/toc/GoosS16}
\bibinfo{author}{Mika \surnamestart G{\"{o}}{\"{o}}s\surnameend} \&
  \bibinfo{author}{Jukka \surnamestart Suomela\surnameend}
  (\bibinfo{year}{2016}): \emph{\bibinfo{title}{Locally Checkable Proofs in
  Distributed Computing}}.
\newblock {\slshape \bibinfo{journal}{Theory of Computing}}
  \bibinfo{volume}{12}(\bibinfo{number}{1}), pp. \bibinfo{pages}{1--33},
  \doi{10.4086/toc.2016.v012a019}.

\bibitemdeclare{inproceedings}{DBLP:conf/podc/HellaJKLLLSV12}
\bibitem[HJK{\etalchar{+}}12]{DBLP:conf/podc/HellaJKLLLSV12}
\bibinfo{author}{Lauri \surnamestart Hella\surnameend}, \bibinfo{author}{Matti
  \surnamestart J{\"{a}}rvisalo\surnameend}, \bibinfo{author}{Antti
  \surnamestart Kuusisto\surnameend}, \bibinfo{author}{Juhana \surnamestart
  Laurinharju\surnameend}, \bibinfo{author}{Tuomo \surnamestart
  Lempi{\"{a}}inen\surnameend}, \bibinfo{author}{Kerkko \surnamestart
  Luosto\surnameend}, \bibinfo{author}{Jukka \surnamestart Suomela\surnameend}
  \& \bibinfo{author}{Jonni \surnamestart Virtema\surnameend}
  (\bibinfo{year}{2012}): \emph{\bibinfo{title}{Weak models of distributed
  computing, with connections to modal logic}}.
\newblock In \bibinfo{editor}{Darek \surnamestart Kowalski\surnameend} \&
  \bibinfo{editor}{Alessandro \surnamestart Panconesi\surnameend}, editors:
  {\slshape \bibinfo{booktitle}{{ACM} Symposium on Principles of Distributed
  Computing, {PODC} '12, Funchal, Madeira, Portugal, July 16-18, 2012}},
  \bibinfo{publisher}{{ACM}}, pp. \bibinfo{pages}{185--194},
  \doi{10.1145/2332432.2332466}.

\bibitemdeclare{article}{DBLP:journals/dc/HellaJKLLLSV15}
\bibitem[HJK{\etalchar{+}}15]{DBLP:journals/dc/HellaJKLLLSV15}
\bibinfo{author}{Lauri \surnamestart Hella\surnameend}, \bibinfo{author}{Matti
  \surnamestart J{\"{a}}rvisalo\surnameend}, \bibinfo{author}{Antti
  \surnamestart Kuusisto\surnameend}, \bibinfo{author}{Juhana \surnamestart
  Laurinharju\surnameend}, \bibinfo{author}{Tuomo \surnamestart
  Lempi{\"{a}}inen\surnameend}, \bibinfo{author}{Kerkko \surnamestart
  Luosto\surnameend}, \bibinfo{author}{Jukka \surnamestart Suomela\surnameend}
  \& \bibinfo{author}{Jonni \surnamestart Virtema\surnameend}
  (\bibinfo{year}{2015}): \emph{\bibinfo{title}{Weak models of distributed
  computing, with connections to modal logic}}.
\newblock {\slshape \bibinfo{journal}{Distributed Computing}}
  \bibinfo{volume}{28}(\bibinfo{number}{1}), pp. \bibinfo{pages}{31--53},
  \doi{10.1007/s00446-013-0202-3}.
\newblock \urlprefix\url{https://arxiv.org/abs/1205.2051}.

\bibitemdeclare{book}{DBLP:books/daglib/Immerman99}
\bibitem[Imm99]{DBLP:books/daglib/Immerman99}
\bibinfo{author}{Neil \surnamestart Immerman\surnameend}
  (\bibinfo{year}{1999}): \emph{\bibinfo{title}{Descriptive complexity}}.
\newblock \bibinfo{series}{Graduate texts in computer science},
  \bibinfo{publisher}{Springer}, \doi{10.1007/978-1-4612-0539-5}.

\bibitemdeclare{article}{DBLP:journals/dc/KormanK07}
\bibitem[KK07]{DBLP:journals/dc/KormanK07}
\bibinfo{author}{Amos \surnamestart Korman\surnameend} \& \bibinfo{author}{Shay
  \surnamestart Kutten\surnameend} (\bibinfo{year}{2007}):
  \emph{\bibinfo{title}{Distributed verification of minimum spanning trees}}.
\newblock {\slshape \bibinfo{journal}{Distributed Computing}}
  \bibinfo{volume}{20}(\bibinfo{number}{4}), pp. \bibinfo{pages}{253--266},
  \doi{10.1007/s00446-007-0025-1}.

\bibitemdeclare{incollection}{Kleene56}
\bibitem[Kle56]{Kleene56}
\bibinfo{author}{Stephen~C. \surnamestart Kleene\surnameend}
  (\bibinfo{year}{1956}): \emph{\bibinfo{title}{Representation of events in
  nerve nets and finite automata}}.
\newblock In \bibinfo{editor}{Claude \surnamestart Shannon\surnameend} \&
  \bibinfo{editor}{John \surnamestart McCarthy\surnameend}, editors: {\slshape
  \bibinfo{booktitle}{Automata studies}}, \bibinfo{series}{Annals of
  mathematics studies, no. 34}, \bibinfo{publisher}{Princeton University Press,
  Princeton, N. J.}, pp. \bibinfo{pages}{3--41}.

\bibitemdeclare{inproceedings}{DBLP:journals/corr/KuusistoR17}
\bibitem[KR17]{DBLP:journals/corr/KuusistoR17}
\bibinfo{author}{Antti \surnamestart Kuusisto\surnameend} \&
  \bibinfo{author}{Fabian \surnamestart Reiter\surnameend}
  (\bibinfo{year}{2017}): \emph{\bibinfo{title}{Emptiness Problems for
  Distributed Automata}}.
\newblock In \bibinfo{editor}{Patricia \surnamestart Bouyer\surnameend},
  \bibinfo{editor}{Andrea \surnamestart Orlandini\surnameend} \&
  \bibinfo{editor}{Pierluigi~San \surnamestart Pietro\surnameend}, editors:
  {\slshape \bibinfo{booktitle}{Proceedings Eighth International Symposium on
  Games, Automata, Logics and Formal Verification, GandALF 2017, Roma, Italy,
  20-22 September 2017.}}, {\slshape \bibinfo{series}{{EPTCS}}}
  \bibinfo{volume}{256}, pp. \bibinfo{pages}{210--222},
  \doi{10.4204/EPTCS.256.15}.

\bibitemdeclare{incollection}{DBLP:series/sci/Kutrib08}
\bibitem[Kut08]{DBLP:series/sci/Kutrib08}
\bibinfo{author}{Martin \surnamestart Kutrib\surnameend}
  (\bibinfo{year}{2008}): \emph{\bibinfo{title}{Cellular Automata - {A}
  Computational Point of View}}.
\newblock In \bibinfo{editor}{Gemma~Bel \surnamestart Enguix\surnameend},
  \bibinfo{editor}{Maria~Dolores \surnamestart
  Jim{\'{e}}nez{-}L{\'{o}}pez\surnameend} \& \bibinfo{editor}{Carlos
  \surnamestart Mart{\'{\i}}n{-}Vide\surnameend}, editors: {\slshape
  \bibinfo{booktitle}{New Developments in Formal Languages and Applications}},
  {\slshape \bibinfo{series}{Studies in Computational Intelligence}}
  \bibinfo{volume}{113}, \bibinfo{publisher}{Springer}, pp.
  \bibinfo{pages}{183--227}, \doi{10.1007/978-3-540-78291-9_6}.

\bibitemdeclare{inproceedings}{DBLP:conf/aiml/Kuusisto08}
\bibitem[Kuu08]{DBLP:conf/aiml/Kuusisto08}
\bibinfo{author}{Antti \surnamestart Kuusisto\surnameend}
  (\bibinfo{year}{2008}): \emph{\bibinfo{title}{A modal perspective on monadic
  second-order alternation hierarchies}}.
\newblock In \bibinfo{editor}{Carlos \surnamestart Areces\surnameend} \&
  \bibinfo{editor}{Robert \surnamestart Goldblatt\surnameend}, editors:
  {\slshape \bibinfo{booktitle}{Advances in Modal Logic 7, papers from the
  seventh conference on "Advances in Modal Logic," 9-12 September 2008, Nancy,
  France}}, \bibinfo{publisher}{College Publications}, pp.
  \bibinfo{pages}{231--247}.
\newblock \urlprefix\url{http://www.aiml.net/volumes/volume7/Kuusisto.pdf}.

\bibitemdeclare{inproceedings}{DBLP:conf/csl/Kuusisto13}
\bibitem[Kuu13a]{DBLP:conf/csl/Kuusisto13}
\bibinfo{author}{Antti \surnamestart Kuusisto\surnameend}
  (\bibinfo{year}{2013}): \emph{\bibinfo{title}{Modal Logic and Distributed
  Message Passing Automata}}.
\newblock In \bibinfo{editor}{Simona Ronchi~Della \surnamestart
  Rocca\surnameend}, editor: {\slshape \bibinfo{booktitle}{Computer Science
  Logic 2013 {(CSL} 2013), {CSL} 2013, September 2-5, 2013, Torino, Italy}},
  {\slshape \bibinfo{series}{LIPIcs}}~\bibinfo{volume}{23},
  \bibinfo{publisher}{Schloss Dagstuhl - Leibniz-Zentrum fuer Informatik}, pp.
  \bibinfo{pages}{452--468}, \doi{10.4230/LIPIcs.CSL.2013.452}.

\bibitemdeclare{article}{Kuusisto13}
\bibitem[Kuu13b]{Kuusisto13}
\bibinfo{author}{Antti \surnamestart Kuusisto\surnameend}
  (\bibinfo{year}{2013}): \emph{\bibinfo{title}{Modal Logics and
  Definability}}.
\newblock {\slshape \bibinfo{journal}{TamPub, Institutional repository of the
  University of Tampere}}, pp. \bibinfo{pages}{1--6}.
\newblock \urlprefix\url{http://urn.fi/URN:ISBN:978-951-44-9157-3}.

\bibitemdeclare{inproceedings}{DBLP:journals/corr/Kuusisto14a}
\bibitem[Kuu14]{DBLP:journals/corr/Kuusisto14a}
\bibinfo{author}{Antti \surnamestart Kuusisto\surnameend}
  (\bibinfo{year}{2014}): \emph{\bibinfo{title}{Infinite Networks, Halting and
  Local Algorithms}}.
\newblock In \bibinfo{editor}{Adriano \surnamestart Peron\surnameend} \&
  \bibinfo{editor}{Carla \surnamestart Piazza\surnameend}, editors: {\slshape
  \bibinfo{booktitle}{Proceedings Fifth International Symposium on Games,
  Automata, Logics and Formal Verification, GandALF 2014, Verona, Italy,
  September 10-12, 2014.}}, {\slshape \bibinfo{series}{{EPTCS}}}
  \bibinfo{volume}{161}, pp. \bibinfo{pages}{147--160},
  \doi{10.4204/EPTCS.161.14}.

\bibitemdeclare{article}{DBLP:journals/apal/Kuusisto15}
\bibitem[Kuu15]{DBLP:journals/apal/Kuusisto15}
\bibinfo{author}{Antti \surnamestart Kuusisto\surnameend}
  (\bibinfo{year}{2015}): \emph{\bibinfo{title}{Second-order propositional
  modal logic and monadic alternation hierarchies}}.
\newblock {\slshape \bibinfo{journal}{Ann. Pure Appl. Logic}}
  \bibinfo{volume}{166}(\bibinfo{number}{1}), pp. \bibinfo{pages}{1--28},
  \doi{10.1016/j.apal.2014.08.003}.

\bibitemdeclare{article}{Lenzi05}
\bibitem[Len05]{Lenzi05}
\bibinfo{author}{Giacomo \surnamestart Lenzi\surnameend}
  (\bibinfo{year}{2005}): \emph{\bibinfo{title}{The Modal {\(\mu\)}-Calculus: a
  Survey}}.
\newblock {\slshape \bibinfo{journal}{TASK Quarterly -- Scientific Bulletin of
  the Academic Computer Centre in Gdansk}}
  \bibinfo{volume}{9}(\bibinfo{number}{3}), pp. \bibinfo{pages}{293--316}.
\newblock \urlprefix\url{http://task.gda.pl/quart/05-3.html}.

\bibitemdeclare{book}{DBLP:books/sp/Libkin04}
\bibitem[Lib04]{DBLP:books/sp/Libkin04}
\bibinfo{author}{Leonid \surnamestart Libkin\surnameend}
  (\bibinfo{year}{2004}): \emph{\bibinfo{title}{Elements of Finite Model
  Theory}}.
\newblock \bibinfo{series}{Texts in Theoretical Computer Science. An {EATCS}
  Series}, \bibinfo{publisher}{Springer}, \doi{10.1007/978-3-662-07003-1}.

\bibitemdeclare{incollection}{DBLP:books/ws/automata2012/Loding12}
\bibitem[L{\"{o}}d12]{DBLP:books/ws/automata2012/Loding12}
\bibinfo{author}{Christof \surnamestart L{\"{o}}ding\surnameend}
  (\bibinfo{year}{2012}): \emph{\bibinfo{title}{Basics on Tree Automata}}.
\newblock In \bibinfo{editor}{Deepak \surnamestart D'Souza\surnameend} \&
  \bibinfo{editor}{Priti \surnamestart Shankar\surnameend}, editors: {\slshape
  \bibinfo{booktitle}{Modern Applications of Automata Theory}}, {\slshape
  \bibinfo{series}{IISc Research Monographs Series}}~\bibinfo{volume}{2},
  \bibinfo{publisher}{World Scientific}, pp. \bibinfo{pages}{79--109},
  \doi{10.1142/9789814271059_0003}.

\bibitemdeclare{inproceedings}{DBLP:conf/ifipTCS/LodingT00}
\bibitem[LT00]{DBLP:conf/ifipTCS/LodingT00}
\bibinfo{author}{Christof \surnamestart L{\"{o}}ding\surnameend} \&
  \bibinfo{author}{Wolfgang \surnamestart Thomas\surnameend}
  (\bibinfo{year}{2000}): \emph{\bibinfo{title}{Alternating Automata and Logics
  over Infinite Words}}.
\newblock In \bibinfo{editor}{Jan \surnamestart van Leeuwen\surnameend},
  \bibinfo{editor}{Osamu \surnamestart Watanabe\surnameend},
  \bibinfo{editor}{Masami \surnamestart Hagiya\surnameend},
  \bibinfo{editor}{Peter~D. \surnamestart Mosses\surnameend} \&
  \bibinfo{editor}{Takayasu \surnamestart Ito\surnameend}, editors: {\slshape
  \bibinfo{booktitle}{Theoretical Computer Science, Exploring New Frontiers of
  Theoretical Informatics, International Conference {IFIP} {TCS} 2000, Sendai,
  Japan, August 17-19, 2000, Proceedings}}, {\slshape \bibinfo{series}{Lecture
  Notes in Computer Science}} \bibinfo{volume}{1872},
  \bibinfo{publisher}{Springer}, pp. \bibinfo{pages}{521--535},
  \doi{10.1007/3-540-44929-9_36}.

\bibitemdeclare{book}{DBLP:books/mk/Lynch96}
\bibitem[Lyn96]{DBLP:books/mk/Lynch96}
\bibinfo{author}{Nancy~A. \surnamestart Lynch\surnameend}
  (\bibinfo{year}{1996}): \emph{\bibinfo{title}{Distributed Algorithms}}.
\newblock \bibinfo{publisher}{Morgan Kaufmann}.

\bibitemdeclare{article}{DBLP:journals/tcs/Matz02}
\bibitem[Mat02]{DBLP:journals/tcs/Matz02}
\bibinfo{author}{Oliver \surnamestart Matz\surnameend} (\bibinfo{year}{2002}):
  \emph{\bibinfo{title}{Dot-depth, monadic quantifier alternation, and
  first-order closure over grids and pictures}}.
\newblock {\slshape \bibinfo{journal}{Theor. Comput. Sci.}}
  \bibinfo{volume}{270}(\bibinfo{number}{1-2}), pp. \bibinfo{pages}{1--70},
  \doi{10.1016/S0304-3975(01)00277-8}.

\bibitemdeclare{article}{DBLP:journals/iandc/MatzST02}
\bibitem[MST02]{DBLP:journals/iandc/MatzST02}
\bibinfo{author}{Oliver \surnamestart Matz\surnameend}, \bibinfo{author}{Nicole
  \surnamestart Schweikardt\surnameend} \& \bibinfo{author}{Wolfgang
  \surnamestart Thomas\surnameend} (\bibinfo{year}{2002}):
  \emph{\bibinfo{title}{The Monadic Quantifier Alternation Hierarchy over Grids
  and Graphs}}.
\newblock {\slshape \bibinfo{journal}{Inf. Comput.}}
  \bibinfo{volume}{179}(\bibinfo{number}{2}), pp. \bibinfo{pages}{356--383},
  \doi{10.1006/inco.2002.2955}.

\bibitemdeclare{inproceedings}{DBLP:conf/lics/MatzT97}
\bibitem[MT97]{DBLP:conf/lics/MatzT97}
\bibinfo{author}{Oliver \surnamestart Matz\surnameend} \&
  \bibinfo{author}{Wolfgang \surnamestart Thomas\surnameend}
  (\bibinfo{year}{1997}): \emph{\bibinfo{title}{The Monadic Quantifier
  Alternation Hierarchy over Graphs is Infinite}}.
\newblock In: {\slshape \bibinfo{booktitle}{Proceedings, 12th Annual {IEEE}
  Symposium on Logic in Computer Science, Warsaw, Poland, June 29 - July 2,
  1997}}, \bibinfo{publisher}{{IEEE} Computer Society}, pp.
  \bibinfo{pages}{236--244}, \doi{10.1109/LICS.1997.614951}.

\bibitemdeclare{article}{Nerode58}
\bibitem[Ner58]{Nerode58}
\bibinfo{author}{Anil \surnamestart Nerode\surnameend} (\bibinfo{year}{1958}):
  \emph{\bibinfo{title}{Linear automaton transformations}}.
\newblock {\slshape \bibinfo{journal}{Proc. Amer. Math. Soc.}}
  \bibinfo{volume}{9}, pp. \bibinfo{pages}{541--544}, \doi{10.2307/2033204}.

\bibitemdeclare{book}{Peleg00}
\bibitem[Pel00]{Peleg00}
\bibinfo{author}{David \surnamestart Peleg\surnameend} (\bibinfo{year}{2000}):
  \emph{\bibinfo{title}{Distributed Computing: A Locality-Sensitive Approach}}.
\newblock {\slshape \bibinfo{series}{SIAM Monographs on Discrete Mathematics
  and Applications}}~\bibinfo{volume}{5}, \bibinfo{publisher}{Society for
  Industrial and Applied Mathematics (SIAM)}, \doi{10.1137/1.9780898719772}.

\bibitemdeclare{inproceedings}{DBLP:conf/lics/Reiter15}
\bibitem[Rei15]{DBLP:conf/lics/Reiter15}
\bibinfo{author}{Fabian \surnamestart Reiter\surnameend}
  (\bibinfo{year}{2015}): \emph{\bibinfo{title}{Distributed Graph Automata}}.
\newblock In: {\slshape \bibinfo{booktitle}{30th Annual {ACM/IEEE} Symposium on
  Logic in Computer Science, {LICS} 2015, Kyoto, Japan, July 6-10, 2015}},
  \bibinfo{publisher}{{IEEE} Computer Society}, pp. \bibinfo{pages}{192--201},
  \doi{10.1109/LICS.2015.27}.
\newblock \urlprefix\url{https://arxiv.org/abs/1408.3030}.

\bibitemdeclare{article}{DBLP:journals/corr/Reiter16}
\bibitem[Rei16]{DBLP:journals/corr/Reiter16}
\bibinfo{author}{Fabian \surnamestart Reiter\surnameend}
  (\bibinfo{year}{2016}): \emph{\bibinfo{title}{Alternating Set Quantifiers in
  Modal Logic}}.
\newblock {\slshape \bibinfo{journal}{CoRR}} \bibinfo{volume}{abs/1602.08971}.
\newblock \urlprefix\url{http://arxiv.org/abs/1602.08971}.

\bibitemdeclare{inproceedings}{DBLP:conf/icalp/Reiter17}
\bibitem[Rei17]{DBLP:conf/icalp/Reiter17}
\bibinfo{author}{Fabian \surnamestart Reiter\surnameend}
  (\bibinfo{year}{2017}): \emph{\bibinfo{title}{Asynchronous Distributed
  Automata: {A} Characterization of the Modal Mu-Fragment}}.
\newblock In \bibinfo{editor}{Ioannis \surnamestart
  Chatzigiannakis\surnameend}, \bibinfo{editor}{Piotr \surnamestart
  Indyk\surnameend}, \bibinfo{editor}{Fabian \surnamestart Kuhn\surnameend} \&
  \bibinfo{editor}{Anca \surnamestart Muscholl\surnameend}, editors: {\slshape
  \bibinfo{booktitle}{44th International Colloquium on Automata, Languages, and
  Programming, {ICALP} 2017, July 10-14, 2017, Warsaw, Poland}}, {\slshape
  \bibinfo{series}{LIPIcs}}~\bibinfo{volume}{80}, \bibinfo{publisher}{Schloss
  Dagstuhl - Leibniz-Zentrum fuer Informatik}, pp.
  \bibinfo{pages}{100:1--100:14}, \doi{10.4230/LIPIcs.ICALP.2017.100}.
\newblock \urlprefix\url{http://arxiv.org/abs/1611.08554}.

\bibitemdeclare{article}{DBLP:journals/dmtcs/SchwentickB99}
\bibitem[SB99]{DBLP:journals/dmtcs/SchwentickB99}
\bibinfo{author}{Thomas \surnamestart Schwentick\surnameend} \&
  \bibinfo{author}{Klaus \surnamestart Barthelmann\surnameend}
  (\bibinfo{year}{1999}): \emph{\bibinfo{title}{Local Normal Forms for
  First-Order Logic with Applications to Games and Automata}}.
\newblock {\slshape \bibinfo{journal}{Discrete Mathematics {\&} Theoretical
  Computer Science}} \bibinfo{volume}{3}(\bibinfo{number}{3}), pp.
  \bibinfo{pages}{109--124}.
\newblock \urlprefix\url{http://dmtcs.episciences.org/254}.

\bibitemdeclare{inproceedings}{DBLP:conf/csl/Schweikardt97}
\bibitem[Sch97]{DBLP:conf/csl/Schweikardt97}
\bibinfo{author}{Nicole \surnamestart Schweikardt\surnameend}
  (\bibinfo{year}{1997}): \emph{\bibinfo{title}{The Monadic Quantifier
  Alternation Hierarchy over Grids and Pictures}}.
\newblock In \bibinfo{editor}{Mogens \surnamestart Nielsen\surnameend} \&
  \bibinfo{editor}{Wolfgang \surnamestart Thomas\surnameend}, editors:
  {\slshape \bibinfo{booktitle}{Computer Science Logic, 11th International
  Workshop, {CSL} '97, Annual Conference of the EACSL, Aarhus, Denmark, August
  23-29, 1997, Selected Papers}}, {\slshape \bibinfo{series}{Lecture Notes in
  Computer Science}} \bibinfo{volume}{1414}, \bibinfo{publisher}{Springer}, pp.
  \bibinfo{pages}{441--460}, \doi{10.1007/BFb0028030}.

\bibitemdeclare{article}{DBLP:journals/csur/Suomela13}
\bibitem[Suo13]{DBLP:journals/csur/Suomela13}
\bibinfo{author}{Jukka \surnamestart Suomela\surnameend}
  (\bibinfo{year}{2013}): \emph{\bibinfo{title}{Survey of local algorithms}}.
\newblock {\slshape \bibinfo{journal}{{ACM} Comput. Surv.}}
  \bibinfo{volume}{45}(\bibinfo{number}{2}), pp. \bibinfo{pages}{24:1--24:40},
  \doi{10.1145/2431211.2431223}.

\bibitemdeclare{incollection}{DBLP:reference/nc/Terrier12}
\bibitem[Ter12]{DBLP:reference/nc/Terrier12}
\bibinfo{author}{V{\'{e}}ronique \surnamestart Terrier\surnameend}
  (\bibinfo{year}{2012}): \emph{\bibinfo{title}{Language Recognition by
  Cellular Automata}}.
\newblock In \bibinfo{editor}{Grzegorz \surnamestart Rozenberg\surnameend},
  \bibinfo{editor}{Thomas \surnamestart B{\"{a}}ck\surnameend} \&
  \bibinfo{editor}{Joost~N. \surnamestart Kok\surnameend}, editors: {\slshape
  \bibinfo{booktitle}{Handbook of Natural Computing}},
  \bibinfo{publisher}{Springer}, pp. \bibinfo{pages}{123--158},
  \doi{10.1007/978-3-540-92910-9_4}.

\bibitemdeclare{inproceedings}{DBLP:conf/icalp/Thomas91}
\bibitem[Tho91]{DBLP:conf/icalp/Thomas91}
\bibinfo{author}{Wolfgang \surnamestart Thomas\surnameend}
  (\bibinfo{year}{1991}): \emph{\bibinfo{title}{On Logics, Tilings, and
  Automata}}.
\newblock In \bibinfo{editor}{Javier~Leach \surnamestart Albert\surnameend},
  \bibinfo{editor}{Burkhard \surnamestart Monien\surnameend} \&
  \bibinfo{editor}{Mario \surnamestart Rodr{\'{\i}}guez{-}Artalejo\surnameend},
  editors: {\slshape \bibinfo{booktitle}{Automata, Languages and Programming,
  18th International Colloquium, ICALP91, Madrid, Spain, July 8-12, 1991,
  Proceedings}}, {\slshape \bibinfo{series}{Lecture Notes in Computer Science}}
  \bibinfo{volume}{510}, \bibinfo{publisher}{Springer}, pp.
  \bibinfo{pages}{441--454}, \doi{10.1007/3-540-54233-7_154}.

\bibitemdeclare{inproceedings}{DBLP:conf/dimacs/Thomas96}
\bibitem[Tho96]{DBLP:conf/dimacs/Thomas96}
\bibinfo{author}{Wolfgang \surnamestart Thomas\surnameend}
  (\bibinfo{year}{1996}): \emph{\bibinfo{title}{Elements of an automata theory
  over partial orders}}.
\newblock In \bibinfo{editor}{Doron~A. \surnamestart Peled\surnameend},
  \bibinfo{editor}{Vaughan~R. \surnamestart Pratt\surnameend} \&
  \bibinfo{editor}{Gerard~J. \surnamestart Holzmann\surnameend}, editors:
  {\slshape \bibinfo{booktitle}{Partial Order Methods in Verification,
  Proceedings of a {DIMACS} Workshop, Princeton, New Jersey, USA, July 24-26,
  1996}}, {\slshape \bibinfo{series}{{DIMACS} Series in Discrete Mathematics
  and Theoretical Computer Science}}~\bibinfo{volume}{29},
  \bibinfo{publisher}{{DIMACS/AMS}}, pp. \bibinfo{pages}{25--40}.

\bibitemdeclare{inproceedings}{DBLP:conf/tapsoft/Thomas97}
\bibitem[Tho97a]{DBLP:conf/tapsoft/Thomas97}
\bibinfo{author}{Wolfgang \surnamestart Thomas\surnameend}
  (\bibinfo{year}{1997}): \emph{\bibinfo{title}{Automata Theory on Trees and
  Partial Orders}}.
\newblock In \bibinfo{editor}{Michel \surnamestart Bidoit\surnameend} \&
  \bibinfo{editor}{Max \surnamestart Dauchet\surnameend}, editors: {\slshape
  \bibinfo{booktitle}{TAPSOFT'97: Theory and Practice of Software Development,
  7th International Joint Conference CAAP/FASE, Lille, France, April 14-18,
  1997, Proceedings}}, {\slshape \bibinfo{series}{Lecture Notes in Computer
  Science}} \bibinfo{volume}{1214}, \bibinfo{publisher}{Springer}, pp.
  \bibinfo{pages}{20--38}, \doi{10.1007/BFb0030586}.

\bibitemdeclare{incollection}{Thomas97b}
\bibitem[Tho97b]{Thomas97b}
\bibinfo{author}{Wolfgang \surnamestart Thomas\surnameend}
  (\bibinfo{year}{1997}): \emph{\bibinfo{title}{Languages, Automata, and
  Logic}}.
\newblock In \bibinfo{editor}{Grzegorz \surnamestart Rozenberg\surnameend} \&
  \bibinfo{editor}{Arto \surnamestart Salomaa\surnameend}, editors: {\slshape
  \bibinfo{booktitle}{Handbook of Formal Languages: Volume 3 Beyond Words}},
  \bibinfo{publisher}{Springer, Berlin}, pp. \bibinfo{pages}{389--455},
  \doi{10.1007/978-3-642-59126-6_7}.

\bibitemdeclare{article}{Trakhtenbrot61}
\bibitem[Tra61]{Trakhtenbrot61}
\bibinfo{author}{Boris~A. \surnamestart Trakhtenbrot\surnameend}
  (\bibinfo{year}{1961}): \emph{\bibinfo{title}{Finite automata and the logic
  of single-place predicates}}.
\newblock {\slshape \bibinfo{journal}{Soviet Physics Dokl.}}
  \bibinfo{volume}{6}, pp. \bibinfo{pages}{753--755}.

\bibitemdeclare{article}{DBLP:journals/mst/ThatcherW68}
\bibitem[TW68]{DBLP:journals/mst/ThatcherW68}
\bibinfo{author}{James~W. \surnamestart Thatcher\surnameend} \&
  \bibinfo{author}{Jesse~B. \surnamestart Wright\surnameend}
  (\bibinfo{year}{1968}): \emph{\bibinfo{title}{Generalized Finite Automata
  Theory with an Application to a Decision Problem of Second-Order Logic}}.
\newblock {\slshape \bibinfo{journal}{Mathematical Systems Theory}}
  \bibinfo{volume}{2}(\bibinfo{number}{1}), pp. \bibinfo{pages}{57--81},
  \doi{10.1007/BF01691346}.

\end{thebibliography}
